\documentclass[a4paper,twocolumn]{revtex4}
\usepackage[latin1]{inputenc}
\usepackage{amsmath}
\usepackage{amsfonts}
\usepackage{amssymb}
\usepackage{amsthm}
\usepackage{graphicx}

\newtheorem{Proposition}{Proposition}
\newtheorem{Lemma}{Lemma}
\newtheorem{Corollary}{Corollary}
\newtheorem{Definition}{Definition}
\newtheorem{Assumptions}{Assumptions}

\newcommand{\Real}{{\mathrm{Re}}}
\newcommand{\Imag}{{\mathrm{Im}}}
\newcommand{\Tr}{{\mathrm{Tr}}} 
\newcommand{\Sp}{{\mathrm{Sp}}}

\begin{document}
\title{Fully quantum fluctuation theorems}
\author{Johan {\AA}berg}
\email{johan.aberg@uni-koeln.de}
\affiliation{Institute for Theoretical Physics, University of Cologne, Z\"ulpicher Strasse 77, D-50937, Cologne, Germany}

\begin{abstract}
Systems that are driven out of thermal equilibrium typically dissipate random quantities of energy on microscopic scales. Crooks fluctuation theorem relates the distribution of these random work costs with the corresponding distribution for the reverse process. By an analysis that explicitly incorporates the energy reservoir that donates the energy, and the control system that implements the dynamic, we here obtain a quantum generalization of Crooks theorem that not only includes  the energy changes in the reservoir, but the full description of its evolution, including coherences. This approach moreover opens up for generalizations of the concept of fluctuation relations. Here we introduce `conditional' fluctuation relations that are applicable to non-equilibrium systems, as well as approximate fluctuation relations that allow for the analysis of autonomous evolution  generated by global time-independent Hamiltonians. We furthermore extend these notions to Markovian master equations, implicitly modeling the influence of the heat bath.
\end{abstract}

\maketitle

\section{Introduction}

Imagine a physical system with a Hamiltonian  $H_S(x)$ that depends on some external parameter $x$, e.g., electric or magnetic fields that we can vary at will. By changing $x$ we can push the system out of thermal equilibrium. This would typically require work that may be dissipated due to interactions with the surrounding heat bath. The latter may also make the dissipation random, in the sense that the work cost $w$ is different each time we implement the same change of the Hamiltonian   \cite{Bustamante05,Rondoni07,ReviewJarzynski,ReviewSeifert,ReviewMarconi,ReviewFluctThm}. (Think of a spoon pushed through syrup. On microscopic scales the friction resolves into random molecular collisions.) These fluctuations can be described via a probability distribution $P_{+}(w)$. Crooks theorem \cite{CrooksTheorem} shows that there is a surprisingly simple relation between $P_{+}$ and the corresponding distribution  $P_{-}$ obtained for the process run in reverse (i.e., where the time-schedule for the change of $x$ is mirrored) namely
\begin{equation}
\label{StandardCrooks}
Z(H^i_S)P_{+}(w) =  Z(H^f_S)e^{\beta w}P_{-}(-w),
\end{equation}
where $H^i_S$ and $H^f_S$ are the initial and final Hamiltonians, respectively, and 
  $Z(H)$ is the partition function, which in the quantum case takes the form $Z(H) = \Tr e^{-\beta H}$. Moreover, $\beta = 1/(kT)$, with $k$ Boltzmann's constant, and where $T$ is the absolute temperature of the heat bath. In this investigation we only consider a single heat bath with a fixed temperature $T$. 

Crooks theorem was originally derived in a classical setting \cite{CrooksTheorem} (for reviews on classical fluctuation theorems, see \cite{Bustamante05,Rondoni07,ReviewJarzynski,ReviewSeifert,ReviewMarconi,ReviewFluctThm}) but there has accumulated a considerable body of quantum fluctuation relations (for reviews see \cite{Esposito09,Campisi11a,Hanggi15}). 
However, the latter often include measurements, and in particular energy measurements (see e.g.~\cite{Tasaki00,Kurchan01,Mukamel03b,Jarzynski04,Talkner05,Saito08,Quan08,Talkner08,Talkner09,Andrieux09,Campisi10,Campisi10b,Deffner11b,Campisi11b,Talkner13}) that typically destroy coherences and quantum correlations. Here we avoid such auxiliary components, and obtain  fluctuation relations that retain all quantum aspects of the evolution.

The key is to explicitly model all degrees of freedom involved in the process.  This includes the control mechanism that implements the change of the external parameter, i.e., $x$ in $H_S(x)$, as well as the `energy reservoir', e.g., a battery or an excited atom, which donates the energy for the work $w$.
Since work corresponds to a change of energy in this reservoir, it follows that the reservoir  itself is forced to evolve as it fuels the dynamics. The general theme of this investigation is to formulate fluctuation relations in terms of this induced evolution. To make this more concrete, let us briefly display the first of the fluctuation theorems that we will derive. The roles of the probability distributions $P_{+}$ and $P_{-}$ are here taken over by  the channels $\mathcal{F}_{+}$ and  $\mathcal{F}_{-}$ that are induced on the energy reservoir by the forward and reverse processes, respectively. As we shall see, these channels can, under suitable conditions, be related by the following quantum Crooks relation 
\begin{equation}
\label{QuantumCrooks}
Z(H^i_{S}) \mathcal{F}_{+} = Z(H^f_{S})\mathcal{J}_{\beta H_E}\mathcal{F}_{-}^{\ominus}\mathcal{J}_{\beta H_E}^{-1},
\end{equation}
where $H_E$ denotes the Hamiltonian of the energy reservoir $E$.
The combined application of $\mathcal{J}_{\beta H_E}$ and $\mathcal{J}_{\beta H_E}^{-1} $, where $\mathcal{J}_{\beta H_E}(Q) = e^{-\beta H_E/2}Q e^{-\beta H_E/2}$, can be viewed as a  counterpart to the term $e^{\beta w}$ in (\ref{StandardCrooks}).  The mapping from $\mathcal{F}_{-}$ to $\mathcal{F}^{\ominus}_{-}$  (to be described in detail later) is related to time-reversals, and can in some sense be regarded as a generalization of the transformation of $P_{-}(w)$ to $P_{-}(-w)$ at the right hand side of (\ref{StandardCrooks}). Note that the right hand side of (\ref{QuantumCrooks}) should be interpreted in the sense of compositions of functions, i.e., (\ref{QuantumCrooks}) can be written more explicitly as 
$Z(H^i_{S}) \mathcal{F}_{+}(\sigma) = Z(H^f_{S})\mathcal{J}_{\beta H_E}\Big(\mathcal{F}_{-}^{\ominus}\big(\mathcal{J}_{\beta H_E}^{-1}(\sigma)\big)\Big)$, with $\sigma$ being an arbitrary operator on the energy reservoir.

To get an alternative perspective on fully quantum fluctuation theorems and their relation to the second law, the reader is encouraged to consult \cite{Alhambra16}.  While we here primarily consider generalizations of Crooks theorem, \cite{Alhambra16} mainly focuses on equalities, including generalizations of the Jarzynski equality. However, our paths do  occasionally cross, for example in section \ref{SecMainEnergytransl} where we share the focus on energy translation invariance, and section  \ref{SecMainFlctnMstrEq} where thermal operations play an important role.

In the following section we derive (\ref{QuantumCrooks}) and moreover show that it can be decomposed into diagonal and off-diagonal Crooks relations, thus yielding fluctuation relations for  coherences. We also derive  quantum Jarzynski equalities, as well as   bounds on the work cost. In section \ref{SecMainEnergytransl} we regain the classical Crooks relation (\ref{StandardCrooks}) from (\ref{QuantumCrooks}) via the additional assumption of energy translation invariance.  We next turn to  generalizations of  (\ref{QuantumCrooks}), where section \ref{MainConditional} introduces conditional fluctuation relations, and section \ref{MainApproximate} approximate versions. In section \ref{SecMainFlctnMstrEq} we formulate fluctuation theorems for master equations.

\section{\label{MainAFluctThm}A quantum fluctuation theorem}

\subsection{\label{MainTheModel}The model}

To derive the quantum Crooks relation in (\ref{QuantumCrooks}) we employ a general class of models that previously has been used in the context of quantum thermodynamics to analyze, e.g., work extraction, information erasure, and coherence  \cite{Janzing00,Horodecki11,Brandao13,Brandao13b,Skrzypczyk13,Aberg13,Ng14,Lostaglio14a,Lostaglio14b,Lostaglio15c,Cwiklinski15,Narasimhachar15,Korzekwa15,Masanes15,Wehner15,Gallego15,Woods15,Alhambra15}. The main idea is that we include all the relevant degrees of freedom (which in our case consist of four subsystems, see Fig.~\ref{FigSystems}) and assign a global time-independent Hamiltonian $H$ to account for energy. On this joint system  we are allowed to act with any unitary operation $V$ that conserves energy, which is formalized by the condition that $V$ commutes with the total Hamiltonian $[H,V] = 0$. (See \cite{Skrzypczyk14} for an alternative notion of energy conservation.) 

At first sight it may be difficult to see how such a manifestly time-independent Hamiltonian can be used to describe the evolving Hamiltonians in Crooks theorem.  
 For this purpose we introduce a control system $C$ such that the Hamiltonian of $S' = SB$ depends on the state of $C$. As an illustration, suppose that we wish to describe a transition from an initial Hamiltonian $H^{i}_{S'}$ to a final Hamiltonian $H^{f}_{S'}$. One possibility would be to define a joint Hamiltonian of the form $H_{S'C} = H^{f}_{S'}\otimes |c_f\rangle\langle c_f| + H^{i}_{S'}\otimes |c_i\rangle\langle c_i|$, where $|c_f\rangle$ and $|c_i\rangle$ are two normalized and orthogonal states of the control system $C$. If the evolution would change the control from $|c_i\rangle$ to $|c_f\rangle$, then this would effectively change the Hamiltonian of $S'$ from  $H^i_{S'}$ to $H^{f}_{S'}$. However, this change can in general not be achieved by an energy conserving unitary operation on $S'C$ alone, since the transition from $H^{i}_{S'}$ to $H^{f}_{S'}$ typically will involve a change of energy. The role of the energy reservoir $E$ is to make these transitions possible by  donating or absorbing the necessary energy. With a suitable choice of Hamiltonian $H_E$ of the energy reservoir, the global Hamiltonian $H = H_{S'C}\otimes \hat{1}_E + \hat{1}_{S'C}\otimes H_E$ allows for non-trivial energy conserving unitary operators $V$. 

\begin{figure}[h]
 \includegraphics[width= 8cm]{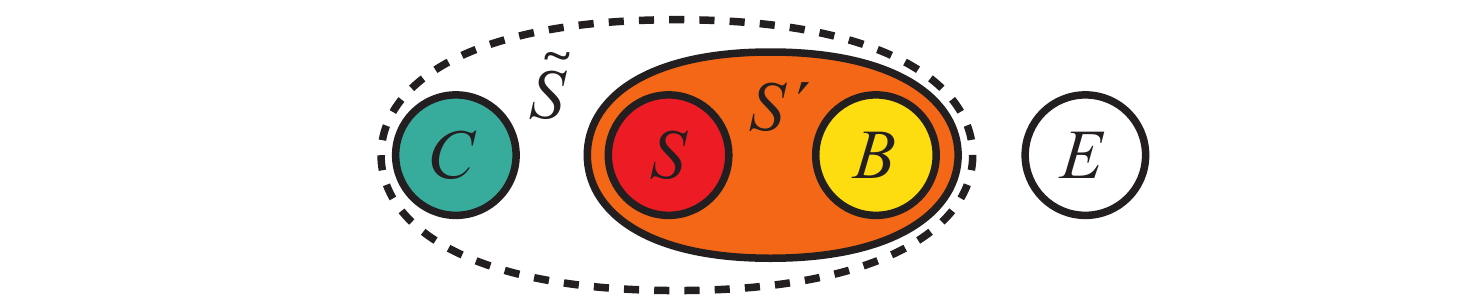} 
   \caption{\label{FigSystems}  {\bf The systems.} 
The model explicitly includes the `system' $S$ on which we operate, the heat bath $B$, the control $C$ that implements the change of the Hamiltonian on $SB$, and the energy reservoir $E$ that donates or accepts the energy required to drive the processes. Most of our fluctuation relations are expressed in terms of the dynamics induced on the energy reservoir $E$. \\
For the main part of this investigation (Secs.~\ref{MainAFluctThm} and \ref{MainConditional}) we assume that the total system is described by a time-independent Hamiltonian that is non-interacting between $SBC$ and $E$, i.e., is of the form $H = H_{SBC}\otimes \hat{1}_E + \hat{1}_{SBC}\otimes H_E$. We moreover model the evolution by energy conserving unitary operators $V$, i.e., $[H,V] = 0$. (We go beyond these assumptions when we consider approximate fluctuation relations in section \ref{MainApproximate}.)\\
For most of the derivations it is not necessary to make any distinction between $S$ and $B$, and for this reason we will often bundle them together into an extended system $S' = SB$. This `rationalization' can be taken one step further to  $\tilde{S} = S'C = SBC$ when we turn to the conditional fluctuation theorems in section \ref{MainConditional}.  
 }
\end{figure}

In this investigation the energy reservoir does not only serve as a source or sink for energy, but also acts as a probe of the dynamics of the system. Instead of using auxiliary measurements, which inevitably would introduce additional interactions that potentially may disturb the dynamics, we here let the energy reservoir take over this role. We do, so to speak, make an additional use of a component that anyway has to be included.  A further benefit of using an explicit quantum probe is that we not only capture the flow of energy, but also the change in coherence and  
correlations. Needless to say, standard macroscopic energy reservoirs would be of little use for the latter purposes, due to the levels of decoherence that they are exposed to. However, our notion of energy reservoirs  includes nano-systems like, e.g., single atoms or spins.

\subsection{\label{MainSecInterm}An intermediate version}

In this section we derive a simplified version of our fluctuation theorem. This serves as a convenient intermediate step, and illustrates why time-reversal symmetry is useful.

We wish to determine the channel induced on the energy reservoir by a non-equilibrium process {\it a la} Crooks, where the system initially is in equilibrium, and is forced out of it. To model this we let the control system initially be in state $|c_i\rangle$, and we let the combined system and heat bath $S'$ be in the Gibbs state $G(H^{i}_{S'})$, where $G(H) = e^{-\beta H}/Z(H)$. Moreover, we let the reservoir $E$ be in an arbitrary state $\sigma$. Hence, the joint system is initially in the state $G(H^{i}_{S'})\otimes |c_i\rangle\langle c_i|\otimes \sigma$.  After a global energy-conserving unitary operation $V$, the state of the energy reservoir is consequently given by the channel
 \begin{equation}
\label{nsdlkvn}
\mathcal{F}(\sigma) = \Tr_{S'C}\boldsymbol{(}V[G(H^i_{S'})\otimes|c_{i}\rangle\langle c_{i}| \otimes \sigma]V^{\dagger}\boldsymbol{)}.
\end{equation} 
To obtain a Crooks-like relation we should somehow relate this forward process with a `reversed' process, where the system instead is initiated in the Gibbs state of the final Hamiltonian, i.e., the global system should be in the joint state $G(H^{f}_{S'})\otimes |c_f\rangle\langle c_f|\otimes \sigma$. The question is how to formalize the idea of a reversed process acting on this initial state. A rather brutal interpretation would be to simply invert the entire global evolution, thus substituting $V$ with $V^{\dagger}$  (see Fig.~\ref{FigSimulation}). The resulting channel on the reservoir would in this case be
\begin{equation}
\label{sfldkbn}
\mathcal{R}(\sigma) = \Tr_{S'C}\boldsymbol{(}V^{\dagger}[G(H^f_{S'})\otimes|c_{f}\rangle\langle c_{f}| \otimes \sigma]V\boldsymbol{)}.
\end{equation} 
Let us now assume that we have a perfectly functioning control system, in the sense that the state $|c_i\rangle$ is transformed into $|c_f\rangle$ with certainty.  More precisely, we demand that the unitary operator $V$ should satisfy the condition
\begin{equation}
\label{MainPerfectControl}
[\hat{1}_{S'}\otimes |c_{f}\rangle\langle c_{f}|\otimes \hat{1}_E]V = V[\hat{1}_{S'}\otimes |c_{i}\rangle\langle c_{i}|\otimes \hat{1}_E].
\end{equation}
  One can verify (see Appendix \ref{SecIdealized}) that this assumption implies the following relation between the channels $\mathcal{F}$ and $\mathcal{R}$   
\begin{equation}
\label{Preliminary}
Z(H^i_S) \mathcal{F} = Z(H^f_S)\mathcal{J}_{\beta H_E}\mathcal{R}^{*}\mathcal{J}_{\beta H_E}^{-1}.
\end{equation}
Here $\mathcal{R}^{*}$ denotes the conjugate \cite{KrausBook} of the channel $\mathcal{R}$. If the channel has Kraus representation $\mathcal{R}(\rho) = \sum_{j}V_{j}\rho V_{j}^{\dagger}$, then the conjugate can be written $\mathcal{R}^{*}(\rho) = \sum_{j}V_{j}^{\dagger} \rho V_{j}$ \cite{KrausBook}. 

As the reader may have noticed, we have written $Z(H^i_S)$ and $Z(H^f_S)$ in (\ref{Preliminary}) rather than the more generally valid  $Z(H^i_{S'})$ and $Z(H^f_{S'})$. To obtain the former we can assume that $H^{i}_{S'} = H^{i}_{S}\otimes \hat{1}_B + \hat{1}_S\otimes H_{B}$ and $H^{f}_{SB} = H^{f}_{S}\otimes \hat{1}_B + \hat{1}_S\otimes H_{B}$,  and thus $Z(H^{i}_{S'}) = Z(H^{i}_S)Z(H_B)$ and $Z(H^{f}_{S'}) = Z(H^{f}_S)Z(H_B)$.

One may observe that the right hand side of  (\ref{Preliminary}) is similar to Crooks' quantum operation time reversal \cite{Crooks08}  (not to be confused with the time reversals $\mathcal{T}$ or the mapping $\ominus$ discussed in the next section) and closely related to Petz recovery channel \cite{Petz86,Petz88,Barnum02,Hayden04}, where a relation to work extraction was identified recently   \cite{Wehner15}. See also the discussion on time-reversals for quantum channels in \cite{Ticozzi09}. For further comments on this topic, see Appendix \ref{SecPetzRecovery}.

The approach of this investigation can be compared with another construction that also focuses on quantum channels \cite{Albash13,Rastegin13}, where the starting point is to assume a property of a channel (unitality) and derive fluctuation relations for the resulting probability distribution for suitable classes of initial and final measurements that sandwich the channel (see also the generalization in \cite{Rastegin14}).
Sequences of channels are considered in \cite{Manzano15}.

Although (\ref{Preliminary}) indeed can be regarded as a kind of quantum Crooks relation, it does suffer from an inherent flaw.  The swap between $V$ and $V^{\dagger}$ means that we invert the entire evolution of all involved systems, including the heat bath. Apart from assuming an immense level of control, this assumption does not quite fit with the spirit of Crooks relation. The latter only assumes a reversal of the time-schedule of the control parameters, not a reversal of the entire evolution. In the following we will resolve this issue by invoking time-reversal symmetry.

 \begin{figure}[t]
 \includegraphics[width= 8cm]{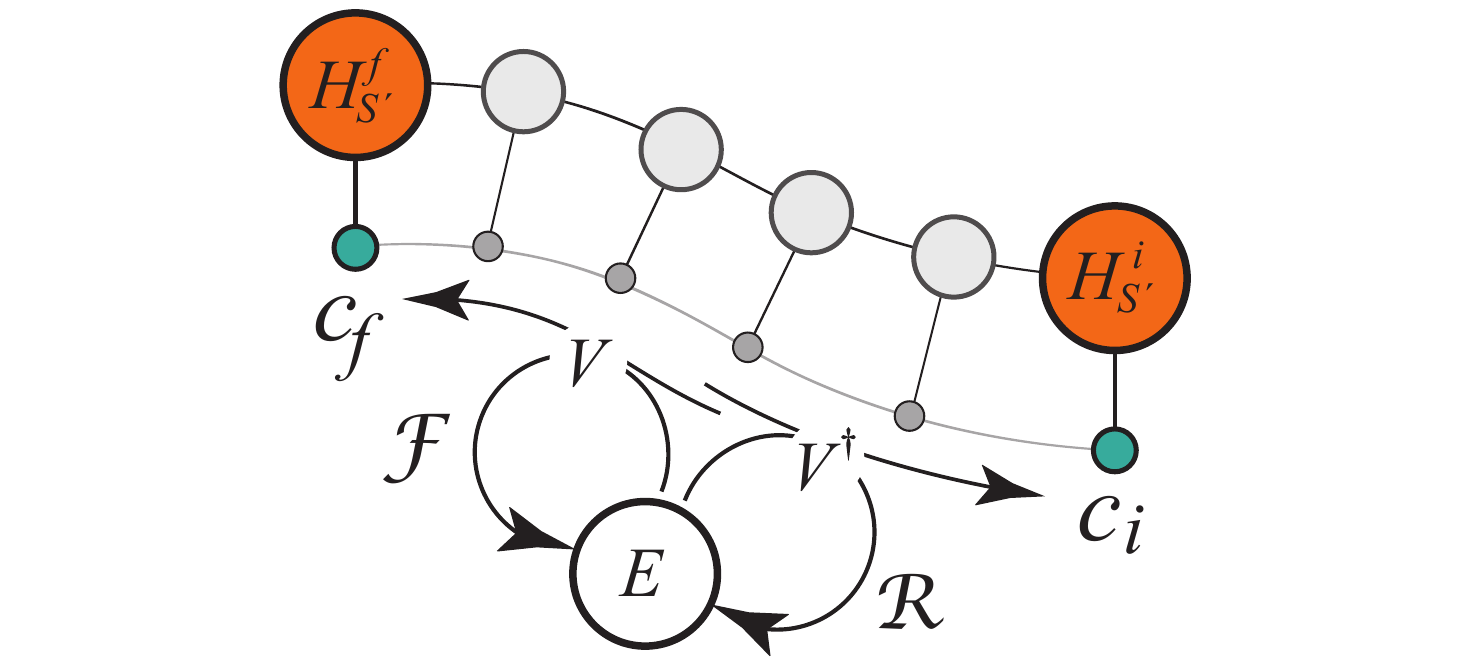} 
   \caption{\label{FigSimulation}  {\bf An intermediate quantum Crooks relation.}
Our quantum Crooks theorem is based on the idealized idea of a perfect control mechanism. This assumes that the energy conserving unitary evolution $V$ on the joint system $S'CE$ turns the control state $|c_i\rangle$ into $|c_f\rangle$ with certainty, thus implementing the change of the Hamiltonian $H^{i}_{S'}$ to $H^{f}_{S'}$ perfectly. To model the forward process of the Crooks relation we assume that 
 $S'C$ starts in the `conditional' equilibrium state $G(H^{i}_{S'})\otimes |c_i\rangle \langle c_i|$. The subsequent evolution under $V$ induces a channel $\mathcal{F}$ on the energy reservoir $E$. 
As an intermediate step towards the quantum Crooks relation  (\ref{QuantumCrooks}) we assume that the reverse process is given by the globally reversed evolution $V^{\dagger}$.  With $S'C$ in the initial state $G(H^{f}_{S'})\otimes |c_f\rangle \langle c_f|$, this results in a channel $\mathcal{R}$ on $E$. The preliminary Crooks relation in Eq.~(\ref{Preliminary}) relates the channels $\mathcal{F}$ and $\mathcal{R}$.\\
This result does not rely on any additional assumptions on how $V$ comes about, or on the nature of the dynamics at intermediate times. However, if one so wishes, it is possible to explicitly model the path $H(x)$ in the space of Hamiltonians. This path can be discretized as $H_l$ for $l=0,\ldots, L$ with $H_0 = H^{i}_{S'}$ and $H_L = H^{f}_{S'}$. To these Hamiltonians we associate states $|c_l\rangle$ in the control $C$, and construct the joint Hamiltonian  $H_{S'C} = \sum_{l=0}^{L}H_l\otimes |c_l\rangle\langle c_l|$. Hence,  if the evolution traverses the family of control states, then we progress along the path of Hamiltonians. An alternative method to model evolving Hamiltonians is discussed in section \ref{MainApproximate}.
 }
\end{figure}

\subsection{Time-reversals and time-reversal symmetry} 

In standard textbooks on quantum mechanics  (see e.g.~\cite{Sakurai}) time-reversals are often introduced on the level of Hilbert spaces via complex conjugation of wave-functions (also applied in the context of quantum fluctuation relations, see e.g.~\cite{Jarzynski04,Andrieux08,Campisi10b,Campisi11a,Campisi11b}). Here  we  instead regard time-reversals as acting on  operators (cf.~\cite{Sanpera97,Bullock05}). Moreover, our time-reversals have the flavor of  transpose operations rather than complex conjugations. Since density operators and observables are Hermitian, the choice between conjugation or transposition is largely a matter of taste. Here we opt for the transpose, since this choice avoids the inconvenient anti-linearity of complex conjugation. A couple of key properties of our time-reversals  $\mathcal{T}$  are $\mathcal{T}(AB) = \mathcal{T}(B)\mathcal{T}(A)$ and $\mathcal{T}(A)^{\dagger} = \mathcal{T}(A^{\dagger})$.
(For the complete characterization, see Appendix \ref{SecTimeReversal}.)

To see why it is reasonable to refer to  $\mathcal{T}$ as a `time-reversal', let us assume that the evolution operator $V$ is time-reversal invariant, i.e., $\mathcal{T}(V) = V$. If an initial state $\rho_i$  is evolved into $\rho_f = V\rho_i V^{\dagger}$, then the properties of $\mathcal{T}$ mentioned above yield $\mathcal{T}(\rho_i) = V\mathcal{T}(\rho_f)V^{\dagger}$. In other words, the reversed final state $\mathcal{T}(\rho_f)$  evolves to the reversed initial state $\mathcal{T}(\rho_i)$. Note that this `backwards' transformation is implemented by the forward evolution $V\rho V^{\dagger}$,  rather than the inverted evolution $V^{\dagger}\rho V$. 

The $\ominus$-transformation that appears in (\ref{QuantumCrooks}) can be defined via the chosen time-reversal as 
\begin{equation}
\label{MainDefominus}
\mathcal{F}^{\ominus} = \mathcal{T}\mathcal{F}^{*}\mathcal{T}.
\end{equation}
 As mentioned earlier, if $\mathcal{F}(\sigma) = \sum_{j}V_{j}\sigma V_{j}^{\dagger}$, then the conjugate is given by $\mathcal{F}^{*}(\sigma) = \sum_{j}V_{j}^{\dagger}\sigma V_{j}$. By the properties of $\mathcal{T}$ mentioned above it thus follows that $\mathcal{F}^{\ominus}(\sigma) = \sum_{j}\mathcal{T}(V_j)\sigma \mathcal{T}(V_{j})^{\dagger}$, i.e., we time-reverse the operators in the Kraus representation. 

In the finite-dimensional case each time-reversal can be implemented by a transpose with respect to some orthonormal basis, followed by a special class of unitary operations. (Hence, our time-reversals do strictly speaking include a bit larger set of operations than proper transposes, see Appendix \ref{SecTimeReversal}.) Since the transpose is a positive but not completely positive map (as is illustrated by the effect of the partial transpose on entangled states \cite{Peres96,Horodecki96}) we can conclude that time-reversals generally are not physical operations. Hence, one should not be tempted to think of time-reversals as  something that we actually apply to a system. However, given a \emph{description} of a state $\rho$, there is nothing that  in principle prevents us from preparing the time reversed state $\mathcal{T}(\rho)$, which is sufficient for our purposes.

\subsection{\label{DerivingQuantumCrooks}Deriving the quantum Crooks relation}

Similar to our intermediate version in section \ref{MainSecInterm} we here assume a non-interacting global Hamiltonian $H = H_{S'C}\otimes\hat{1}_E + \hat{1}_{S'C}\otimes H_E$, and a global energy conserving unitary evolution $[V,H]=0$. As a substitute for the global inversion of the evolution we  assume time-reversal symmetry in terms of a suitably chosen time-reversal of the form $\mathcal{T} = \mathcal{T}_{S'C}\otimes\mathcal{T}_E$, i.e., the global time reversal is composed of  local reversals on $S'C$ and on $E$. We impose the time reversal symmetry by assuming that $\mathcal{T}(V) = V$, $\mathcal{T}_{S'C}(H_{S'C}) = H_{S'C}$, and $\mathcal{T}_E(H_E) = H_E$.

Another novel component compared to the intermediate version is that we associate \emph{pairs} of control states to the Hamiltonians on $S'$. We assign the pair $|c_{i+}\rangle$, $|c_{i-}\rangle$ for the initial Hamiltonian $H^{i}_{S'}$, and the pair  $|c_{f+}\rangle$, $|c_{f-}\rangle$ for the final Hamiltonian $H^{f}_{S'}$. The general idea is that the members of these pairs correspond to the forward and reverse evolution of the control system. One way to  formalize this notion is via the following relations
\begin{equation}
\label{MainAssumpIdc}
\begin{split}
\mathcal{T}_{S'C}(\hat{1}_{S'}\otimes |c_{i+}\rangle\langle c_{i+}|) =  & \hat{1}_{S'}\otimes |c_{i-}\rangle\langle c_{i-}|,\\
 \mathcal{T}_{S'C}(\hat{1}_{S'}\otimes |c_{f+}\rangle\langle c_{f+}|) =  &\hat{1}_{S'}\otimes |c_{f-}\rangle\langle c_{f-}|.
\end{split}
\end{equation} 
Hence, the time-reversal swaps the two control states into each other. Note that if $V$ transforms $|c_{i+}\rangle$ to $|c_{f+}\rangle$ perfectly, then $V$ also transforms $|c_{f-}\rangle$ to $|c_{i-}\rangle$ perfectly. In the latter case one should keep in mind the reverse ordering of $i$ and $f$; the state $|c_{f-}\rangle$ is the initial control state of the reverse process (see Fig.~\ref{FigTimeReversal}).  

\begin{figure}[t]
 \includegraphics[width= 8cm]{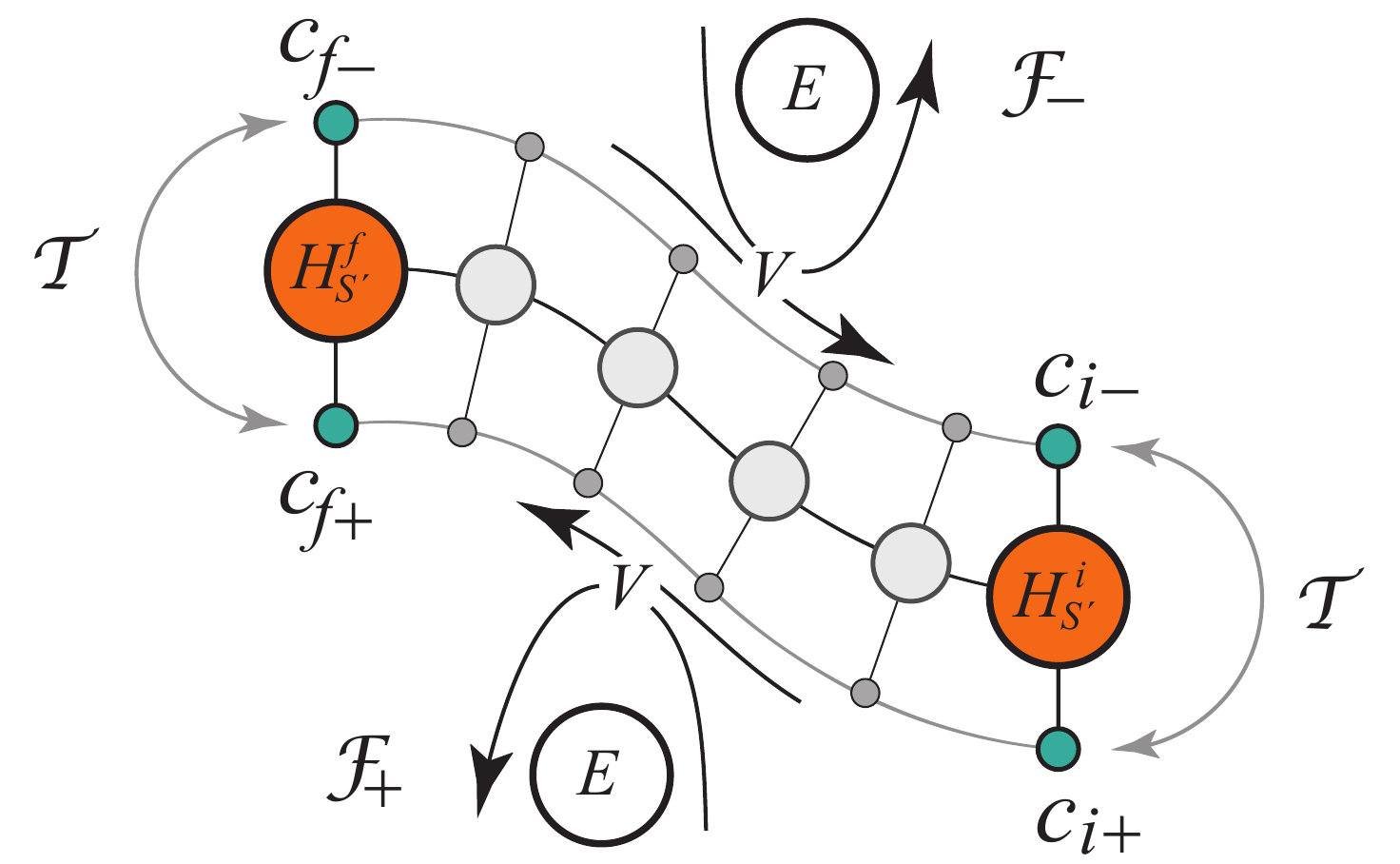} 
   \caption{\label{FigTimeReversal}  {\bf A quantum Crooks relation.}  
   Derivations of fluctuation relations often rely on time-reversals and time reversal symmetry in some form \cite{ReviewJarzynski,ReviewSeifert,ReviewMarconi,ReviewFluctThm,Esposito09,Campisi11a,Hanggi15}. In our case, time-reversal symmetry enables us to in effect let the evolution run `backwards' by swapping the control states from a `forward track' to a `reverse track'. By this we amend for the over-ambitious global reversal of the evolution that we employed for the intermediate fluctuation relation, and obtain a quantum counterpart to the reversal of the control parameter in the classical Crooks relation.\\
  Analogous to the setup in Fig.~\ref{FigSimulation}, the initial control state $|c_{i+}\rangle$ yields  an evolution where the initial Hamiltonian $H^{i}_{S'}$ is transformed to the final Hamiltonian $H^{f}_{S'}$ when the control state $|c_{f+}\rangle$ is reached. This process induces a channel $\mathcal{F}_{+}$ on the energy reservoir $E$. 
  To each of these control states there exists a time-reversed partner starting with $|c_{f-}\rangle$ and ending with $|c_{i-}\rangle$ thus bringing the Hamiltonian $H^{f}_{S'}$ back to $H^{i}_{S'}$. This  induces a channel $\mathcal{F}_{-}$ on the reservoir. By assuming time-reversal symmetry, the channels  $\mathcal{F}_{+}$ and $\mathcal{F}_{-}$ can be related via the quantum Crooks relation in (\ref{QuantumCrooks}).
   }
\end{figure}

Let us assume that the unitary operator $V$ and the initial and final control states $|c_{i+}\rangle, |c_{f+}\rangle$ satisfy the conditions of the intermediate fluctuation relation. The channels $\mathcal{F}_{+}$ and $\mathcal{R}_{+}$, obtained by substituting $|c_i\rangle$ with $|c_{i+}\rangle$, and $|c_f\rangle$ with $|c_{f+}\rangle$ in equations (\ref{nsdlkvn}) and (\ref{sfldkbn}), respectively, are according to  (\ref{Preliminary}) related by 
\begin{equation}
\label{lvnlnfnf}
Z(H^i_{S'}) \mathcal{F}_{+} = Z(H^f_{S'})\mathcal{J}_{\beta H_E}\mathcal{R}_{+}^{*}\mathcal{J}_{\beta H_E}^{-1}.
\end{equation}
 For the reverse process we obtain the channels $\mathcal{F}_{-}$ and $\mathcal{R}_{-}$, where $|c_i\rangle$ this time is substituted by $|c_{f-}\rangle$, and $|c_f\rangle$ with $|c_{i-}\rangle$.  By using the properties of the time-reversal one can show (for details, see Appendix \ref{SecAQuantumFlctnThrm}) that $\mathcal{R}_{+}\mathcal{T}_E = \mathcal{T}_E\mathcal{F}_{-}$.
 By applying this equality to 
(\ref{lvnlnfnf}) our desired quantum Crooks relation (\ref{QuantumCrooks}) follows immediately.
Note that both $\mathcal{F}_{+}$ and $\mathcal{F}_{-}$ correspond to the `forward' evolution, i.e., we no longer need to impose a global inversion of the dynamics.

After having established (\ref{QuantumCrooks}) an immediate question is how to interpret it physically.
The classical Crooks relation (\ref{StandardCrooks}) can be phrased as a comparison between two experiments, one for the forward path of control parameters, and one for the reversed. For each of these two setups we repeat the experiment in order to obtain good statistical estimates of the probability distributions $P_{+}$ and $P_{-}$. The claim, so to speak, of Crooks theorem is that these two probability distributions are related as in (\ref{StandardCrooks}). Our quantum fluctuation relation in (\ref{QuantumCrooks}) can be interpreted in a similar manner. Here the channels $\mathcal{F}_{+}$ and $\mathcal{F}_{-}$ induced by the forward and reverse process can in principle be determined by process tomography (see e.g.~\cite{Nielsen96,DAriano02,Altepeter03,Mohseni08} and Sec.~8.4.2 in \cite{NielsenChuang} for an overview). Analogous to the classical Crooks relation, our quantum fluctuation theorem tells us that these two channels should be related as in (\ref{QuantumCrooks}). One may note that the mapping $\mathcal{F}_{-}\mapsto \mathcal{J}_{\beta H_E}\mathcal{F}_{-}^{\ominus}\mathcal{J}_{\beta H_E}^{-1}$ in this setup is something that we \emph{calculate} given an experimental estimate of the channel $\mathcal{F}_{-}$, analogous to how we would calculate the function $e^{\beta w}P_{-}(-w)$ from the estimated distribution $P_{-}(w)$. In section \ref{MainSecGlobal} we consider a reformulation of  (\ref{QuantumCrooks}) that is associated with a scenario that does not require such a post-processing.

As the reader may have noticed, our quantum Crooks relation (\ref{QuantumCrooks}) does, in contrast to the classical Crooks relation (\ref{StandardCrooks}), not contain any explicit reference to `work'.
Hence, we do in essence circumvent the issue of how to translate the classical notion of work into the quantum setting; a question that has been discussed rather extensively in relation to fluctuation theorems \cite
{Bochkov81,Yukawa00,Monnai03,Engel07,Allahverdyan05,Gelin08,Chernyak04,Chernyak06,Li13,Kurchan01,Tasaki00,Mukamel03,Roeck04,Monnai05,Talkner05,Talkner07,Allahverdyan14}. 
We can nevertheless in some sense associate work with the loss of energy in the reservoir, which we employ in section \ref{SecMainEnergytransl} where we regain the classical Crooks relation (\ref{StandardCrooks}).

\subsection{\label{MainSecDiagonalOffdiagonal}Diagonal and off-diagonal Crooks relations}

Quantum systems do not only carry energy, but also coherence, i.e., superpositions between energy eigenstates. (An atom can, apart from being in the ground state or an exited state, also be in a superposition between the two.) Such coherences manifest themselves as non-zero off-diagonal elements in the matrix representation of the density operator in the energy eigenbasis. In the quantum regime, coherence emerges as a relevant resource alongside energy \cite{Janzing00,Scully03,Skrzypczyk14,Brandao13b,Aberg13,Rosario13,Lostaglio14a,Lostaglio14b,Klimowski14,Uzdin15,Cwiklinski15, Narasimhachar15,Korzekwa15,Streltsov16,Lostaglio17,Cirstoiu17,Gour17}.
While the classical Crooks theorem relates the distribution of work (and thus change of energy in the implicit energy reservoir) of the forward and reverse processes, the quantum version (\ref{QuantumCrooks}) additionally incorporates the evolution of coherence, as we shall see here.

The energy conservation and the diagonal initial state with respect to the Hamiltonian of $SBC$ conspire to decouple the dynamics on $E$ with respect to the `modes of coherence'  \cite{Lostaglio14b} (which in turn is a  part of the wider context of symmetry preserving operations \cite{Marvian14}). To make this more concrete, assume that $H_E$ has a complete set of energy eigenstates $|n\rangle$ with corresponding non-degenerate energies $E_n$. Let us also assume that  $\mathcal{T}_E$ acts as the transposition with respect to this energy eigenbasis. The decoupling does in this case mean that the channels $\mathcal{F}_{\pm}$ can map an element $|n\rangle\langle n'|$  to $|m\rangle\langle m'|$ only if $E_{n} -E_{n'} = E_{m}- E_{m'}$ (see Appendix \ref{SecDiagonalAndOffdiagonal} for further details). Hence, each mode of coherence is characterized by an `offset' $\delta = E_{n}-E_{n'}$ from the main diagonal (see Fig.~\ref{FigDiagonalOffdiagonal}). 

In particular, the main diagonal is given by the offset $\delta = 0$.
The dynamics of the diagonal elements under the forward process can be described via a conditional probability distribution  
\begin{equation}
\label{DiagonalConditional}
p_{\pm}(m|n) = \langle m|\mathcal{F}_{\pm}(|n\rangle\langle n|)|m\rangle =  \langle n|\mathcal{F}^{*}_{\pm}(|m\rangle\langle m|)|n\rangle.
\end{equation}
The application of (\ref{QuantumCrooks}) yields
\begin{equation}
\label{diagonalCrooks}
Z(H^i_S)p_{+}(m|n)  = e^{\beta (E_{n}-E_{m})} Z(H^f_S)p_{-}(n|m),
\end{equation}
which thus can be regarded as a diagonal Crooks relation. In section \ref{SecMainEnergytransl} we shall see how this in turn can be used to re-derive (\ref{StandardCrooks}) via the additional assumption of energy translation invariance. 

For each permissible value of $\delta$ (which depends on the spectrum of $H_E$)  we can characterize the dynamics of the corresponding off-diagonal mode by
\begin{equation*} 
q_{\pm}^{\delta}(m|n) =  \langle m|\mathcal{F}_{\pm}(|n\rangle\langle n'|)|m'\rangle
=  \langle n'|\mathcal{F}^{*}_{\pm}(|m'\rangle\langle m|)|n\rangle,
\end{equation*}
where the assumption of non-degeneracy implies that $n'$ is uniquely determined by $n,\delta$, and similarly $m'$ by $m,\delta$.
We obtain an off-diagonal analogue of (\ref{diagonalCrooks}) in the form of
\begin{equation}
\label{offdiagonalCrooks}
Z(H^i_S)q^{\delta}_{+}(m|n)  = e^{\beta (E_{n}-E_{m})} Z(H^f_S)q^{\delta}_{-}(n|m).
\end{equation}
In other words, the functional form of the relations for the main diagonal and all the off-diagonals are identical. However, one should keep in mind that the numbers $q^{\delta}_{\pm}(m|n)$ in general are complex, and thus cannot be interpreted as conditional probability distributions. 

\begin{figure}[t]
 \includegraphics[width= 6cm]{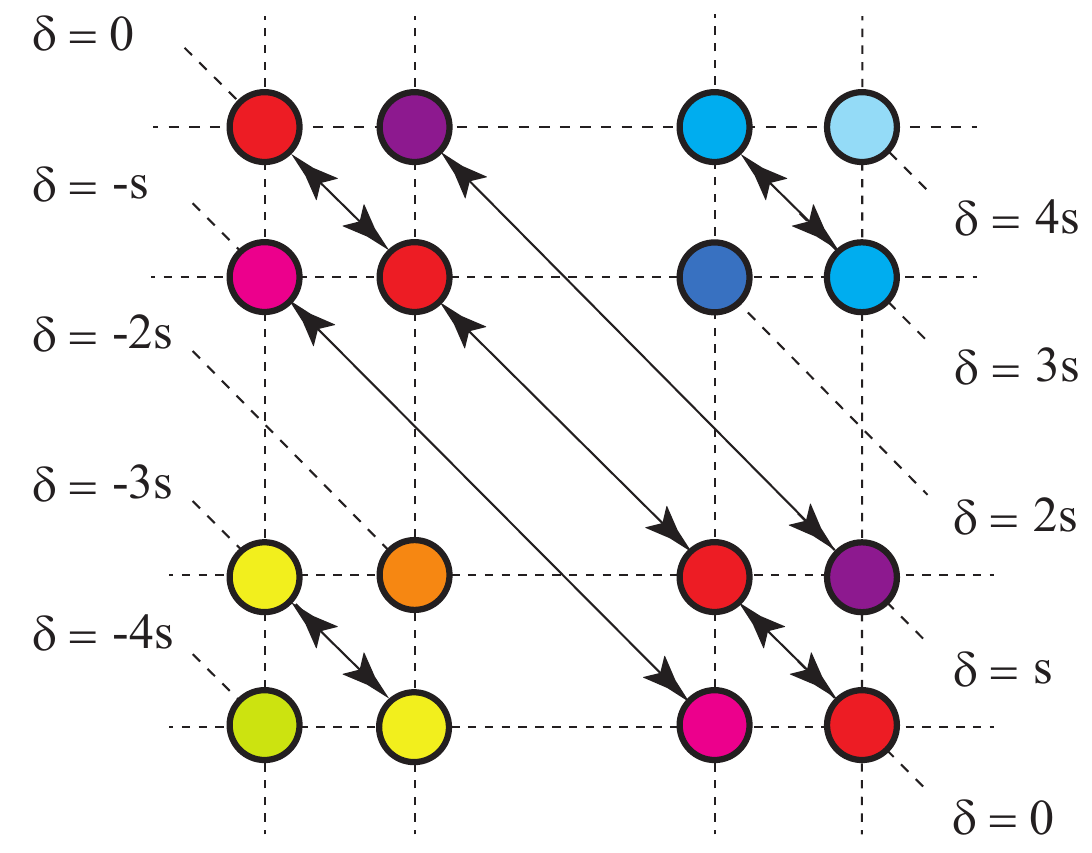} 
   \caption{\label{FigDiagonalOffdiagonal}  {\bf Diagonal and off-diagonal Crooks relations.}
Due to energy conservation in combination with the particular class of initial states, it turns out that the dynamics of the energy reservoir decouples along the modes of coherence \cite{Lostaglio14b}. One can think of these as the collection of diagonals (main and off-diagonals) of the density operator represented in an energy eigenbasis. Due to the decoupling it follows that the quantum fluctuation relation (\ref{QuantumCrooks}) can be separated into individual Crooks relations for each mode of coherence. As an example, suppose that the spectrum of the reservoir would include the energy levels  $E_1 = 0, E_2 = s, E_3 = 3s, E_4 = 4s$ for some $s>0$. 
The dynamics on the main diagonal, $\delta =0$, satisfies the relation in (\ref{diagonalCrooks}). Each of the diagonals with offsets $\delta = \pm s, \pm 2s,\pm 3s, \pm 4s$ satisfies a Crooks relation as in (\ref{offdiagonalCrooks}). For a more concrete example of the decoupling, see Fig.~\ref{FigDecoupling}.
   }
\end{figure}

The conditional distributions $p_{\pm}(m|n)$ and all  $q^{\delta}_{\pm}(m|n)$ can be regarded as representing different aspects of the channels $\mathcal{F}_{\pm}$. However, one can also give more direct physical interpretations to these quantities. To this end one should keep in mind that the energy reservoir can be viewed as being part of our experimental equipment, and we thus in principle are free to prepare and measure  it in any way that we wish. The conditional distribution $p_{\pm}(m|n)$ can be interpreted as the probability to measure the energy $E_m$ in the reservoir after the experiment has been completed, given that we prepared the reservoir in energy $E_n$ before we connected it to the system. Equivalently, we could replace the initial preparation with an energy measurement, thus approaching the  `two-time' or `two-point' measurements that play a central role in several investigations on quantum fluctuation relations  (see e.g.~\cite{Tasaki00,Kurchan01,Mukamel03b,Jarzynski04,Talkner05,Saito08,Quan08,Talkner08,Talkner09,Andrieux09,Campisi10,Campisi10b,Deffner11b,Campisi11b,Talkner13}, and for  overviews see \cite{Esposito09,Campisi11a}), although we here apply these measurements on the energy reservoir, rather than on the system. An obvious issue with sequential quantum measurements is that the first measurement typically perturbs the statistics of the second measurement. However, in our case the decoupling comes to our aid.  If we measure the energy in the reservoir after the process, then the statistics of this measurement would not be affected if we would insert an additional energy measurement before the process, irrespective of how off-diagonal the initial state may be. The decoupling thus justifies the use of two-point measurements on the energy reservoir, if our purpose is to determine $p_{\pm}$. However, the  very same argument also implies that if we wish to determine $q^{\delta}_{\pm}$, then we have to use non-diagonal initial states, as well as non-diagonal measurements. For an example of such a setup, see Appendix \ref{Secoffdiagonal}.

\subsection{Jarzynski equalities}

From the classical Crooks relation one can directly derive the Jarzynski equality $\langle e^{-\beta W} \rangle = Z(H^f)/Z(H^i)$ \cite{Jarzynski97,ReviewFluctThm}. As mentioned above, it is not clear how to translate the random variable $W$ into the quantum setting, and consequently it is also not evident how to translate the expectation value $\langle e^{-\beta W} \rangle$ (which yields a variety of approaches to the quantum Jarzynski equality, see overviews in \cite{Esposito09,Campisi11a,Hanggi15}). Like for the Crooks relation, we do not attempt a direct translation, but rather focus on the general dynamics in the energy reservoir. 

By applying (\ref{QuantumCrooks}) to an initial state $\sigma$ of the reservoir and taking the expectation value of the operator $e^{\beta H_E}$, it follows that 
\begin{equation}
\label{nvanlkavlkn}
\frac{\Tr[e^{\beta H}\mathcal{F}_{+}(\sigma)]}{\Tr[e^{\beta H_E/2}\mathcal{R}_{+}(\hat{1}) e^{\beta H_E/2}\sigma]} 
=  \frac{Z(H^f_{S})}{Z(H^i_{S}) }.
\end{equation} 
Although this equality has the flavor of a Jarzynski equality, we can bring it one step further by additionally assuming that $\mathcal{R}_{+}(\hat{1}) = \hat{1}$, or equivalently $\mathcal{F}_{-}(\hat{1}) = \hat{1}$, i.e., that these channels are unital. This assumption yields 
\begin{equation}
\label{fsjaflkslk}
\frac{\Tr[e^{\beta H_E}\mathcal{F}_{+}(\sigma)]}{\Tr(e^{\beta H_E}\sigma)} = \frac{Z(H^{f}_{S'})}{Z(H^{i}_{S'})}.
\end{equation}
This relation has a clear physical interpretation in terms of an experiment where we measure the expectation value of $e^{\beta H_E}$ on the input state of the reservoir (e.g., from the statistics of energy measurements), and in a separate experiment measure the same observable on the evolved state $\mathcal{F}_{+}(\sigma)$. 

In a similar fashion one can derive a whole family of Jarzynski-like equalities (with as well as without the assumption $\mathcal{R}_{+}(\hat{1}) = \hat{1}$, see Appendix  \ref{SecQuantumJarzynski}) and another member of this family is
\begin{equation}
\label{eq2jarzynski1}
\Tr[e^{\beta H_E}\mathcal{F}_{+}(e^{-\beta H_E/2} \sigma e^{-\beta H_E/2})] = \frac{Z(H^f_{S})}{Z(H^i_{S})}.
\end{equation}
This equality does not have a quite as direct physical interpretation as (\ref{fsjaflkslk}), but one can resort two-point energy measurements, resulting in fluctuation relations akin to e.g.~\cite{Tasaki00,Kurchan01,Mukamel03b,Monnai05,Talkner07}. To see this, assume that $H_E$ has a point spectrum. Then equation (\ref{eq2jarzynski1}) takes the form $\sum_{mn}e^{\beta (E_m-E_n)}p_{+}(m|n)  \langle n|\sigma|n\rangle =  Z(H^f_{S})/Z(H^i_{S})$, where we have made use of the decoupling of the diagonal mode of coherence. This can alternatively be obtained directly from the diagonal Crooks relation (\ref{diagonalCrooks}), where the unitality of $\mathcal{F}_{-}$ implies that $p_{-}$ is not only stochastic, but doubly stochastic.

There are other types of modified fluctuation relations that, in the spirit of Maxwell's demon, incorporate feedback control, resulting in efficacy factors (see e.g.~\cite{Sagawa10,Toyabe10,Morikuni11}). The fluctuation relations considered in this section do not include explicit feedback processes (as opposed to the conditional fluctuation relations in section \ref{MainConditional}, which can be said to represent a crude form of feedback, where we throw away results if we get the `wrong' outcome in a control measurement).  Nevertheless, one may wonder whether the global unitary evolution under some circumstances can be regarded as  implicitly representing a feedback process on the system. If so, this could potentially pave the way for a merging of these approaches. However,  we leave this as an open question.

\subsection{\label{MainSecBounds}Bounds on the work cost}

As mentioned above, the classical Jarzynski equality follows directly from Crooks relation. It is therefore perhaps a bit surprising that the condition of unitality of the channel $\mathcal{R}_{+}$ emerges in the derivations of our Jarzynski equalities. 
Classical fluctuation relations imply various inequalities, often regarded as manifestations, or refinements, of the second law (see e.g.~\cite{ReviewFluctThm,ReviewJarzynski}). Here we explore the condition of unitality of $\mathcal{R}_{+}$ by establishing quantum counterparts of some of these classical bounds. (For other discussions on the relation between the second law and fully quantum fluctuation theorems, see \cite{Alhambra16}.)

In macroscopic thermodynamics the second law implies that the amount of work required to change a system from one equilibrium state to another, when in contact with a single heat bath,  is at least equal to the free energy difference. 
In a statistical setting we would thus expect  (e.g.~as a special case of more general bounds  \cite{Procaccia76,Takara10,Esposito11,Lindblad83}) that the average work cost should satisfy 
\begin{equation}
\label{ClassicalBound}
\langle W\rangle \geq F(H^f)-F(H^i),
\end{equation}
 i.e., the work cost is at least equal to the difference in  (equilibrium) free energy between the final and initial Hamiltonian (where equality typically is reached in the limit of quasi-static processes).

The classical Jarzynski equality  implies (\ref{ClassicalBound}) via convexity of the exponential function \cite{Jarzynski97}. 
It turns out that one similarly can obtain a quantum counterpart to (\ref{ClassicalBound}). More precisely, one can use the diagonal Crooks relation (\ref{diagonalCrooks}) to show that
\begin{equation}
\label{zthetzj}
\begin{split}
 \Tr(H_E\sigma) - \Tr\big(H_E \mathcal{F}_{+}(\sigma)\big)  \geq & F(H^{f}_{S'})  -F(H^{i}_{S'}) \\
& - kT\ln \Tr\big(\sigma\mathcal{R}_{+}(\hat{1}_E)\big).
\end{split}
\end{equation}
(As a technical remark,  we do for this derivation additionally assume that $H_E$ has a non-degenerate point spectrum. For details see Appendix \ref{SecBoundsOnWork}.) If we identify the decrease in average energy of the reservoir, $\Tr(H_E\sigma)-\Tr\big(H_E\mathcal{F}_{+}(\sigma)\big)$  with the  average work cost $\langle W\rangle$, we thus regain the standard bound if $\mathcal{R}_{+}$ is unital. Hence, the unitality of $\mathcal{R}_{+}$ appear again, this time to guarantee the standard work bound. 

The inequality (\ref{zthetzj}) does not necessarily mean that the standard bound is violated if $\mathcal{R}_{+}$ fails to be unital. However, one can construct an explicit example where the work cost  is smaller than the standard bound for a process with a non-unital $\mathcal{R}_{+}$  (see Appendix \ref{SecViolationStndrdWorkB}). At first this may seem a bit alarming since it would appear to suggest violations of basic thermodynamics. However, one has to keep in mind that we here include the energy reservoir as a physical system, and that we cannot expect to regain standard bounds if we allow ourselves to use the energy reservoir \emph{per se} as a resource (cf.~discussions in \cite{Skrzypczyk14}), and one may suspect that the unitality of $\mathcal{R}_{+}$ is related to this issue. A further indication in this direction is that the energy translation invariance, which is the topic of the next section, not only allows us to re-derive the classical Crooks relation (\ref{StandardCrooks}), but also guarantees that the channel $\mathcal{R}_{+}$ is unital (see Appendix \ref{SecEnergyIndependence}).

For another counterpart to a classical bound, one can combine the forward process with the reverse processes to obtain a cycle (cf.~section 5 in  \cite{ReviewJarzynski}). 
In the spirit of the Kelvin-Planck formulation of the second law (see e.g.~\cite{BookHaar}), one would not expect that such a cyclic process would result in a net gain of work, and thus $\langle W_{+}\rangle + \langle W_{-}\rangle \geq 0$, with $W_{+}$ and $W_{-}$ being the (random) work cost of the forward and the reverse process, respectively. In the quantum case one can obtain a counterpart of this relation (see Appendix \ref{KelvinPlanckBounds}). We find, perhaps unsurprisingly, that this analogue is valid when $\mathcal{R}_{+}$ and $\mathcal{R}_{-}$ are unital.

The inequality in (\ref{ClassicalBound}) includes an expectation value, and thus opens up for the possibility that the (macroscopic) bound is violated for single instances of repeated experiments. In the classical case, fluctuation relations do imply more refined estimates of the work regarded as random variables. The probability that the work $W$  violates the macroscopic bound $W \geq F(H^f)-F(H^i)$ with more than $\zeta$ can be bounded as $P[W < F(H^f)-F(H^i)-\zeta]\leq e^{-\beta \zeta}$, see \cite{Jarzynski99}, or \cite{ReviewJarzynski}. In other words, the probability of a `second law violation' is exponentially suppressed in the size of that violation.   In the quantum case one can formulate a counterpart to this bound (see Appendix \ref{SecViolationBound}).

As suggested by the discussions above, the appearance of $\mathcal{R}_{+}$ is due to the explicit energy reservoir. One may thus suspect that this phenomenon would also occur in the purely classical case if explicit energy carriers are included. If  so, one may wonder what property would take over the role of the unitality of the channels $\mathcal{R}_{\pm}$. However, we will not consider these questions further in this investigation.

\subsection{\label{SecMainEnergytransl}Energy translation invariance: A bridge to standard fluctuation relations}

Let us for a moment reconsider the classical Crooks relation in equation (\ref{StandardCrooks}). As formulated,  it makes no explicit reference to any energy reservoir. Moreover, the distributions $P_{\pm}$ are typically phrased as functions only of the work $w$, and thus only functions of the \emph{change} of energy in the implicit reservoir (although it is difficult to see that anything in principle would prevent $P_{\pm}$ from depending on additional parameters, like the actual energy level of the reservoir). One could even argue that for a well-designed experiment $P_{\pm}$ should not depend on how much energy there is in the  reservoir (as long as there is enough), since this in some sense is a property of the experimental equipment rather than a property of the system under study. This is in contrast with our more general formulation in terms of channels on the reservoir, which very explicitly allows for the possibility that the  exact choice of initial state of the reservoir (and thus the energy) can make a difference. In order to formalize the idea that the experiment should be independent of the amount of energy in the reservoir, we do in this section assume energy translation invariance. This enables us to derive the classical Crooks relation (\ref{StandardCrooks}) from the quantum relation  (\ref{QuantumCrooks}). In this context one should note the approach in \cite{Alhambra16}, where energy translation invariance plays an important role.

The technique for implementing energy translation invariance has previously been used to study work extraction and coherence \cite{Aberg13,Skrzypczyk14}. Here we use a model where the spectrum of $H_E$ forms an equi-spaced doubly infinite energy-ladder \cite{Aberg13} (see also the continuum version in \cite{Skrzypczyk14}), i.e., the energy eigenstates $|n \rangle$ correspond to energies $E_{n} = sn$ for some fixed number $s>0$ and $n \in \mathbb{Z}$. We impose the energy translation invariance by assuming that the unitary operators $V$ should commute with energy translations along the energy-ladder, $[V,\Delta] = 0$, where $\Delta = \sum_{n}|n+1\rangle\langle n|$. (Technically, we also assume that the eigenvalues of $H^{i}_{S'}$ and $H^{f}_{S'}$ are multiples of $s$.)
As a consequence, the channels $\mathcal{F}_{\pm}$ also become energy-translation invariant in the sense that
\begin{equation}
\label{channeltranslinv}
\Delta^{j}\mathcal{F}_{\pm}(\sigma){\Delta^{\dagger}}^{k} =\mathcal{F}_{\pm}\big( \Delta^{j}\sigma{\Delta^{\dagger}}^{k}\big), 
\end{equation}
for all $j,k\in\mathbb{Z}$ (see Appendix \ref{SecEnergyIndependence}).  A further consequence is that the diagonal transition probabilities $p_{\pm}$, defined in  (\ref{DiagonalConditional}), inherit the translation invariance
\begin{equation}
\label{nvfleqlv}
p_{\pm}(m|n) = p_{\pm}(m-n|0) = p_{\pm}(0|n-m).
\end{equation}
Moreover, the translation invariance (\ref{channeltranslinv}) conspires with the decoupling of the diagonals described in section \ref{MainSecDiagonalOffdiagonal} such that we can rewrite $\mathcal{F}_{\pm}$ as
\begin{equation}
\label{ruetwporiu}
\mathcal{F}_{\pm}(\rho) =  \sum_{n,n',k\in\mathbb{Z}} p_{\pm}(k|0)\langle n|\rho|n'\rangle |n+k\rangle\langle n'+k|.
\end{equation}
Hence, the channels $\mathcal{F}_{\pm}$ are completely determined by the transition probabilities $p_{\pm}(k|0)$ from $|0\rangle\langle 0|$ to the other energy eigenstates. Alternatively, one may say that the dynamics is identical along all of the diagonals (which does \emph{not} imply that the different diagonals of the density matrix are identical).

Our primary goal in this section is to regain  (\ref{StandardCrooks}) from (\ref{QuantumCrooks}). However, we face an immediate obstacle in that  (\ref{StandardCrooks}) very explicitly refers to `work', while we in our constructions deliberately have avoided specifying what the quantum version of work exactly is supposed to be. A way to test Crooks theorem in a macroscopic setting would be to measure the energy content in the energy source before and after the experiment, in order to see how much energy that has been spent.  This macroscopic type of two-point measurements makes sense since we can expect a macroscopic source to be in states  that would not be significantly disturbed by energy measurements (e.g.~in the sense of pointer states and quantum Darwinism \cite{Zurek03,BlumeKohout06,Zurek09}). Recalling the discussion in section \ref{MainSecDiagonalOffdiagonal}, the decoupling of the diagonal and off-diagonal modes of coherence can be regarded as yielding a weaker form of stability. In this case all coherences are blatantly annihilated, but the evolution of the diagonal distribution is unaffected by repeated energy measurements. Since our aim is to regain the classical Crooks relation, this appears as an acceptable form of stability, and it also seems reasonable to identify work with the energy loss in the reservoir, $w = E_n -E_m = s(n-m)$, for the initial measurement outcome $E_n$ and the final outcome $E_m$. For these two-point measurements, the probability $P_{\pm}(w)$ is obtained by summing over all possible transitions that result in the work $w$ for the given initial state $\sigma$, and thus
\begin{equation}
\label{weeqwre}
P_{\pm}(w) = \sum_{m,n:E_n-E_m = w}p_{\pm}(m|n)\langle n|\sigma|n\rangle.
\end{equation}
To be able  to treat the energy reservoir as an implicit object, as it is in the standard classical case, the distributions  $P_{\pm}$ should not be functions of the initial state $\sigma$. Here the energy translation invariance comes into play. By combining (\ref{weeqwre}) with (\ref{nvfleqlv}) we find that $P_{\pm}(w) =  p_{\pm}(-w/s|0) = p_{\pm}(0|w/s)$, which thus removes the dependence on the reservoir. In the final step we combine the last equality with the diagonal Crooks relation  (\ref{diagonalCrooks}) and obtain the classical Crooks relation (\ref{StandardCrooks}).

\section{\label{MainConditional}Conditional fluctuation theorems}

There are several reasons for why it is useful to generalize the type of quantum fluctuation theorems that we have considered so far. First of all, the fluctuation relation (\ref{QuantumCrooks}) and its accompanying setup is in many ways an idealization. For example, the requirement of perfect control together with energy conservation is a strong assumption that easily leads to an energy reservoir that has to have an energy spectrum that is unbounded from both above and below (see Appendix \ref{SecConditionsOnE}) and a bottomless spectrum is certainly not very reasonable from a physical point of view. A further consequence of the idealization, which may not be apparent unless one dives into the technicalities in the appendices, is that (\ref{QuantumCrooks}) is based on rather elaborate assumptions. Apart from resolving these issues, the conditional fluctuation theorems naturally incorporate non-equilibrium initial states. 

There exist previous generalizations to non-canonical initial states. See e.g.~\cite{Cleuren06} for a classical case, and \cite{Talkner08,Campisi08,Talkner13,Allahverdyan14} for the quantum case. In the classical context of stochastic thermodynamics there are also fluctuation relations valid for arbitrary initial distributions \cite{Seifert05,Seifert12}.
Another example is the classical notion of extended fluctuation relations (EFRs) \cite{Maragakis08,Junier09,Alemany12,Liphardt12,Roldan14,Gavrilov17}, typically based on metastable states, or partial equilibrium conditions (akin to the  conditional equilibrium states discussed in Appendix \ref{TheInitalStates}). In contrast, the conditional fluctuation relation that we consider here allows for arbitrary non-equilibrium initial quantum states. In section \ref{SecMainQEFT}  we consider a quantum generalization of the classical EFRs.

To see how the conditional fluctuation relations can come about, let us again consider the perfect control mechanism, i.e., that the evolution transforms the initial control state into the final with certainty. Imagine now that we abandon this assumption, and instead include a measurement device that checks whether the control system succeeds in reaching the final control state or not. The idea is that we only accept the particular run of the experiment if the control measurement is successful. Analogous to our previous fluctuation relations, the conditional fluctuation theorem relates the induced dynamics of the forward and reverse processes, but which now are conditioned on successful control measurements (see Fig.~\ref{FigConditional}). 

It is useful to keep in mind that one should be cautious when interpreting post-selected dynamics, as it easily can end up in seemingly spectacular results. The conditional dynamics can for example create states with off-diagonal elements from globally diagonal states  (see Appendix \ref{ExampleTwoQubits} for an explicit example). However, this should not be viewed as a mysterious creation of coherence, but simply as the result of a post selection with respect to a non-diagonal measurement operator.

\begin{figure}
 \includegraphics[width= 8cm]{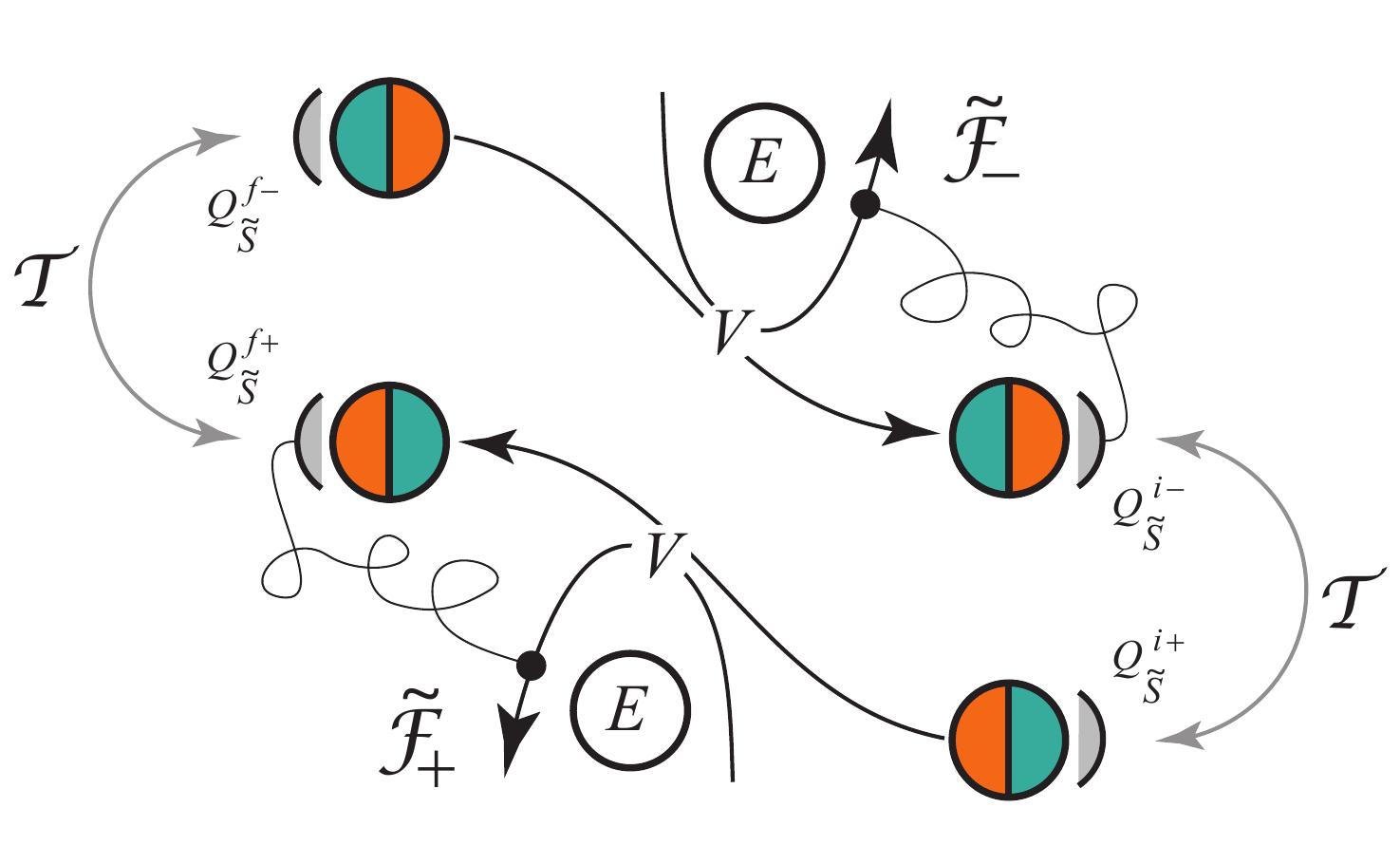} 
   \caption{\label{FigConditional}  {\bf Conditional fluctuation theorems.}   
 Similar to our previous fluctuation theorems we here relate the dynamics induced on the energy reservoir $E$ by the forward and reverse process, but conditioned on successful control measurements on $\widetilde{S} = SBC$. 
  The forward process is characterized by a pair of positive semi-definite operators  $Q_{\widetilde{S}}^{i+}, Q_{\widetilde{S}}^{f+}$ on $\widetilde{S}$, such that the initial state is given via the Gibbs map $\mathcal{G}_{\beta H_{\widetilde{S}}}(Q_{\widetilde{S}}^{i+})$, and the control measurement is represented by $Q_{\widetilde{S}}^{f+}$. More precisely, $Q_{\widetilde{S}}^{f+}$ corresponds to the `successful' outcome in the POVM $\{Q_{\widetilde{S}}^{f+},\hat{1}_{\widetilde{S}}-Q_{\widetilde{S}}^{f+}\}$, and results in the induced CPM $\tilde{\mathcal{F}}_{+}$ on the reservoir. One should keep in mind that the mapping $Q_{\widetilde{S}}^{i+}\mapsto \mathcal{G}_{\beta H_{\widetilde{S}}}(Q_{\widetilde{S}}^{i+})$ only serves as a convenient description of the initial state, and should not be interpreted as representing a particular physical process, or measurement.\\
 The reverse process, resulting in the CPM $\tilde{\mathcal{F}}_{-}$, is correspondingly characterized by the pair $Q_{\widetilde{S}}^{f-} = \mathcal{T}_{\widetilde{S}}(Q_{\widetilde{S}}^{f+})$, $Q_{\widetilde{S}}^{i-} = \mathcal{T}_{\widetilde{S}}(Q_{\widetilde{S}}^{i+})$, where $ \mathcal{G}_{\beta H_{\widetilde{S}}}(Q_{\widetilde{S}}^{f-})$ is the initial state, and $Q_{\widetilde{S}}^{i-}$ gives the control measurement. The CPMs $\tilde{\mathcal{F}}_{\pm}$ are related via the conditional fluctuation theorem in (\ref{ConditionalFluctuationThm}).
\newline
We are free to choose the operators $Q_{\widetilde{S}}^{i+}, Q_{\widetilde{S}}^{f+}$ in any way we wish, and in the finite-dimensional case we can obtain arbitrary initial states. Hence, we are not restricted to initial equilibrium states. The price  that we pay for this additional freedom is that non-equilibrium initial states are translated to  non-trivial control  measurements in the reverse process.
 }
\end{figure}

In the introduction it was pointed out that an advantage with fluctuation theorem (\ref{QuantumCrooks}) is that it does not rely on any auxiliary measurements. It may thus appear a bit contradictory that we here  re-introduce measurements as a fundamental component in the formalism. However, these only serve as control measurements upon which we condition the quantum evolution in the energy reservoir. This is in contrast to  approaches where the purpose of the measurements is to generate classical outcomes  (as for  two-point energy measurements \cite{Esposito09,Tasaki00,Kurchan01,Mukamel03b,Jarzynski04,Talkner05,Saito08,Quan08,Talkner08,Talkner09,Andrieux09,Campisi10,Campisi10b,Deffner11b,Campisi11b,Talkner13}) and where the quantum fluctuation relations are based on the resulting probability distributions, much in the spirit of classical fluctuation relations.

As a final general remark one should note that it turns out to be unnecessary to restrict the control measurements to the control system; we can let them act on the joint system $\tilde{S} = S'C=SBC$  in any manner  that we wish, resulting in a conditional evolution on $E$.

\subsection{The Gibbs map and the partition map}

Each control measurement has the two possible outcomes `yes' or `no', and can be described via measurement operators $0 \leq Q\leq \hat{1}$, where $\Tr(Q\rho)$ is the probability of a `yes' when measured on a state $\rho$. Similarly $\Tr[(1-Q)\rho]$ is the probability of the outcome `no'. In other words, $\{Q,\hat{1}-Q\}$ is a binary  positive operator valued measure (POVM) \cite{Davies76,Holevo82,KrausBook}. These operators serve a dual purpose in our formalism; not only describing control measurements, but also parameterizing initial states. The latter is done via what we here refer to as the `Gibbs map' $\mathcal{G}_{\beta H}$ and the `partition map' $\mathcal{Z}_{\beta H}$ defined by 
\begin{equation}
\label{GibbsMapDef}
\mathcal{G}_{\beta H}(Q) =  \frac{1}{\mathcal{Z}_{\beta H}(Q)}\mathcal{J}_{\beta H}(Q),\quad \mathcal{Z}_{\beta H}(Q)  =   \Tr \mathcal{J}_{\beta H}(Q).
\end{equation}
 For  finite-dimensional Hilbert spaces all density operators can be  reached via $\mathcal{G}_{\beta H}$ for suitable choices of $Q$. These mappings are generalizations of the Gibbs state and the partition function in the sense that $\mathcal{G}_{\beta H}(\hat{1}) = G(H)$ and $\mathcal{Z}_{\beta H}(\hat{1}) = Z(H)$.

\subsection{A conditional fluctuation relation}

As before, we assume a global non-interacting Hamiltonian of the form 
\begin{equation}
\label{StructureGlobalHamil}
H = H_{\widetilde{S}}\otimes \hat{1}_E + \hat{1}_{\widetilde{S}}\otimes H_E.
\end{equation}
We furthermore assume time-reversals $\mathcal{T}_{\widetilde{S}}$ and $\mathcal{T}_E$ such that $\mathcal{T}_{\widetilde{S}}(H_{\widetilde{S}}) = H_{\widetilde{S}}$, and $\mathcal{T}_E(H_E) = H_E$. We also assume that the global evolution $V$ preserves energy $[V,H] = 0$, and that it is time-reversal symmetric, $\mathcal{T}(V) = V$, with respect to  $\mathcal{T} = \mathcal{T}_{\widetilde{S}}\otimes\mathcal{T}_E$.

The forward process is characterized by a pair of positive semi-definite operators $Q_{\widetilde{S}}^{i+}, Q_{\widetilde{S}}^{f+}$ on system $\widetilde{S}$. The operator $Q_{\widetilde{S}}^{f+}$ characterizes a  measurement at the end of the process (and it does not have to be related to energy measurements). The operator $Q_{\widetilde{S}}^{i+}$ has the rather different role of characterizing the initial state $\rho_{\widetilde{S}}^{i+}$ via the Gibbs map, such that  $\rho_{\widetilde{S}}^{i+} =  \mathcal{G}_{\beta H_{\widetilde{S}}}(Q_{\widetilde{S}}^{i+})$.
One should note that the initial state $\rho_{\widetilde{S}}^{i+}$ can be prepared in any manner that we wish.
The Gibbs map is only used as a convenient method to describe the initial state, and is not intended to imply any particular choice of preparation procedure. For example, one should not be tempted to interpret $\mathcal{G}_{\beta H_{\widetilde{S}}}(Q_{\widetilde{S}}^{i+})$ as necessarily involving a measurement of $Q_{\widetilde{S}}^{i+}$. 
Even though $Q_{\widetilde{S}}^{i+}, Q_{\widetilde{S}}^{f+}$ do not always represent measurements, we will nevertheless for the sake of convenience refer to them as `measurement operators' (simply meaning that they satisfy $0\leq Q\leq \hat{1}$).

The global initial state $\rho_{\widetilde{S}}^{i+}\otimes \sigma$ evolves via an energy-preserving unitary operation $V$, and we condition the result on a successful measurement of $Q_{\widetilde{S}}^{f+}$. The resulting 
completely positive map (CPM) on $E$ is  given by
\begin{equation}
\label{tildeFplus}
\begin{split}
\tilde{\mathcal{F}}_{+}(\sigma) = & \Tr_{\widetilde{S}}\boldsymbol{(}[Q_{\widetilde{S}}^{f+}\otimes\hat{1}_E]V[ \rho_{\widetilde{S}}^{i+} \otimes \sigma]V^{\dagger}\boldsymbol{)}.
\end{split}
\end{equation}
That we generically get a CPM rather than a channel corresponds to the fact that the control measurement may fail. The quantity  $\Tr\tilde{\mathcal{F}}_{+}(\sigma)$ is the probability that the control measurement succeeds.

To obtain the reversed process we define the operators $Q_{\widetilde{S}}^{i-} = \mathcal{T}_{\widetilde{S}}(Q_{\widetilde{S}}^{i+})$ and $Q_{\widetilde{S}}^{f-} = \mathcal{T}_{\widetilde{S}}(Q_{\widetilde{S}}^{f+})$, and use $Q_{\widetilde{S}}^{f-}$ to generate the initial state, and $Q_{\widetilde{S}}^{i-}$ to characterize the final measurement. In other words, we not only time-reverse the measurement operators, but also swap their roles. The resulting CPM on $E$ for the reversed process thus becomes 
\begin{equation}
\label{tildeFminus}
\begin{split}
\tilde{\mathcal{F}}_{-}(\sigma) =  & \Tr_{\widetilde{S}}\boldsymbol{(}[Q_{\widetilde{S}}^{i-}\otimes\hat{1}_E]V[\rho_{\widetilde{S}}^{f-}\otimes \sigma]V^{\dagger}\boldsymbol{)},\\
\rho_{\widetilde{S}}^{f-}= &\mathcal{G}_{\beta H_{\widetilde{S}}}(Q_{\widetilde{S}}^{f-}).
\end{split}
\end{equation}
Hence, for the reverse process we do again have a single control measurement at the end of the process, but now characterized by $Q_{\widetilde{S}}^{i-}$.

With the above assumptions one can show (see Appendix \ref{SecConditional}) that the CPMs $\tilde{\mathcal{F}}_{\pm}$ are related by the following `conditional' fluctuation theorem 
 \begin{equation}
 \label{ConditionalFluctuationThm}
\mathcal{Z}_{\beta H_{\widetilde{S}}}(Q_{\widetilde{S}}^{i})\tilde{\mathcal{F}}_{+} = \mathcal{Z}_{\beta H_{\widetilde{S}}}(Q_{\widetilde{S}}^{f})\mathcal{J}_{\beta H_E}\tilde{\mathcal{F}}_{-}^{\ominus}\mathcal{J}_{\beta H_E}^{-1},
\end{equation}
where we use the notation $\mathcal{Z}_{\beta H_{\widetilde{S}}}(Q_{\widetilde{S}}^{i}) = \mathcal{Z}_{\beta H_{\widetilde{S}}}(Q_{\widetilde{S}}^{i\pm})$ and $\mathcal{Z}_{\beta H_{\widetilde{S}}}(Q_{\widetilde{S}}^{f}) = \mathcal{Z}_{\beta H_{\widetilde{S}}}(Q_{\widetilde{S}}^{f\pm})$.

The conditional fluctuation relation (\ref{ConditionalFluctuationThm}) allows (up to infinite-dimensional technicalities) for arbitrary initial states on $\widetilde{S}$, and thus in particular arbitrary coherences. We only need to choose the appropriate initial operator $Q_{\widetilde{S}}^{i+}$ that generates the desired initial state via the Gibbs map. 

The quantum fluctuation relation in  (\ref{QuantumCrooks}) can be regained from (\ref{ConditionalFluctuationThm}), and as
 a first step we choose the measurement operators to be $Q_{\widetilde{S}}^{i\pm} = \hat{1}_{S'}\otimes |c_{i\pm}\rangle\langle c_{i\pm}|$ and $Q_{\widetilde{S}}^{f\pm} = \hat{1}_{S'}\otimes |c_{f\pm}\rangle\langle c_{f\pm}|$. For the appropriate Hamiltonian $H_{\widetilde{S}}$, the Gibbs map gives the conditional equilibrium states $\mathcal{G}_{\beta H_{\widetilde{S}}}(Q_{\widetilde{S}}^{i\pm}) = G(H^{i}_{S'})\otimes |c_{i\pm}\rangle\langle c_{i\pm}|$ and $\mathcal{G}_{\beta H_{\widetilde{S}}}(Q_{\widetilde{S}}^{f\pm}) = G(H^{f}_{S'})\otimes |c_{i\pm}\rangle\langle c_{i\pm}|$. Due to the additional assumption of perfect control, the evolution $V$  is such that we can replace the final measurement operators with identity operators, i.e., we do not have to perform any measurements at the end of the processes.

This sketchy re-derivation indicates that we can abolish the assumption of perfect control if we instead keep the measurements at the end of the protocol. Hence, these control measurements may fail, e.g., if the energy reservoir does not contain sufficient energy, which thus allows us to avoid an energy spectrum that is unbounded from below. (For an explicit example, see Appendix \ref{AbolishPerfectControl}.)

For the conditional maps $\tilde{\mathcal{F}}_{\pm}$ there is in general no decomposition of the dynamics into different diagonals. However, in the special case that $Q_{\widetilde{S}}^{i+}$ and $Q_{\widetilde{S}}^{f+}$ are diagonal with respect to an eigenbasis of $H_{\widetilde{S}}$ one can regain the decomposition of the dynamics (see Appendix \ref{SecDiagonalMeasuremnts}). In particular, one can additionally impose the energy translation invariance and obtain  a conditional version of the classical  Crooks relation (Appendix \ref{SecDiagonalMeasuremnts}).

\section{\label{MainSecGlobal}An alternative formulation}

So far we have formulated our fluctuation relations solely in terms of channels and CPMs induced on the energy reservoir.
Here we consider a reformulation that highlights two elementary properties on which our fluctuation relations in some sense are based.

Given some channel or CPM $\mathcal{F}$, a pair of measurement operators $Q^i$ and $Q^f$, and a Hamiltonian $H$, we define the `transition probability' from $Q^i$  to $Q^f$ by 
\begin{equation}
\label{DefTransitionProb}
P^{\mathcal{F}}_{\beta H}[Q^i\rightarrow Q^{f}]= \Tr\Big(Q^f\mathcal{F}\big(\mathcal{G}_{\beta H}(Q^{i})\big)\Big).
\end{equation}
Hence, $P^{\mathcal{F}}_{\beta H}[Q^i\rightarrow Q^{f}]$ is the probability that we would detect $Q^i$ after we have evolved the initial state $\mathcal{G}_{\beta H}(Q^{i})$ under the map $\mathcal{F}$. The Gibbs map thus again appears as a method to parametrize the set of initial states. 

The quantum Crooks relation (\ref{QuantumCrooks}) can be rephrased in terms of these transition probabilities as
\begin{equation}
\label{ReformulationQuantumCrooks}
\begin{split}
& Z(H^i_{S})\mathcal{Z}_{\beta H_E}(Q^{i}_E)P^{\mathcal{F}_{+}}_{\beta H_E}[Q^{i+}_E\rightarrow Q^{f+}_E] \\
&= Z(H^f_{S})\mathcal{Z}_{\beta H_E}(Q^{f}_E)P^{\mathcal{F}_{-}}_{\beta H_E}[Q^{f-}_E\rightarrow Q^{i-}_E], 
\end{split}
\end{equation} 
where $Q^{i+}_E, Q^{f+}_E$ are measurement operators on the energy reservoir (not on $\widetilde{S}$), and where, as one may expect,  $Q^{i-}_E = \mathcal{T}_E(Q^{i+}_E)$ and $Q^{f-}_E = \mathcal{T}_E(Q^{f+}_E)$.

In section \ref{DerivingQuantumCrooks} we noted the possibility to interpret the quantum Crooks theorem (\ref{QuantumCrooks}) in terms of process tomography of the two induced channels $\mathcal{F}_{\pm}$. In that interpretation the mapping $\mathcal{F}_{-}\mapsto \mathcal{J}_{\beta H_E}\mathcal{F}^{\ominus}_{-}\mathcal{J}^{-1}_{\beta H_E}$ is something that we implement `by hand' on the experimentally determined description of the channel $\mathcal{F}_{-}$. The reformulation (\ref{ReformulationQuantumCrooks}) suggests a different interpretation where the fluctuation relation can be tested more directly via the estimated transition probabilities. The mapping $\mathcal{J}_{\beta H_E}$  (as well as the time-reversal $\mathcal{T}$ and the conjugate $*$ contained in $\ominus$) is in some sense put in by hand also in this scenario, but enters in the parametrization of the initial states via the Gibbs map, rather than via a post processing of the measurement data.

The interpretation of the quantum Crooks theorem (\ref{QuantumCrooks}), or conditional fluctuation theorem (\ref{ConditionalFluctuationThm}), in terms of process tomography may suggest rather extensive experimental procedures, since the number of measurement settings required to perform such a tomography increases rapidly with the size of the involved systems. (We need sufficiently many initial states and final observables in order to span the input and output spaces of $\mathcal{F}_{\pm}$ or $\tilde{\mathcal{F}}_{\pm}$. See e.g.~Sec.~8.4.2 in \cite{NielsenChuang}.) However, the reformulation (\ref{ReformulationQuantumCrooks}) suggests milder tests, where we only make one single choice of initial state and final measurement for the forward and reverse process. The estimation of free energy differences, discussed in section \ref{SecMainQEFT}, is an example.

It is no coincidence that the measurement operators $Q^{i\pm}_{E}, Q^{f\pm}_E$ undergo the same transformations as the operators $Q^{i\pm}_{\widetilde{S}},Q^{f\pm}_{\widetilde{S}}$ that we know from the conditional fluctuation relations.  
These similarities stem from the fact that all of our fluctuation relations can be regarded as special cases of a `global' fluctuation relation, defined for global measurement operators $Q^{i+},Q^{f+}$ on the whole of $SBCE$. 
It is straightforward to confirm (or to consult Appendix \ref{SecAltformTransProb}) that energy conservation $[H,V] = 0$ combined with time reversal symmetry $\mathcal{T}(H) = H$, $\mathcal{T}(V) = V$, yield
\begin{equation}
\label{MainGlobal}
\mathcal{Z}_{\beta H}(Q^{i})P^V_{\beta H}[Q^{i+}\!\rightarrow\!Q^{f+}] = \mathcal{Z}_{\beta H}(Q^{f})P^V_{\beta H}[Q^{f-}\!\rightarrow\! Q^{i-}],
\end{equation}
where $Q^{f-} = \mathcal{T}(Q^{f+})$ and $Q^{i-} = \mathcal{T}(Q^{i+})$. 

The global relation (\ref{MainGlobal}) can alternatively be phrased as the invariance of the quantity  $\Tr\big(Q^{f}V\mathcal{J}_{\beta H}(Q^{i})V^{\dagger}\big)$ under the transformation $(Q^i,Q^f)\mapsto (Q^{i'},Q^{f'}) = \big(\mathcal{T}(Q^f),\mathcal{T}(Q^i)\big)$.  All our fluctuation relations do in some sense  express this basic invariance in different guises.

One can indeed re-derive  (\ref{QuantumCrooks}) and  (\ref{ConditionalFluctuationThm}) from (\ref{MainGlobal}). However, when doing so one quickly realizes that this hinges on another property (that we have used repeatedly without comment). Namely, if the global Hamiltonian $H$ is non-interacting over two subsystems, i.e. $H = H_1\otimes \hat{1}_2 + \hat{1}_1\otimes H_2$, then the mapping $\mathcal{J}$ satisfies the factorization property
\begin{equation}
\label{MainFactorization}
\mathcal{J}_{\beta H}(Q_1\otimes Q_2) = \mathcal{J}_{\beta H_1}(Q_1)\otimes \mathcal{J}_{\beta H_2}(Q_2).
\end{equation}
Although elementary, it is in some sense this property (in conjunction with product time-reversals $\mathcal{T} = \mathcal{T}_{1}\otimes\mathcal{T}_2$) that makes it possible to formulate fluctuation relations solely in terms of the dynamics of the energy reservoir, and also to eliminate unaccessible degrees of freedom (see Appendix \ref{SecAltformTransProb} for further details). The central role of this property may become more apparent when we abandon it in section \ref{MainApproximate} .

As a side-remark one can note that if one attempts to formulate quantum fluctuation relations for non-exponential forms of generalized equilibrium distributions, then the lack of a factorization property for non-exponential functions makes the generalization problematic (see Appendix \ref{SecGeneralizedGibbs}).

\section{\label{MainPrecorrelations} Correlated initial states}
 
The conditional relations extend the notion of fluctuation theorems to non-equilibrium initials states on $\widetilde{S}$. An obvious question is if also entanglement and general pre-correlations can be included. (One could even argue that in order to deserve the label `fully quantum' our fluctuation relations must include  quantum correlations.) The global fluctuation relation (\ref{MainGlobal}) already provides an affirmative answer to this question, in the sense that we can generate arbitrary global initial states via suitable choices of measurement operators (including quantum correlations to external reference systems). However, (\ref{MainGlobal}) is not the only option for how to describe such scenarios, and here we highlight three special cases that focus on descriptions via channels or CPMs.

\subsection{Correlations between $E$ and an external reference}

Imagine that the energy reservoir not only carries energy and coherences, but also carries entanglement and correlations with an external reference $R$ (not included in $SBCE$).  One option for obtaining a fluctuation relation that incorporates this scenario would be to regard $ER$ as a new extended energy reservoir, with a Hamiltonian $H_{ER}$ and a new time-reversal $\mathcal{T}_{ER}$, with  $\mathcal{T}_{ER}(H_{ER}) = H_{ER}$.  Both the  quantum fluctuation relation (\ref{QuantumCrooks}) and the conditional fluctuation relation (\ref{ConditionalFluctuationThm}) are applicable to this scenario.
 In the latter case, the induced CPMs  $\mathcal{F}_{\pm}$ on the new energy reservoir $ER$ satisfies the corresponding quantum Crooks relation $\mathcal{Z}_{\beta H_{\widetilde{S}}}(Q_{\widetilde{S}}^{i})\mathcal{F}_{+} = \mathcal{Z}_{\beta H_{\widetilde{S}}}(Q_{\widetilde{S}}^{f})\mathcal{J}_{\beta H_{ER}}\mathcal{F}_{-}^{\ominus}\mathcal{J}_{\beta H_{ER}}^{-1}$.

\subsection{Correlations between $\widetilde{S}$ and an external reference}

An obvious alternative to the energy reservoir carrying correlations with an external reference is that  $\widetilde{S}$ carries these correlations. 
We can analogously incorporate $R$ into  $\widetilde{S}$, with the joint Hamiltonian $H_{\widetilde{S}R}$ and an extended time-reversal $\mathcal{T}_{\widetilde{S}R}$ such that $\mathcal{T}_{\widetilde{S}R}(H_{\widetilde{S}R}) = H_{\widetilde{S}R}$. The requirements for the conditional fluctuation theorem (\ref{ConditionalFluctuationThm}) are satisfied for this extended system, thus resulting in the relation $\mathcal{Z}_{\beta H_{\widetilde{S}R}}(Q_{\widetilde{S}R}^{i})\tilde{\mathcal{F}}_{+} = \mathcal{Z}_{\beta H_{\widetilde{S}R}}(Q_{\widetilde{S}R}^{f})\mathcal{J}_{\beta H_E}\tilde{\mathcal{F}}_{-}^{\ominus}\mathcal{J}_{\beta H_E}^{-1}$.

\subsection{\label{MainCorrelatedSE}Pre-correlated $\widetilde{S}$ and $E$}

The two previous examples can both be regarded as variations of the conditional fluctuation relation, and in this sense do not add anything essentially new to the general picture. A maybe more interesting case is if we allow for pre-correlations between $\widetilde{S}$ and $E$. (It is not difficult to imagine cases where the energy reservoir $E$ interacts repeatedly with $\widetilde{S}$, and thus may build up correlations.) Apart from the global fluctuation relation (\ref{MainGlobal}) this case does not fit very well with our previous scenarios. However, it is straightforward to adapt our general formalism to find a fluctuation relation for the evolution on the joint system $SE$.

For the sake of illustration we consider one particular case related to the setup 
in section \ref{DerivingQuantumCrooks}, but where we allow for initial correlations between $S$ and $E$. In essence (for details on the setup, see Appendix \ref{SecPrecorrelations}) we have an initial Hamiltonian $H^{i}_{SE}$ and a final Hamiltonian $H^{f}_{SE}$ and assume perfect control that transfers one into the other, resulting in the induced channels
\begin{equation*}
\begin{split}
\overline{\mathcal{F}}_{+}(\chi) = & \Tr_{CB}(V [|c_{i+}\rangle\langle c_{i+}|\otimes G(H_B)\otimes\chi] V^{\dagger}),\\
\overline{\mathcal{F}}_{-}(\chi) = & \Tr_{CB}(V [|c_{f-}\rangle\langle c_{f-}|\otimes G(H_B)\otimes\chi] V^{\dagger})\\
\end{split}
\end{equation*}
on the combined system $SE$. For these channels one can derive the fluctuation relation
\begin{equation}
\label{fdbvlqfv}
\overline{\mathcal{F}}_{+} = \mathcal{J}_{\beta H_{SE}^{f}}\overline{\mathcal{F}}_{-}^{\ominus}\mathcal{J}^{-1}_{\beta H_{SE}^{i}}.
\end{equation}  
The partition maps on the left and right hand side cancel due to the identical initial and final Hamiltonian for the heat bath, which starts in the Gibbs state. Note also that the two applications of the $\mathcal{J}$-map may potentially involve an initial Hamiltonian $H_{SE}^{i}$ that is different from the final Hamiltonian $H_{SE}^{f}$.

\section{\label{MainApproximate}Approximate fluctuation relations} 

The global Hamiltonian $H$ has so far in this investigation only characterized the notion of energy conservation, while the dynamics has been modeled by a unitary operator $V$, with the only restriction that $H$ and $V$ should commute. In view of standard textbook quantum mechanics it would be very reasonable to demand a much more tight connection, where the global Hamiltonian $H$ induces the evolution according to Schr\"odinger's equation, thus yielding $V = e^{-it H/\hbar}$. 
An additional benefit with the latter arrangement would be that it does not require any further interventions beyond the preparation of the initial state, the measurement of the final state, and the time-keeping for when to do the measurement. In other words, once we have started the global system it evolves autonomously.  (For  discussions on autonomous and clock controlled thermal machines in the quantum regime, see \cite{Malabarba15,Frenzel15}.)  This should be compared with the models that we have employed so far, where one in principle should analyze the mechanism that implements the evolution $V$.

The problem is that if we would impose the condition  $V = e^{-it H/\hbar}$ for non-interacting Hamiltonians as in (\ref{StructureGlobalHamil}), the resulting dynamics on the energy reservoir would become trivial. If we on the other hand would abandon the non-interacting Hamiltonians, then we also have to abandon the factorization property (\ref{MainFactorization}), which, as pointed out in section \ref{MainSecGlobal}, plays an important role in our derivations. The main observation in this section is that we can obtain approximate fluctuation relations as long as the factorization holds approximately, which also is enough to obtain a non-trivial evolution. In the following section we shall illustrate the general ideas with two special cases. (For a more general version that includes both of them, see Appendix \ref{SecApproximateFluct}.)

\subsection{\label{MainApproxConditional}Approximate conditional fluctuation relations}

\begin{figure}[t]
 \includegraphics[width= 7cm]{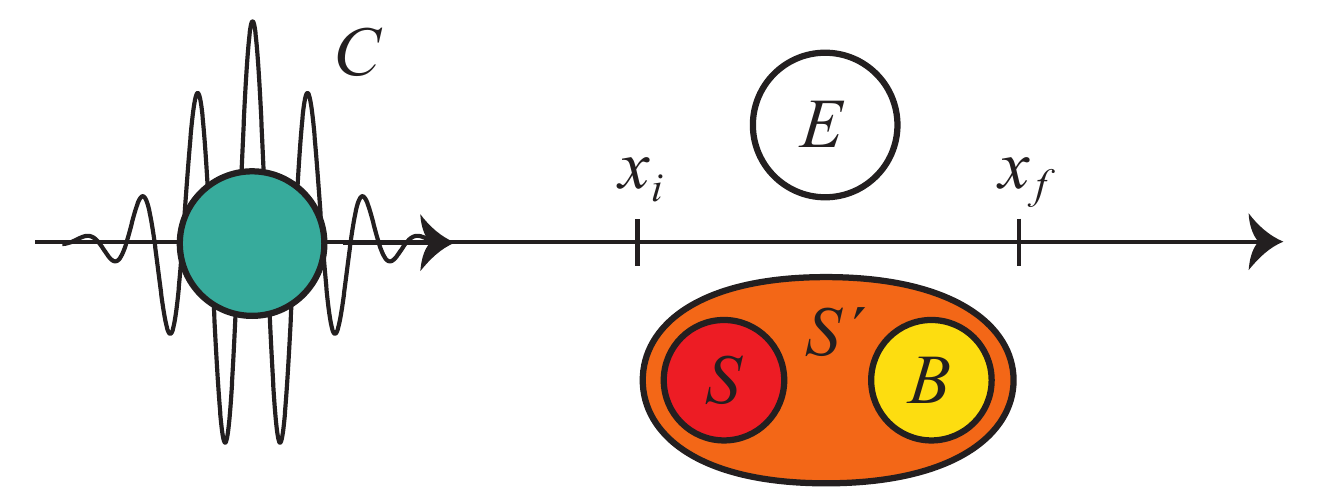} 
   \caption{\label{FigApproxC}  {\bf A control particle.}
One way to `quantize' the control parameter $x$ of the classical Crooks relation would be to regard it as the position of a quantum particle. The propagation of the particle would approximately implement the time dependent Hamiltonian for suitable wave packets. 
 A particularly clean example is obtained if $E$ and $S'$ only interact when the control particle is in an `interaction region' $[x_i,x_f]$  (cf.~the construction in \cite{Malabarba15}). Since the global Hamiltonian contains interactions between $E$ and $S'$ it does not fit with our previous classes of models, and in particular, it does not satisfy the factorization property (\ref{MainFactorization}). However, for suitable choices of measurement operators $Q_{C}^{i+}$ and $Q_C^{f+}$ that are localized outside the interaction region, the systems becomes approximately non-interacting, and the factorization property (\ref{MainFactorization}) is approximately satisfied. This makes it possible to derive an approximate conditional fluctuation relation on the energy reservoir. 
   }
\end{figure}

There is an additional reason for why it is useful to go beyond our previous settings. Namely that there exist rather evident ways to `quantize' the classical control mechanism in Crooks relation that do not fit particularly well within the machinery that we have employed so far. These quantized control mechanisms  moreover provide a good starting point for introducing approximate fluctuation relations. 

Imagine a particle whose position $x$ determines the parameter in the family of Hamiltonians $H_{S'E}(x)$. More precisely, assume a joint Hamiltonian of the form
\begin{equation}
\label{MainGlobalHamiltonianScenario1}
H = \frac{1}{2M_C}\hat{P}_{C}^2\otimes\hat{1}_{S'E} + H_{S'E}(\hat{X}_{C}),
\end{equation}
where $M_C$ is the mass, and $\hat{X}_{C},\hat{P}_{C}$ are the canonical position and momentum operators of the control particle. If the control particle is reasonably well localized in both space and momentum, then one can imagine that its propagation approximately implements the evolving control parameter $x$ in the family of Hamiltonians $H_{S'E}(x)$. For the sake of simplicity, let us assume that the control particle only affects the other systems in an `interaction region' (see Fig.~\ref{FigApproxC}) corresponding to an interval  $[x_i,x_f]$, while outside of this region it is the case that
\begin{equation}
\label{FinalInitialMain}
H_{S'E}(x) = 
 \left\{ \begin{matrix}
 H^{i}_{S'}\otimes\hat{1}_E + \hat{1}_{S'} \otimes H_E,\quad x \leq x_{i},\\
 H^{f}_{S'}\otimes\hat{1}_E + \hat{1}_{S'} \otimes H_E,\quad x \geq x_{f}.
\end{matrix}\right.
\end{equation}
Although the systems thus are non-interacting outside the interaction region $[x_i,x_f]$  one should keep in mind that  the kinetic energy term $\hat{K}  = \hat{P}_{C}^2/(2M_C)$ in (\ref{MainGlobalHamiltonianScenario1}) does not commute with $X_C$, and thus prohibits a factorization as in  (\ref{MainFactorization}). In some sense it is thus the kinetic operator $\hat{K}$ that causes the failure of the factorization property. On the other hand, it is also the kinetic operator that yields a non-trivial evolution of the control mechanism. It is this tension that we strive to handle via the approximate fluctuation relations.
The basic idea is that if the measurement operators $Q^{i+}_C$ and $Q^{f+}_C$ are well localized outside the interaction region, then the systems are approximately non-interacting, and consequently the factorization (\ref{MainFactorization}) should hold approximately. It is worth keeping in mind that even if the measurement operators are well localized in non-interacting regions, this does not exclude the possibility that the system evolves into states where the interactions are strong, even at the very moment of the control measurement. (For an explicit example, see Fig.~\ref{FigDynamics} in Appendix \ref{SecNumericalEvaluation}.)

Although the special case of the control particle provides intuition, the more general setting of the conditional fluctuation theorems yields a more concise description. 
For the chosen measurement operators  $Q_{\widetilde{S}}^{i\pm}$ and $Q_{\widetilde{S}}^{f\pm}$, let us assume that there exist local approximate Hamiltonians  $H_{\widetilde{S}}^{i}$, $H_{\widetilde{S}}^{f}$, $H_{E}^{i}$, $H_{E}^{f}$, such that the factorization holds approximately
\begin{equation}
\label{ApproxMain}
\begin{split}
 & \mathcal{J}_{\beta H}(Q_{\widetilde{S}}^{i+}\otimes Q)   \approx \mathcal{J}_{\beta H_{\widetilde{S}}^{i}}(Q_{\widetilde{S}}^{i+})\otimes \mathcal{J}_{\beta H^i_E}(Q),\\
 & \mathcal{J}_{\beta H}(Q_{\widetilde{S}}^{f+}\otimes Q)   \approx \mathcal{J}_{\beta H_{\widetilde{S}}^{f}}(Q_{\widetilde{S}}^{f+})\otimes \mathcal{J}_{\beta H^f_E}(Q).
\end{split}
\end{equation}
For time-reversal symmetric systems one can show that this leads to the approximate fluctuation relation
\begin{equation}
\label{MainApproxFluctThm}
\mathcal{Z}_{\beta H_{\widetilde{S}}^{i}}(Q_{\widetilde{S}}^{i+})\tilde{\mathcal{F}}_{+}\mathcal{J}_{\beta H^i_E} \approx \mathcal{Z}_{\beta H_{\widetilde{S}}^{f}}(Q_{\widetilde{S}}^{f-})\mathcal{J}_{\beta H^f_E}\tilde{\mathcal{F}}^{\ominus}_{-}.
\end{equation}
This can be turned into a quantitative statement where the size of the error in  (\ref{MainApproxFluctThm}) is bounded by the errors in (\ref{ApproxMain}), see  Appendix \ref{SecApproxConditional}.

\subsection{Joint control system and energy reservoir}

Compared to the control particle in Fig.~\ref{FigApproxC} there exists an even simpler setup where the control particle simultaneously serves as the energy reservoir (see Fig.~\ref{FigApproxCE}). This corresponds to the global Hamiltonian 
\begin{equation}
\label{MainGlobalHamiltonianScenario2}
H = \frac{1}{2M_{CE}}\hat{P}_{CE}^2\otimes\hat{1}_{S'} + H_{S'}(\hat{X}_{CE}),
\end{equation}
again with an interaction region
\begin{equation}
\label{MainFinalInitialCE}
H_{S'}(x) = \left\{ \begin{matrix}H^{i}_{S'},\quad x \leq x_{i},\\
 H^{f}_{S'},\quad x \geq x_{f},
\end{matrix}\right.
\end{equation}
with a non-trivial dependence on $x$ inside $[x_i,x_f]$. 
As opposed to the previous setting, we here would have to perform the control measurement on the energy reservoir itself.  This situation is conveniently described in terms of the transition probabilities discussed in section \ref{MainSecGlobal}, resulting in the approximate fluctuation relation
\begin{equation}
\label{MainApproxFluctQuantitative}
\begin{split}
 & \mathcal{Z}_{\beta H^i_{S'}}(Q_{S'}^{i})\mathcal{Z}_{\beta H^i_{CE}}(Q_{CE}^{i}) \\
&\quad \times P^{V}_{\beta H^i}[ Q_{S'}^{i+}\otimes Q_{CE}^{i+}\rightarrow Q_{S'}^{f+}\otimes Q_{CE}^{f+}] \\
&  \approx \mathcal{Z}_{\beta H^f_{S'}}(Q_{S'}^{f})\mathcal{Z}_{\beta H^f_{CE}}(Q_{CE}^{f})\\
&\quad \times P^{V}_{\beta H^f}[ Q_{S'}^{f-}\otimes Q_{CE}^{f-}\rightarrow Q_{S'}^{i-}\otimes Q_{CE}^{i-}],
\end{split}
\end{equation}
where for the particular choice of Hamiltonian (\ref{MainGlobalHamiltonianScenario2}) we have $H^i_{CE} = H^f_{CE} = \hat{P}_{CE}^2/(2M_{CE})$.
This approximate relation can also be made quantitative, see Appendix \ref{SecApproximateTransProb}. For a numerical evaluation of the errors in a concrete model, see Appendix \ref{SecNumericalEvaluation}.

\begin{figure}[t]
 \includegraphics[width= 7cm]{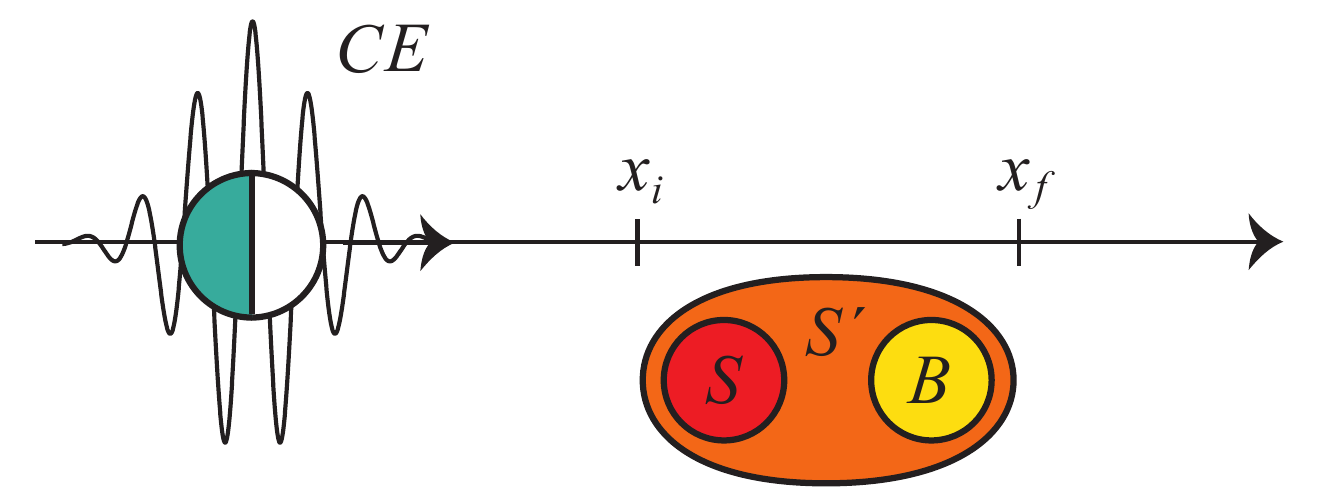} 
   \caption{\label{FigApproxCE}  {\bf A control particle that also serves as energy reservoir.}
An alternative to the setup in Fig.~\ref{FigApproxC} is a particle that simultaneously serves as both control system and energy reservoir. Hence, it is the kinetic energy of the control particle that drives the non-equilibrium process (if the other systems would start in equilibrium). This is another example of a system that would generally not satisfy the factorization property (\ref{MainFactorization}), but for measurement operators $Q_{CE}^{i+}$ and $Q_{CE}^{f+}$ that are localized outside the interaction region, the systems become approximately non-interacting. For this setup one can obtain an approximate fluctuation relation in terms of the transition probabilities introduced in section \ref{MainSecGlobal}. 
   }
\end{figure}

\section{\label{SecMainFlctnMstrEq} Fully quantum fluctuation relations for Markovian master equations}

We have step by step extended the range of applicability of the fully quantum fluctuation relations, away from the idealized setting of (\ref{QuantumCrooks}), towards the purely Hamiltonian evolution in section \ref{MainApproximate}. The general approach may nevertheless appear abstract and rather remote from the standard machinery of open quantum systems. Here we aim to include a fundamental component in the toolbox of the latter, namely master equations \cite{BreuerPetruccione}, which opens up for the application of a range of modelling techniques. As such, this extension aligns with recent efforts to bridge resource theoretic approaches to quantum thermodynamics with the master equation formalism \cite{Lostaglio17}.

So far we have explicitly included all degrees of freedom that have some role to play, including the heat bath,
which allows us to keep track of the evolution of all resources, and in particular coherences. Here we turn to the question of effective models where the heat bath is only included implicitly via master equations. Under suitable conditions, the notion of fully quantum fluctuation relations can be extended to this setting. Apart from providing a bridge to standard notions of open systems theory, an additional advantage of this generalization is that it avoids the rather extensive models that the `all inclusive' approach yields for realistic heat baths. Although the latter is no issue for the type of purely theoretical questions that we have focused on so far, it can be problematic, e.g., for numerical assessments of how good the approximations in the approximate fluctuation relations are.

One should note previous approaches to fluctuation relations for master equations \cite{Esposito06,Kawamoto11,Chetrite12,Leggio13a,Suomela14}, including unravellings and quantum jump methods 
\cite{Leggio13b,Suomela15}, and Brownian motion models \cite{Bai14}. Like for the previous sections, a prominent difference is that we here explicitly focus on the system that delivers the energy, and the changes that the thermal process induces on this energy reservoir.

\subsection{\label{SecMainTimeRevesalSymmetricThermalOperations} Fluctuation relations for time-reversal symmetric thermal operations}

Since we aim at removing the heat bath $B$ from the explicit description, we first  consider the fluctuation relation obtained on the remaining systems $SCE$. One should note that this setting, and the corresponding fluctuation relation  (\ref{ReductionHeatBath}), is closely related to a fluctuation relation obtained in \cite{Alhambra16}, with the difference that we here include time-reversals.

We assume that the heat bath is non-interacting with respect to the rest of the systems, $H = H_{SCE}\otimes \hat{1}_B + \hat{1}_{SCE}\otimes H_B$. We furthermore assume a time reversal of the form $\mathcal{T} = \mathcal{T}_{SCE}\otimes\mathcal{T}_B$ with $\mathcal{T}_{SCE}(H_{SCE}) = H_{SCE}$ and $\mathcal{T}_{B}(H_B) = H_B$. The global unitary evolution $V$ is energy conserving  $[H,V] = 0$, and time-reversal symmetric $\mathcal{T}(V) = V$.
With the heat bath initially in the equilibrium state, and with no measurements performed on it, the resulting induced channel on $SCE$ is
\begin{equation}
\label{TTO}
\mathcal{F}(\rho) = \Tr_{B}(V[\rho\otimes G_{\beta}(H_B)]V^{\dagger}).
\end{equation}
Much in line with our previous derivations, one finds  that 
\begin{equation}
\label{ReductionHeatBath}
\mathcal{F}\mathcal{J}_{\beta H_{SCE}}=  \mathcal{J}_{\beta H_{SCE}}\mathcal{F}^{\ominus}.
\end{equation}
In other words, (\ref{ReductionHeatBath}) is the fluctuation relation that would be satisfied on the joint system $SCE$, assuming an  energy conserving time-reversal symmetric dynamics. 
Due to the energy conserving dynamics with an equilibrium heat bath, the channels (\ref{TTO}) belong to the class of thermal operations \cite{Janzing00,Janzing06,Horodecki11,Brandao13b,Gour,Brandao13,Renes14,Faist15,Lostaglio14b,Perry16,Scharlau16,Lostaglio16}, although we here additionally require the time-reversal symmetric implementation described above.
In the following we refer to this as `time-reversal symmetric thermal operations'. As one might expect, time-reversal symmetric thermal operations form a proper subset to the set of thermal operations (see Appendix \ref{SecComp}).

\subsection{\label{SecMainTurningTables}Yet another extension: Assuming a global fluctuation relation}

In the previous section we found that a time-reversal symmetric energy conserving dynamics with a thermal heat bath leads to the fluctuation relation  (\ref{ReductionHeatBath}). In the following we shall consider a generalization where we turn things around and instead \emph{assume} that the dynamics on $SCE$ is such that it satisfies the relation (\ref{ReductionHeatBath}).  Hence, we enlarge the set of channels that we allow for, beyond the set of time-reversal symmetric thermal operations. (It is indeed an extension, see Appendix \ref{SecComp}.) It turns out that many (but not all) of the results of the previous sections can be regained for any channel $\mathcal{F}$ on $SCE$ that satisfies (\ref{ReductionHeatBath}), for some Hermitian operator $H_{SCE}$ and time-reversal $\mathcal{T}_{SCE}$ such that $T_{SCE}(H_{SCE}) = H_{SCE}$.

In particular, with this extended assumption as starting point, and assuming a non-interacting Hamiltonian $H_{SCE} = H_{SC}\otimes\hat{1}_E + \hat{1}_{SC}\otimes H_E$, and $\mathcal{T}_{SCE} = \mathcal{T}_{SC}\otimes\mathcal{T}_E$ with $\mathcal{T}_{SC}(H_{SC}) = H_{SC}$ and $\mathcal{T}_E(H_E) = H_E$, we can  re-derive the conditional fluctuation relations in 
section \ref{MainConditional} (see Appendix \ref{SecGlobalChnlIdeal}). In other words,
\begin{equation}
\label{fbfbbfsMain}
\mathcal{Z}_{\beta H_{SC}}(Q^{i}_{SC})  \tilde{\mathcal{F}}_{+} \mathcal{J}_{\beta H_{E}} = \mathcal{Z}_{\beta H_{SC}}(Q^{f}_{SC})\mathcal{J}_{\beta H_{E}}\tilde{\mathcal{F}}^{\ominus}_{-},
\end{equation}
holds for the CPMs
\begin{equation}
\label{euezuie}
\begin{split}
\tilde{\mathcal{F}}_{+}(\sigma) = & \Tr_{SC}\Big([\hat{1}_{E}\otimes Q^{f+}_{SC}]\mathcal{F}\big(\sigma\otimes\mathcal{G}_{\beta H_{SC}}(Q^{i+}_{SC}) \big)\Big),\\
\tilde{\mathcal{F}}_{-}(\sigma) = & \Tr_{SC}\Big([\hat{1}_{E}\otimes Q^{i-}_{SC}]\mathcal{F}\big(\sigma\otimes \mathcal{G}_{\beta H_{SC}}(Q^{f-}_{SC}) \big)\Big).
\end{split}
\end{equation}
Alternatively one can use the formulation in terms of transition the probabilities in section \ref{MainSecGlobal}
\begin{equation}
\label{nbmvbmn}
\begin{split}
 & \mathcal{Z}_{\beta H_{SC}}(Q^{i}_{SC})\mathcal{Z}_{\beta H_E}(Q^{i}_E)\\
& \quad\times  P^{\mathcal{F}}_{\beta H_{SCE}}[Q^{i+}_{SC} \otimes Q^{i+}_E\rightarrow Q^{f+}_{SC} \otimes Q^{f+}_E]\\
& =   \mathcal{Z}_{\beta H_{SC}}(Q^{f}_{SC})\mathcal{Z}_{\beta H_E}(Q^{f}_E)\\
& \quad\times P^{\mathcal{F}}_{\beta H_{SCE}}[Q^{f-}_{SC} \otimes Q^{f-}_E\rightarrow Q^{i-}_{SC} \otimes Q^{i-}_E].
\end{split}
\end{equation}
In the spirit of section \ref{MainApproximate}, one can also obtain approximate versions of these fluctuation relations, with error bounds (see Appendix \ref{GlobalFlctnThmApproximate}).

\subsection{\label{SecMainMarkovian} A condition on generators}
The primary reason for why we consider the extension provided by (\ref{ReductionHeatBath}) is that it can be translated to a convenient condition on the generators of Markovian master equations. 
Here we consider master equations $\frac{d}{dt}\mathcal{F}_t = \mathcal{L}\mathcal{F}_t$ with $\mathcal{F}_{0} = \mathcal{I}$, where the generator $\mathcal{L}$ is time-independent, and can be written on the Lindblad form 
\begin{equation}
\label{nvsfklbnmain}
\begin{split}
\mathcal{L}(Q) =&  -\frac{i}{\hbar}[H,Q] \\
& +  \sum_{k}L_k Q L_{k}^{\dagger} -\frac{1}{2}\sum_{k}L_{k}^{\dagger}L_{k}Q - \frac{1}{2}Q \sum_{k}L_{k}^{\dagger}L_{k},
\end{split}
\end{equation}
where $H$ is a Hermitian operator and $L_k$ are general operators, which guarantee trace preservation and complete positivity of the solution $\mathcal{F}_t$ \cite{Lindblad76,Gorini76}.

If we assume that the generator $\mathcal{L}$ satisfies 
\begin{equation}
\label{GeneratorRelationMain}
\mathcal{L}\mathcal{J}_{\beta H_{SCE}}=  \mathcal{J}_{\beta H_{SCE}}\mathcal{L}^{\ominus}.
\end{equation} 
for some Hermitian $H_{SCE}$  and some time-reversal $\mathcal{T}_{SCE}$, then it follows that the solution $\mathcal{F}_t = e^{t\mathcal{L}}$ satisfies (\ref{ReductionHeatBath}) for each time $t\geq 0$. (See Appendix \ref{ExamplesSatisfyingCond} for some examples of generators that satisfy the type of condition in (\ref{GeneratorRelationMain}).)
Consequently, (\ref{GeneratorRelationMain}) guarantees that we can apply the observations in the previous section. 
Hence, in the non-interacting case, the conditional evolution on system $E$ satisfies the conditional fluctuation relation (\ref{fbfbbfsMain}), or equivalently the global evolution satisfies (\ref{nbmvbmn}). In other words, under the condition that the generator of the Markovian evolution satisfies (\ref{GeneratorRelationMain}), we regain the results from sections \ref{MainConditional} and \ref{MainSecGlobal}, and moreover one can also regain the notion of approximate fluctuation relations of section \ref{MainApproximate}.

The relation (\ref{GeneratorRelationMain}) is,  up to the application of time-reversals, similar to quantum detailed balance for master equations (see e.g.~Definition 2 in \cite{Temme13}). In Appendix \ref{SecModelForThrmls} we consider a simple model of  thermalization, where the generator satisfies  (\ref{GeneratorRelationMain}) if the transition rates of the diagonal elements of the density matrix satisfy classical detailed balance.

\subsection{\label{SecMainGluingGen} Constructing generators}

When constructing models we may need to combine different components, e.g., a generator that models thermal relaxation of the system, and another that affects the energy reservoir, as well as some interaction between the two. In order to apply our machinery, we need to know that the total generator satisfies (\ref{GeneratorRelationMain}). It turns out that one indeed can construct such generators in a systematic manner.

Suppose  that we have two different subsystems with generators $\mathcal{L}_1$ and $\mathcal{L}_2$. Moreover, suppose that $\mathcal{L}_1\mathcal{J}_{\beta H_{1}}=  \mathcal{J}_{\beta H_{1}}\mathcal{L}_1^{\ominus}$ and $\mathcal{L}_2\mathcal{J}_{\beta H_{2}}=  \mathcal{J}_{\beta H_{2}}\mathcal{L}_2^{\ominus}$ with respect to some Hermitian operator $H_1$ and $H_2$, respectively. The generator $\mathcal{L}_1\otimes \mathcal{I}_2 + \mathcal{I}_1\otimes\mathcal{L}_2$ corresponds to  an independent evolution of the two systems. However, let us assume that there additionally is an interaction Hamiltonian $H_{\mathrm{int}}$, with corresponding generator $\mathcal{L}_{\mathrm{int}}(\rho) = -\frac{i}{\hbar}[H_{\mathrm{int}},\rho]$.
If $[H_{\mathrm{int}},  H_1\otimes\hat{1}_2 + \hat{1}_1\otimes H_2] = 0$, then it turns out that $\mathcal{L} =  \mathcal{L}_1\otimes \mathcal{I}_2 + \mathcal{I}_1\otimes\mathcal{L}_2 + \mathcal{L}_{\mathrm{int}}$ satisfies (\ref{GeneratorRelationMain}) with respect to $H_{SCE} = H_1\otimes\hat{1}_2 + \hat{1}_1\otimes H_2$ (see Appendix \ref{SecAssemblingGenerators} for details). In other words, we can use suitable interaction Hamiltonians in order to `glue' local generators in such a way that the result satisfies (\ref{GeneratorRelationMain}). We make use of this technique for the examples in sections \ref{SecMainTwoCoupledSpins}, \ref{SecMainQEFT}, and \ref{SecMainJCwithDissipation}.

\subsection{\label{SecMainDecoupling} Decoupling again}

A property that is no longer guaranteed to be true when allowing for all generators that satisfy (\ref{GeneratorRelationMain}) is the decomposition of the fluctuation relation into modes of coherence that was discussed in section \ref{MainSecDiagonalOffdiagonal} and in Appendix \ref{SecDiagonalMeasuremnts}. However, if the generator $\mathcal{L}$ in addition is time-translation symmetric  \cite{Lostaglio17}, then the  decomposition is regained (see Appendix \ref{SecDecouplingDiagAgain}). Moreover, analogous to the gluing of generators that satisfy (\ref{GeneratorRelationMain}), one can combine generators that satisfy time-translation symmetry. We will apply these observations in section \ref{SecMainJCwithDissipation}.

\subsection{\label{SecMainGenThrmOp}Generators of thermal and time-reversal symmetric thermal operations}
The condition (\ref{ReductionHeatBath}) for channels on $SCE$, and (\ref{GeneratorRelationMain}) for generators, have the advantage that they increase the range of applicability of the fully quantum fluctuation theorems. However, a drawback is that it is unclear what these conditions implicitly assume concerning the global initial state, as well as the evolution, on the complete system $SCEB$. Consequently, it is also unclear what these conditions imply concerning the requirements for initial resources and their evolution. Although clarifying such implications could be a subject of future studies, an alternative approach for gaining better control would be to instead impose stricter conditions than (\ref{GeneratorRelationMain}). In view of section \ref{SecMainTimeRevesalSymmetricThermalOperations}, 
a reasonable requirement would be that the generators induce time-reversal symmetric thermal operations, i.e., that $e^{t\mathcal{L}}$ is a time-reversal symmetric thermal operation for each $t\geq 0$. Although one indeed can find such generators (see Appendix \ref{SecGeneratorThrmOp}) it appears more tractable to drop the requirement of time-reversal symmetry, and only require generators of thermal operations (see Appendix \ref{SecGeneratorThrmOp}). Analogous to the gluing of generators that satisfy (\ref{GeneratorRelationMain}), described in section \ref{SecMainGluingGen}, it turns out that one also can glue generators of thermal operations (see Appendix \ref{SecGeneratorThrmOp}). One can apply these concepts to the model in the next section.

\subsection{\label{SecMainTwoCoupledSpins}Two coupled thermalizing spins}
As an illustration of the concepts introduced in the previous sections we here consider a model consisting of two two-level system, e.g.,  two resonantly coupled spins, where one spin acts as the energy reservoir of the other. More precisely, the spins have the Hamiltonians $H_{1} = \frac{1}{2}E\sigma_{z1}$ and $H_{2} = \frac{1}{2}E\sigma_{z2}$. We let $\mathcal{T}_1$ and $\mathcal{T}_2$ be the transposes in the eigenbasis of $\sigma_{z1}$ and $\sigma_{z2}$, respectively.
The spins are furthermore affected by a heat bath, and on each separate spin we assume that the resulting open system evolution is obtained via the generators
\begin{equation}
\label{MainTwoSpin}
\begin{split}
\mathcal{L}_1(\rho) = & -\frac{iE}{2\hbar}[\sigma_{z1},\rho] +r_1G_{\beta}(H_1)\Tr(\rho) -r_1\rho,\\
\mathcal{L}_2(\rho) = & -\frac{iE}{2\hbar}[\sigma_{z2},\rho] +r_2G_{\beta}(H_2)\Tr(\rho) -r_2\rho.
\end{split}
\end{equation}
The master equation corresponding to each of these generators drives the systems toward the Gibbs states $G_{\beta}(H_1)$ and $G_{\beta}(H_2)$, respectively. (For further details on this example, see  Appendix \ref{DissipativeSpins}.)
Moreover, it can be shown (see Appendix \ref{SecModelForThrmls}) that each of these generators separately satisfies  (\ref{GeneratorRelationMain}) with respect to $H_1$ and $H_2$, respectively.
Let us now further assume that the spins interact via the Hamiltonian $H_{\mathrm{int}} = \lambda |01\rangle\langle 10| + \lambda|10\rangle\langle 01|$ (where $|0\rangle$ and $|1\rangle$ are the eigenstates of $\sigma_z$). By construction,  $H_{\mathrm{int}}$ commutes with $H_1\otimes\hat{1}_2 + \hat{1}_1\otimes H_2$, and by the technique outlined in section \ref{SecMainGluingGen}, it follows that $\mathcal{L} = \mathcal{L}_1\otimes\mathcal{I}_2 + \mathcal{I}_1\otimes\mathcal{L}_2 + \mathcal{L}_{\mathrm{int}}$ satisfies (\ref{GeneratorRelationMain}) with respect to $H_{SCE} =  H_1\otimes\hat{1}_2 + \hat{1}_1\otimes H_2$. Consequently, the reduced conditional dynamics (\ref{euezuie}) on one of the spins satisfy the conditional  fluctuation relation (\ref{fbfbbfsMain}).

It turns out that the two generators in (\ref{MainTwoSpin}) are generators of thermal operations. Moreover, the interaction is such that the gluing mentioned in the previous section is applicable (see Appendix \ref{DissipativeSpins} for details).  Hence, the generator $\mathcal{L}$ not only satisfies  (\ref{GeneratorRelationMain}), but is in addition a generator of  thermal operations.

As an additional example we do in Appendix \ref{TwoSpinGlobalThrm} consider a system of two weakly coupled spins affected by a global thermalization. This system satisfies the approximate fluctuation relation developed in Appendix \ref{GlobalFlctnThmApproximate}. In Appendix \ref{SecMoreWidelyAppl} we also discuss the prospects of finding more widely applicable approximate fluctuation relations.

\subsection{\label{SecMainQEFT} Free energy differences}

As an application we here discuss the estimation of free energy differences, via a quantum generalization of the notion of extended fluctuation relations (EFRs) \cite{Maragakis08,Junier09,Alemany12,Liphardt12,Roldan14,Gavrilov17}.

 Free energy differences are traditionally measured via quasi-static changes of external parameters, in which case we can identify the resulting work cost with the change of free energy. A prominent feature of classical fluctuation relations is that they offer alternative means to determine free energy differences for arbitrary driving forces \cite{Bustamante05,Rondoni07,ReviewJarzynski,ReviewSeifert,ReviewMarconi,ReviewFluctThm}.    That our quantum fluctuation relations also can be used for this purpose may have been slightly obscured by the fact that we phrase our relations in terms of partition functions rather than free energies. However, since the (equilibrium) free energy $F(H)$ is related to the partition function via $F(H) = -kT\ln Z(H)$, it follows that the change of free energy directly corresponds to the quotient of partition functions via $F(H^f)-F(H^i) = -kT\ln[Z(H^f)/Z(H^i)]$.

\begin{figure}
 \includegraphics[width= 8cm]{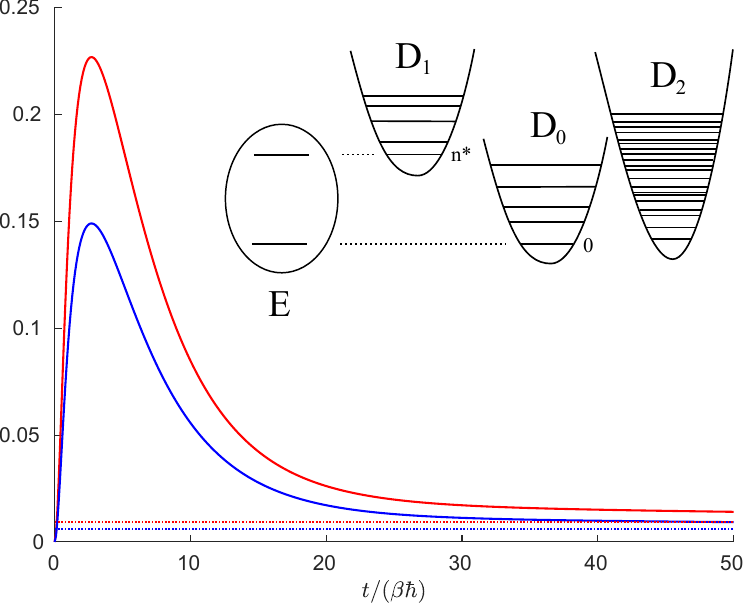} 
   \caption{\label{FigQEFT2}  
{\bf Free energy difference.}   We wish to determine the free energy difference between two basins $D_0$ and $D_1$ of states, embedded in a larger collection of configurations $D_2$. Within each such collection there is a fast thermalization process, while there is a slow global thermalization. 
An external energy reservoir $E$, in the form of a single two-level system, is in resonance with the transition between the global ground state $|0\rangle$ in basin $D_0$ and the local ground state $|n^{*}\rangle$ in basin $D_1$. 
The energy reservoir acts as the counterpart to the external control in the classical EFRs. 
The transition probability $\mathcal{P}^{+}$ (red solid curve) of the forward process,  and the transition probability $\mathcal{P}^{-}$ (blue solid curve) of the reverse process, as defined by (\ref{sdvlkvsdmain}), are plotted as functions of time, here expressed in a unit-free manner via $t/(\beta \hbar)$. The forward process is defined by $Q^{i+}_E = |1\rangle\langle 1|$ and $Q^{f+}_E = \hat{1}_E$. The transition probabilities $\mathcal{P}^{+}$ and $\mathcal{P}^{-}$ can be estimated by repeated experiments,
and determine, via (\ref{ndfklbklnbfdmain}), the quotient $Z_1/Z_0$ of the partition functions $Z_0$ and $Z_1$ of $D_0$ and $D_1$, respectively, and thus also the desired free energy difference. The dotted lines correspond to the transition probabilities in the limit of infinite evolution times, where the system has reached the fixpoint of the master equation. 
}
\end{figure}

By inspection one can realize that the quantum Crooks relation (\ref{QuantumCrooks}), as well as the quantum Jarzynski equality (\ref{fsjaflkslk}), in principle could be used to determine $Z(H^f)/Z(H^i)$, and thus the free energy difference, although as discussed in previous sections, the assumptions behind these relations are rather idealized.  (For a discussion on how one can determine approximate  free energy differences via the approximate fluctuation relations, see Appendix \ref{SecApproxFreeEnergyDiff}.)  However, the approach via master equations also allows us to determine these quantities, and as an application we here consider a variation on this theme, where we in addition take the opportunity to generalize the classical notion of extended fluctuation relations \cite{Maragakis08,Junier09,Alemany12,Liphardt12,Roldan14,Gavrilov17} to the quantum regime.

Extended fluctuation relations describe the transitions between meta-stable regions of configuration space, where we can associate a free energy to each such region \cite{Maragakis08,Junier09,Alemany12,Liphardt12,Roldan14,Gavrilov17}. By application of external forces, the system can be pushed between these meta-stable regions, and by recording the work cost, the EFRs can be used to determine the free energy difference. An example is the unfolding and refolding of DNA-strings by optical tweezers \cite{Alemany12}.

To mimic the classical setup of meta-stable configurations we partition the eigenstates of the Hamiltonian $H_{SC}$ into groups of states. One such collection, $D_0$, contains the ground state, and $D_1$ is the desired collection of target states, with the corresponding projectors $P_0$ and $P_1$. The set $D_2$ consists of all the remaining states. We wish to determine the quotient between the partition functions $Z_0 = \mathcal{Z}_{\beta H_{SC}}(P_0)$ and $Z_1 = \mathcal{Z}_{\beta H_{SC}}(P_1)$, and thus the free energy difference.

We model the evolution via a Markovian master equation where the generator $\mathcal{L}$ contains components that causes a fast thermalization within the groups of states, as well as a slow global thermalization, thus representing the meta-stability. We include an energy reservoir in the form of a single two-level system, and a resonant interaction that induces transitions between the ground state and the `local ground state' in the desired target basin. See Fig.~\ref{FigQEFT2} for a schematic illustration. The details of the model can be found in Appendix \ref{SecQEFT}, where the main observation is that the global generator  satisfies  (\ref{GeneratorRelationMain}), and thus all the induced channels $\mathcal{F}_{t}$ satisfy (\ref{nbmvbmn}).

The fluctuation relation (\ref{nbmvbmn}) can be rewritten as
\begin{equation}
\label{ndfklbklnbfdmain}
\frac{Z_1}{Z_0} = \frac{\mathcal{Z}_{\beta H_{E}}(Q^{i+}_E)}{\mathcal{Z}_{\beta H_{E}}(Q^{f-}_E)}\frac{\mathcal{P}^{+}(t)}{\mathcal{P}^{-}(t)},
\end{equation}
where the transition probabilities are
\begin{equation}
\label{sdvlkvsdmain}
\begin{split}
\mathcal{P}^{+}(t) =&  P^{\mathcal{F}_t}_{\beta H_{SCE}}[P_1\otimes Q^{i+}_E \rightarrow P_2\otimes Q^{f+}_E],\\
\mathcal{P}^{-}(t) = & P^{\mathcal{F}_t}_{\beta H_{SCE}}[P_2\otimes Q^{f-}_E \rightarrow P_1\otimes Q^{i-}_E].
\end{split}
\end{equation}
and $\mathcal{F}_t = e^{t\mathcal{L}}$ are the induced channels. 
Hence, in order to determine the desired quotient $Z_1/Z_0$ we have to repeat the forward experiment to estimate the transition probability $\mathcal{P}^{+}$, as well as the reverse experiment to estimate $\mathcal{P}^{-}$. As long as we know the values of $\mathcal{Z}_{\beta H_{E}}(Q^{i+}_E)$ and $\mathcal{Z}_{\beta H_{E}}(Q^{f-}_E)$ we can thus obtain $Z_1/Z_0$ via (\ref{ndfklbklnbfdmain}).
Another way to phrase (\ref{ndfklbklnbfdmain}) is to say that $\mathcal{P}^{+}$ and $\mathcal{P}^{-}$ are proportional for all times (cf.~the red and blue curve in Fig.~\ref{FigQEFT2}), and that this proportionality can be used determine the desired quotient.

We are in principle free to choose at which time to evaluate the transition probabilities, as well as  the initial state and control measurements on the energy reservoir. However, some choices may result in  low transition probabilities, e.g., if the initial state does not contain sufficient of energy to reach the desired excited basin $D_1$. Keeping this observation in mind, we do for the calculation of Fig.~\ref{FigQEFT2} choose $Q^{i+}_E = |1\rangle\langle 1|$ and $Q^{f+}_E = \hat{1}_E$. In other words, for the forward process, when the system starts in the conditional equilibrium of the ground state basin $D_0$, we let the energy reservoir start in the excited state $|1\rangle$. By this arrangement we compensate for the low initial energy in the system, with a high initial energy in the energy reservoir. For the reverse process this translates to the energy reservoir initially being in its equilibrium state, and the process being conditioned on the energy reservoir at the end being found in the excited state. However, in this case the system is initially in the excited basin $D_1$, and we only wish to reach the ground state basin $D_0$. One may thus expect that this arrangement could result in reasonable transition probabilities, which also is confirmed by Fig.~\ref{FigQEFT2}, where both the forward and reverse process yield transient transition probabilities that reach beyond their long term limits.

One may observe the correspondence with the classical scenario. There we also need some method to detect the transitions between the two relevant basins, as well as some means to keep track of the work implicitly provided by the external controls. Moreover, without the work input from the external control, we would have to wait passively to observe fortuitous thermal fluctuations that result in transitions between the two desired basins.

\subsection{\label{SecMainJCwithDissipation}  Jaynes-Cummings with dissipation: Fluctuation relations for coherences}

\begin{figure}
 \includegraphics[width= 8cm]{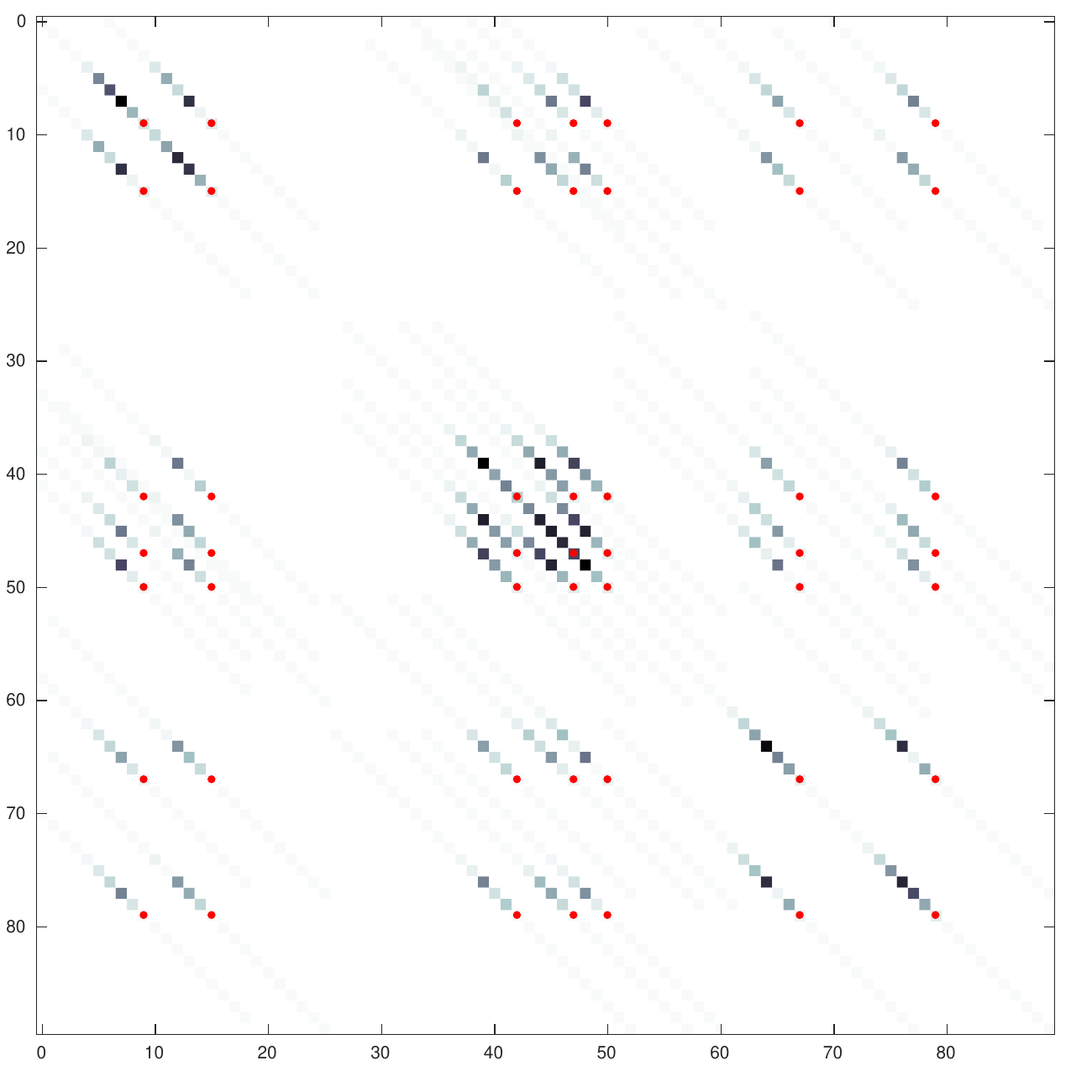} 
   \caption{\label{FigDecoupling}
  {\bf Decoupling of the modes of coherence in a dissipative Jaynes-Cummings model.}   
A harmonic oscillator, with energy eigenbasis $|n\rangle$, serves as the energy reservoir in interaction with an open two-level system. As the initial state of the energy reservoir we choose the (rather arbitrary) superposition of number states $|\psi\rangle = (|9\rangle + |15\rangle + |42\rangle +|47\rangle +|50\rangle + |67\rangle +|79\rangle)/\sqrt{7}$. The red dots represent the non-zero matrix elements of the density operator $|\psi\rangle\langle\psi|$ with respect to the number basis. The conditional evolution $\tilde{\mathcal{F}}_{+}$ on the energy reservoir, defined in (\ref{euezuie}), is calculated for the evolution time $t/(\beta\hbar) =1.5$, and the dark colors correspond to the values $|\langle n|\tilde{\mathcal{F}}_{+}(|\psi\rangle\langle\psi|)|n'\rangle|$.
The upper left corner, $n = n' = 0$, corresponds the probability to find the oscillator in the ground state, and the main diagonal $n = n'$ the probabilities for the excited states. The off-diagonal elements, $n \neq n'$, represent the coherences between the energy eigenstates. One can recognize the leakage of energy out of the oscillator, as well as the decay of coherence. It is also clearly visible how each initial off-diagonal element only evolves and spreads along the particular diagonal that it belongs to, thus illustrating the decoupling of the modes of coherence.
}
\end{figure}

Here we illustrate the conditional fluctuation theorem (\ref{fbfbbfsMain}), with focus on the decoupling into modes of coherence, and the fluctuation relations for coherences.

The Jaynes-Cummings (JC) model \cite{Jaynes63,Shore93} of two-level systems interacting with harmonic oscillators is a common approach to study atom-field interactions. This can be  further generalized to master equations, including various open system effects, such as dissipation and decoherence. Such types of models have been considered, e.g. for superconducting qubits interacting with field modes \cite{Gambetta08, Schmidt10,Houck12}, electron spins coupled to nano-mechanical vibrations \cite{Struck14,Palyi12}, and a nano-mechanical oscillator interacting with a Cooper pair box \cite{Tiwari08}. Here we consider a particular case of this general class of models, where we let a single field mode serve as the energy reservoir, and where this interacts with a single two-level system, which in turn is affected by thermalization and decoherence.

We let the Hamiltonian of the two-level system be $H_{SC} = \frac{1}{2}E\sigma_z = -\frac{1}{2}E|0\rangle\langle 0| +\frac{1}{2}E|1\rangle\langle 1|$, and let the Hamiltonian of the energy reservoir be $H_{E} = Ea^{\dagger}a$, with $E>0$. We also  include an interaction Hamiltonian of the form $H_{\textrm{int}} = \lambda|0\rangle\langle 1|\otimes a^{\dagger} + \lambda|1\rangle\langle 0|\otimes a$. 
The two latter Hamiltonians correspond to the generators $\mathcal{L}_E(\rho) =  -E\frac{i}{\hbar}[a^{\dagger}a,\rho]$ and $\mathcal{L}_{\mathrm{int}}(\rho) =  -\frac{i}{\hbar}[H_{\mathrm{int}},\rho]$.
We choose the time-reversal $\mathcal{T}_{SC}$ as the transpose with respect to the eigenbasis  $\{|0\rangle,|1\rangle\}$ of $H_{SC}$, and $\mathcal{T}_E$ is the transpose with respect to the number basis of the harmonic oscillator.

The two-level system is furthermore affected by an environment, which we model via the generator
\begin{equation}
\label{ndfbkjnjkdfb}
\begin{split}
  \mathcal{L}_{SC}(\rho) & =  -E\frac{i}{2\hbar}[\sigma_z,\rho] \\
& + re^{-\beta E/2} |1\rangle \langle 0|\rho |0\rangle\langle 1|   +  re^{\beta E/2} |0\rangle \langle 1|\rho |1\rangle\langle 0|  \\
& -\frac{1}{2}re^{-\beta E/2}|0\rangle\langle 0|\rho   -\frac{1}{2}re^{-\beta E/2}\rho|0\rangle\langle 0|\\
& -\frac{1}{2} re^{\beta E/2}|1\rangle\langle 1|\rho   -\frac{1}{2} re^{\beta E/2}\rho|1\rangle\langle 1|\\
&  + \kappa \sigma_z\rho\sigma_z -\kappa \rho,
\end{split}
\end{equation}
where  $\sigma_{+} =|1\rangle\langle 0|$ and $\sigma_{-} =|0\rangle\langle 1|$.
The first line in  (\ref{ndfbkjnjkdfb}) corresponds to the Hamiltonian evolution on the spin, while the next three lines model thermalization with rate $r>0$ (see Appendices \ref{SecModelForThrmls} and \ref{SecGenForThrmOp}), and the last line corresponds to an additional decoherence of rate $\kappa >0$ with respect to the energy eigenbasis (Appendix \ref{SecModelDecoherence}). The total generator is 
$\mathcal{L} = \mathcal{L}_{SC}\otimes\mathcal{I}_E + \mathcal{I}_{SC}\otimes\mathcal{L}_{E} + \mathcal{L}_{\mathrm{int}}$.
 For the particular case illustrated in figures \ref{FigDecoupling} and \ref{FigJCcoff} we choose the parameters such that $\beta E = 1$, $r\beta \hbar = 1$, $\kappa \beta\hbar = 0.1$, $\lambda\beta = 1$.

One can show (see Appendix \ref{DetailsJC}) that $\mathcal{L}$ satisfies (\ref{GeneratorRelationMain}) with respect to $H_{SC}\otimes\hat{1}_E + \hat{1}_{SC}\otimes H_{E}$, and thus the reduced conditional dynamics on the energy reservoir satisfies the conditional fluctuation relation (\ref{fbfbbfsMain}). 
Moreover, the generator $\mathcal{L}$ satisfies the time-translation symmetry mentioned in section \ref{SecMainDecoupling}, and thus $\tilde{\mathcal{F}}_{\pm}$ are decoupled into modes of coherence, if $Q^{i\pm}_{SC}$ and $Q^{f\pm}_{SC}$ commute with $H_{SC}$ (see  Appendix  \ref{DetailsJC}).
This is illustrated in Fig.~\ref{FigDecoupling}, where we choose $Q^{i+}_{SC} = |0\rangle\langle 0|$ and $Q_{SC}^{f+} = 0.05|0\rangle\langle 0| + |1\rangle\langle 1|$. (The term $0.05|0\rangle\langle 0|$ is only there in order to avoid that the red and blue curves in Fig.~\ref{FigJCcoff} overlap.)

Analogous to what we found in section \ref{MainSecDiagonalOffdiagonal} (and  Appendix \ref{SecDiagonalAndOffdiagonal})  each mode of coherence obeys a fluctuation relation of its own, implied by the general conditional fluctuation relation.
Let us now zoom in, so to speak, to a very detailed view of how (\ref{fbfbbfsMain}) constrains the dynamics of the coherences.  We  select two elements along a displaced diagonal, corresponding to $|e_0 \rangle\langle e_0 + d|$ and $| e_0 + \delta \rangle\langle e_0 + \delta +d|$. In other words, we select two elements along the diagonal with offset $d$, which are separated by the  energy difference $\delta$. The relation (\ref{fbfbbfsMain}) yields 
\begin{equation}
\label{mainfhdjbhfd}
\begin{split}
&\mathcal{Z}_{\beta H_{SC}}(Q^{i}_{SC})q_{+} = e^{-\beta E\delta}\mathcal{Z}_{\beta H_{SC}}(Q^{f}_{SC})q_{-},\\
& q_{+} =  \langle e_0 + \delta| \tilde{\mathcal{F}}_{+} \big(|e_0 \rangle\langle e_0+ d|\big)|e_{0} +\delta +d\rangle,\\
& q_{-} =\langle e_0|\tilde{\mathcal{F}}_{-}\big(|e_0 + \delta\rangle\langle e_0 + \delta + d|\big) |e_0 + d\rangle.
\end{split}
\end{equation}
 In other words, $q_{\pm}$ are the `amplitudes' for the transitions between the elements along the displaced diagonal. Equation  (\ref{mainfhdjbhfd}) predicts that these, generally complex, amplitudes are strictly related for all times, as is illustrated in Fig.~\ref{FigJCcoff}, where we have chosen $e_{0} = 20$, $d=  3$, and $\delta = -1$. 
Hence, $q_{+} =  \langle 19| \tilde{\mathcal{F}}_{+} (|20 \rangle\langle 23|)|22\rangle$ and 
$q_{-} =\langle 20|\tilde{\mathcal{F}}_{-}(|19\rangle\langle 22|) |23\rangle$.
It should be emphasized that neither does  $|20 \rangle\langle 23|$ correspond to a proper state, nor does $|22\rangle\langle 19|$ correspond to a POVM element.  Nevertheless, $q_{+}$ can be determined via a sufficient number of expectation values measured on the output, for a sufficiently large collection of different input states (see a similar discussion in Appendix \ref{Secoffdiagonal}).

It may be worth noting that the relation (\ref{mainfhdjbhfd})  remains valid even in the case when there is no decoupling between the diagonals. The effect of the decoupling is rather to turn `cross-mode relations' trivial. We discuss this further in Appendix \ref{DetailsJC}.

\begin{figure}
 \includegraphics[width= 8cm]{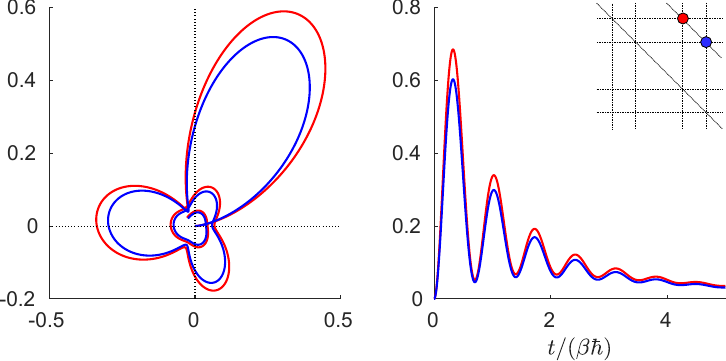} 
   \caption{\label{FigJCcoff} 
 {\bf Off-diagonal fluctuation relation.}
The fluctuation theorem (\ref{mainfhdjbhfd}) relates the evolution of the coherences carried by two off-diagonal elements along a displaced diagonal.
The plot on the right displays $|q_{+}|$ (red curve) and $|q_{-}|$ (blue curve) plotted as functions of time. The plot on the left depicts the trajectories that $q_{+}$ (red curve) and $q_{-}$ (blue curve) sweep in the complex plane during the same time interval. 
As predicted by (\ref{mainfhdjbhfd}), $|q_{+}|$  and $|q_{-}|$ are proportional to each other, and the phase factors of  $q_{+}$  and $q_{-}$ are identical. Analogous to how the classical Crooks relation relates the probability distribution of the work, and thus the change of energy in the reservoir, of the forward and reverse processes, the off-diagonal Crooks relation relates the changes of coherence. 
}
\end{figure}

\section{Conclusions and outlook}

We have generalized Crooks fluctuation theorem to a genuine quantum regime that incorporates the full quantum dynamics. This leads to a decomposition into diagonal and off-diagonal Crook's relations, one for each mode of coherence. We have also derived Jarzynski equalities, and re-derived  standard bounds on the average work cost under an additional assumption of unitality of a certain induced channel. We have furthermore shown that the classical Crooks relation can be regained under the additional assumption of energy translation invariant dynamics on the energy reservoir. The general approach moreover leads to the concept of conditional fluctuation relations, where a pair of measurement operators characterizes the initial state and the final measurement, and where the transformation from the forward to the reverse process corresponds to a transformation of the pair of measurement operators. This generalization allows for non-equilibrium initial states, and can also be extended to include correlations. We have demonstrated that by allowing for errors in the fluctuation relation, we can incorporate the `natural' setting where the global dynamics is determined by a single time-independent Hamiltonian. Finally, we have shown that the notion of fully quantum fluctuation theorems can be extended to master equations that implicitly model the influence of the heat bath.

Although we in this investigation have regained both the standard Crooks and Jarzynski relations, it would nevertheless be useful to obtain connections between the formalism of this investigation and the multitude of established quantum fluctuation relations (see e.g. the overviews in \cite{Esposito09,Campisi11a,Hanggi15}).

For the fluctuation relations for master equations we have here focused entirely on the time-independent Markovian case. It seems likely that the approach could be generalized to time-dependent generators, and it would be interesting to consider the extension to non-Markovian master equations. 
One may also speculate if it would be possible to use the global Hamiltonian of the approximate fluctuation relations as a starting point to obtain a version for master equations, via a reduction of the heat bath, along the lines of standard derivations of master equations (see e.g.~\cite{BreuerPetruccione}).

On a more general level it could prove fruitful to explore the  intersection between the resource theoretic perspective and the standard machinery of open quantum systems (see \cite{Lostaglio17} for recent contributions in this direction). 
A concrete question is how to construct generators that yield thermal or time-reversal symmetric thermal operations (see discussions in Appendix \ref{SecGeneratorThrmOp}). Another question is what the condition  (\ref{GeneratorRelationMain}) implies concerning the underlying global evolution, and the evolution of resources.

Classical fluctuation relations have been subject to several experimental tests \cite{Liphardt02,Collin05,Douarche05,Harris07,Toyabe10,Gieseler14}. Various setups have been suggested for the quantum case \cite{Huber08,Heyl12,Pekola13,Dorner13,Mazzola13,Campisi13b,Campisi}, with recent experimental implementations in NMR \cite{Batalhao14}, as well as in trapped ion systems \cite{An15}. See also \cite{Saira12} for tests in a single-electron box, and \cite{Schuler05} for an experiment on a relation for non-thermal noise. The conditional fluctuation theorems, and in particular the approximate version, as well as those based on master equations, allow for a considerable flexibility, which suggests that experimental tests may be feasible. 
Since the quantum fluctuation relation (\ref{QuantumCrooks}) and the conditional version (\ref{ConditionalFluctuationThm}) are phrased in terms of channels and CPMs, one may be tempted to conclude that every test of these relations necessarily would require a full process tomography, which generally is very demanding. However, the global fluctuation relation (\ref{MainGlobal}) suggests `milder' tests based on small sets of suitable chosen measurement operators (corresponding to a partial process tomography). An example is the determination of free energy differences discussed in section \ref{SecMainQEFT}.

One could also consider the possibility to experimentally verify cases of Jarzynski equalities with and without the condition of unitality of the channel $\mathcal{R}_{+}$, such as in (\ref{nvanlkavlkn}) and 
(\ref{fsjaflkslk}). It would also be desirable to get a better theoretical understanding of the role of the unitality of the channel $\mathcal{R}_{+}$, which one may suspect is related to the energy reservoir regarded as a resource. In this context one may also ask for the general conditions for the non-violation of standard bounds, and how this relates to the energy translation invariance (cf.~discussions in \cite{Skrzypczyk14}).

In this investigation we have tacitly assumed that the heat bath can be taken as initially being in the Gibbs state. It would be desirable to let go of this assumption, e.g.~via typicality \cite{Popescu06,Goldstein06} (see further discussions in Appendix \ref{SecAddtnlRmrks}).  

On a more technical note one may observe that whenever we  actively have referred to the properties of the spectrum of the Hamiltonian of the reservoir, we have always assumed that it is a point spectrum. An  analysis that explicitly investigates the effects of reservoir spectra that contain a continuum could potentially be useful.

One can imagine several generalizations of the results in this investigation. It does for example seem plausible with a grand canonical version, thus not only including energy flows, but also the flow of particles. (For a previous grand canonical fluctuation relation, see \cite{Yi12}.) More generally one could consider  settings with multiple conserved quantities \cite{Lostaglio15d,Halpern15c,Guryanova15,PerernauLlobet15}.

Another potential generalization concerns a classical version of the conditional fluctuation relations. We have already obtained a particular class of classical conditional fluctuation theorems  (Appendix \ref{SecClassicalConditional}). However, these are classical in the sense of being diagonal with respect to a fixed energy eigenbasis, which should not be confused with a classical phase space setting. It seems reasonable that the structure of pairs of measurement operators, translated to functions over phase space, with classical counterparts for the Gibbs and partition maps, could combine with phase space flows to yield classical conditional fluctuation theorems. It may also be possible to bridge such a classical phase space approach to the quantum setting via Wigner functions and other phase space representations of quantum states. 

The intermediate fluctuation relation (\ref{Preliminary}) can be rephrased in terms of Petz recovery channel  \cite{Petz86,Petz88,Barnum02,Hayden04}, which reminds of the recent finding  in  \cite{Wehner15} that relates Petz recovery channel with work extraction. One may wonder if these results hint at a deeper relation. See further comments in Appendix \ref{SecPetzRecovery}.

Several recent contributions to quantum thermodynamics have focused on resource theories, single-shot statistical mechanics, quantum correlations, and coherence \cite{Janzing00,Brandao13b,Horodecki13,Gour,Dahlsten,delRio,Aberg12,Horodecki11,Egloff,EgloffThesis,Faist,Skrzypczyk13,Aberg13,Lostaglio15c,Lostaglio14b,Lostaglio14a,Brandao13,Cwiklinski15,Narasimhachar15,Korzekwa15,Masanes15,Wehner15,Gallego15,Woods15,Ng14,Alhambra15,Hovhannisyan13,Reeb14,Renes14,Halpern14,Kammerlander15, Gemmer15,PerernauLlobet15b,PerernauLlobet15c}. (For general overviews on recent developments in quantum thermodynamics, see  \cite{Goold15,Vinjanampathy15}.)
The fluctuation relations in this investigation are at their core statements about dynamics, rather than about resources. Nevertheless, one could consider formulating quantitative characterizations of the evolution of resources in the spirit of Crooks theorem, and our fluctuation relations may serve as a starting point for such an analysis. In this context one should note recent efforts to link single-shot quantities and fluctuation theorems \cite{Halpern15a,Salek15,Dahlsten15,Halpern15b}. The fact that the present investigation is based on energy conserving dynamics, and thus brings the notion of fluctuation theorems under the same umbrella as previous investigations on quantum thermodynamics and coherence  \cite{Janzing00,Horodecki11,Brandao13,Brandao13b,Skrzypczyk13,Aberg13,Ng14,Lostaglio14a,Lostaglio14b,Lostaglio15c,Cwiklinski15,Narasimhachar15,Korzekwa15,Masanes15,Wehner15,Gallego15,Woods15,Alhambra15},  may further facilitate the merging of these subjects.

\section*{Acknowledgements}

This work was supported by the Excellence Initiative of the German Federal and State Governments (Grants ZUK 43 \& 81), and the DFG.

The author thanks Kilian Mitterweger for discussions, and \'Edgar Rold\'an for pointing out the literature on extended fluctuation relations.






\begin{widetext}
\end{widetext}

\begin{appendix}

\section{\label{SecIdealized} An intermediate quantum fluctuation theorem}

\subsection{Setting the stage}

The standard classical Crooks theorem compares the probability distributions of the random work costs of a forward and reverse process where the system is driven by external fields. Often this  external field is taken as a parameter $x$ in a Hamilton function. The system is usually imagined to additionally interact with a heat bath of a given temperature. The time-schedule of the parameter $x$ is implemented as a function $x_t$  of time $t$, which runs from $t = 0$ to $t = T$. At $t = 0$ we assume that the initial system is in equilibrium with the heat bath. For the reverse process, the external parameter evolves as $x'_t := x_{T-t}$ for $t = 0$ to $t = T$. In other words, the time-schedule of the parameter is run in reverse. Again we assume that the system initially is in equilibrium with the heat bath, but now for the parameter value $x'_0 = x_T$. It is useful to keep in mind that these initial equilibrium distributions are conditioned on the value of the control parameter. The aim of the following sections is to make a quantum version of this classical setup.

The model consists of four components, the `system' $S$, the heat bath $B$, the control $C$, and the energy reservoir $E$. Assumptions \ref{Def1} below does not mention the system $S$ or the heat bath $B$. The reason for this is that the main part of the derivations does not require any distinction between these subsystems, so they can be regarded as one single system $S' := SB$.

For a Hamiltonian $H$ and $\beta = 1/(kT)$, for Boltzmann's constant $k$ and the absolute temperature $T$, we denote the partition function by $Z_{\beta}(H) = \Tr e^{-\beta H}$, and (assuming that $Z_{\beta}(H)$ is finite) we denote the Gibbs state by $G_{\beta}(H) = e^{-\beta H}/Z_{\beta}(H)$. Since we here only consider heat baths with one single temperature, we will often suppress the subscript and write $G(H)$ and $Z(H)$.

\begin{Assumptions}
\label{Def1}
Let $\mathcal{H}_{S'}$, $\mathcal{H}_C$, and $\mathcal{H}_{E}$ be complex Hilbert spaces. Let $|c_i\rangle,|c_f\rangle\in\mathcal{H}_C$ be normalized and orthogonal to each other, and define the projector $P_{C}^{\perp} := \hat{1}_C -|c_i\rangle\langle c_i|-|c_f\rangle\langle c_f|$. 
\begin{itemize}
\item  Let $H^{i}_{S'}$ and $H^{f}_{S'}$ be Hermitian operators on $\mathcal{H}_{S'}$, such that $Z_{\beta}(H^i_{S'})$ and $Z_{\beta}(H^f_{S'})$ are finite. (This guarantees that $G_{\beta}(H^i_{S'})$ and $G_{\beta}(H^f_{S'})$ exist.) Let $H_E$ be a Hermitian operator on $\mathcal{H}_{E}$.  Let $H^{\perp}$ be a Hermitian operator on $\mathcal{H}_{S'}\otimes\mathcal{H}_C$ such that $[\hat{1}_{S'}\otimes P_{C}^{\perp}]H^{\perp}[\hat{1}_{S'}\otimes P_{C}^{\perp}] = H^{\perp}$, and define 
\begin{equation*}
H_{S'C} :=  H^{i}_{S'}\otimes |c_i\rangle\langle c_i|+ H^{f}_{S'}\otimes |c_f\rangle\langle c_f| + H^{\perp}
\end{equation*}
and $H :=  H_{S'C}\otimes\hat{1}_E + \hat{1}_{S'C}\otimes H_E$.
\item $V$ is a unitary operator on $\mathcal{H}_{S'}\otimes\mathcal{H}_C\otimes\mathcal{H}_E$ such that $[V,H] = 0$, 
 and 
\begin{equation}
\label{PerfectControl}
 V[\hat{1}_{S'}\otimes|c_i\rangle\langle c_i|\otimes \hat{1}_{E}] = [\hat{1}_{S'}\otimes|c_f\rangle\langle c_f|\otimes \hat{1}_{E}]V.
\end{equation}
\end{itemize}
\end{Assumptions}
In the following we briefly discuss the rationale behind these assumptions. 

The Hamiltonian $H_{S'C}$ describes how the state of the control system $C$ changes the Hamiltonian of $S'$ (see Fig.~\ref{FigSupplStructure}).
Since $|c_i\rangle$ is orthogonal to $|c_f\rangle$, and these in turn are orthogonal to the support of $H^{\perp}$, it follows that if $C$ is in state $|c_{i}\rangle$, then the Hamiltonian of $S'$ is $H^{i}_{S'}$. Similarly, if $C$ is in state $|c_f\rangle$, then $S'$ has Hamiltonian $H^{f}_{S'}$.  The Hamiltonian $H^{\perp}$ allows for the possibility of having intermediate Hamiltonians between the initial and final one (see Fig.~\ref{FigSupplStructure}). The global Hamiltonian is the sum of $H_{S'C}$ and the Hamiltonian $H_E$ of the energy reservoir, which thus by construction are non-interacting. 

The global evolution is given by unitary operations that conserve energy, which here is modeled via unitary operators $V$ on $S'CE$ such that $[H,V] =0$. (For an alternative notion of energy conservation, see \cite{Skrzypczyk14}.)  In addition to being energy conserving, we also require $V$ to satisfy (\ref{PerfectControl}). In other words, $V$ should rotate the subspace 
$\mathcal{H}_{S'E}\otimes\Sp\{|c_i\rangle\}$  to the subspace $\mathcal{H}_{S'E}\otimes\Sp\{|c_f\rangle\}$.
 This models the idealization of a perfect control mechanism, meaning that the evolution with certainty will bring the initial control state $|c_i\rangle$ to the final control state $|c_f\rangle$, and thus with certainty will transform the initial Hamiltonian  $H^{i}_{S'}$ to the final Hamiltonian $H^{f}_{S'}$.
In Appendices \ref{minimalwithout} and \ref{discretizedwithout} we demonstrate that there exist setups that satisfy all conditions in Assumptions \ref{Def1}.

As mentioned above, we do not need to make a distinction between the system $S$ and the heat bath $B$ in most of these derivations. However, to obtain fluctuation relations where the partition functions only refer to system $S$, we can additionally  assume that the initial and final Hamiltonians of the system and the heat bath are non-interacting. More precisely,  we would assume that there exist Hermitian operators $H^{i}_{S}$ and $H^{f}_{S}$ on $\mathcal{H}_{S}$ and a Hermitian operator $H_B$ on $\mathcal{H}_{B}$ such that
\begin{equation}
\label{Def1HeatBath}
\begin{split}
H^{i}_{S'} = & H^{i}_{S}\otimes\hat{1}_{B} + \hat{1}_{S}\otimes H_B,\\
H^{f}_{S'} = & H^{f}_{S}\otimes\hat{1}_{B} + \hat{1}_{S}\otimes H_B,
\end{split}
\end{equation}
and such that $Z_{\beta}(H^i_{S})$, $Z_{\beta}(H^f_{S})$,  and $Z_{\beta}(H_B)$ are finite.

\begin{figure}[t]
 \includegraphics[width= 6cm]{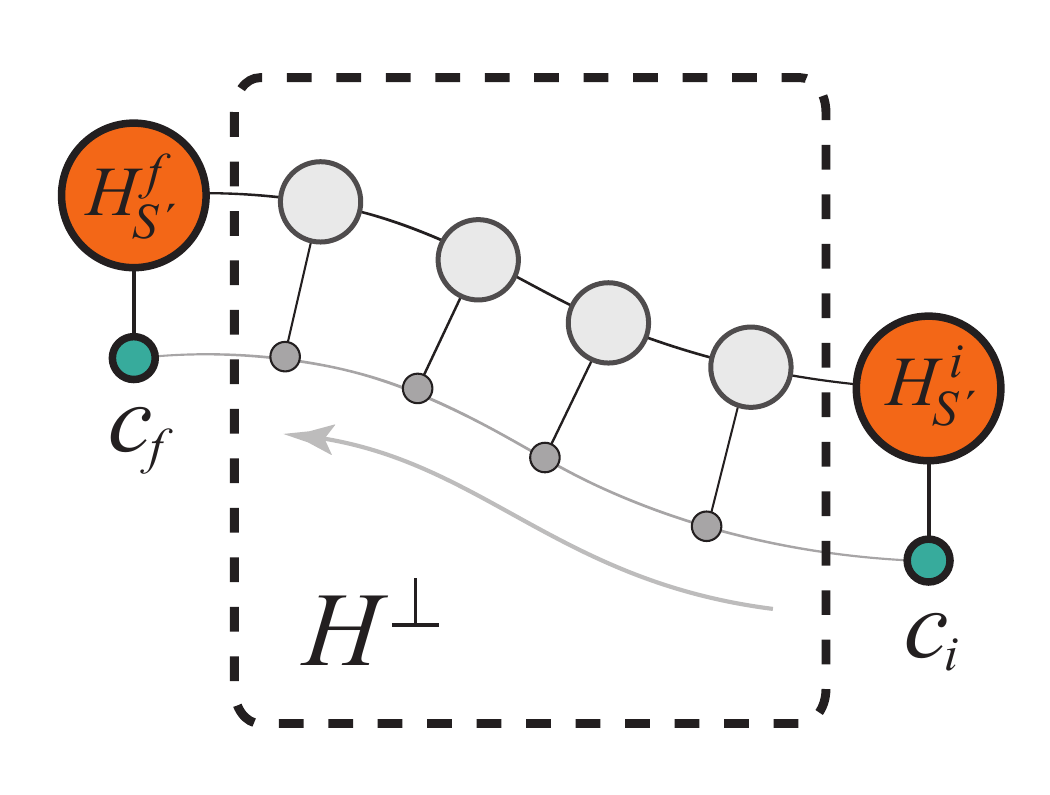} 
   \caption{\label{FigSupplStructure}  {\bf Structure of the Hamiltonian.}
The Hamiltonian of the extended system $S' = SB$ and the control has the form $H_{S'C}  = H^{i}_{S'}\otimes |c_i\rangle\langle c_i|+ H^{f}_{S'}\otimes |c_f\rangle\langle c_f| + H^{\perp}$. Here $H^{i}_{S'}$ and $H^{f}_{S'}$ are the initial and final Hamiltonians, respectively, and $|c_i\rangle, |c_f\rangle$ the corresponding orthonormal control states. If the control is in state $|c_i\rangle$ the Hamiltonian of $S'$ is  $H^{i}_{S'}$, while it is $H^{f}_{S'}$ if the control is in state $|c_i\rangle$. The Hamiltonian $H^{\perp}$, which has orthogonal support to $H^{i}_{S'}\otimes |c_i\rangle\langle c_i|$ and $H^{f}_{S'}\otimes |c_f\rangle\langle c_f|$, corresponds to possible intermediate stages. In the proofs of our fluctuation theorems $H^{\perp}$ plays no particular role. However, it can be used to simulate a path of Hamiltonians $H_{S'}(x)$, e.g., via  a  discretization (see Appendix \ref{discretizedwithout}).
   }
\end{figure}

\subsection{\label{MakeSense} The global Hamiltonian $H$ and the global evolution $V$}

In typical textbook quantum mechanics, the Hamiltonian defines the notion of energy and energy conservation, as well as being the generator of time-evolution. Here we do in some sense separate these two roles, since we let the time evolution be given by $V$ with the only restriction that it commutes with $H$, without demanding that $V = e^{-itH/\hbar}$. This separation is very convenient since it gives tractable models (compared to introducing an interaction term in the Hamiltonian and try to analyze the resulting evolution via Schr\"odinger's equation) and has successfully been employed in several previous studies \cite{Janzing00,Horodecki11,Brandao13,Brandao13b,Skrzypczyk13,Aberg13,Ng14,Lostaglio14a,Lostaglio14b,Lostaglio15c,Cwiklinski15,Narasimhachar15,
Korzekwa15,Masanes15,Wehner15,Gallego15,Woods15,Alhambra15}. 

It is maybe worth emphasizing that when we in this investigation refer to two systems as being `non-interacting', this only means that the energy observable is of the form $H = H_1\otimes\hat{1}_2 + \hat{1}_1\otimes H_2$. It does \emph{not} imply that the evolution is of a trivial product form $V_1\otimes V_2$, since (depending on the combination of the spectra of $H_1$ and $H_2$) there may exist non-trivial unitary operators $V$ that commute with $H$.

One way to understand the separation of roles between $H$ and $V$ is to imagine that the global evolution is generated by a Hamiltonian $H_{\textrm{evol}}$, i.e., that  $V = e^{-itH_{\textrm{evol}}/\hbar}$, where we let $H_{\textrm{evol}} = H + H',\quad [H,H'] = 0$. A possible justification would be if  $H'$ is `small', thus leaving $H$ as the dominant contribution to the energy. It should be emphasized that the derivations of our fluctuation relations do not require us to know how $V$ comes about, or what happens at intermediate times when the system evolves from the initial to the final state. We only need to know that $V$ commutes with $H$.

Although the above reasoning may serve as a possible justification, one may nevertheless wonder how to incorporate more `standard'  assumption that $V = e^{-itH/\hbar}$. This topic is discussed in Appendix \ref{SecApproximateFluct}. 

For a final observation concerning the structure of the energy conserving unitary operators, assume that the global non-interacting Hamiltonian $H$ (in the finite-dimensional case) is non-degenerate. Then its eigenvectors are all of the product form $|1_n\rangle|2_m\rangle$, for the two local eigenbases $\{|1_n\rangle\}_n$ and $\{|2_m\rangle\}_m$ of $H_1$ and $H_2$, respectively. Hence, all unitary operators that commute with $H$ can be written  $V = \sum_{mn}e^{i\theta_{mn}}|1_n\rangle\langle 1_n|\otimes |2_m\rangle\langle 2_m|$, for arbitrary real numbers $\theta_{mn}$. 
Although typically not product operators, these do not have the power to change the occupancies of the  product energy eigenstates. In particular, these unitaries cannot transfer energy between the subsystems, which is not particularly satisfying as a model of thermodynamic processes. 
However, if the global Hamiltonian has degeneracies due to matchings of transition energies in the local spectra,  then there exist energy conserving unitary operators that transfer energy. (For a simple example, see Appendix \ref{ExampleTwoQubits}). This matching is clearly a rather brittle assumption, and one could relax this idealization by allowing for transitions within a narrow energy shell. However, we shall not consider such generalizations in this investigation, but rather stick to perfect energy conservation. One may also note that the approximate approach, presented in Appendix \ref{SecApproximateFluct}, provides an alternative route to handle this issue. There we assume $V = e^{-itH/\hbar}$ and allow for interacting Hamiltonians, thus enforcing perfect energy conservation, as well as removing the need for matching of local spectra.

\subsection{\label{TheInitalStates} The initial states}

In the typical derivation of Crooks theorem one assumes that the system initially is at equilibrium with respect to the initial value of the control parameter. In other words, for $x = x_i$, the system should start in state $G_{\beta}\big(H_{S'}(x_i)\big)$. When one considers a  more explicit model that includes the degrees of freedom of the control system it becomes clear that the initial state of the system and control combined cannot be in a global equilibrium. 
For example, in our case the global equilibrium state on $S'C$ would be
\begin{equation}
\label{GlobalEquilibrium}
\begin{split}
G(H_{S'C}) =  & \frac{Z(H^{i}_{S'})}{Z(H_{S'C})}G(H^{i}_{S'})\otimes |c_i\rangle\langle c_i| \\
& +  \frac{Z(H^{f}_{S'})}{Z(H_{S'C})}G(H^{f}_{S'})\otimes |c_f\rangle\langle c_f|\\
& + \frac{Z(H^{\perp})}{Z(H_{S'C})}G(H^{\perp}),
\end{split}
\end{equation} 
where $Z(H_{S'C}) =  Z(H^{i}_{S'})  +  Z(H^{f}_{S'}) + Z(H^{\perp})$.
Hence, the global equilibrium is a weighted average over all the control states and the corresponding conditional equilibrium states in $S'$. We rather have to think of system $S'$ as being in a `conditional' equilibrium $G_{\beta}(H^{i}_{S'})\otimes |c_i\rangle\langle c_i|$.  The conditional equilibrium corresponds to a projection (and subsequent normalization) of the global equilibrium onto an eigenspace of $H_{S'C}$.

The initial state of the forward process is the conditional equilibrium state of $S'C$ and an arbitrary state $\sigma$ of the reservoir, i.e., 
$G_{\beta}(H^{i}_{S'})\otimes |c_i\rangle \langle c_i| \otimes \sigma$.
 In an analogous fashion, the reverse process should start in a conditional equilibrium with respect to the final Hamiltonian, thus corresponding to the global state  
$G_{\beta}(H^{f}_{S'})\otimes |c_f\rangle \langle c_f| \otimes \sigma$.

The fact that we here assume that  $E$ and $S'C$ initially are uncorrelated makes it possible to formulate our quantum fluctuation theorems in terms of quantum channels on the energy reservoir alone. In Appendix \ref{SecPrecorrelations} we discuss a particular case of pre-correlations.

Although the validity of the fluctuation theorems  \emph{per se} does not rely on how these conditional equilibrium states come about, or whether it would be difficult or easy to prepare them, it is nevertheless justified to ask how they are supposed to be obtained. In the typical narrative surrounding the classical Crooks theorem it appears to be taken for granted that the system eventually settles at the equilibrium  distribution $G_{\beta}\big(H_{S'}(x_i)\big)$ if $x_i$ is kept fixed. When turning the control mechanism explicit (both in the classical and quantum case) it is clear that this is not an entirely innocent statement, as it suggests that there is a separation of time-scales, where the equilibration of the system is much faster than the equilibration of the controlling degrees of freedom.  It is not difficult to imagine classical models where this assumption would make sense. Suppose for example that we would have a polymer with the ends attached to two (comparably) massive bodies immersed in a liquid (e.g.~in the spirit of the experimental setup in \cite{Liphardt02}). The equilibrium state of the polymer depends on the distance between the anchor points, and it seems intuitively reasonable that the polymer typically would equilibrate on very short time scales compared to the Brownian motion of the massive bodies. These notions could also be adapted to the quantum case, akin to what we do in Appendices \ref{controlparticle1}, \ref{SecControlparticleAgain}, and \ref{SecNumericalEvaluation}. One may even argue that this separation of time-scales should be a requirement for well-designed control mechanisms. It is also clear that we cannot expect this to hold in general, but that it implies conditions on the nature of the interactions between system, control, and heat bath, as well as on the initial states.

In relation to these questions it may be useful to note the similarities between the type of  conditional equilibrium that we consider here, and some of the settings in the literature on classical fluctuation relations for partial equilibrium conditions, or extended fluctuation relations \cite{Maragakis08,Junier09,Alemany12,Liphardt12,Roldan14,Gavrilov17}.
 One may ask similar questions  as above, concerning the consequences of including explicit control systems, also in the classical scenario. Although it indeed would be relevant to elucidate the general conditions for well-functioning control systems in both the classical and quantum case, these questions will not be covered in this investigation.

\subsection{Induced channels on the energy reservoir}
 
For the standard formulations of Crooks theorem, the change of the external control parameters  would typically push the system out of equilibrium at the expense of work. In our quantum treatment we wish to describe all aspects of how the state of the energy reservoir changes, which conveniently can be captured by the channels (trace preserving completely positive maps \cite{KrausBook}) induced on the reservoir.

More precisely, we wish to describe how the state of the energy reservoir evolves under the action of a global energy conserving unitary operation $V$ that additionally exhibits perfect control (\ref{PerfectControl}). 
We furthermore assume that $S'$ starts in the conditional equilibrium with respect to the initial control state $|c_{i}\rangle$ as described in the previous section. 
 The state of the reservoir after the evolution can thus be written 
\begin{equation}
\label{Fdef}
\mathcal{F}(\sigma) :=  \Tr_{S'C}(V [G_{\beta}(H^{i}_{S'})\otimes |c_i\rangle \langle c_i| \otimes \sigma]V^{\dagger}).
\end{equation}
Hence, $\mathcal{F}$ describes the change of state induced on the energy reservoir $E$
 due to the global dynamics $V$ for this particular class of initial states.

For this intermediate version we reverse the entire evolution on the global system. More precisely, we replace $V\cdot V^{\dagger}$ with $V^{\dagger}\cdot V$. For a $V$ generated by Hamiltonian evolution $V= e^{-it H_{\textrm{evol}}/\hbar}$, this corresponds to a replacement of $t$ with $-t$. The reverse process starts in the local equilibrium with respect to the final Hamiltonian $H^{f}_{S'}$, which results in the channel
\begin{equation}
\label{Rdef}
\mathcal{R}(\sigma) :=  \Tr_{S'C}(V^{\dagger} [G_{\beta}(H^{f}_{S'})\otimes |c_{f}\rangle \langle c_{f}| \otimes \sigma] V). 
\end{equation}
 Although this indeed guarantees that the evolution is reversed in a very concrete sense, one can argue that it does not quite correspond  to the spirit of Crooks relation, which only requires a reversal of the control parameters.  In Appendix \ref{SecAQuantumFlctnThrm} we will remove this idealization. The purpose of the following sections is to establish a relation (Proposition \ref{PreliminaryQuantumCrooks}) between the `forward' channel $\mathcal{F}$ and the `reverse' channel $\mathcal{R}$.

As a further remark one can compare the type of channels defined in (\ref{Fdef}) and (\ref{Rdef}) with thermal operations \cite{Janzing00,Janzing06,Horodecki11,Brandao13b,Gour,Brandao13,Renes14,Faist15,Lostaglio14b,Perry16,Scharlau16,Lostaglio16}.
Thermal operations are, as in (\ref{Fdef}), obtained when initially uncorrelated ancillary systems interact with the system of interest via energy conserving unitary operations (for non-interacting Hamiltonians). However, as opposed to (\ref{Fdef}), thermal operations require all the ancillary systems to initially be in their Gibbs states. In other words, the above channels would be thermal operations were it not for the control system, which initially is in a non-equilibrium state (as discussed in Appendix \ref{TheInitalStates}). 
Nevertheless, (\ref{Fdef}) is `almost' a  thermal operation in the sense that $S'C$ is in a conditional equilibrium state. 
An extension of the theory of thermal operations to these types of initial states could provide an alternative route to study fully quantum fluctuation relations, and could potentially yield insights on the violations of the standard bounds discussed in Section \ref{MainSecBounds}. However, we will not consider such generalizations in this investigation.

\subsection{\label{ConjugateCPMs} Conjugate CPMs}

The conjugate map $\phi^{*}$ of a completely positive map (CPM) $\phi$ can be defined via $\Tr\big(Y\phi(\sigma)\big) = \Tr\big(\phi^{*}(Y)\sigma\big)$, where $Y$ are arbitrary (bounded) Hermitian operators, and $\sigma$ arbitrary density operators. A convenient alternative characterization is via Kraus representations \cite{KrausBook} $\phi(\sigma) = \sum_{k}V_{k}\sigma V_{k}^{\dagger}$, where the conjugate map is given by $\phi^{*}(Y) = \sum_{k}V^{\dagger}_{k} Y V_{k}$. 

For the derivations it will be convenient to keep in mind the following observation. Suppose that a CPM $\phi$ is defined via a unitary $V:\mathcal{H}_a\otimes\mathcal{H}\rightarrow \mathcal{H}_a\otimes\mathcal{H}$ as
\begin{equation}
\label{conjremark1}
\phi(\sigma) := \Tr_{a}([Q_a\otimes \hat{1}]V[\eta_a\otimes \sigma] V^{\dagger}),
\end{equation}
where $Q_a$ (bounded) and $\eta_a$ (trace class) are positive operators on an ancillary Hilbert space $\mathcal{H}_a$.
It follows that the conjugate CPM $\phi^{*}$ can be written
\begin{equation}
\label{conjremark2}
\phi^{*}(Y) = \Tr_{a}([\eta_a \otimes\hat{1}]V^{\dagger}[Q_{a}\otimes Y]V).
\end{equation}
One should note that the definition of the conjugate $*$ via $\Tr(Y\phi(\sigma)) = \Tr(\phi^{*}(Y)\sigma)$ is not restricted to $\phi$ being a CPM. For example, if $\phi(\sigma): = A\sigma B$, for some operators $A,B$ (not necessarily Hermitian), then $\phi^{*}(Y) = BYA$.

\subsection{\label{TheMappingJ}The mapping $\mathcal{J}$}

For an operator $A$ we define the mapping
\begin{equation}
\mathcal{J}_{A}(Q) := e^{-A/2}Q{e^{ -A^{\dagger}/2}}.
\end{equation}
The reason for why we here choose the exponent to be $- A/2$, rather than say $A$, is only to make it more directly related to Gibbs states in the special case that $A := \beta H$ and $Q:= \hat{1}$, and thus $\mathcal{J}_{\beta H}(\hat{1}) = Z_{\beta}(H)G_{\beta}(H)$. The mapping $\mathcal{J}_{A}$ is a CPM, but is in general not trace preserving. 

The mapping $\mathcal{J}_{\beta H}$ does often occur together with its inverse $\mathcal{J}_{\beta H}^{-1}$, in such a way that $ \mathcal{J}_{\beta H}\circ\phi\circ\mathcal{J}_{\beta H}^{-1}$ for some CPM $\phi$ (see e.g.~Proposition \ref{PreliminaryQuantumCrooks}).  
This combination can in some sense be viewed as a quantum version of the term $e^{\beta w}$ in the classical Crooks relation in (\ref{StandardCrooks}). To see this, let us consider the special case that $H_E$ has a pure point spectrum, i.e., there exists an orthonormal basis of eigenvector $|n\rangle$ with corresponding eigenvalues $E_n$.
For mappings from diagonal elements to diagonal elements we would get $\langle m|\mathcal{J}_{\beta H}\big(\phi\big(\mathcal{J}_{\beta H}^{-1}(|n\rangle\langle n|)\big)\big)|m\rangle = e^{\beta(E_{n}-E_{m})} \langle m|\phi\big(|n\rangle\langle n|\big)|m\rangle$. The term $E_n-E_{m}$ is the decrease of energy in the reservoir, and by identifying this loss with the work performed, the analogy becomes evident. For the general transition between arbitrary matrix elements, the corresponding expression reads $\langle m|\mathcal{J}_{\beta H}\big(\phi\big(\mathcal{J}_{\beta H}^{-1}(|n\rangle\langle n'|)\big)\big)|m'\rangle = e^{\beta(E_{n}-E_{m})/2} e^{\beta(E_{n'}-E_{m'})/2} \langle m|\phi\big(|n\rangle\langle n'|\big)|m'\rangle$. The fact that the off-diagonal case is governed by two energy differences, rather than one, corresponds to the evolution of the coherences in the energy reservoir.

\subsection{Derivation of an intermediate fluctuation relation}

\begin{Lemma}
\label{Transition}
With Assumptions \ref{Def1} it is the case that 
$ V [e^{\alpha H^{i}_{S'}}\otimes |c_i\rangle\langle c_i|\otimes \hat{1}_E]   = [e^{\alpha H^f_{S'}}\otimes |c_f\rangle\langle c_f|\otimes e^{\alpha H_E}]V  [\hat{1}_{S'}\otimes \hat{1}_C\otimes e^{-\alpha H_E}]$ 
and $[e^{\alpha H^{f}_{S'}}\otimes |c_f\rangle\langle c_f|\otimes \hat{1}_E]V  =  [\hat{1}_{S'}\otimes \hat{1}_C\otimes  e^{-\alpha H_E}]V [e^{\alpha H^{i}_{S'}}\otimes |c_i\rangle\langle c_i|\otimes e^{\alpha H_E }]$,
for all $\alpha \in \mathbb{C}$.
\end{Lemma}

\begin{proof}
We only prove the first equality, since the proof of the second is analogous.  
First note that $e^{\alpha H^{i}_{S'}}\otimes |c_i\rangle\langle c_i|\otimes \hat{1}_E = e^{\alpha H^{i}_{S'}\otimes |c_i\rangle\langle c_i|\otimes \hat{1}_E}[\hat{1}_{S'}\otimes |c_i\rangle\langle c_i|\otimes \hat{1}_E]$. 
Next we can use the fact that $H^{i}_{S'}\otimes |c_i\rangle\langle c_i|\otimes \hat{1}_E = H - H^{f}_{S'}\otimes |c_f\rangle\langle c_f| \otimes \hat{1}_E  - H^{\perp}\otimes \hat{1}_E-\hat{1}_{S'C}\otimes H_E$. 
Note that these summands commute with each other. Moreover, 
$H^{f}_{S'}\otimes |c_f\rangle\langle c_f| \otimes \hat{1}_E$ and $H^{\perp}\otimes \hat{1}_E$ have orthogonal supports compared to $\hat{1}_{S'}\otimes |c_i\rangle\langle c_i|\otimes \hat{1}_E$. By these observations it follows that we can write
\begin{equation*}
\begin{split}
& V [e^{\alpha H^{i}_{S'}}\otimes |c_i\rangle\langle c_i|\otimes \hat{1}_E] \\
& =  V e^{\alpha H}e^{-\alpha \hat{1}_{S'C}\otimes H_E}[\hat{1}_{S'}\otimes|c_i\rangle\langle c_i| \otimes \hat{1}_E] \\
& =  e^{\alpha H}V [\hat{1}_{S'}\otimes|c_i\rangle\langle c_i| \otimes \hat{1}_E] e^{-\alpha \hat{1}_{S'C}\otimes H_E},
\end{split}
\end{equation*}
where we in the second equality have used $[H,V] = 0$, and the fact that $\hat{1}_{S'}\otimes|c_i\rangle\langle c_i| \otimes \hat{1}_E$ commutes with $e^{-\alpha \hat{1}_{S'C}\otimes H_E}$, as well as orthogonality of various terms. 

Next we use the assumed property of perfect control in Eq.~(\ref{PerfectControl}), i.e.,  $V[\hat{1}_{S'}\otimes |c_i\rangle\langle c_i|\otimes \hat{1}_{E}] =  [\hat{1}_{S'}\otimes|c_f\rangle\langle c_f|\otimes \hat{1}_{E}]V$. When $e^{\alpha H}$ on the left hand side of $V$ `meets' $[\hat{1}_{S'}\otimes|c_f\rangle\langle c_f|\otimes \hat{1}_{E}]$, only the terms $e^{\alpha H^{f}_{S'}\otimes P_C^{f}\otimes \hat{1}_E  }e^{\alpha\hat{1}_{S'C}\otimes H_E}$ survive. This leads to the first equality.
 The proof of the second equality is analogous. 
\end{proof}

\begin{Proposition}[An intermediate quantum Crooks relation]
\label{PreliminaryQuantumCrooks}
With the definitions as in \ref{Def1}, the channels $\mathcal{F}$ and $\mathcal{R}$ defined in (\ref{Fdef})  and (\ref{Rdef}) are related as
\begin{equation}
\label{Fequiv}
Z(H^i_{S'}) \mathcal{F} = Z(H^f_{S'})\mathcal{J}_{\beta H_E}\mathcal{R}^{*}\mathcal{J}_{\beta H_E}^{-1}.
\end{equation}
With the separation of $S'$ into system $S$ and the heat bath $B$ as in equation (\ref{Def1HeatBath}) we thus get
\begin{equation}
\label{FequivS}
Z(H^i_{S}) \mathcal{F} = Z(H^f_{S})\mathcal{J}_{\beta H_E}\mathcal{R}^{*}\mathcal{J}_{\beta H_E}^{-1}.
\end{equation}
\end{Proposition}

\begin{proof}
By comparing the definition (\ref{Rdef}) of channel $\mathcal{R}$, with Eqs.~(\ref{conjremark1}) and (\ref{conjremark2}) in Appendix \ref{ConjugateCPMs} we can conclude that 
\begin{equation*}
\begin{split}
&\mathcal{R}^{*}(Y)=  \Tr_{S'C}\big( [G(H^{f}_{S'})\otimes |c_{f}\rangle \langle c_{f}| \otimes \hat{1}_E]\\
& \quad\quad\quad\quad\quad\quad \times V[\hat{1}_{S'}\otimes\hat{1}_C\otimes Y]V^{\dagger}\big)\\
& =  \frac{1}{Z(H^{f}_{S'})}\Tr_{S'C}\Big( [e^{-\beta H^{f}_{S'}/2}\otimes |c_{f}\rangle \langle c_{f}| \otimes \hat{1}_E]V\\
& \quad\quad\quad\quad\quad\quad \times [\hat{1}_{S'}\otimes\hat{1}_C\otimes Y]\\
& \quad\quad\quad\quad\quad\quad\times V^{\dagger}[e^{-\beta H^{f}_{S'}/2}\otimes |c_{f}\rangle \langle c_{f}| \otimes \hat{1}_E]\Big)\\
&[\textrm{By Lemma \ref{Transition}}] \\
& =  \frac{1}{Z(H^{f}_{S'})}e^{\beta H_E/2}\Tr_{S'C}\Big(V\\
&  \times [e^{-\beta H^{i}_{S'}}\otimes  |c_i\rangle\langle c_i|\otimes  e^{-\beta H_E/2 }Ye^{-\beta H_E/2 }]V^{\dagger}\Big)e^{\beta H_E/2}\\
& =  \frac{Z(H^{i}_{S'})}{Z(H^{f}_{S'})} \mathcal{J}^{-1}_{\beta H_E}\circ\mathcal{F}\circ \mathcal{J}_{\beta H_E}(Y).
\end{split}
\end{equation*}
This can be rewritten as (\ref{Fequiv}).

With the additional assumption in equation (\ref{Def1HeatBath}) we get $Z(H^{i}_{S'}) = Z(H^{i}_{S})Z(H_{B})$ and $Z(H^{f}_{S'}) = Z(H^{f}_{S})Z(H_{B})$. From this observation we obtain (\ref{FequivS}) from (\ref{Fequiv}).
\end{proof}

Equation (\ref{FequivS}) in Proposition \ref{PreliminaryQuantumCrooks} already has the flavor of a Crooks relation. However, as already has been mentioned, it relies on a too ambitious reversal where we invert the entire evolution. In Appendix  \ref{SecAQuantumFlctnThrm} we shall remove this idealization, and let both the forward and reverse evolution be governed by the same `direction of time', i.e., in both cases the global evolution is given by the map $Q\mapsto VQV^{\dagger}$.

\subsection{Control in fluctuation theorems versus control in thermal protocols}

In relation to the remarks concerning the swap of $V$ to $V^{\dagger}$ as being a `too ambitious' reversal, one may note that in other contexts, such as work extraction and information erasure, one often imagines to be able to choose freely among all energy-conserving global unitary operations over all the involved degrees of freedom, including additional equilibrium systems. One may thus wonder why such a detailed control is acceptable in those scenarios, but not for fluctuation theorems. One should keep in mind that work extraction and information erasure in some sense are engineering tasks, where the purpose is to construct an optimal  protocol or machinery. Moreover, the free equilibrium resources can be viewed as engineered ancillary systems that are thermalized by being put in contact with the heat bath, rather than themselves constituting the heat bath. This should be put in contrast with fluctuation theorems, where our task is not to design optimal procedures in engineered systems, but to make general statements about the nature of the given dynamics in arbitrary thermal systems. This may include systems provided by nature, about which we may have a very limited knowledge, and where our means of control are restricted to designated control systems (e.g.~external fields).

\subsection{\label{SecPetzRecovery}Remarks concerning Crooks operation time reversal and Petz recovery channel}

The notion of `operation time reversals' was introduced in \cite{Crooks08} as a quantum generalization of time-reversals of classical Markov chains. Given a channel $\Phi$ with fix-point density operator $\rho$, i.e., $\Phi(\rho) = \rho$,  the operation time reversal of $\Phi$ is defined by the mapping $\sigma \mapsto \sqrt{\rho}\Phi^{*}(\sqrt{\rho}^{-1}\sigma\sqrt{\rho}^{-1})\sqrt{\rho}$. 
Let us now compare this with the right hand side of (\ref{Fequiv}).
By construction $\mathcal{R}$ is a channel. However, one can confirm that it is the case that
\begin{equation}
\mathcal{R}(e^{-\beta H_E}) =  \frac{Z(H^{i}_{S'})}{Z(H^{f}_{S'})} e^{-\beta H_E}. 
\end{equation}
Hence,  $e^{-\beta H_E}$ is \emph{not} a fixpoint of $\mathcal{R}$, and thus the conditions for Crooks time reversal are not quite satisfied (unless $Z(H^{i}_{S'}) = Z(H^{f}_{S'})$, which would be the case for a cyclic process, i.e., if $H^f_{S'} = H^i_{S'}$). 

There exists a more general construction introduced in information theory, namely Petz recovery channel  \cite{Petz86,Petz88,Barnum02,Hayden04}. 
 Given a channel $\Phi$ and a reference state $\rho$ (that does not have to be a fixpoint of the channel $\Phi$) Petz recovery channel is defined as
\begin{equation}
\label{avva}
\widehat{\Phi}(\sigma)  = \sqrt{\rho}\Phi^{*}\big(\Phi(\rho)^{-1/2}\sigma \Phi(\rho)^{-1/2}\big)\sqrt{\rho}.
\end{equation}
Hence, Crooks time reversal emerges as a special case when $\rho$ is a fixpoint of $\Phi$. In this context one may also note the discussions on time-reversals for quantum channels in \cite{Ticozzi09}.

If we take $e^{-\beta H_E}$ as the reference operator, it is straightforward to confirm that the intermediate fluctuation relation (\ref{Fequiv}) can be rephrased as 
$\mathcal{F} = \widehat{\mathcal{R}}$,
i.e., that the forward channel is equal to the Petz-transformation (\ref{avva}) of the reverse channel. 

In  \cite{Wehner15} it was shown that  the work gain in work extraction can be bounded by how well the initial state can be reconstructed via Petz recovery channel. One may in particular note the similarity between our channel $\mathcal{R}$ and the recovery channel $\mathcal{R}_{\rho\rightarrow \sigma}$ defined in \cite{Wehner15}, although the definition of $\mathcal{R}$ contains the control system that is subject to perfect control. In view of these structural similarities it is tempting to speculate on  deeper relations between these results. Fluctuation relations can be viewed as statistical manifestations of the second law \cite{ReviewFluctThm}; an observation that makes the connection to work extraction more plausible. Moreover, there are investigations that hint or elaborate on connections. One example is the generalized Jarzynski relations for feedback control \cite{Sagawa08,Sagawa10,Toyabe10,Morikuni11}. Moreover, in \cite{Aberg13} a classical Crook's relation was used as a component in a proof about single-shot work extraction, and recent investigations \cite{Halpern15a,Salek15,Dahlsten15,Halpern15b} have focused on exploring links between single-shot concepts and fluctuation theorems. On a similar note one may wonder whether there exists a more operational characterization of fluctuation relations. Although intriguing questions, we will not consider them further in this investigation.

\section{\label{SecTimeReversal} Time-reversal and time-reversal symmetry}

Here we discuss the notion of time-reversals that we shall use for obtaining the quantum Cooks relation.
We begin with a quick reminder of the essence of the standard notion of time-reversals, which is closely related to complex conjugation. As mentioned in the main text, our time-reversals are primarily related to transposes rather than to complex conjugation (cf.~the discussions on time-reversals in \cite{Sanpera97,Bullock05}), and in section \ref{ComplVsTransp} we compare complex conjugation to transposes regarded as time-reversal operations. After these preludes, we do in section  \ref{ReverseUseful} turn to the actual definition of time-reversal that we employ in this investigation.

\subsection{A brief reminder of the standard notion of time-reversals}

Here we briefly recollect the notion of time-reversals in classical and quantum mechanics. More thorough discussions can be found, e.g.,  in chapter 26 of \cite{Wigner56} and  section 4.4 of \cite{Sakurai}. These references also include the case of electric and magnetic fields, as well as angular momenta, which we do not cover. For an overview of time-reversals in classical systems, see e.g.~\cite{Roberts92,Lamb98}.
 See also \cite{Chetrite08} for various notions of time-reversals in the context of classical fluctuation relations.

In classical mechanics, time-reversals are defined via the reversal of momenta, e.g., for a system of particles, their positions are left intact, but all the velocities are reversed. If the underlying Hamilton function is time-reversal symmetric, i.e., is invariant under $\overline{p}\mapsto -\overline{p}$, then this implies that  the reversed particles follow their reversed trajectories, and thus effectively behave as if time was running backwards. As an example one can consider a particle of mass $M$ in a potential $V$, with Hamilton function   $H(\overline{x},\overline{p}) = \overline{p}^2/(2M)+V(\overline{x})$. If $\big(\overline{x}(t),\vec{p}(t)\big)$ is a solution to the corresponding equation of motions, then $\big(\overline{x}(-t),-\vec{p}(-t)\big)$ is also a solution, thus representing the particle moving backwards with reversed momentum.

Quantum mechanics does not possess a very crisp notion of phase space, or phase-space trajectories, due to the canonical commutation relation for the position and momentum operators. One can nevertheless introduce a notion of time-reversals. Suppose that the system has the Hamilton operator $H = \hat{P}^2/(2M) + V(\hat{X})$, and that $\psi(\overline{x},t)$ is a solution to the corresponding Schr\"odinger equation $i\hbar\dot{\psi} = - \hbar^2 \nabla^2\psi/(2M) + V(\overline{x})\psi$. The function $\psi(\overline{x},-t)$ is generally not a solution, while $\psi(\overline{x},-t)^{*}$ is. In other words,  a complex-conjugated wave-function evolves `backwards', which suggests that time-reversals in quantum mechanics are related to complex conjugation. 
As further indications one may note that a plane wave $\psi(x) = e^{ipx}$ gets mapped to $\psi(x)^{*} = e^{-ipx}$, thus changing the sign of the momentum of the momentum eigenstates. Another example is the family of coherent states $\{|\alpha\rangle\}_{\alpha\in \mathbb{C}}$, with wave-functions $\langle x|\alpha\rangle = \exp[-\Imag(\alpha)^2-(x/\sigma-2\alpha)^2/4]/[(2\pi)^{1/4}\sqrt{\sigma}]$.  Coherent states can in some sense be regarded as representing fuzzy phase space points (with an as sharp simultaneous position and momentum as quantum mechanics allows for), where the real and imaginary part of $\alpha$ can be associated to the average position and momentum, respectively (see Appendix \ref{SecNumericalEvaluation}). Since $\langle x|\alpha\rangle^{*} =\langle x|\alpha^{*}\rangle$, the effect of the conjugation is to swap the sign of the imaginary part of $\alpha$, while leaving the real part intact, thus emulating the classical procedure of swapping momenta at fixed positions.

There is also a more abstract line of reasoning, arguing that transformations that leave the magnitude of inner products on the Hilbert space invariant either should be unitary or anti-unitary (see the Appendix to chapter 20 in \cite{Wigner56}, or chapter 27 in \cite{Gottfried66}), and that time-reversals fall in the category of anti-unitary operators (see chapter 26 of \cite{Wigner56}, or  section 4.4 in \cite{Sakurai}). Moreover, anti-unitary operators can be written as a complex conjugation composed with a unitary operation (chapter 27 in \cite{Gottfried66}). Hence, on the level of Hilbert spaces, time-reversals are closely related to the complex conjugate of the wave-function.

The above remarks have been focused on cases where there are no external parameters that break time-reversal symmetry. The typical example of such symmetry breaking is external magnetic fields. In such cases, the time-reversal operation would not only include a change of the state of the system, but also a change of the Hamiltonian (e.g.~swapping the directions of the external magnetic fields) which thus means that we have to intervene and change the nature of the dynamics of the system. This goes somewhat against the general spirit of the present investigation, where we employ time-reversal symmetry precisely in order to avoid having to make such interventions. (It would be more in spirit to explicitly include the systems and currents that generate the magnetic fields.)  It nevertheless seems reasonable that one could generalize the type of fluctuation relations that we consider here in order to incorporate external time-reversal breaking parameters. However, we leave this as an open question.

\subsection{\label{ComplVsTransp}Complex conjugation vs transpose}

As discussed in the previous section, the standard notion of time-reversals is in quantum mechanics typically formulated on the level of Hilbert spaces via anti-unitary operators, and can be expressed via complex conjugation of wave-functions. Here we shall make a slight shift of perspective, and consider the action of the complex conjugation on the level of density operators, and compare this to transposes. We will see that both these operations in some sense can be regarded as time-reversal operations.

The standard time reversal can be expressed in terms of the complex conjugation $\psi^{*}(x)$ of wave functions $\psi(x)$, or via  an orthonormal basis as $|\psi^{*}\rangle = \sum_{n}|n\rangle\langle n|\psi\rangle^{*}$.
On the level of operators this translates to $Q^{*} = \int\!\! \int \langle x|Q|x'\rangle^{*}|x\rangle \langle x'|dxdx'$, or  $Q^{*} = \sum_{nn'}\langle n|Q|n'\rangle^{*}|n\rangle\langle n'|$.
In comparison, the transpose, $\tau$, acts as $Q^{\tau} = \int\!\!\int |x\rangle \langle x'|Q|x\rangle \langle x'|dxdx'$, or  $Q^{\tau} = \sum_{nn'}|n\rangle \langle n'|Q|n\rangle \langle n'|$.

Both the complex conjugate and the transpose implement time-reversals, but in a slightly different manner.
Suppose that a Hermitian generator $H_{\textrm{evol}}$ for the time evolution  satisfies 
$H_{\textrm{evol}}^{*} = H_{\textrm{evol}}$, or equivalently $H_{\textrm{evol}}^{\tau} = H_{\textrm{evol}}$.
 (One should not confuse $H_{\textrm{evol}}$, discussed in Appendix \ref{MakeSense}, with the Hamiltonians $H$, $H_{S}$, $H_{E}$ etc.) The time-evolution operator consequently transforms as
\begin{equation}
\label{nklvndy}
\begin{split}
(e^{-itH_{\textrm{evol}}/\hbar})^{*} = & e^{itH_{\textrm{evol}}/\hbar},\\
 (e^{-it H_{\textrm{evol}}/\hbar})^{\tau} = & e^{-itH_{\textrm{evol}}/\hbar}.
\end{split}
\end{equation}
Hence, complex conjugation inverts the evolution operator, while the transpose leaves it intact. At first sight it may thus seem a bit odd that the transpose can implement any form of time-reversal. To understand this we should consider the manner in which these mappings act on products of operators, namely
\begin{equation}
\label{yfmbnf}
(AB)^{*} = A^{*}B^{*},\quad (AB)^{\tau} = B^{\tau}A^{\tau}.
\end{equation}
In other words, complex conjugation leaves the operator ordering intact, while the transpose reverses the ordering. In some sense, (\ref{nklvndy}) and (\ref{yfmbnf}) complement each other when it comes to the reversal of the evolution. To see this, assume that an initial state $\rho_{i}$ is evolved into the state $\rho_{f} = e^{-itH_{\textrm{evol}}/\hbar}\rho_i e^{itH_{\textrm{evol}}/\hbar}$. For the complex conjugate we get
\begin{equation*}
\begin{split}
\rho_{f}^{*} =  & (e^{-itH_{\textrm{evol}}/\hbar})^{*}\rho_i^{*}(e^{itH_{\textrm{evol}}/\hbar})^{*} \\
= & e^{itH_{\textrm{evol}}/\hbar}\rho_i^{*} e^{-itH_{\textrm{evol}}/\hbar},
\end{split}
\end{equation*}
and hence $\rho_i^{*} = e^{-itH_{\textrm{evol}}/\hbar}\rho_{f}^{*}e^{itH_{\textrm{evol}}/\hbar}$. Similarly, 
\begin{equation*}
\begin{split}
\rho_{f}^{\tau}  = &  (e^{itH_{\textrm{evol}}/\hbar})^{\tau} \rho_i^{\tau}(e^{-itH_{\textrm{evol}}/\hbar})^{\tau}\\
 = & e^{itH_{\textrm{evol}}/\hbar}\rho_i^{\tau} e^{-itH_{\textrm{evol}}/\hbar},
 \end{split}
 \end{equation*}
and consequently $\rho_i^{\tau} = e^{-itH_{\textrm{evol}}/\hbar}\rho_{f}^{\tau}e^{itH_{\textrm{evol}}/\hbar}$.
Hence, we do again obtain the effective reversal of the evolution. 
We can conclude that for conjugation the time-reversal is  due to the inversion of the time evolution operator, while for the transpose it is due to the inversion of the operator ordering.

\subsection{\label{ReverseUseful} What we require from time-reversals}
Instead of directly defining time-reversals in terms of transposes we here rather define it via a `wish-list' of properties.  By inspection one can see that transposes satisfy these conditions, although the latter allow for a slightly larger class of operations (see  Proposition \ref{TcharactFinite}).
One can also see that this definition immediately excludes the complex conjugation (due to the assumed linearity).
Hence, one should not take this list as the ultimate and most general definition of what a time-reversal possibly could be, but rather as convenient set of assumptions that is sufficient for our purposes and makes the book-keeping in the proofs simple. It may potentially be the case that a more general notion of time-reversals could extend the resulting family of quantum Crooks relations. Although an interesting question, it will not be pursued  further in this investigation.
\begin{Definition}
\label{PropertiesTimeReversal}
A linear map $\mathcal{T}$ is called a  time-reversal if
\begin{subequations}  
\begin{equation}
\label{Prop3} \mathcal{T}(AB) = \mathcal{T}(B)\mathcal{T}(A),
\end{equation}
\begin{equation}
\label{Prop4}
\mathcal{T}(A^{\dagger}) = \mathcal{T}(A)^{\dagger},
\end{equation}
\begin{equation}
\label{Prop2} \Tr[\mathcal{T}(\sigma)] = \Tr(\sigma),
\end{equation}
\begin{equation}
\label{Prop6} \mathcal{T}^2 = I.
\end{equation}
\end{subequations}
\end{Definition}
It is certainly justified to ask in what sense a map $\mathcal{T}$ with the above properties deserves to be called a `time-reversal'. Much analogous to the discussions of complex conjugates and transposes in the previous section, let us now assume that $\rho_f := V\rho_i V^{\dagger}$, for some unitary operator $V$ and initial state $\rho_i$. We get
\begin{equation}
\label{ndlflknlknbd}
\begin{split}
\mathcal{T}(\rho_f) = & \mathcal{T}(V\rho_i V^{\dagger})= \mathcal{T}(V)^{\dagger}\mathcal{T}(\rho_i) \mathcal{T}(V),
\end{split}
\end{equation}
where we in the second equality make use of the reordering-property (\ref{Prop3}) and of property (\ref{Prop4}). Our first observation is that $\mathcal{T}(V)$ is a unitary operator (see Lemma \ref{PreservesIdentity} in Appendix \ref{TimeRevFinite}). Moreover, $\mathcal{T}(\rho_i)$ and $\mathcal{T}(\rho_f)$ are density operators, due to preservation of trace (\ref{Prop2}) and preservation of positivity (by Lemma \ref{ReorderHermImplyPos}), although $\mathcal{T}$ is generally not completely positive. 
By the unitarity of $\mathcal{T}(V)$ we can rewrite (\ref{ndlflknlknbd}) as $\mathcal{T}(\rho_i) =  \mathcal{T}(V)\mathcal{T}(\rho_f)\mathcal{T}(V)^{\dagger}$, which thus describes the unitary evolution of a well defined quantum state. If it additionally would be the case that the time-reversal leaves the evolution operator $V$ invariant, $\mathcal{T}(V) = V$, then we get $\mathcal{T}(\rho_i) = V\mathcal{T}(\rho_f)V^{\dagger}$. Hence, the time-reversed final state $\mathcal{T}(\rho_f)$ evolves to the time-reversed initial state $\mathcal{T}(\rho_i)$ under the forward evolution, if the time-reversal leaves the evolution operator invariant. By comparison with the discussion in the previous section, we can conclude that $\mathcal{T}$, much in analogy with the transpose, obtain its  capacity to `time-reverse' from the reordering property (\ref{Prop3}), as opposed to complex conjugation that obtains this power due to its capability to invert the time-evolution operator. As we will see in Appendix \ref{TimeRevFinite}, the time-reversals $\mathcal{T}$ are indeed very closely related to transposes.

As a bit of a technical remark, in the infinite-dimensional case one may additionally require that  $\mathcal{T}$ maps bounded operators to bounded operators, and trace class operators to trace class operators.
If one restricts to bounded $A,B$ in (\ref{Prop3}) it follows that $\mathcal{T}(A)$, $\mathcal{T}(B)$, and $\mathcal{T}(AB)$ are bounded. 
By demanding that $A$ in (\ref{Prop4})  is bounded we make sure that the Hilbert adjoint $A^{\dagger}$ is well defined and bounded (see Theorem 3.9-2 in \cite{Kreyszig}).
By the requirement that $\mathcal{T}$ maps bounded operators to bounded operators  we know that $\mathcal{T}(A)$ is bounded, and thus  $\mathcal{T}(A)^{\dagger}$ is also well defined.
If one restricts $\sigma$ to be trace class in (\ref{Prop2}) it follows that $\Tr(\sigma)$ is well defined, and if $\mathcal{T}$ maps trace class operators to trace class operators, $\Tr[\mathcal{T}(\sigma)]$ is also well defined. Although this is a reasonable collection of assumptions, one should keep in mind that we here tend to apply these maps also to unbounded operators.

\subsection{The $\ominus$-transformation}

For a CPM $\phi$ we define $\phi^{\ominus}$ as
\begin{equation}
\label{Defominus}
\phi^{\ominus} := \mathcal{T}\phi^{*}\mathcal{T},
\end{equation}
where $\mathcal{T}$ is a given  time-reversal, and where $\phi^{*}$ is the conjugation discussed in Appendix \ref{ConjugateCPMs}. It is a straightforward application of the properties of the time-reversal $\mathcal{T}$ to show the following alternative definition
\begin{equation}
\label{Altominusdef}
\phi^{\ominus} = (\mathcal{T}\phi\mathcal{T})^{*}.
\end{equation}
It is also straightforward to confirm the following lemma.
\begin{Lemma}
If $\phi$ is a CPM with Kraus decomposition $\phi(\sigma) = \sum_{k}V_{k}\sigma V_{k}^{\dagger}$, and $\mathcal{T}$ is a time-reversal, then
\begin{equation}
\phi^{\ominus}(\sigma) =  \sum_{k}\mathcal{T}(V_{k})\sigma \mathcal{T}(V_k)^{\dagger}.
\end{equation}
In other words,  if $\{V_{k}\}_{k}$ is a Kraus representation of $\phi$, then $\{\mathcal{T}(V_{k})\}_{k}$ is a Kraus representation of $\phi^{\ominus}$. Hence, if $\phi$ is a CPM, then $\phi^{\ominus}$ is a CPM.
If $\phi$ is a channel (trace preserving CPM), then $\phi^{\ominus}$ is not necessarily a channel.
However, if $\phi$ is a unital channel ($\phi(\hat{1}) = \hat{1}$), then $\phi^{\ominus}$ is  a channel, and moreover a unital channel.
\end{Lemma}

\subsection{\label{TimeRevFinite} Characterization of $\mathcal{T}$ in finite dimensions}

The purpose of this section is to make more precise what kind of mappings that the list of properties in Definition \ref{PropertiesTimeReversal} specifies, how they relate to transposes, as well as deriving some further properties that will be useful for the subsequent derivations. Throughout this section we assume that the underlying Hilbert space is finite-dimensional, although some of the results would be straightforward to extend to the infinite-dimensional case.

For a finite-dimensional Hilbert space $\mathcal{H}$, we do in the following let $\mathcal{L}(\mathcal{H})$ denote the set of linear operators on $\mathcal{H}$. Our first general observation is that if  $\mathcal{T}_1$ and $\mathcal{T}_2$ are time-reversals on two different finite-dimensional Hilbert spaces,  then $\mathcal{T}_1\otimes\mathcal{T}_2$ is also a time-reversal. It turns out that each single property in Definition \ref{PropertiesTimeReversal} is separately preserved under the tensor product. The proof can be obtained via decompositions $Q = \sum_{mn}Q_1^{(m)}\otimes Q_2^{(n)}$, where $Q_1^{(m)}$ and $Q_2^{(n)}$ are operators on $\mathcal{H}_1$ and $\mathcal{H}_2$, respectively.

\begin{Lemma}
\label{PreservesIdentity}
If $\mathcal{T}$ is a linear map that satisfies conditions (\ref{Prop3}) and (\ref{Prop6}), then $\mathcal{T}(\hat{1}) = \hat{1}$.\\
Moreover if $\mathcal{T}$ is a linear map that satisfies conditions (\ref{Prop3}), (\ref{Prop4}), and  (\ref{Prop6}), then $\mathcal{T}$ maps unitary operators to unitary operators.
\end{Lemma}
\begin{proof}
By applying $\mathcal{T}$ to the trivial identity $\mathcal{T}(\hat{1}) = \mathcal{T}(\hat{1})\hat{1}$ and use property  (\ref{Prop6})
it follows that 
$\hat{1} = \mathcal{T}\big(\mathcal{T}(\hat{1})\big) =   \mathcal{T}\big(\mathcal{T}(\hat{1})\hat{1}\big) 
=  \mathcal{T}(\hat{1})\mathcal{T}(\mathcal{T}(\hat{1}))
=  \mathcal{T}(\hat{1})\hat{1}= \mathcal{T}(\hat{1})$,
where the third equality follows by (\ref{Prop3}).

If we now furthermore assume that $\mathcal{T}$ satisfies  (\ref{Prop4}) and that $V$ is a unitary operator, then $\mathcal{T}(V)\mathcal{T}(V)^{\dagger} = \mathcal{T}(V)\mathcal{T}(V^{\dagger}) = \mathcal{T}(V^{\dagger}V) =  \mathcal{T}(\hat{1}) =\hat{1}$, and analogously $\mathcal{T}(V)^{\dagger}\mathcal{T}(V) = \hat{1}$. Hence, $\mathcal{T}(V)$ is unitary.
\end{proof}

\begin{Lemma}
\label{ReorderHermImplyPos}
Let  $\mathcal{T}$ be linear.  
If $\mathcal{T}$ satisfies (\ref{Prop3})  and (\ref{Prop4}), then $\mathcal{T}$ is positive, i.e., $Q\geq 0$  $ \Rightarrow$  $\mathcal{T}(Q)\geq 0$.
\end{Lemma}
\begin{proof}
If $Q\geq 0$ then there exists $A$ such that $Q = AA^{\dagger}$. Hence
$\mathcal{T}(Q) =  \mathcal{T}(AA^{\dagger})
=  \mathcal{T}(A^{\dagger})\mathcal{T}(A)
=  \mathcal{T}(A)^{\dagger}\mathcal{T}(A)$,
where the second equality follows by  (\ref{Prop3}), and the third by  (\ref{Prop4}). Hence, $\mathcal{T}(Q)\geq 0$. 
\end{proof}

\begin{Lemma}
\label{TpreservesPurity}
Let $\mathcal{T}$ be a linear map.
\begin{itemize}
\item If $\mathcal{T}$ satisfies (\ref{Prop3}) and  (\ref{Prop4}), then $\mathcal{T}$ maps projectors to projectors. Furthermore, pairwise orthogonal projectors are mapped to pairwise orthogonal projectors.
\item If $\mathcal{T}$ satisfies (\ref{Prop3}),  (\ref{Prop4}), and (\ref{Prop2}), then $\mathcal{T}$ preserves the dimension of the projected subspaces. In particular,  $\mathcal{T}$ preserves purity, i.e., if $|\psi\rangle\in\mathcal{H}$ is normalized, then there exists a normalized $|\chi_{\psi}\rangle\in\mathcal{H}$  such that $\mathcal{T}(|\psi\rangle\langle\psi|) = |\chi_{\psi}\rangle\langle\chi_{\psi}|$.
\item If $\mathcal{T}$ satisfies (\ref{Prop3}),  (\ref{Prop4}), (\ref{Prop2}), and if the underlying Hilbert space is finite-dimensional, then $\mathcal{T}(\hat{1}) = \hat{1}$.
\end{itemize}
\end{Lemma}
For the third item it is necessary to restrict to finite dimensions. As an example, let $\{|n\rangle\}_{n\in \mathbb{N}}$ be a complete orthonormal basis, and define $\mathcal{T}(|n\rangle\langle n'|) := |2n'\rangle\langle 2n|$. This satisfies  (\ref{Prop3}),  (\ref{Prop4}) and (\ref{Prop2}), but $\mathcal{T}(\hat{1}) \neq \hat{1}$.
\begin{proof}
A linear operator $P$ is a projector if and only if $P^{2} = P$ and $P^{\dagger} =P$ (and $P$ is bounded in the infinite-dimensional case).
Assuming that $P$ is a projector it follow by properties (\ref{Prop3}) and (\ref{Prop4}) that $\mathcal{T}(P)$ is  also a projector. (If $\mathcal{T}$ preserves boundedness, then the boundedness of $\mathcal{T}(P)$ is guaranteed.) 
Two projectors are orthogonal if and only if $P_1P_2 =0$. Thus by property (\ref{Prop3}) it follows that $\mathcal{T}(P_2)\mathcal{T}(P_1) = \mathcal{T}(P_1P_2) = 0$. 
The dimension of the subspace onto which a projector  $P$ projects is given by $\Tr(P)$. By assumption (\ref{Prop2}) it follows that  $\Tr(\mathcal{T}(P)) = \Tr(P)$. Hence, the dimension is preserved.
In the case that the Hilbert space is finite-dimensional, then $\Tr\mathcal{T}(\hat{1}) = \Tr\hat{1}$ is the dimension of the Hilbert space. Hence $\mathcal{T}(\hat{1})$ is a projector with the dimension of the Hilbert space, and thus $\mathcal{T}(\hat{1}) = \hat{1}$.
\end{proof}

In the following we denote the standard operator norm by $\Vert Q\Vert := \sup_{\psi:\vert\psi\Vert = 1}\Vert Q|\psi\rangle\Vert$, and the trace norm by $\Vert Q\Vert_1 := \Tr\sqrt{QQ^{\dagger}} = \Tr\sqrt{Q^{\dagger}Q}$ (where the last equality in the finite-dimensional case follows by the singular value decomposition of $Q$).
\begin{Lemma}
\label{Tpreservesbound}
Let $\mathcal{T}$ be a time-reversal as in Definition \ref{PropertiesTimeReversal}, then 
$\Vert \mathcal{T}(Q)\Vert =  \Vert Q\Vert$, $\Vert \mathcal{T}(Q)\Vert_1 =  \Vert Q\Vert_1$.
\end{Lemma}
\begin{proof}
First we note that 
$\Vert \mathcal{T}(Q)|\psi\rangle\Vert^{2} 
=  \Tr\big(|\psi\rangle\langle\psi|\mathcal{T}(QQ^{\dagger})\big) 
=  \Tr\big(QQ^{\dagger}|\chi_{\psi}\rangle\langle\chi_{\psi}|\big)
=  \Vert Q^{\dagger}|\chi_{\psi}\rangle\Vert^{2}$, where $|\chi_{\psi}\rangle$ is such that $|\chi_{\psi}\rangle\langle\chi_{\psi}| = \mathcal{T}(|\psi\rangle\langle\psi|)$ as in Lemma \ref{TpreservesPurity}. Consequently, $\Vert \mathcal{T}(Q)\Vert 
\leq  \Vert Q^{\dagger}\Vert$. By $ \Vert Q^{\dagger}\Vert= \Vert Q\Vert$ (see, e.g., Theorem 3.9-2 in \cite{Kreyszig}) it  thus follows that  $\Vert \mathcal{T}(Q)\Vert \leq  \Vert Q\Vert$.
By substituting $Q$ with $\mathcal{T}(Q)$ in the above reasoning, and using $\mathcal{T}^2 = I$ one obtains
$\Vert Q\Vert  \leq  \Vert \mathcal{T}(Q)\Vert$.
Hence, $\Vert \mathcal{T}(Q)\Vert =  \Vert Q\Vert$.

Next we make the observation that
$\Vert \mathcal{T}(Q)\Vert_1 =\Tr\sqrt{\mathcal{T}(Q)\mathcal{T}(Q)^{\dagger}} = \Tr\sqrt{\mathcal{T}(Q^{\dagger}Q)}$. By Lemma \ref{ReorderHermImplyPos} we know that $\mathcal{T}$ maps positive operators to positive operators. Hence, $\mathcal{T}(Q^{\dagger}Q)$ is a positive operator, and thus $\sqrt{\mathcal{T}(Q^{\dagger}Q)}$ is well defined and positive (see e.g.~section 9.4 in \cite{Kreyszig}).
By Lemma  \ref{ReorderHermImplyPos} we  know that  $\mathcal{T}(\sqrt{Q^{\dagger}Q})\geq 0$. Moreover, $\mathcal{T}(\sqrt{Q^{\dagger}Q})\mathcal{T}(\sqrt{Q^{\dagger}Q})
=  \mathcal{T}(Q^{\dagger}Q)$. By the reasoning above we thus know that both $\sqrt{\mathcal{T}(Q^{\dagger}Q)}$ and 
  $\mathcal{T}(\sqrt{Q^{\dagger}Q})$ are positive square roots of $\mathcal{T}(Q^{\dagger}Q)$.
However, the positive square root of a positive operator is unique (Theorem 9.4-2 in \cite{Kreyszig}), and thus  
$\sqrt{\mathcal{T}(Q^{\dagger}Q)} =  \mathcal{T}(\sqrt{Q^{\dagger}Q})$. Consequently $\Vert \mathcal{T}(Q)\Vert_1 
 =  \Tr\sqrt{\mathcal{T}(Q^{\dagger}Q)}
=  \Tr \mathcal{T}(\sqrt{Q^{\dagger}Q})
=  \Vert Q\Vert_1$.
\end{proof}

\begin{Lemma}
\label{TOnHermitean}
On a finite-dimensional complex Hilbert space, let $H$ be Hermitian with an orthogonal family of eigenprojectors $\{P_m\}_m$ and corresponding eigenvalues $h_m$, such that $h_m\neq h_{m'}$ whenever $m\neq m'$, and $H = \sum_{m}h_mP_m$. If $\mathcal{T}$ is  a time reversal such that $\mathcal{T}(H) = H$, then $\mathcal{T}(P_m) = P_m$.  Hence, $\mathcal{T}$ preserves the eigenspaces.
\end{Lemma}
\begin{proof}
By Lemma \ref{TpreservesPurity} we know that each $\mathcal{T}(P_m)$ is a projector, and that it projects onto a subspace of the same dimension as $P_m$. Next one can confirm that $H\mathcal{T}(P_m) = \mathcal{T}(H)\mathcal{T}(P_m) = \mathcal{T}(P_mH) = h_m\mathcal{T}(P_m)$. Hence, $\mathcal{T}(P_m)$ must be an eigenprojector corresponding to eigenvalue $h_m$. Since $\mathcal{T}(P_m)$ and $P_m$ projects on spaces of the same dimension, we must have $\mathcal{T}(P_m) = P_m$.
\end{proof}

Given an orthonormal basis $\{|k\rangle\}_{k=1}^{K}$ of a finite-dimensional Hilbert space $\mathcal{H}$, we define the transpose with respect to this basis as
\begin{equation}
\label{DefTranspose}
Q^{\tau} := \sum_{kk'}|k'\rangle\langle k|Q|k'\rangle\langle k|.
\end{equation}
Since the transpose depends on the choice of basis, an obvious question is what happens when we make a change of basis. The following lemma, which we state without proof, specifies how one can express the new transpose in terms of the old.
\begin{Lemma}
\label{dkjgdg}
On a finite-dimensional complex Hilbert space $\mathcal{H}$ let the transpose $\tau_{\mathrm{old}}$ be defined with respect to an orthonormal basis  $\{|\mathrm{old}_k\rangle\}_{k}$. Let the transpose  $\tau_{\mathrm{new}}$ be defined with respect to the orthonormal basis $\{|\mathrm{new}_k\rangle\}_{k}$, where $|\mathrm{new}_k\rangle := W|\mathrm{old}_k\rangle$ for some unitary operator $W$ on $\mathcal{H}$. Then $W^{\tau_{\textrm{new}}} =  WW^{\tau_{\mathrm{old}}}W^{\dagger}$, and the new transpose $\tau_{\mathrm{new}}$ can be expressed in terms of the old basis as 
\begin{equation}
Q^{\tau_{\mathrm{new}}} = WW^{\tau_{\mathrm{old}}}Q^{\tau_{\mathrm{old}}}(WW^{\tau_{\mathrm{old}}})^{\dagger}.
\end{equation}
Similarly, the old transpose $\tau_{\textrm{old}}$ can be expressed in terms of the new basis as
\begin{equation}
Q^{\tau_{\textrm{old}}} = (W^{\tau_{\textrm{new}}}W)^{\dagger} Q^{\tau_{\textrm{new}}}W^{\tau_{\textrm{new}}}W.
\end{equation}
\end{Lemma}

The following Lemma is a special case of Autonne-Takagi's decomposition, see e.g.~Corollary 4.4.4 in \cite{HornJohnson}. 
\begin{Lemma}[Special case of Autonne-Takagi's decomposition]
\label{TakagiUnitarySymmetric}
Let $\mathcal{H}$ be a finite-dimensional complex Hilbert space. Let $U$ be a unitary operator on $\mathcal{H}$. 
Then $U^{\tau} = U$ (with respect to a given orthonormal basis of $\mathcal{H}$) if and only if there exists a unitary operator $W$ on $\mathcal{H}$ such that $
U = WW^{\tau}$.
\end{Lemma}
By combining  Lemma \ref{TakagiUnitarySymmetric} with Lemma \ref{dkjgdg} we can conclude that transformations of transposes are characterized by complex symmetric unitary operators. 
\begin{Proposition}
\label{TcharactFinite}
Let $\mathcal{H}$ be a finite-dimensional complex Hilbert space, and let $\mathcal{B}:=\{|k\rangle\}_{k=1}^{K}$ be an orthonormal basis of $\mathcal{H}$.  Let $\tau$ denote the transpose with respect to the basis $\mathcal{B}$. Let  $\mathcal{T}$ be a linear map on $\mathcal{L}(\mathcal{H})$.
\begin{enumerate}
\item $\mathcal{T}$ satisfies (\ref{Prop3}),  (\ref{Prop4}),  and (\ref{Prop2})  if and only if there exists a unitary operator $U$ on $\mathcal{H}$ such that 
\begin{equation}
\label{lknfnfd}
\mathcal{T}(|k\rangle\langle k'|) = U|k'\rangle\langle k|U^{\dagger},\quad \forall k,k',
\end{equation}
or equivalently
\begin{equation}
\label{ReprTranspose}
\mathcal{T}(Q) = UQ^{\tau}U^{\dagger},\quad \forall Q\in\mathcal{L}(\mathcal{H}).
\end{equation}
Moreover, $U$ is uniquely determined by $\mathcal{T}$ and $\mathcal{B}$ up to a global phase factor.
\item If $\mathcal{T}$ satisfies (\ref{Prop3}),  (\ref{Prop4}),  and (\ref{Prop2}), then the following are equivalent:
\begin{itemize}
\item $\mathcal{T}$ satisfies (\ref{Prop6}).
\item The unitary operator $U$ in (\ref{lknfnfd}) satisfies $U^{\tau} = \pm U $, i.e., $U$ is complex symmetric or complex skew-symmetric. (The choice of global phase factor in $U$ does not affect the property of being symmetric or skew-symmetric.)
\end{itemize}
\item  If $\mathcal{T}$ satisfies (\ref{Prop3}),  (\ref{Prop4}),  and (\ref{Prop2}), then the following are equivalent:
\begin{itemize}
\item There exists an orthonormal basis $\{|\xi_k\rangle\}_{k}$ of $\mathcal{H}$ such that 
\begin{equation}
\label{fghjjjjj}
\mathcal{T}(|\xi_k\rangle\langle\xi_{k'}|) = |\xi_{k'}\rangle\langle\xi_{k}|,\quad \forall k,k'.
\end{equation}
\item The unitary operator $U$ in (\ref{lknfnfd}) satisfies $U^{\tau} = U$, i.e., $U$ is complex symmetric.
\end{itemize}
\end{enumerate}
\end{Proposition}
As a further remark one may note that (\ref{ReprTranspose}) directly implies that $\mathcal{T}$ is a positive but not completely positive map, since it is a composition of a unitary operation and a transpose, and the transpose is not completely positive \cite{Peres96,Horodecki96}.
\begin{proof}[Proof of Proposition \ref{TcharactFinite}]
We start by proving that properties
(\ref{Prop3}),  (\ref{Prop4}), (\ref{Prop2}) implies equation (\ref{lknfnfd}).  
From Lemma \ref{TpreservesPurity} we know that $\{\mathcal{T}(|k\rangle\langle k|)\}_k$ is a set of pairwise orthogonal projectors onto one-dimensional subspaces that span the whole space. This means that there exists a unitary operator $\tilde{U}$ such that 
$\mathcal{T}(|k\rangle\langle k|) = \tilde{U}|k\rangle\langle k|\tilde{U}^{\dagger}$.
Moreover, by (\ref{Prop3}) it follows that $\mathcal{T}(|k\rangle\langle k'|) =  \mathcal{T}(|k\rangle\langle k||k\rangle\langle k'||k'\rangle\langle k'|) 
=  z_{kk'}\tilde{U}|k'\rangle\langle k|\tilde{U}^{\dagger}$, where $z_{kk'}: = \langle k'|\tilde{U}^{\dagger}\mathcal{T}(|k\rangle\langle k'|)\tilde{U}|k\rangle$. By  (\ref{Prop4}) it follows that $z_{kk'}^{*} = z_{k'k}$.
 By using $\mathcal{T}(|k\rangle\langle k|) = \tilde{U}|k\rangle\langle k|\tilde{U}^{\dagger}$ and  (\ref{Prop3})  it follows that 
$z_{k'k}z_{kk''} =  \langle k''|\tilde{U}^{\dagger}\mathcal{T}(|k\rangle\langle k''|)\mathcal{T}(|k\rangle\langle k|)\mathcal{T}(|k'\rangle\langle k|)\tilde{U}|k'\rangle=  z_{k'k''}$.
 One can realize that these two last conditions together imply that $z_{k'k''} = z_{1k'}^{*}z_{1k''}$.
Moreover, $|z_{1k}|^2 = z_{kk} = 1$. Hence, there exist real numbers $\theta_k$ such that 
$z_{k'k''} = e^{i(\theta_{k'}-\theta_{k''})}$. By putting $U := \tilde{U}\sum_{l}e^{-i\theta_l}|l\rangle\langle l|$, we find that $\mathcal{T}(|k\rangle\langle k'|) =   z_{kk'}\tilde{U}|k'\rangle\langle k|\tilde{U}^{\dagger} = U|k'\rangle\langle k|U^{\dagger}$, and thus
(\ref{lknfnfd}) holds. For the opposite implication, assume that there exists a unitary $U$ such that (\ref{ReprTranspose}) holds. 
It is straightforward to confirm that each of the properties (\ref{Prop3}),  (\ref{Prop4}), (\ref{Prop2}) is satisfied. 

For uniqueness, suppose that that there exist two unitary operators $U_1, U_2$ that both satisfy  (\ref{ReprTranspose}).
Consequently, $U_1Q^{\tau}U_1^{\dagger} = U_2Q^{\tau}U_2^{\dagger}$ for all $Q\in\mathcal{L}(\mathcal{H})$, from which it follows that $U_1 = e^{i\chi}U_2$ for some $\chi\in \mathbb{R}$

Next, we turn to the second item of the proposition. Assume that (\ref{Prop3}),  (\ref{Prop4}),  and (\ref{Prop2}) are satisfied.  We know that there exists a unitary operator $U$ such that equation  (\ref{ReprTranspose}) holds. If we use this observation twice we find
\begin{equation}
\label{jdbkjvdkds}
\begin{split}
\mathcal{T}(\mathcal{T}(Q)) = & \mathcal{T}(UQ^{\tau}U^{\dagger})=  U(UQ^{\tau}U^{\dagger})^{\tau}U^{\dagger}\\
 = & (U^{\tau}U^{\dagger})^{\dagger} QU^{\tau}U^{\dagger}.
\end{split}
\end{equation}
This implies that $\mathcal{T}^2 = I$ if and only if $U^{\tau} = e^{i\theta}U$ for some real number $\theta$.
By the definition of the transpose it follows that $\langle k'|U|k\rangle = \langle k|U^{\tau}|k'\rangle = e^{i\theta}  \langle k|U|k'\rangle$ for all $k,k'$. 
If  this equality is iterated we obtain $\langle k'|U|k\rangle  = e^{i\theta}  \langle k|U|k'\rangle =  e^{2i\theta}  \langle k'|U|k\rangle$, and thus $(1- e^{2i\theta})  \langle k'|U|k\rangle = 0$, for all $k,k'$. 
Hence, either $1- e^{2i\theta} = 0$, or $\langle k'|U|k\rangle = 0$ for all  $k,k'$. However, the latter is not possible since $U$ is unitary. We can conclude that $e^{2i\theta} = 1$, and thus $e^{i\theta} = \pm 1$. This combined with $U^{\tau} = e^{i\theta}U$ yields $U^{\tau} = \pm U$.

Finally we turn to the third item of the proposition.  Let $\tau'$ denote the transpose with respect to $\{|\xi_k\rangle\}_{k}$. Then  (\ref{fghjjjjj}) is the same as saying that $\mathcal{T}(Q) = Q^{\tau'}$. By Lemma \ref{dkjgdg} we know that we can express the  transpose $\tau'$ in terms of the transpose $\tau$ (with respect to the basis $\{|k\rangle\}_{k}$) as $Q^{\tau'} = WW^{\tau}Q^{\tau}(WW^{\tau})^{\dagger}$, for a unitary operator $W$ such that $|\xi_{k}\rangle = W|k\rangle$. In terms of the original basis $\{|k\rangle\}_{k}$ we know that $\mathcal{T}(Q) = UQ^{\tau}U^{\dagger}$. Hence,  $UQ^{\tau}U^{\dagger} = WW^{\tau}Q^{\tau}(WW^{\tau})^{\dagger}$, and thus $U = WW^{\tau}e^{i\chi}$ for some $\chi\in\mathbb{R}$. Hence, we can conclude that the unitary operator $U$ satisfies $U^{\tau} = U$ and thus is complex symmetric.

To derive the opposite implication, assume that there exists a unitary operator $U$ such that $U^{\tau} = U$ and  (\ref{lknfnfd}) holds.  By Lemma \ref{TakagiUnitarySymmetric} we know that there exists a unitary operator $W$ on $\mathcal{H}$ such that $U = WW^{\tau}$.
Define $|\xi_{l}\rangle := W|l\rangle$ for all $l$.
Since $W$ is unitary it follows that $\{|\xi_l\rangle\}_{l}$ is an orthonormal basis of $\mathcal{H}$. One can verify that  (\ref{fghjjjjj}) holds. 
\end{proof}

As a corollary of Proposition \ref{TcharactFinite} it follows that if $\mathcal{T}$ leaves the elements of an orthonormal basis intact, then $\mathcal{T}$ can be written as a transpose in this basis followed by phase shifts.
\begin{Corollary}
\label{CorrBasisInvariant}
Let $\{|n\rangle\}_n$ be an orthonormal basis and assume that the time-reversal $\mathcal{T}$ is such that
$\mathcal{T}(|n\rangle\langle n|) = |n\rangle\langle n|$.
Then there exists real numbers $\theta_n$ such that 
$\mathcal{T}(|n\rangle\langle n'|) = e^{i(\theta_{n'}-\theta_{n})}|n'\rangle\langle n|$,
and thus with $|\xi_n\rangle := e^{i\theta_n/2}|n\rangle$ it is the case that $\mathcal{T}(|\xi_n\rangle\langle \xi_{n'}|) = |\xi_{n'}\rangle\langle \xi_{n}|$.
\end{Corollary}

\section{\label{SecAQuantumFlctnThrm}A quantum fluctuation theorem}

\subsection{\label{SecResetting}Re-setting the stage}

Here we construct a new set of assumptions that includes time reversal symmetry (see Fig.~\ref{FigSupplStructureTime}). 
\begin{Assumptions}
\label{Def2}
Let $\mathcal{H}_{S'}$, $\mathcal{H}_C$, and $\mathcal{H}_{E}$ be complex Hilbert spaces.
Let $|c_{i+}\rangle,|c_{i-}\rangle,|c_{f+}\rangle,|c_{f-}\rangle\in\mathcal{H}_C$ be normalized and such that the linear span $\mathcal{H}_{C}^{i} := \Sp\{|c_{i+}\rangle,|c_{i-}\rangle\}$ is orthogonal to $\mathcal{H}_{C}^{f} := \Sp\{|c_{f+}\rangle,|c_{f-}\rangle\}$. Let $P^{i}_{C}$ and $P^{f}_{C}$ denote the projectors onto these two subspaces, and define $P^{\perp}_{C} := \hat{1}_{C}-P^{i}_C-P^{f}_C$.
\begin{itemize}
\item  $H^{i}_{S'}$ and $H^{f}_{S'}$ are Hermitian operators on $\mathcal{H}_{S'}$ such that $Z(H^i_{S'})$ and $Z(H^f_{S'})$ are finite. Let $H_E$ be a Hermitian operator on $\mathcal{H}_{E}$. Let $H^{\perp}$ be a Hermitian operator on $\mathcal{H}_{S'}\otimes\mathcal{H}_C$ such that $[\hat{1}_{S'}\otimes P_{C}^{\perp}]H^{\perp}[\hat{1}_{S'}\otimes P_{C}^{\perp}] = H^{\perp}$, and define
\begin{equation*}
H_{S'C} :=  H^{i}_{S'}\otimes P^{i}_C+ H^{f}_{S'}\otimes P^f_C + H^{\perp}.
\end{equation*}
and $H :=  H_{S'C}\otimes\hat{1}_E + \hat{1}_{S'C}\otimes H_E$.
\item  $\mathcal{T}_{S'C}$ and $\mathcal{T}_{E}$ are time-reversals, and $\mathcal{T} := \mathcal{T}_{S'C}\otimes \mathcal{T}_{E}$. We assume
\begin{equation}
\mathcal{T}_{S'C}(H_{S'C}) = H_{S'C},\quad \mathcal{T}_{E}(H_E) = H_E,
\end{equation}
and
\begin{equation}
\label{AssumpIdc}
\begin{split}
\mathcal{T}_{S'C}(\hat{1}_{S'}\otimes |c_{i+}\rangle\langle c_{i+}|) =  & \hat{1}_{S'}\otimes |c_{i-}\rangle\langle c_{i-}|,\\
 \mathcal{T}_{S'C}(\hat{1}_{S'}\otimes |c_{f+}\rangle\langle c_{f+}|) =  &\hat{1}_{S'}\otimes |c_{f-}\rangle\langle c_{f-}|.
\end{split}
\end{equation} 
\item $V$ is a unitary operator on $\mathcal{H}_{S'}\otimes\mathcal{H}_C\otimes\mathcal{H}_E$ such that 
$[V,H] = 0$, $\mathcal{T}(V) = V$,
and
\begin{equation}
\label{PerfectControlTimeRev}
\begin{split}
& V[\hat{1}_{S'}\otimes |c_{i+}\rangle\langle c_{i+}|\otimes \hat{1}_{E}] \\
& =  [\hat{1}_{S'}\otimes |c_{f+}\rangle\langle c_{f+}|\otimes \hat{1}_{E}]V.
\end{split}
\end{equation}
\end{itemize}
\end{Assumptions}
In Appendices \ref{minimalwith} and \ref{discretizedwith} it is demonstrated that there exist setups that satisfy all  conditions in Assumptions \ref{Def2}.

At first sight it may seem a bit counterintuitive that the time-reversal should leave the unitary evolution $V$  invariant, i.e., that $\mathcal{T}(V) = V$, rather than inverting the evolution, $V\mapsto V^{\dagger}$, which in essence is how the standard complex-conjugation-based time-reversal works. However, as discussed in Appendix \ref{ReverseUseful}, the capacity of  $\mathcal{T}$ to time-reverse stems from the reordering property $\mathcal{T}(AB) = \mathcal{T}(B)\mathcal{T}(A)$, rather than from the inversion of the evolution operator.

The pair of control states $|c_{i+}\rangle$ and $|c_{i-}\rangle$ (as well as $|c_{f+}\rangle$ and $|c_{f-}\rangle$) can be thought of as abstractions of the idea of a wave-packets with fairly well defined momenta, where the time-reversal changes the direction of motion. However, we are not tied to any such specific scenario, and can apply the formalism whenever Assumptions \ref{Def2} is valid for some choice of time-reversal $\mathcal{T}$.

Assumptions \ref{Def2} are constructed in such a way that the unitary $V$, the initial state $|c_i\rangle := |c_{i+}\rangle$, and the final state $|c_{f}\rangle := |c_{f+}\rangle$ satisfy Assumptions \ref{Def1}. This translation only amounts to redefining the projector $P_{C}^{\perp}$ and the Hamiltonian $H^{\perp}$. More precisely, starting with Assumptions \ref{Def2} we can define the projectors $P_{ci+}: = P^{i}_C- |c_{i+}\rangle\langle c_{i+}|$ and $P_{cf+}: = P^{f}_C- |c_{f+}\rangle\langle c_{f+}|$. With $P^{\perp}_C$ being the projector in Assumptions \ref{Def2}, we can define the new projector in Assumptions \ref{Def1} as $\overline{P}^{\perp}_C := P^{\perp}_C + P_{ci+} + P_{cf+}$. Similarly, given the Hamiltonian $H^{\perp}$ in Assumptions \ref{Def2} we can define the new $\overline{H}^{\perp}:= H^{i}_{S'}\otimes P_{ci+}+ H^{f}_{S'}\otimes P_{cf+} + H^{\perp}$ in Assumptions \ref{Def1}. Hence, the restriction to Assumptions \ref{Def1} only requires us to reshuffle the Hamiltonians.
In an analogous manner, $V$, $|c_{i}\rangle := |c_{f-}\rangle$, and $|c_f\rangle := |c_{i-}\rangle$ also form a valid triple in  Assumptions \ref{Def1}. Note that in this case $|c_{f-}\rangle$ is the initial state of the effectively reversed evolution.

One should note that one can consider several variations on Assumptions \ref{Def2}. For example, one could imagine an alternative to (\ref{AssumpIdc}) where we instead assume a product time-reversal $\mathcal{T}_{S'C} = \mathcal{T}_{S'}\otimes \mathcal{T}_C$, and demand  $\mathcal{T}_{C}(|c_{i+}\rangle\langle c_{i+}|) =  |c_{i-}\rangle\langle c_{i-}|$, $\mathcal{T}_{C}(|c_{f+}\rangle\langle c_{f+}|) =  |c_{f-}\rangle\langle c_{f-}|$. However, the assumption in (\ref{AssumpIdc}) is more general, and provides a rather useful flexibility.

One should keep in mind that although $\Sp\{|c_{i+}\rangle,|c_{i-}\rangle\}$ is orthogonal to $\Sp\{|c_{f+}\rangle,|c_{f-}\rangle\}$ we do not necessarily assume that $|c_{i+}\rangle$ is orthogonal to  $|c_{i-}\rangle$ (In principle, Assumptions \ref{Def2} even allows for the possibility that $|c_{i+}\rangle$ is parallel to $|c_{i-}\rangle$.) The reason is that a generic state is not orthogonal to its time-reversal. (The same remark applies to the standard notion of time-reversal via complex conjugation.) To see this, suppose that $\mathcal{T}$ can be implemented as the  transpose with respect to some orthonormal basis $\{|\xi_n\rangle\}_n$. For $|c_{i+}\rangle =  \sum_n c_n |\xi_n\rangle$ we would thus have $|c_{i-}\rangle =  e^{i\theta}\sum_n c_n^{*} |\xi_n\rangle$, for some arbitrary real number $\theta$. Hence, $\langle c_{i-}|c_{i+}\rangle = e^{-i\theta}\sum_n c_n^2$, which would be zero only for a particular sub-manifold of states. Even though the time-reversal  \emph{per se}  does not force $|c_{i+}\rangle$ and  $|c_{i-}\rangle$ to be orthogonal, one may still wonder if the conditions in  Assumptions \ref{Def2} would `conspire' to enforce this. However, this is not the case, as is shown by an explicit example in Appendix \ref{nonorthogonalminimalwith}. Nevertheless, if we wish to incorporate certain additional features, like sequential paths of orthogonal control spaces, as we do in Appendix \ref{discretizedwith}, then we do need orthogonality between the control states and their time-reversals (see further discussions in Appendix \ref{nonorthogonalminimalwith}).

Due to the time-reversal symmetry, a perfect transition from $|c_{i+}\rangle$ to $|c_{f+}\rangle$ implies a perfect transition of $|c_{f-}\rangle$ to $|c_{i-}\rangle$. More precisely, 
by combining (\ref{PerfectControlTimeRev}) with the properties $\mathcal{T}(AB) =\mathcal{T}(B)\mathcal{T}(A)$, $\mathcal{T}_{E}(\hat{1}_E) = \hat{1}_E$, as well as the assumptions $\mathcal{T}(V)= V$ and (\ref{AssumpIdc}), one obtains
\begin{equation}
\label{PerfectControlTimeRev2}
\begin{split}
& V[\hat{1}_{S'}\otimes |c_{f-}\rangle\langle c_{f-}|\otimes \hat{1}_{E}]\\
 & =   [\hat{1}_{S'}\otimes |c_{i-}\rangle\langle c_{i-}|\otimes \hat{1}_{E}] V.
\end{split}
\end{equation}

As the reader may have noticed, a considerable part of Assumptions \ref{Def2} deals with the control system, which is due to the rather strong idealization that perfect control entails. In Appendix \ref{SecConditional} we will abandon this idealization, and as a bonus we also obtain a leaner set of assumptions (cf.~Assumptions \ref{AssumpCondTimeRev}).  

\begin{figure}[t]
 \includegraphics[width= 6cm]{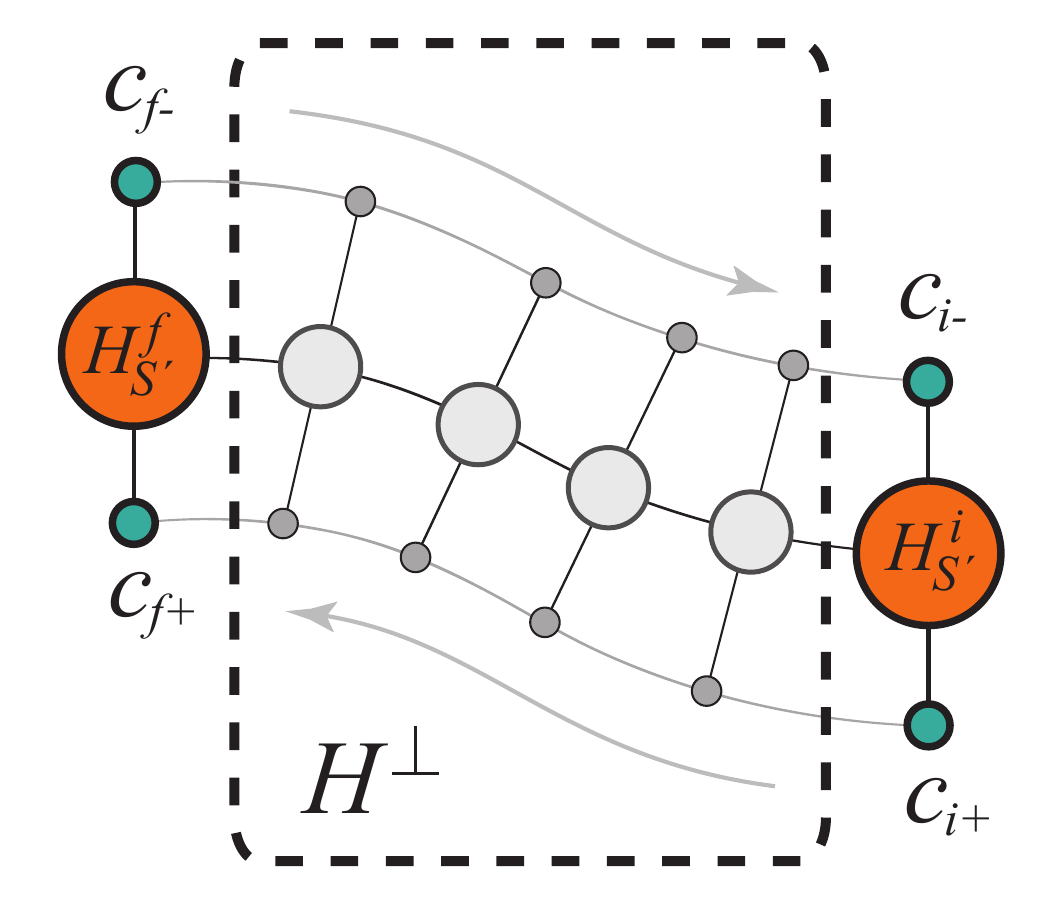} 
   \caption{\label{FigSupplStructureTime}  {\bf Structure of the Hamiltonian.}
The Hamiltonian $S'C$ is of the form $H_{S'C}  = H^{i}_{S'}\otimes P_C^i+ H^{f}_{S'}\otimes P_C^f + H^{\perp}$. Hence, whether the state of the control is in subspace $\mathcal{H}^i_C$ onto which $P_C^i$ projects, or in  $\mathcal{H}^f_C$ onto which $ P_C^f $ projects, determines the Hamiltonian of $S'$.
The Hamiltonian $H^{\perp}$ corresponds to any possible intermediate stages.
The initial control space  $\mathcal{H}^i_C$ is spanned by two states $|c_{i+}\rangle$ and $|c_{i-}\rangle$, which are the time-reversals of each other.
Analogously, the final control space $\mathcal{H}^f_C$ is spanned by $|c_{f+}\rangle$ and $|c_{f-}\rangle$. The global evolution $V$ is such that it brings control state $|c_{i+}\rangle$ into  $|c_{f+}\rangle$, while it brings $|c_{f-}\rangle$ into  $|c_{i-}\rangle$, thus implementing both the forward and the reverse process.  
   }
\end{figure}

\begin{Lemma}
\label{TimeRevOnControlledGibbs}
With Assumptions \ref{Def2} it is the case that 
\begin{equation}
\label{AssumpGc}
\begin{split}
\mathcal{T}_{S'C}\big( G(H^i_{S'})\otimes |c_{i\pm}\rangle\langle c_{i\pm}|\big) =  & G(H^i_{S'})\otimes |c_{i\mp}\rangle\langle c_{i\mp}|,\\
 \mathcal{T}_{S'C}\big( G(H^f_{S'})\otimes |c_{f\pm}\rangle\langle c_{f\pm}|\big) =  &G(H^f_{S'})\otimes |c_{f\mp}\rangle\langle c_{f\mp}|.
\end{split}
\end{equation}
\end{Lemma}
One may wonder why we do not directly assume (\ref{AssumpGc}) in Assumptions \ref{Def2} rather than (\ref{AssumpIdc}). The reason is partially that the latter choice defines the action of the time reversal on the control states in a cleaner manner, but also because it aligns with the more general set of assumptions that we will use in Appendix  \ref{SecCondWithTimeRev}.

\begin{proof}
We first note that $\mathcal{T}_{S'C}(H_{S'C}) = H_{S'C}$ and $\mathcal{T}_{S'C}(\hat{1}_{S'C}) = \hat{1}_{S'C}$ implies
$\mathcal{T}_{S'C}(e^{\alpha H_{S'C}}) =  e^{\alpha H_{S'C}}$  (e.g.~via a Taylor expansion)
for all $\alpha \in \mathbb{C}$. Moreover, due to the orthogonal supports of $H^{i}_{S'}\otimes P^{i}_C$, $H^{f}_{S'}\otimes P^f_C$, and  $H^{\perp}$  we can conclude that 
\begin{equation}
\label{dfksvn}
\begin{split}
e^{-\beta H^i_{S'}}\otimes |c_{i\pm}\rangle\langle c_{i\pm}| = & e^{-\beta H_{S'C}}[\hat{1}_{S'}\otimes |c_{i\pm}\rangle\langle c_{i\pm}|] \\
= & [\hat{1}_{S'}\otimes |c_{i\pm}\rangle\langle c_{i\pm}|] e^{-\beta H_{S'C}}.
\end{split}
\end{equation}
We exemplify the rest of the derivation with the transformation of $G_{\beta}(H^i_{S'})\otimes |c_{i+}\rangle\langle c_{i+}|$.
\begin{equation*}
\begin{split}
& \mathcal{T}_{S'C}\big( G(H^i_{S'})\otimes |c_{i+}\rangle\langle c_{i+}|\big)\\
& [\textrm{By Eq.~(\ref{dfksvn})}]\\
= & \frac{1}{Z(H^i_{S'})}\mathcal{T}_{S'C}(e^{-\beta H_{S'C}}[\hat{1}_{S'}\otimes |c_{i+}\rangle\langle c_{i+}|])  \\
& [\textrm{By property (\ref{Prop3}) and $\mathcal{T}_{S'C}(e^{\alpha H_{S'C}}) =  e^{\alpha H_{S'C}}$}]\\
= & \frac{1}{Z(H^i_{S'})} \mathcal{T}_{S'C}(\hat{1}_{S'}\otimes |c_{i+}\rangle\langle c_{i+}|)e^{-\beta H_{S'C}} \\
= & \frac{1}{Z(H^i_{S'})}[\hat{1}_{S'}\otimes |c_{i-}\rangle\langle c_{i-}|]e^{-\beta H_{S'C}} \\
& [\textrm{By Eq.~(\ref{dfksvn})}]\\
=  & G(H^i_{S'})\otimes |c_{i-}\rangle\langle c_{i-}|.
\end{split}
\end{equation*}
The other identities can be derived in an analogous manner.
\end{proof}

\subsection{The induced channels}

Given the initial state $|c_{i+}\rangle$ and final state $|c_{f+}\rangle$ we define the channels
\begin{equation}
\label{DefFRplus}
\begin{split}
\mathcal{F}_{+}(\sigma) :=  \Tr_{S'C}(V [G_{\beta}(H^{i}_{S'})\otimes |c_{i+}\rangle \langle c_{i+}| \otimes \sigma] V^{\dagger}),\\
\mathcal{R}_{+}(\sigma) :=  \Tr_{S'C}(V^{\dagger} [G_{\beta}(H^{f}_{S'})\otimes |c_{f+}\rangle \langle c_{f+}| \otimes \sigma] V).
\end{split}
\end{equation}
For the initial state $|c_{f-}\rangle$ and the final state $|c_{i-}\rangle$ we similarly define the channels
\begin{equation}
\label{DefFRminus}
\begin{split}
\mathcal{F}_{-}(\sigma) :=  \Tr_{S'C}(V [G_{\beta}(H^{f}_{S'})\otimes |c_{f-}\rangle \langle c_{f-}| \otimes \sigma] V^{\dagger}),\\
\mathcal{R}_{-}(\sigma) :=  \Tr_{S'C}(V^{\dagger} [G_{\beta}(H^{i}_{S'})\otimes |c_{i-}\rangle \langle c_{i-}| \otimes \sigma] V).
\end{split}
\end{equation}

\subsection{Deriving a quantum Crooks relation}

\begin{Lemma}
\label{ChannelReversal}
Given Assumptions \ref{Def2}, then the channels $\mathcal{R}_{+}$, $\mathcal{F}_{-}$, as defined in Eqs.~(\ref{DefFRplus})  and (\ref{DefFRminus}) are related as
\begin{equation}
\label{fnmbsf}
\mathcal{T}_E\mathcal{R}_{+} = \mathcal{F}_{-}\mathcal{T}_E
\end{equation}
and thus 
\begin{equation}
\label{RconjFominus}
\mathcal{R}_{+}^{*} = \mathcal{F}_{-}^{\ominus}.
\end{equation}
\end{Lemma}
Hence, under the assumption of time-reversal symmetry we can in effect simulate the reversed time evolution (i.e., the replacement of $V$ with $V^{\dagger}$) via the `forward' evolution $V$.  
By applying the property $\mathcal{T}_E^2 = I$ to (\ref{fnmbsf}) one can also show $\mathcal{R}_{+}\mathcal{T}_E = \mathcal{T}_E\mathcal{F}_{-}$ and $\mathcal{R}_{+} = \mathcal{T}_E\mathcal{F}_{-}\mathcal{T}_E$.

\begin{proof}[Proof of Lemma \ref{ChannelReversal}]
We use the definition of $\mathcal{R}_{+}$ in Eq.~(\ref{DefFRplus}) and the general relation $\Tr_{2}\big([\mathcal{T}_1\otimes\mathcal{T}_2](\rho)\big) = \mathcal{T}_1\big(\Tr_2(\rho)\big)$ to obtain
\begin{equation}
\begin{split}
& \mathcal{T}_E\circ\mathcal{R}_{+}(\sigma)\\
  = & \Tr_{S'C}\Big( \mathcal{T}\big(V^{\dagger} [G(H^{f}_{S'})\otimes |c_{f+}\rangle \langle c_{f+}| \otimes \sigma] V\big)\Big)\\
  & [\textrm{By  (\ref{Prop3}), (\ref{Prop4}), and $\mathcal{T}(V) = V$}]\\
     = & \Tr_{S'C}\Big(   V\mathcal{T}\big(G(H^{f}_{S'})\otimes |c_{f+}\rangle \langle c_{f+}| \otimes \sigma\big)V^{\dagger} \Big)\\
           &[\textrm{By Lemma \ref{TimeRevOnControlledGibbs}}] \\
       = & \Tr_{S'C}\big(   V[G(H^{f}_{S'})\otimes |c_{f-}\rangle \langle c_{f-}| \otimes \mathcal{T}_{E}(\sigma)]V^{\dagger} \big)\\  
       = &  \mathcal{F}_{-}\circ\mathcal{T}_{E}(\sigma).
\end{split}
\end{equation}
By multiplying  (\ref{fnmbsf}) from the left with $\mathcal{T}$ and using the relation $\mathcal{T}^2 = I$ and the 
  the alternative definition of $\ominus$ in (\ref{Altominusdef}) we obtain (\ref{RconjFominus}).
\end{proof}

\begin{Proposition}[Quantum Crooks relation]
\label{PropQuantumCrooks} With Assumptions \ref{Def2}, the channels $\mathcal{F}_{\pm}$ as defined in Eqs.~(\ref{DefFRplus}) and (\ref{DefFRminus}) satisfy
\begin{equation}
\label{FplusequivFminus}
Z(H^{i}_{S'})\mathcal{F}_{+} = Z(H^{f}_{S'})\mathcal{J}_{\beta H_E}\mathcal{F}_{-}^{\ominus}\mathcal{J}_{\beta H_E}^{-1}.
\end{equation}
With the separation of $S'$ into system $S$ and the heat bath $B$ as in equation (\ref{Def1HeatBath}) we thus get
\begin{equation}
\label{FplusequivFminusS}
Z(H^{i}_{S})\mathcal{F}_{+} = Z(H^{f}_{S})\mathcal{J}_{\beta H_E}\mathcal{F}_{-}^{\ominus}\mathcal{J}_{\beta H_E}^{-1}.
\end{equation}
\end{Proposition}
\begin{proof}
The triple $V$, $|c_i\rangle := |c_{i+}\rangle$, $|c_f\rangle :=|c_{f+}\rangle$ from Assumptions \ref{Def2} satisfies Assumptions \ref{Def1}. It follows that we can apply  
Proposition \ref{PreliminaryQuantumCrooks} on the pair of channels $\mathcal{F}_{+}$ and $\mathcal{R}_{+}$ and thus obtain 
$Z(H^i_{S'}) \mathcal{F}_{+} = Z(H^f_{S'})\mathcal{J}_{\beta H_E}\mathcal{R}_{+}^{*}\mathcal{J}_{\beta H_E}^{-1}$.
Next, we use Eq.~(\ref{RconjFominus}) to obtain Eq.~(\ref{FplusequivFminus}).

With the additional assumption in equation (\ref{Def1HeatBath}) we get $Z(H^{i}_{S'}) = Z(H^{i}_{S})Z(H_{B})$ and $Z(H^{f}_{S'}) = Z(H^{f}_{S})Z(H_{B})$. From this it follows that (\ref{FplusequivFminus}) yields  (\ref{FplusequivFminusS}).
\end{proof}

\subsection{\label{SecConditionsOnE} Unbounded $H_E$}

The requirement of perfect control, i.e., that $V$ satisfies (\ref{PerfectControlTimeRev}) puts rather stringent conditions on the properties of $H_E$. To see this, let us assume  that $H_E$ has a pure point spectrum corresponding to the orthonormal eigenvectors $\{|n\rangle\}_{n}$ with respect to the energy eigenvalues $E_n$. Let us furthermore assume that $\mathcal{H}_{S'}$ is finite-dimensional. Here we shall see that for generic choices of initial and final Hamiltonians $H^{i}_{S'}$ and $H^{f}_{S'}$, the perfect control implies that the spectrum of $H_E$ must be unbounded from both above and below.

Due to the assumption of energy conservation, the energy reservoir has to compensate for any change in energy in the transition from the initial to the final state. Let $h^{i}_{n}$ be the eigenvalues of $H^{i}_{S'}$ and similarly $h^{f}_{m}$ the eigenvalues of $H^{f}_{S'}$. 
Suppose that $h^{f}_{m}\neq h^{i}_{n}$ for all $m,n$. This means that every possible transition either must cost or yield energy, which has to be drawn from or deposited in the reservoir $E$.  
Imagine now that $S'$ initially is in an eigenstate $|h^i_n\rangle$. Suppose that at the end of the process there is a non-zero probability for finding $S'$ in the state $|h^{f}_{m}\rangle$ with $h^{f}_{m}> h^{i}_{m}$. For this to happen, the reservoir  has to donate the energy $q := h^{i}_{m} - h^{f}_{n}$.  
Suppose that the spectrum of $H_E$ would be bounded  from below, i.e.,  $E_{\textrm{lower}} = \inf_{n}E_{n} > -\infty$.
This means that there exists some state $|k\rangle$ of the reservoir such that all transitions downwards in energy (if any available) would be smaller than $q$. In other words, if the energy reservoir would start in state $|k\rangle$, then it cannot donate the energy $q$, and the transition cannot occur. For a reservoir with a spectrum bounded from below, the only way to avoid this would be if all transitions in $S'$ always would go downwards in energy.  Generic choices of  $H^{i}_{S'}$ to $H^{f}_{S'}$ would involve both increases and  decreases in energy, and thus the spectrum of $H_E$ must be unbounded from both above and below. The key point behind the unboundedness is the demand that the control system always should succeed in its task irrespective of the state of the system and the energy reservoir. It would be reasonable with a control system fails in some cases, e.g., if the energy in the reservoir is too low (i.e., too close to the ground state). In Appendix  \ref{SecConditional} we introduce conditional fluctuation relations that allows for failing control systems.  (For an explicit example, see Appendix \ref{AbolishPerfectControl}.)

\section{\label{SecDiagonalAndOffdiagonal}Diagonal and off-diagonal Crooks relations}

\subsection{\label{DecouplingDiagonals} Decoupling of diagonals}

We here demonstrate the useful fact that the dynamics under the induced channels $\mathcal{F}_{\pm}$ and $\mathcal{R}_{\pm}$ decouples along different diagonals or modes of coherence \cite{Lostaglio14b}. We first show that the channels $\mathcal{F}_{\pm}$ and $\mathcal{R}_{\pm}$ commute with the commutator with respect to $H_E$.
\begin{Lemma}
\label{ComWithCom}
With Assumptions \ref{Def2}, the channels $\mathcal{F}_{\pm}$ and $\mathcal{R}_{\pm}$ as defined in Eqs.~(\ref{DefFRplus}) and (\ref{DefFRminus}), satisfy 
\begin{equation}
\begin{split}
[H_E,\mathcal{F}_{\pm}(\sigma)] = & \mathcal{F}_{\pm}([H_E,\sigma]),\\
 [H_E,\mathcal{R}_{\pm}(\sigma)] = & \mathcal{R}_{\pm}([H_E,\sigma]).
\end{split}
\end{equation}
\end{Lemma}
\begin{proof}
Here we only show the relation $[H_E,\mathcal{F}_{+}(\sigma)] = \mathcal{F}_{+}([H_E,\sigma])$.
By the definition of the Hamiltonian $H$ in Assumptions \ref{Def2} it follows that $\hat{1}_{S'}\otimes |c_{f+}\rangle \langle c_{f+}|\otimes H_E = (H- H^f_{S'}\otimes |c_{f+}\rangle \langle c_{f+}|\otimes \hat{1}_E)[\hat{1}_{S'}\otimes |c_{f+}\rangle \langle c_{f+}|\otimes \hat{1}_E]$. By combining this observation with the perfect control (\ref{PerfectControlTimeRev}) one can show that 
\begin{equation}
\label{utoioitoiu}
\begin{split}
& [H_E,\mathcal{F}_{+}(\sigma)] \\
= & \Tr_{S'C}\Big( \Big[H, V [G_{\beta}(H^{i}_{S'})\otimes |c_{i+}\rangle \langle c_{i+}| \otimes \sigma] V^{\dagger}\Big]\Big)\\
& -\Tr_{S'C}\Big( \Big[H^f_{S'}\otimes |c_{f+}\rangle \langle c_{f+}|\otimes \hat{1}_E, \\
&\quad\quad \quad \quad V [G_{\beta}(H^{i}_{S'})\otimes |c_{i+}\rangle \langle c_{i+}| \otimes \sigma] V^{\dagger}\Big]\Big),
\end{split}
\end{equation}
where the last term becomes zero due to the cyclic property of the partial trace $\Tr_{S'C}$ with respect to $H^f_{S'}\otimes |c_{f+}\rangle \langle c_{f+}|\otimes \hat{1}_E$. By the definition of the global Hamiltonian $H$ in Assumptions \ref{Def2}
\begin{equation*}
\begin{split}
& [H, G_{\beta}(H^{i}_{S'})\otimes |c_{i+}\rangle \langle c_{i+}| \otimes \sigma] \\
& =  [H^{i}_{S'},G_{\beta}(H^{i}_{S'})]\otimes |c_{i+}\rangle \langle c_{i+}| \otimes \sigma \\
&\quad  + G_{\beta}(H^{i}_{S'})\otimes |c_{i+}\rangle \langle c_{i+}| \otimes [H_E,\sigma]\\
& =  G_{\beta}(H^{i}_{S'})\otimes |c_{i+}\rangle \langle c_{i+}| \otimes [H_E,\sigma].
\end{split}
\end{equation*}
By combining this with  $[H,V] = 0$ in (\ref{utoioitoiu})
the lemma follows.
\end{proof}

\begin{Corollary}
\label{MappingOffDiagonal} 
Suppose that $H_E$ has a complete orthonormal eigenbasis $\{|n\rangle\}_{n}$ with corresponding eigenvalues $E_n$. Then
\begin{equation}
\label{kjsdfbvab}
\begin{split}
& \langle m|\mathcal{F}_{\pm}(|n\rangle\langle n'|)|m'\rangle = 0\\
& \quad \textrm{if}\quad E_{m}-E_{n}\neq E_{m'}-E_{n'}.
\end{split}
\end{equation}
 The analogous statement holds for $\mathcal{R}_{\pm}$.
\end{Corollary}
\begin{proof}
By Lemma \ref{ComWithCom} 
\begin{equation*}
\begin{split}
& (E_{m}-E_{m'})\langle m|\mathcal{F}_{\pm}(|n\rangle\langle n'|)|m'\rangle \\
 & =  \langle m|[H,\mathcal{F}_{\pm}(|n\rangle\langle n'|)]|m'\rangle\\
 & =  \langle m|\mathcal{F}_{\pm}([H,|n\rangle\langle n'|])|m'\rangle\\
 & =  (E_{n}-E_{n'})\langle m|\mathcal{F}_{\pm}(|n\rangle\langle n'|)|m'\rangle.
\end{split}
\end{equation*}
Thus 
$(E_{m}-E_{m'}-E_{n}+E_{n'})\langle m|\mathcal{F}_{\pm}(|n\rangle\langle n'|)|m'\rangle = 0$.
\end{proof}

If $H_E$ is non-degenerate (i.e., $E_{n}=E_{n'}$ if and only if $n = n'$) then it follows by Corollary \ref{MappingOffDiagonal} that 
$\langle m|\mathcal{F}_{\pm}(Q)|m\rangle  =  \sum_{n} \langle m|\mathcal{F}_{\pm}\big(|n\rangle\langle n|Q|n\rangle\langle n|\big)|m\rangle$.

Hence, $\mathcal{F}_{\pm}$ cannot `create' off-diagonal elements with respect to the energy eigenbasis. Moreover, if we are only interested in the diagonal elements of the output, we only need to consider the diagonal elements of the input. Another way to put this is that the statistics of an energy measurement on the output is unaffected by an additional energy measurement on the input.

\subsection{\label{DiagonalCase} Diagonal Crooks relations}

Let us assume that $H_E$ has a pure non-degenerate point spectrum with eigenenergies $E_n$ corresponding to the complete orthonormal eigenbasis $\{|n\rangle\}_{n}$. We furthermore assume that $\mathcal{T}_E(|n\rangle\langle n|) = |n\rangle\langle n|$. (Due to Corollary \ref{CorrBasisInvariant} this is only a very minor generalization compared to assuming that $\mathcal{T}_E$ is the transpose with respect to $\{|n\rangle\}_{n}$.)
 
Imagine now that we represent the density operator of the energy reservoir as a matrix with respect to the basis  $\{|n\rangle\}_{n}$.
Since  $\mathcal{F}_{\pm}$ are channels it follows that
the numbers
\begin{equation} 
\label{Pplusminusdef}
\begin{split}
p_{\pm}(m|n) :=  \langle m|\mathcal{F}_{\pm}(|n\rangle\langle n|)|m\rangle,
\end{split}
\end{equation}
can be interpreted as conditional probability distributions.

\begin{Proposition}
\label{PropdiagonalCrooks}
With Assumptions \ref{Def2}, assume that $H_E$ has a complete orthonormal eigenbasis $\{|n\rangle\}_n$ with corresponding eigenvalues $E_n$, and let $\mathcal{T}_E$ be such that $\mathcal{T}_E(|n\rangle\langle n|) = |n\rangle\langle n|$. Then the conditional distributions $p_{+}(m|n)$ and $p_{-}(n|m)$ defined in (\ref{Pplusminusdef}) satisfy
 \begin{equation}
\label{diagonalCrooksRel}
 Z(H^i_{S})  p_{+}(m|n) = Z(H^f_{S}) e^{\beta (E_n-E_m)}p_{-}(n|m).
\end{equation}
\end{Proposition}
One may note that (\ref{diagonalCrooksRel}) holds for any pair of eigenvectors $|n\rangle, |m\rangle$ that are invariant under the time-reversal, irrespective of whether $H_E$ possesses a complete eigenbasis or not. (Similar remarks also apply to Corollary \ref{MappingOffDiagonal}  and Proposition \ref{Propoffdiag}.) 
\begin{proof}[Proof of Proposition \ref{PropdiagonalCrooks}]
If we apply both sides of equation (\ref{FplusequivFminusS}) on the operator $|n\rangle\langle n|$, and operate on both sides of the resulting equality with $\langle m|\cdot |m\rangle$ we obtain
\begin{equation*}
\begin{split}
& Z(H^{i}_{S})\langle m|\mathcal{F}_{+}(|n\rangle\langle n|)|m\rangle\\
& = Z(H^{f}_{S})e^{\beta (E_n-E_m)}\langle m|\mathcal{F}_{-}^{\ominus}(|n\rangle\langle n|)|m\rangle.
\end{split}
\end{equation*}
With the invariance of $|n\rangle\langle n|$ under the time-reversal, we find $\langle m|\mathcal{F}_{-}^{\ominus}(|n\rangle\langle n|)|m\rangle = \langle n|\mathcal{F}_{-}(|m\rangle\langle m|)|n\rangle$. With the identifications in (\ref{Pplusminusdef}) the proposition follows. 
\end{proof}

In Appendix \ref{SecRegainingClassicalCrooks} we shall use the additional assumption of energy translation invariance on the energy reservoir to show how (\ref{diagonalCrooksRel}) leads to the standard classical Crooks relation.

\subsection{\label{Secoffdiagonal}Off-diagonal Crooks relations}

Like in Appendix \ref{DiagonalCase} we here assume a discrete non-degenerate spectrum of $H_E$, with corresponding orthonormal eigenbasis $\{|n\rangle\}_{n}$ and energy eigenvalues $E_n$. We also assume that $\mathcal{T}_E$ is the transpose with respect to this basis, and thus $\mathcal{T}_E(|n\rangle\langle n'|) = |n'\rangle\langle n|$. 

As discussed in Appendix \ref{DecouplingDiagonals}, the channel $\mathcal{F}_{+}$ can only induce transitions between $|n\rangle\langle n'|$ and $|m\rangle\langle m'|$ if $E_{n} -E_{n'} = E_{m}-E_{m'}$. For each $\delta$ we can thus define a corresponding set of operators $\{|n\rangle\langle n'|\}_{n,n':E_{n}-E_{n'} = \delta}$. (This set would be empty for many values of $\delta$.) For each such $\delta$, we will construct a Crooks relation, analogous to what we did for the diagonal case.

As the generalization of $p_{+}(m|n)$ and $p_{-}(n|m)$ we define 
\begin{equation}
\label{DefQpm} 
\begin{split}
q_{\pm}^{\delta}(m|n) := & \langle m|\mathcal{F}_{\pm}(|n\rangle\langle n'|)|m'\rangle,
 \end{split}
\end{equation}
where  $q^{0}_{\pm} = p_{\pm}$.
The reason for why it is enough to write `$q_{\pm}^{\delta}(m|n)$' rather than `$q_{\pm}^{\delta}(mm'|nn')$' is that $m'$ and $n'$ are uniquely determined by $\delta$, $m$, and $n$, due to the assumption that $H_E$ is non-degenerate.

The set of numbers $q^{\delta}_{\pm}(m|n)$ represent the channels $\mathcal{F}_{\pm}$ in the sense that
\begin{equation}
\label{nnabna}
\mathcal{F}_{\pm}(\rho) = \sum_{\delta}  \sum_{m,n}\sum_{\substack{n'm':E_n-E_{n'} =\delta\\
 E_{m}-E_{m'}= \delta}} q^{\delta}_{\pm}(m|n)  |m\rangle\langle n|\rho|n'\rangle\langle m'|.
\end{equation}

\begin{Proposition}
\label{Propoffdiag}
With Assumptions \ref{Def2}, assume that $H_E$ has a complete orthonormal eigenbasis $\{|n\rangle\}_n$ with corresponding non-degenerate eigenvalues $E_n$. Let $\mathcal{T}_E$ be the transpose with respect to  $\{|n\rangle\}_n$. Then $q_{\pm}^{\delta}(m|n)$ defined in (\ref{DefQpm}) satisfy
\begin{equation}
\label{grkgsslkjg}
Z(H^{i}_{S})q^{\delta}_{+}(m|n) =  Z(H^{f}_{S}) e^{\beta(E_{n}-E_{m})}  q_{-}^{\delta}(n|m).
\end{equation}
\end{Proposition}
\begin{proof}
Let $n,n'$ and $m,m'$ be such that $E_{n}-E_{n'} = E_{m}-E_{m'} = \delta$. 
If we apply both sides of equation (\ref{FplusequivFminusS}) on the operator $|n\rangle\langle n'|$, and operate on both sides of the resulting equality with $\langle m|\cdot |m'\rangle$ we obtain
\begin{equation*}
\begin{split}
& Z(H^{i}_{S})\langle m|\mathcal{F}_{+}(|n\rangle\langle n'|)|m'\rangle \\
& =   Z(H^{f}_{S}) e^{\beta(E_{n}-E_{m})}\langle m|\mathcal{F}_{-}^{\ominus}(|n\rangle\langle n'|)|m'\rangle,
\end{split}
\end{equation*}
 where we have made use of $E_{n'} = E_{n} - \delta$ and $E_{m'} = E_{m} - \delta$.
 With the assumption that $\mathcal{T}_E$ is the transpose with respect to $\{|n\rangle\}_{n}$, together with the identifications in equation (\ref{DefQpm}), we obtain the proposition. 
\end{proof}

As mentioned in the main text we need to use off-diagonal initial states as well as off-diagonal measurement operators in order to determine the numbers $q_{\pm}^{\delta}(m|n)$ in a `prepare and measure'-experiment. There are many possible arrangements, but let us here construct a setup that determines these numbers via  interference. Let $n,n'$ and $m,m'$ with $n\neq n'$ and $m\neq m'$ be such that $\delta = E_n-E_{n'} = E_{m}-E_{m'}$, with non-degenerate $E_n$. Define the POVM element $A := (|m\rangle+|m'\rangle)(\langle m|+\langle m'|)/2$ and the family of initial states $|\psi_{\theta}\rangle := (|n\rangle + e^{i\theta}|n'\rangle)/\sqrt{2}$. Then the probability of measuring $A$ on the evolved state is 
\begin{equation*}
\begin{split}
& \Tr\big(A\mathcal{F}_{\pm}(|\psi_{\theta}\rangle\langle\psi_{\theta}|)\big)\\
= & \frac{1}{4}\big(p_{\pm}(m|n)+p_{\pm}(m'|n) +p_{\pm}(m|n') +p_{\pm}(m'|n')\big)\\
& + \frac{1}{2}|q^{\delta}_{\pm}(m|n)|\cos\Big(\arg\big(q_{\pm}^{\delta}(m|n)\big)-\theta\Big),
\end{split}
\end{equation*}
where we have made use of the decoupling, and the fact that $\langle m'|\mathcal{F}_{\pm}(|n'\rangle\langle n|)|m\rangle = \langle m|\mathcal{F}_{\pm}(|n\rangle\langle n'|)|m'\rangle^{*}$ by  virtue of the complete positivity of $\mathcal{F}_{\pm}$. Hence, the magnitude and phase of $q^{\delta}_{\pm}(m|n)$ can be determined via the amplitude and phase-shift of the interference pattern with respect to the phase $\theta$.

\section{\label{SecQuantumJarzynski} Jarzynski equalities}

Jarzynski's equality \cite{Jarzynski97} can be formulated as 
$\langle e^{-\beta W}\rangle = Z(H^f)/Z(H^i)$.  
This is often written in the more elegant form  $\langle e^{-\beta (W-\Delta F)}\rangle = 1$, where $\Delta F = F(H^f)-F(H^i)$, with $F(H) = -kT\ln Z(H)$, is the (equilibrium) free-energy difference between the initial and final state. 
Here we obtain the following family of quantum Jarzynski equalities
\begin{Proposition}
\label{CollectionJarzynski}
With Assumptions \ref{Def2}, the channels $\mathcal{F}_{+}$ and $\mathcal{R}_{+}$ as defined in Eq.~(\ref{DefFRplus}) satisfy 
\begin{equation}
\label{eq1jarzynski}
\begin{split}
&\Tr\Big[e^{\beta H_E}\mathcal{F}_{+}\Big(e^{(-\beta  + r + z)H_E/2 } \rho e^{(-\beta+ r -z) H_E/2}\Big) \Big]\\
& = \frac{Z(H^{f}_{S'})}{ Z(H^{i}_{S'})}\Tr[e^{r H_E/2}\mathcal{R}_{+}(\hat{1}) e^{rH_E/2 } \rho],
\end{split}
\end{equation}
for $r\in\mathbb{R}$ and $z\in \mathbb{C}$.

Hence, if $\mathcal{R}_{+}(\hat{1}_{E}) = \hat{1}_{E}$, then
\begin{equation}
\label{eq2jarzynski}
\begin{split}
& \Tr\Big[e^{\beta H_E}\mathcal{F}_{+}\Big(e^{(-\beta  + r + z)H_E/2 } \rho e^{(-\beta+ r -z) H_E/2}\Big) \Big]\\
& = \frac{Z(H^{f}_{S'})}{ Z(H^{i}_{S'})}\Tr(e^{rH_E}\rho).
\end{split}
\end{equation}
The condition $\mathcal{R}_{+}(\hat{1}_{E}) = \hat{1}_{E}$ is equivalent to $\mathcal{F}_{-}(\hat{1}_{E}) = \hat{1}_{E}$, with $\mathcal{F}_{-}$ as defined in Eq.~(\ref{DefFRminus}).
\end{Proposition}
The relations (\ref{fsjaflkslk}) and (\ref{eq2jarzynski1}) in the main text are obtained as special cases of (\ref{eq2jarzynski}).

Strictly speaking, in the infinite-dimensional case the traces in the above expressions are not necessarily well defined and finite for all operators $\rho$. However, we here proceed under the assumption that $\rho$ are chosen such that the traces are well defined.
\begin{proof}
By Assumptions \ref{Def2} it follows that Proposition \ref{PropQuantumCrooks} is applicable, and thus the channels $\mathcal{F}_{\pm}$ as defined in Eqs.~(\ref{DefFRplus}) and (\ref{DefFRminus}) satisfy equation (\ref{FplusequivFminus}). By applying (\ref{FplusequivFminus}) to the operator $e^{(-\beta  + r + z)H_E/2 } \rho e^{(-\beta+ r -z) H_E/2}$, multiplying both sides of the resulting equality with $e^{\beta H_E}$, and take the trace, and divide by $ Z(H^{i}_{S'})$, one obtains
\begin{equation*}
\begin{split}
&\Tr\Big[e^{\beta H_E}\mathcal{F}_{+}\big(e^{(-\beta  + r + z)H_E/2 } \rho e^{(-\beta+ r -z) H_E/2}\big)\Big] \\
& = \frac{Z(H^{f}_{S'})}{ Z(H^{i}_{S'})}\Tr\Big[\mathcal{F}_{-}^{\ominus}\Big(e^{(r + z)H_E/2 } \rho e^{(r -z) H_E/2}\Big)\Big]\\
& [\textrm{By Eq.~(\ref{RconjFominus}) and the definition of `$*$'}]\\
& = \frac{Z(H^{f}_{S'})}{ Z(H^{i}_{S'})}\Tr\Big[\mathcal{R}_{+}(\hat{1}) e^{(r + z)H_E/2 } \rho e^{(r -z) H_E/2}\Big].
\end{split}
\end{equation*}
With the definition of the commutator $\mathcal{C}_{H_E}(\sigma) := [H_E,\sigma]$, we can write $e^{-z \mathcal{C}_{H_E}/2}\big(\mathcal{R}_{+}(\hat{1})\big) = e^{-zH_E/2}\mathcal{R}_{+}(\hat{1})e^{zH_E/2}$. Combined with the fact from Lemma \ref{ComWithCom}, that $\mathcal{C}_{H_E}$ and $\mathcal{R}_{+}$ commute, we thus get $e^{-zH_E/2}\mathcal{R}_{+}(\hat{1})e^{zH_E/2} = \mathcal{R}_{+}(\hat{1})$. This proves (\ref{eq1jarzynski}). 

That $\mathcal{F}_{-}(\hat{1}) = \hat{1}$ if and only in $\mathcal{R}_{+}(\hat{1}) = \hat{1}$, follows from Lemma \ref{ChannelReversal} together with the properties $\mathcal{T}_E(\hat{1}) = \hat{1}$ (Lemma \ref{PreservesIdentity}) and $\mathcal{T}_E^2 = I$.
\end{proof}

One may get the impression that the members in the family of equalities in Proposition \ref{CollectionJarzynski} are independent.
However, at least in the finite-dimensional case one can transform them into each other, and in this sense they should maybe rather be regarded as the same equality in different  guises. To see this, start with assuming that  (\ref{eq1jarzynski}) is true for all operators $\rho$ on a finite-dimensional Hilbert space. We wish to show that this implies  that (\ref{eq1jarzynski}) is also true for $r$ and $z$  substituted with arbitrary $r'$ and $z'$. Let $r = r' + \Delta r$ and $z = z' + \Delta z$ in (\ref{eq1jarzynski}) and define $\rho':= e^{\Delta rH_E/2}\rho e^{\Delta rH_E/2}$. 
This yields the equality 
\begin{equation*}
\begin{split}
&\Tr\Big[e^{\beta H_E}\mathcal{F}_{+}\Big(e^{(-\beta  + r' + z'+\Delta z)H_E/2 } \rho'  e^{(-\beta+ r'-z'-\Delta z) H_E/2}\Big) \Big]\\
& = \frac{Z(H^{f}_{S'})}{ Z(H^{i}_{S'})}\Tr[e^{r' H_E/2}\mathcal{R}_{+}(\hat{1}) e^{r'H_E/2 } \rho' ].
\end{split}
\end{equation*}
We can now use the fact (Lemma \ref{ComWithCom}) that the commutator $\mathcal{C}_{H_E}(\sigma) := [H_E,\sigma]$ commutes with $\mathcal{F}_{+}$, to show that  $\mathcal{F}_{+}\big(e^{(-\beta  + r' + z'+\Delta z)H_E/2 } \rho' e^{(-\beta+ r'  -z'-\Delta z) H_E/2}\big) =  e^{\Delta z H_E/2}\mathcal{F}_{+}\big(e^{(-\beta  + r' + z')H_E/2 } \rho' e^{(-\beta+ r'  -z') H_E/2}\big)e^{-\Delta z H_E/2}$. From this is follows that (\ref{eq1jarzynski}) remains valid with $r,z,\rho$ substituted by $r',z',\rho'$. In the finite-dimensional case it is also clear that the mapping $\rho\mapsto \rho' = e^{\Delta rH_E/2}\rho e^{\Delta rH_E/2}$ for Hermitian $H_E$ and real $\Delta r$  is a bijection on the space of linear operators. Hence,  (\ref{eq1jarzynski}) with $r',z'$ holds for all operators $\rho$.

\section{\label{SecBoundsOnWork} Bounds on the work cost }

Here we investigate quantum analogues to some classical bounds on the work cost of processes. We find that one indeed can obtain such bounds.

\subsection{\label{SecBoundAverage}Bound on the average energy loss in the reservoir}

For processes that start in equilibrium one would expect that the work cost should be bounded from below by the equilibrium free energy difference of the final and initial Hamiltonian of the system $\langle W\rangle \geq  F(H^{f}) -F(H^{i})$. Here we derive a similar expression in our setting, under the assumption that  $\mathcal{R}_{+}$ is unital, and that $H_E$ has a pure non-degenerate point spectrum, i.e., that there exists a complete orthonormal basis $\{|n\rangle\}_{n}$ of eigenvectors to $H_E$ corresponding to distinct eigenvalues $E_n$. (The latter assumption may not necessarily be essential.)

\begin{Proposition}
\label{BoundOnAverageWork}
With Assumptions \ref{Def2}, assume that $H_E$ has a complete orthonormal eigenbasis with corresponding non-degenerate eigenvalues.
 Assume that the initial state is $G(H_{S'}^{i})\otimes \sigma$, with $\sigma$ being a density operator on $\mathcal{H}_E$. Then  
\begin{equation}
\label{svdfhjbvjsv}
\begin{split}
 & \Tr(H_E\sigma) - \Tr\big(H_E \mathcal{F}_{+}(\sigma)\big) \\
 & \geq F(H^{f}_{S'})  -F(H^{i}_{S'}) - \frac{1}{\beta}\ln \Tr(\sigma\mathcal{R}_{+}(\hat{1}_E)). 
\end{split}
\end{equation}
Hence, if $\mathcal{R}_{+}(\hat{1}) = \hat{1}$, then 
\begin{equation}
\label{RegainedStandard}
\begin{split}
   \Tr(H_E\sigma) -\Tr\big(H_E \mathcal{F}_{+}(\sigma)\big) \geq F(H^{f}_{S'})  -F(H^{i}_{S'}).
\end{split}
\end{equation}
\end{Proposition}
Hence, if we identify the loss of average energy  in the energy reservoir, $\Tr(H_E\sigma) -\Tr\big(H_E \mathcal{F}_{+}(\sigma)\big)$, with $\langle W\rangle$, equation (\ref{RegainedStandard}) thus gives the standard bound.
\begin{proof}
Let $E_n$ and $|n\rangle$ be eigenvalues and corresponding orthonormal eigenvectors to $H_E$ such that $\{|n\rangle\}_{n}$ is a complete orthonormal basis to $\mathcal{H}_E$. 
Define  
\begin{equation}
p_{n,m}  := \langle n|\sigma |n\rangle \langle m|\mathcal{F}_{+}(|n\rangle\langle n|)|m\rangle.
\end{equation}
By the fact that $\mathcal{F}_{+}$ is a channel it follows that $\{p_{n,m}\}_{n,m}$ is a probability distribution.  Corollary \ref{MappingOffDiagonal} and the  non-degeneracy of  $H_E$ yield
\begin{equation*}
\begin{split}
& \Tr[e^{\beta H_E}\mathcal{F}_{+}(e^{-\beta H_E/2} \sigma e^{-\beta H_E/2})]  \\
        & =  \sum_{m}\sum_{n} \langle n|\sigma |n\rangle e^{\beta (E_m-E_{n})}\langle m|\mathcal{F}_{+}(|n\rangle\langle n|)|m\rangle\\
       & \geq   \exp\Big[\beta \sum_{m}\sum_{n} (E_m-E_{n})\langle n|\sigma |n\rangle \langle m|\mathcal{F}_{+}(|n\rangle\langle n|)|m\rangle \Big] \\ 
             & [\textrm{By Corollary \ref{MappingOffDiagonal}, and non-degenerate $H_E$} ]\\
= &  \exp\Big[ \beta \Tr(H_E\mathcal{F}_{+}(\sigma))  -  \beta \Tr(\mathcal{F}_{+}(H_E\sigma)) \Big] \\    
    &[\textrm{$\mathcal{F}_{+}$ is trace preserving}] \\
= &  \exp\Big[ \beta \Tr(H_E\mathcal{F}_{+}(\sigma))  -  \beta \Tr(H_E\sigma) \Big],                       
 \end{split}
\end{equation*}
where the inequality follows by the convexity of the exponential function.
By combining this inequality with the quantum Jarzynski equality  (\ref{eq1jarzynski}) in Proposition \ref{CollectionJarzynski} for $r = 0$ and $z =0$, one obtains
\begin{equation*}
 \frac{Z(H^f_{S'})}{Z(H^i_{S'})}\Tr[\sigma\mathcal{R}_{+}(\hat{1})] \geq e^{\beta \Tr(H_E\mathcal{F}_{+}(\sigma))  -  \beta \Tr(H_E\sigma)}.
\end{equation*}
Since the logarithm is monotonically increasing, and thus preserves the inequality, we thus obtain (\ref{svdfhjbvjsv}), where we use $F(H) = -kT\ln Z(H)$.
\end{proof}

\subsection{\label{SecViolationStndrdWorkB}An example}

One should keep in mind that the inequality (\ref{svdfhjbvjsv}) does not \emph{per se} imply that the standard work bound $\langle W\rangle \geq  F(H^{f}) -F(H^{i})$ necessarily is violated when $\mathcal{R}_{+}$ is not unital; it only allows for the possibility. However, here we construct an explicit example where one indeed gets a violation.

We begin with  a general remark to put this, maybe not entirely transparent, example in perspective.
The general goal is to find a case where the joint unitary evolution on $S'E$ is such that the loss of energy in the energy reservoir is too small compared to the standard bound. More precisely, we wish to find a global unitary $V$ such that $D(V,\sigma)$, defined below, violates the standard bound. In Appendix \ref{SecEnergyIndependence} we will show that a specific class of energy translation invariant models yields unital  $\mathcal{R}_{+}$, and thus recovers the standard bound. That model is based on a special Hamiltonian on $E$, as well as a specific class of Hamiltonians on $S'$, such that every possible transition in $S'$ always can be compensated by a corresponding transition in $E$. This results in a particularly simple structure of isomorphic eigenspaces, each enumerated by an integer $j$. Since the global unitary $V$ by assumption is energy conserving, it means that it block-diagonalizes into a collection of smaller unitaries $\{V_j\}_j$ on these eigenspaces. For the energy translation invariant model in Appendix \ref{SecEnergyIndependence}, all these unitary operators are (in certain sense) equal. Here we will use the very same structure of Hamiltonians and eigenspaces, but we let all $V_j$ vary independently, thus increasing the number of free parameters in the minimization of $D(V,\sigma)$ from finite to infinite. From this perspective it is maybe not entirely surprising that this opens up for a violation of the standard bound.

With the setting as in Assumptions \ref{Def2} we define the average energy loss $D(V,\sigma)$ in the energy reservoir 
\begin{equation}
\begin{split}
& D(V,\sigma) := \Tr(H_E\sigma) \\
& \quad \quad\quad \quad- \Tr\big([\hat{1}_{S'}\otimes\hat{1}_{C}\otimes H_E]V\rho_i V^{\dagger}\big),\\
&\quad \quad\quad \rho_i:=G(H_{S'}^i)\otimes |c_{i+}\rangle\langle c_{i+}|\otimes\sigma.
\end{split}
\end{equation}
By energy conservation, the assumption of prefect control, and the general relation  $H = F(H)-kT\ln G(H)$, one can show that 
\begin{equation}
\label{toeeoreoiu}
\begin{split}
& D(V,\sigma) = F(H_{S'}^f)-F(H_{S'}^i)  -\frac{1}{\beta}S\big(G(H_{S'}^i)\big)\\
& -\frac{1}{\beta}\Tr\big([\ln G(H_{S'}^f)\otimes \hat{1}_C\otimes\hat{1}_E]V\rho_i V^{\dagger}\big).
\end{split}
\end{equation}
The strategy will be to construct a model with a Hamiltonian and a particular class of energy conserving unitary operators $V$ that is simple enough that we can determine the corresponding minimum of $D(V,\sigma)$.

Let us assume that $H_E = s\sum_{j\in\mathbb{Z}}j|j\rangle\langle j|$ where $s>0$, for an orthonormal basis $\{|j\rangle\}_{j\in\mathbb{Z}}$. (We will use the very same Hamiltonian in Appendix \ref{SecEnergyIndependence}.)
Moreover, we assume that $\mathcal{H}_{S'}$ is finite-dimensional and that  the eigenvalues of $H^{i}_{S'}$ and $H^{f}_{S'}$ are integer multiples of $s$, i.e., they have the eigenvalues $\{sz^{i}_n\}_{n=1}^{N}$,  $\{sz^{f}_n\}_{n=1}^{N}$, for some $z^{i}_n,z^{f}_n\in\mathbb{Z}$, with corresponding eigenvectors $|\chi^{i}_n\rangle$, $|\chi^{f}_n\rangle$ (again the same as in Appendix \ref{SecEnergyIndependence}). To make the derivations simpler we also assume that these eigenvalues are non-degenerate.

One can realize that for each single $j$ all the vectors 
\begin{equation*}
\begin{split}
|i^{+}_{j,n}\rangle := & |\chi^i_n\rangle |c_{i+}\rangle| j-z^i_n\rangle,\\
 |i^{-}_{j,n}\rangle :=  & |\chi^i_n\rangle |c_{i-}\rangle| j-z^i_n\rangle,\\ 
|f^{+}_{j,m}\rangle := &  |\chi^f_m\rangle|c_{f+}\rangle|j-z^f_m\rangle,\\
 |f^{-}_{j,m}\rangle := & |\chi^f_m\rangle|c_{f-}\rangle|j-z^f_m\rangle,
\end{split}
\end{equation*}
for $m,n = 1,\ldots, N$, correspond to the same global energy $sj$. Hence, any unitary transformation that only transforms within these collections is energy conserving. However, we also have to satisfy perfect control. Hence, the unitary should transfer $|c_{i+}\rangle$ to $|c_{f+}\rangle$ (as well as $|c_{f-}\rangle$ to $|c_{i-}\rangle$).
With the above remarks in mind,
we define the following class of unitary operators on $\mathcal{H}_{S'CE}$
\begin{equation}
\label{kjdsjlkfdlk}
\begin{split}
V :=& \sum_{j\in\mathbb{Z}}\sum_{n,m=1}^{N}U^{(j):f+,i+}_{m,n}|f^{+}_{j,m}\rangle\langle i^{+}_{j,n}|\\
& + \sum_{j\in\mathbb{Z}}\sum_{n,m=1}^{N}U^{(j):i-,f-}_{n,m}|i^{-}_{j,n}\rangle\langle f^{-}_{j,m}|\\
&+ \sum_{j\in\mathbb{Z}}\sum_{n,m=1}^{N}U^{(j):i+,f+}_{n,m}|i^{+}_{j,n}\rangle\langle f^{+}_{j,m}|\\
& + \sum_{j\in\mathbb{Z}}\sum_{n,m=1}^{N}U^{(j):f-,i-}_{m,n}|f^{-}_{j,m}\rangle\langle i^{-}_{j,n}|,
\end{split}
\end{equation}
where the matrices $[U^{(j):f+,i+}_{m,n}]_{m,n}$, $[U^{(j):i-,f-}_{n,m}]_{n,m}$, $[U^{(j):i+,f+}_{n,m}]_{n,m}$, and  $[U^{(j):f-,i-}_{m,n}]_{m,n}$ are unitary for each fixed $j$. (If we additionally assume that these matrices are independent of $j$, then we obtain the class of energy translation invariant unitaries that is considered in Appendix \ref{SecEnergyIndependence}.) 
By construction (\ref{kjdsjlkfdlk}) is energy conserving, and the first line corresponds to transitions from the control state $|c_{i+}\rangle$ to $|c_{f+}\rangle$, thus implementing the desired perfect control.  The second line in (\ref{kjdsjlkfdlk}) analogously describes evolution from $|c_{f-}\rangle$ to $|c_{i-}\rangle$. The two last lines in (\ref{kjdsjlkfdlk}) serve no active role in our protocol, but are there in order to guarantee unitarity, energy conservation, and time-reversal symmetry of the global evolution $V$.

Next we define time-reversals on $S'C$ and $E$. First define  
$Y :=  |c_{i+}\rangle\langle c_{i-}| + |c_{i-}\rangle\langle c_{i+}|
   + |c_{f+}\rangle\langle c_{f-}| + |c_{f-}\rangle\langle c_{f+}|$, which so to speak swaps the control states between the positive and negative `tracks'.
We use this in turn to define $\mathcal{T}_{S'C}(Q) := [\hat{1}_{S'}\otimes Y]Q^{\tau}[\hat{1}_{S'}\otimes Y^{\dagger}]$,
where $\tau$ denotes the transpose with respect to the orthonormal basis
$\mathcal{B}:=  \{|\chi^{i}_{n}\rangle|c_{i+}\rangle\}_{n}\cup\{|\chi^{i}_{n}\rangle|c_{i-}\rangle\}_{n} \cup\{|\chi^{f}_{n}\rangle|c_{f+}\rangle\}_{n}\cup\{|\chi^{f}_{n}\rangle|c_{f-}\rangle\}_{n}$. If this would have been a finite-dimensional case, we could have used Proposition \ref{TcharactFinite} to conclude that $\mathcal{T}_{S'C}$ is a time-reversal. However, it is straightforward to directly check the properties in Definition \ref{PropertiesTimeReversal}. 

Let $\mathcal{T}_E$ be the transpose with respect to the basis $\{|j\rangle\}_{j\in\mathbb{Z}}$ of $\mathcal{H}_E$, and define $\mathcal{T} :=\mathcal{T}_{S'C}\otimes\mathcal{T}_E$. With this definition one can confirm that $\mathcal{T}(V) = V$
if and only if  $U^{(j):i-,f-}_{n,m} = U^{(j):f+,i+}_{m,n}$ and $U^{(j):f-,i-}_{m,n} = U^{(j):i+,f+}_{n,m}$.
Next, note that 
\begin{equation}
\label{yaax}
\begin{split}
& \Tr\big([\ln G(H_{S'}^f)\otimes \hat{1}_C\otimes\hat{1}_E]V\rho_i V^{\dagger}\big)\\
& =   \sum_{j\in\mathbb{Z}}\sum_{n,m=1}^{N}|U^{(j):f+,i+}_{m,n}|^2\ln G_m(H_{S'}^f) \\
& \quad\quad\quad\quad\times G_n(H_{S'}^i)\langle j-z^i_n|\sigma| j-z^i_{n} \rangle.
\end{split}
\end{equation}
 In order to minimize (\ref{toeeoreoiu}) over the time-reversal symmetric operators $V$ in our designated family (\ref{kjdsjlkfdlk}) it is sufficient to minimize over the collection of unitary matrices $U^{(j):f+,i+} := [U^{(j):f+,i+}_{m,n}]_{m,n}$. (Since $\mathcal{T}(V) = V$ if and only if $U^{(j):f+,i+}_{m,n} = U^{(j):i-,f-}_{n,m}$ and $U^{(j):i+,f+}_{n,m} = U^{(j):f-,i-}_{m,n}$, there are no further restrictions.)
Next, insert (\ref{yaax}) into (\ref{toeeoreoiu}) and minimize, which yields
\begin{equation}
\label{eqrewrq}
\begin{split}
& \min_{\{V\}_j}D(V,\sigma) =  F(H_{S'}^f)-F(H_{S'}^i) -\frac{1}{\beta}S\big(G(H_{S'}^i)\big)\\
&\quad\quad -\frac{1}{\beta} \sum_{j}\sum_{m}\lambda^{\downarrow}_{m}\Big(r_j\big(G(H^i_{S'})\big)\Big)\ln\lambda^{\downarrow}_{m}\big( G(H^f_{S'})\big), 
\end{split}
\end{equation}
where 
\begin{equation}
\label{rjdef}
r^{(j)}(\rho) := \sum_{n}\langle j-z^{i}_n|\sigma|j-z^{i}_n\rangle |\chi^{i}_{n}\rangle\langle\chi^{i}_{n}|\rho |\chi^{i}_{n}\rangle\langle\chi^{i}_{n}|,
\end{equation}
and where $\lambda^{\downarrow}_{n}(Q)$ denotes the $n$:th eigenvalue of $Q$, ordered non-increasingly. 
The minimum in (\ref{eqrewrq}) can be obtained by noting that $[|U^{(j):f+,i+}_{m,n}|^2]_{m,n}$ is a doubly stochastic matrix for each $j$. Hence, according to Birkhoff's theorem \cite{MatrixAnalysisBhatia} it can be regarded as a convex combination of permutation matrices. Since every permutation matrix results from a unitary matrix, we know that the maximum of (\ref{yaax}) is given by a permutation. (Alternatively, one can define  $U^{(j)} := \sum_{m,n}U^{(j):f+,i+}_{m,n}|\chi^{f}_m\rangle\langle \chi^{i}_n|$ and 
observe that $\sum_{n,m=1}^{N}|U^{(j):f+,i+}_{m,n}|^2\ln G_m(H_{S'}^f) G_n(H_{S'}^i)\langle j-z^i_n|\sigma| j-z^i_{n} \rangle = \Tr\big[ {U^{(j)}}^{\dagger}\ln G(H^f_{S'}) U^{(j)}r_j\big(G(H^i_{S'})\big)\big]$.
By the general relation $\max_{U}\Tr(U^{\dagger}QUR) = \sum_{m}\lambda^{\downarrow}_m(Q)\lambda^{\downarrow}_m(R)$, see e.g.~Theorem 4.3.53 in \cite{HornJohnson}, it follows that (\ref{eqrewrq}) holds.)
 
Assume that the energy reservoir starts in the specific energy eigenstate $\sigma:=|0\rangle\langle 0|$.
The definition of $r_{j}$ in (\ref{rjdef}), together with the assumed non-degeneracy of $H^i_{S'}$ and thus of $z^{i}_n$, leads to 
\begin{equation*}
\begin{split}
& \sum_{j}\sum_{m}\lambda^{\downarrow}_{m}\Big(r_j\big(G(H^i_{S'})\big)\Big)\ln\lambda^{\downarrow}_{m}\big(G(H^f_{S'})\big)\\
& =  \sum_{n}\sum_{m}\lambda^{\downarrow}_{m}\big(G_n(H^i_{S'})|\chi^i_{n}\rangle\langle\chi^{i}_n|\big)\ln\lambda^{\downarrow}_{m}\big(G(H^f_{S'})\big)\\
& =  \sum_{n}G_n(H^i_{S'})\ln\lambda^{\downarrow}_{1}\big(G(H^f_{S'})\big)\\
& =  \ln \lambda^{\downarrow}_1\big(G(H^{f}_{S'})\big).
\end{split}
\end{equation*}
By combining this observation with (\ref{eqrewrq}) we get
\begin{equation}
\begin{split}
&\min_{\{U^{(j)}\}_j}D(V,|0\rangle\langle 0|)=  F(H_{S'}^f)-F(H_{S'}^i)\\
& \quad\quad\quad  -\frac{1}{\beta}S(G(H_{S'}^i)) -\frac{1}{\beta} \ln \lambda^{\downarrow}_1(G(H^{f}_{S'})).
\end{split}
\end{equation}
In other words, we get a violation of the standard bound whenever
$S\big(G(H^{i}_{S'})\big) + \ln\lambda^{\downarrow}_{1}\big(G(H^{f}_{S'})\big)>0$.
For an explicit example where this is the case, let 
$H^{i}_{S'} = H^{f}_{S'} = s\sum_{k=0}^{K}k|\chi_k\rangle\langle\chi_k|$.
In this particular case we find
\begin{equation*}
\begin{split}
&  S\big(G(H^{i}_{S'})\big) + \ln\lambda^{\downarrow}_{1}\big(G(H^{f}_{S'})\big)\\
& =  \frac{s\beta\sum_{k=0}^{K}ke^{-s\beta k}}{\sum_{k' =0}^{K}e^{-s\beta k'}}\\
& =  \frac{s\beta e^{-s\beta}}{1-e^{-s\beta}} -\frac{s\beta(K+1)e^{-s\beta(K+1)}}{1-e^{-s\beta(K+1)}}. 
\end{split}
\end{equation*}
In the limit of large $K$ this approaches $s\beta e^{-s\beta}/(1-e^{-s\beta})$, which is strictly larger than zero since $s\beta>0$. Thus, for sufficiently large $K$ it follows that $\min D(V,\sigma) <  F(H_{S'}^f)-F(H_{S'}^i)$. Hence, with the identification between $D(V,\sigma)$ and $\langle W\rangle$ we do get a violation of the standard bound.

\subsection{\label{KelvinPlanckBounds}Bound for a closed cycle}

Here we consider the counterpart to the classical bound  $\langle W_{+}\rangle + \langle W_{-}\rangle \geq 0$, where $W_{+}$ and $W_{-}$ are the work costs of the forward and reverse process, respectively.
As in Appendix \ref{SecBoundAverage}, we take  $\Tr(H_E\sigma) - \Tr\big(H_E\mathcal{F}_{+}(\sigma)\big)$ as the counterpart of $\langle W_{+}\rangle$. Assuming that we use the very same energy reservoir also for the reversed process, we let $\langle W_{-}\rangle$ correspond to $\Tr\big(H_E\mathcal{F}_{+}(\sigma)\big) - \Tr\Big(H_E\mathcal{F}_{-}\big(\mathcal{F}_{+}(\sigma)\big)\Big)$, thus assuming that the energy reservoir is initially in state $\mathcal{F}_{+}(\sigma)$ in the second application. 
The inequality (\ref{svdfhjbvjsv}) in Proposition \ref{BoundOnAverageWork} is applied for the forward process, and the analogous inequality is applied for the reversed process, which yields
\begin{equation*}
\begin{split}
\langle W_{+}\rangle + \langle W_{-}\rangle \geq & -\frac{1}{\beta}\ln\Tr\big(\sigma \mathcal{R}_{+}(\hat{1})\big)\\
&  -\frac{1}{\beta}\ln\Tr\big(\mathcal{F}_{+}(\sigma) \mathcal{R}_{-}(\hat{1})\big). 
\end{split}
\end{equation*}
Hence, in this case we regain the standard result if both channels $\mathcal{R}_{+}$ and $\mathcal{R}_{-}$ are unital.

\subsection{\label{SecViolationBound} Bound on `second law violations'}

In the classical case, Crooks theorem and Jarzynski's equality put constraints on the distribution of the work cost. In particular, one can obtain a bound on the probability that the work value in a single run would violate the classical macroscopic bound $W \geq \Delta F$, where $\Delta F  := F(H^f_{S'}) -F(H^i_{S'})$ (for an initial equilibrium distribution). More precisely, regarding the work $W$ as a random variable, we can ask for the probability that the work $W$ is smaller than $\Delta F-\zeta$. In \cite{Jarzynski99} (alternatively, see section 7 of \cite{ReviewJarzynski}) it is shown that $P[W < \Delta F-\zeta]\leq e^{-\beta \zeta}$. In other words, the probability of a such an event is exponentially suppressed in the size of the violation $\zeta$.

Here we obtain an analogous bound in the quantum setting, but we again find that regaining the standard expression requires unitality of  $\mathcal{R}_{+}$.
For this discussion we assume that $H_E$ has a complete orthonormal basis $\{|n\rangle\}_{n}$ of eigenvectors.

We let $P_{\leq E_0}$ denote the projector onto the energy eigenstates of $H_E$ that has \emph{at most} energy $E_0$. In other words, $P_{\leq E_0}\sigma P_{\leq E_0} = \sigma$ implies that the probability to find an energy larger than $E_0$ in $\sigma$ is zero. We similarly let $P_{\geq \zeta-\Delta F+ E_0}$ denote the projector onto the energy eigenstates of $H_E$ that have \emph{at least} the energy $\zeta-\Delta F+ E_0$.

Assuming that $P_{\leq E_0}\sigma P_{\leq E_0} = \sigma$ we can thus interpret 
$\Tr\big(P_{\geq \zeta-\Delta F+ E_0}\mathcal{F}_{+}(\sigma)\big)$ as the probability that we would observe an energy gain in the reservoir that is at least $\zeta-\Delta F$. 
Equivalently, this would be the probability  that we would observe that the work done on the system is at most $\Delta F-\zeta$.
In this sense we regard $\Tr\big(P_{\geq \zeta-\Delta F+ E_0}\mathcal{F}_{+}(\sigma)\big)$ as an analogue of $P[W < \Delta F-\zeta]$, for the class of initial states $\sigma$ such that $P_{\leq E_0}\sigma P_{\leq E_0} = \sigma$.

We first observe the following operator inequality
 $\mathcal{J}_{\beta H_E}(P_{\geq \zeta-\Delta F+ E_0}) \leq  e^{-\beta (\zeta-\Delta F+ E_0)}\hat{1}$.
Since $\mathcal{R}_{+}$ is a completely positive map, it means that it preserves operator inequalities, and thus 
\begin{equation}
\label{adfnklg}
\begin{split}
& \mathcal{R}_{+}\big(\mathcal{J}_{\beta H_E}(P_{\geq \zeta-\Delta F+ E_0}) \big)
\leq  e^{-\beta (\zeta-\Delta F+ E_0)}\mathcal{R}_{+}(\hat{1}).
\end{split}
\end{equation}
By applying the fluctuation relation (\ref{FplusequivFminus}) in Proposition \ref{PropQuantumCrooks} onto $\sigma$ and take the expectation value of the projector $P_{\geq \zeta-\Delta F+ E_0}$ we get
\begin{equation}
\label{dfklbmnlf}
\begin{split}
& \Tr\big(P_{\geq \zeta-\Delta F+ E_0} \mathcal{F}_{+}(\sigma)\big) \\
& =   \frac{Z(H^{f}_{S'})}{Z(H^{i}_{S'})}\Tr\Big( \mathcal{R}_{+}\big(\mathcal{J}_{\beta H_E}(P_{\geq \zeta-\Delta F+ E_0})\big)\mathcal{J}_{\beta H_E}^{-1}(\sigma)\Big),
\end{split}
\end{equation}
where we use Lemma \ref{ChannelReversal}, and the definition of the channel conjugate $\mathcal{R}_{+}^{*}$.
We next note that  $\mathcal{J}_{\beta H_E}^{-1}(\sigma)$ is a positive semidefinite operator, and hence the operator inequality in (\ref{adfnklg}) implies that
\begin{equation}
\label{ydfbm}
\begin{split}
& \Tr\Big( \mathcal{R}_{+}\big(\mathcal{J}_{\beta H_E}(P_{\geq \zeta-\Delta F+ E_0})\big)\mathcal{J}_{\beta H_E}^{-1}(\sigma)\Big)\\
& \leq  e^{-\beta (\zeta-\Delta F+ E_0)}\Tr(e^{\beta H_E}\sigma)\Tr\big(\mathcal{R}_{+}(\hat{1})\tilde{\sigma}\big),
\end{split}
\end{equation}
where $\tilde{\sigma} := \mathcal{J}_{\beta H_E}^{-1}(\sigma)/\Tr(e^{\beta H_E}\sigma)$ is a density operator. 
If we assume $P_{\leq E_0}\sigma P_{\leq E_0} = \sigma$, then it follows that 
$\Tr(e^{\beta H_E}\sigma)\leq e^{\beta E_0}$. By combining this observation with 
(\ref{dfklbmnlf}), (\ref{ydfbm}) and $e^{\beta \Delta F} = Z(H^i_{S'})/Z(H^f_{S'})$, it follows that 
$\Tr\big(P_{\geq \zeta-\Delta F+ E_0} \mathcal{F}_{+}(\sigma)\big) \leq \Tr\big(\mathcal{R}_{+}(\hat{1})\tilde{\sigma}\big) e^{-\beta \zeta}$.
Hence, even without further conditions on the channel $\mathcal{R}_{+}$, there is an exponential suppression of the energy gain in the reservoir. If $\mathcal{R}_{+}$ is unital, $\mathcal{R}_{+}(\hat{1}) = \hat{1}$, we obtain the counterpart  $\Tr\big(P_{\geq \zeta-\Delta F+ E_0} \mathcal{F}_{+}(\sigma)\big) \leq e^{-\beta \zeta}$ to the classical bound.

\section{\label{SecEnergyIndependence} Energy translation invariance}

In this investigation, we have allowed for the possibility that the processes depend non-trivially on the amount of energy in the energy reservoir. Here we consider a further restriction that implements the idea that the experiment does not depend on the energy level. This model has previously been used in  \cite{Aberg13} to analyze coherence and work extraction. Here we only describe the most essential aspects of this model. For a more detailed description, see \cite{Aberg13}. 
First of all, imagine the Hamiltonian $H_E$ of the energy reservoir as a doubly infinite ladder of energy levels
\begin{equation}
\label{EnergyLadder}
H_{E} = s\sum_{j\in \mathbb{Z}} j|j\rangle\langle j|
\end{equation}
with energy spacing $s>0$. (See also the continuum version in \cite{Skrzypczyk14}.) As one can see, this Hamiltonian has a bottomless spectrum (which echoes the discussions in Appendix \ref{SecConditionsOnE}).  Although this is not the most physically satisfying assumption, one can view it as an idealization of  a `battery' that has a much higher energy content than the characteristic scale of energy costs in the experiment.
We furthermore assume that the Hamiltonian $H_{\widetilde{S}}$ of system $\widetilde{S}$ (which includes all systems that are not $E$, i.e., in our case, system $S$, the heat bath $B$, and the control $C$) is such that all its eigenvalues (we assume a finite-dimensional Hilbert space $\mathcal{H}_{\widetilde{S}}$ with dimension $N$) are integer multiples of the energy spacing $s$. (Due to this assumption it becomes easy to construct non-trivial energy conserving unitary operations.) In other words, we assume that $H_{\widetilde{S}}$ has an eigenbasis $\{|\psi_n\rangle\}_{n=1}^{N}$ with corresponding eigenvalues $sz_n$ where $z_n\in \mathbb{Z}$ for each $n$. Note that we allow $H_{\widetilde{S}}$ to be degenerate, in which case  $\{|\psi_{n}\rangle\}_{n}$ is an eigenbasis of our choice.

In this section we not only demand that the global unitary operations are energy conserving, $[H,V] = 0$, but also that they are energy translation invariant. 
To define what we mean by this, we  introduce the energy translation operator 
$\Delta = \sum_{j}|j+1\rangle\langle j|$
 on the energy reservoir. We say that a unitary operator $V$ on $\mathcal{H}_{\widetilde{S}}\otimes\mathcal{H}_E$ is energy translation invariant if $[\hat{1}_{\widetilde{S}}\otimes \Delta^{a},V] = 0$ for all $a\in\mathbb{Z}$. It turns out \cite{Aberg13} that that all energy conserving and energy translation invariant unitary operators in this model can be written in the following way
\begin{equation}
\label{DefVU}
V(U) = \sum_{n,n'=1}^{N}|\psi_{n}\rangle\langle\psi_{n}|U|\psi_{n'}\rangle\langle\psi_{n'}|\otimes \Delta^{z_{n'} -z_{n}},
 \end{equation}
 where $U$ is an arbitrary unitary operator on $\mathcal{H}_{\widetilde{S}}$. If there are degeneracies in the Hamiltonian $H_{\widetilde{S}}$ then $V(U)$ is independent of the choice of energy eigenbasis $\{|\psi_{n}\rangle\}_{n}$. In particular, if $\{P_m\}_{m}$ is a collection of eigenprojectors of $H_{\widetilde{S}}$, then one can alternatively write
\begin{equation}
\label{DefVUalt}
V(U) = \sum_{m,m'}P_mUP_{m'}\otimes \Delta^{z_{m'} -z_{m}}.
 \end{equation}
A useful property of $V(U)$ is that it preserves products $V(U_2U_1) = V(U_2)V(U_1)$. (In contrast to the time-reversals $\mathcal{T}$, there is no swap of the ordering.)

We also need to incorporate time-reversals (an aspect not included in \cite{Aberg13}). 
\begin{Lemma}
\label{kldsdslk}
Let $\mathcal{T}_E$ be defined as the transpose with respect to the orthonormal basis  $\{|j\rangle\}_{j\in \mathbb{Z}}$. Let $\mathcal{H}_{\widetilde{S}}$ be finite-dimensional and let $\mathcal{T}_{\widetilde{S}}$ be such that $\mathcal{T}_{\widetilde{S}}(H_{\widetilde{S}}) = H_{\widetilde{S}}$. Then $\mathcal{T} :=\mathcal{T}_{\widetilde{S}}\otimes\mathcal{T}_E$ satisfies
\begin{equation}
\mathcal{T}\big(V(U)\big) = V\big(\mathcal{T}_{\widetilde{S}}(U)\big),
\end{equation}
for all operators $U$ on $\mathcal{H}_{\widetilde{S}}$.
\end{Lemma}
The proof is a direct application of the properties of time-reversals combined with Lemma \ref{TOnHermitean} and (\ref{DefVUalt}). 
 
The following lemma shows that the induced channels $\mathcal{R}_{\pm}$ and $\mathcal{F}_{\pm}$ are unital in this model. Hence, they automatically satisfy the condition  that emerged in the considerations on Jarzynski relations and work bounds in Appendices \ref{SecQuantumJarzynski} and \ref{SecBoundsOnWork}, respectively. (The unitality of this type of induced channels was previously observed in section II.C in the Supplementary Material of \cite{Aberg13}.)
\begin{Lemma}
\label{FpmRpmUnital}
For the channels $\mathcal{F}_{\pm}$ and $\mathcal{R}_{\pm}$ defined in equations (\ref{DefFRplus}) and (\ref{DefFRminus}), with $V := V(U)$ as in (\ref{DefVU}),  it is the case that 
$\mathcal{F}_{\pm}(\hat{1}_E) = \hat{1}_E$ and $\mathcal{R}_{\pm}(\hat{1}_E) = \hat{1}_E$.
\end{Lemma}
The proof is obtained by inserting the definition (\ref{DefVU}) of $V(U)$ into the definitions of $\mathcal{F}_{\pm}$ and $\mathcal{R}_{\pm}$ in equations (\ref{DefFRplus}) and (\ref{DefFRminus}), and apply these to the identity operator.
\begin{Lemma}
\label{ChannelsEnergyTranslInv}
With $\mathcal{F}_{\pm}$ and $\mathcal{R}_{\pm}$ as defined in equations (\ref{DefFRplus}) and (\ref{DefFRminus}), with $V := V(U)$ as in (\ref{DefVU}), it is the case that 
\begin{equation}
\label{ChannelTransl}
\begin{split}
& \Delta^{j}\mathcal{F}_{\pm}(\sigma){\Delta^{\dagger}}^{k} =\mathcal{F}_{\pm}\big( \Delta^{j}\sigma{\Delta^{\dagger}}^{k}\big),\\
& \Delta^{j}\mathcal{R}_{\pm}(\sigma){\Delta^{\dagger}}^{k} =\mathcal{R}_{\pm}\big( \Delta^{j}\sigma{\Delta^{\dagger}}^{k}\big).
\end{split}
\end{equation}
\end{Lemma}
\begin{proof}
Here we only show the equality for $\mathcal{F}_{+}$. The others are obtained analogously. The proof is based on the fact that $V(U)$ commutes with $\hat{1}_{\widetilde{S}}\otimes  \Delta^{j}$. With the  notation $\eta_{S'C} := G(H^{i}_{S'})\otimes |c_{i+}\rangle \langle c_{i+}|$ we can write
\begin{equation*}
\begin{split}
& \Delta^{j}\mathcal{F}_{+}(\sigma){\Delta^{\dagger}}^{k} =  \Tr_{S'C}\Big([\hat{1}_{\widetilde{S}}\otimes  \Delta^{j}] V(U)\\
& \quad\quad\quad\times[\eta_{S'C}\otimes \sigma] V^{\dagger}(U)[\hat{1}_{\widetilde{S}}\otimes{\Delta^{\dagger}}^{k}]\Big)\\
& =  \Tr_{S'C}\big(V(U)[\eta_{S'C} \otimes  \Delta^{j} \sigma {\Delta^{\dagger}}^{k}] V^{\dagger}(U)\big)\\
& =  \mathcal{F}_{+}(\Delta^{j}\sigma{\Delta^{\dagger}}^{k}).
\end{split}
\end{equation*}
\end{proof}

We can regard   $\langle m|\mathcal{F}_{\pm}(|n\rangle\langle n'|)|m'\rangle$ as the matrix elements in a matrix representation of the linear maps $\mathcal{F}_{\pm}$.
It turns out that the translation invariance in Lemma \ref{ChannelsEnergyTranslInv} in conjunction with the decoupling between the coherence modes described  in Appendix \ref{DecouplingDiagonals} reduces the number of independent parameters in this representation. More precisely, 
\begin{equation}
\label{Reduction}
\langle m|\mathcal{F}_{\pm}(|n\rangle\langle n'|)|m'\rangle = \left\{\begin{matrix}
p_{\pm}(m-n|0), &  n-m = n'-m',\\
0, & n-m \neq n'-m',
\end{matrix}\right.
\end{equation}
where $p_{\pm}(m|n) = \langle m|\mathcal{F}_{\pm}(|n\rangle\langle n|)|m\rangle$ are the diagonal transition probabilities as defined in (\ref{Pplusminusdef}). The mapping $\mathcal{F}_{\pm}$ is thus determined by the probabilities by which $|0\rangle\langle 0|$ is mapped to the other eigenstates of $H_E$. This can equivalently be expressed as
\begin{equation}
\label{rkjgrkgsal}
\mathcal{F}_{\pm}(\rho) =  \sum_{n,n',k\in\mathbb{Z}} p_{\pm}(k|0)\langle n|\rho|n'\rangle |n+k\rangle\langle n'+k|.
\end{equation}
 The expression in (\ref{rkjgrkgsal}) can be compared with the more general case in (\ref{nnabna}). 
In Appendix \ref{Secoffdiagonal} it was demonstrated that the off-diagonal modes of coherence satisfy Crooks relations that are structurally identical to the one along the diagonal. Equations (\ref{Reduction}), or equivalently (\ref{rkjgrkgsal}), implies a stronger statement for the special case of the energy translation invariant model. Namely, that the dynamical map along each the off-diagonal modes is identical to the one along the main diagonal. (This does not imply that the elements of the density matrix along the different diagonals are the same.)

\subsection{\label{SecRegainingClassicalCrooks}Regaining the standard Crooks and Jarzynski relations}

The standard Crooks relation can be written 
\begin{equation}
\label{StandardForm}
Z(H^i_S)P_{+}(w) =  e^{\beta w}Z(H^f_S)P_{-}(-w),
\end{equation}
where $P_{\pm}(w) := P(W_{\pm}=w)$, and where $W_{+}$ and $W_{-}$ are the work costs of the forward and reverse processes regarded as random variables. 
In the following we shall see how one can regain (\ref{StandardForm}) from the diagonal Crooks relation (\ref{diagonalCrooksRel}) in Appendix \ref{DiagonalCase} by additionally assuming energy translation invariance in the energy-ladder model.

If we identify the loss of energy in the reservoir with the work done, then a transition from energy level $n$ to $m$ in the reservoir corresponds to the work $w = E_{n}-E_{m}$ (cf.~\cite{Campisi11a}). The probability $P_{\pm}(w)$ is obtained by summing up the probabilities of all the transitions that generate the work cost $w$.
More precisely,
\begin{equation}
\label{bvdfsmbn}
P_{\pm}(w) := \sum_{n,m: E_{n}-E_{m} = w}p_{\pm}(m|n)\langle n|\sigma|n\rangle,
\end{equation}
where $p_{\pm}(m|n)$ are the conditional probability distributions, defined in (\ref{Pplusminusdef}), that describe the transitions among the diagonal elements. 

The energy-ladder model yields
$E_{n} = sn$. Hence, $n-m = w/s$.
 By (\ref{Reduction}) it follows that
\begin{equation}
\label{PplusTranslInvSpec}
p_{\pm}(m|n) = p_{\pm}(m-n |0) = p_{\pm}(0|n-m).
\end{equation}
A direct consequence is that the probability distribution $P_{\pm}(w)$, defined in (\ref{bvdfsmbn}), becomes independent of the initial  $\sigma$,
\begin{equation}
\label{skgdjfsg}
P_{\pm}(w) =   p_{\pm}(-w/s|0) = p_{\pm}(0|w/s).
\end{equation}
These observations can be used to regain  the classical Crooks relation (\ref{StandardForm}),
\begin{equation}
\label{odsodsod}
\begin{split}
Z(H^i_{S})  P_{+}(w) = & Z(H^i_{S})  p_{+}(0|w/s)\\
= & Z(H^f_{S}) e^{\beta w}p_{-}(w/s|0)\\
= &  Z(H^f_{S})e^{\beta w} P_{-}(-w),
\end{split}
\end{equation}
where the second inequality is due to the diagonal Crooks relation in (\ref{diagonalCrooksRel}).
Here  $\sum_w$ means that we sum over the set of possible energy changes, which for this particular model is  $s\mathbb{Z}$. 

Since we have re-derived the classical Crooks relation,  we do also  more or less automatically obtain the classical Jarzynski equality 
$\langle e^{-\beta W}\rangle = Z(H^f)/Z(H^i)$ \cite{Jarzynski97}. This can be done via the `standard' derivation $\sum_{w}e^{-\beta w}Z(H^i)P_{+}(w) = \sum_{w}Z(H^f)P_{-}(-w) = Z(H^f)$, where the first equality is due to the Crooks relation (\ref{odsodsod}). One can alternatively use the fact that $\mathcal{R}_{+}(\hat{1}_E) = \hat{1}_E$ (by Lemma \ref{FpmRpmUnital}) and use the quantum Jarzynski equality  (\ref{eq2jarzynski}) to derive the classical Jarzynski relation by again making use of the decoupling of the diagonals, and the energy translation invariance.

\subsection{$\mathcal{F}_{-}\mapsto \mathcal{F}^{\ominus}_{-}$ as a generalization of $P_{-}(w)\mapsto P_{-}(-w)$}

In the main text it is claimed that the mapping $\mathcal{F}_{-}\mapsto \mathcal{F}^{\ominus}_{-}$ can be regarded as a generalization of the map $P_{-}(w)\mapsto P_{-}(-w)$.
This generalization becomes  evident for the energy translation invariant model.
Let $\mathcal{T}_E$ be the transpose with respect to $\{|n\rangle\}_n$. 
By the definition of  $\ominus$ it follows that
\begin{equation}
\label{adslkjalkd}
\begin{split}
\langle m|\mathcal{F}^{\ominus}_{\pm}(|n\rangle\langle n'|)|m'\rangle  = & \langle n| \mathcal{F}_{\pm}(|m\rangle\langle m'|)|n'\rangle  \\
= & \langle -m|\mathcal{F}_{\pm}(|-n\rangle\langle -n'|)|-m'\rangle,
\end{split}
\end{equation}
where the last equality follows from the translation invariance in Lemma \ref{ChannelsEnergyTranslInv}. Hence, the action of $\ominus$ can in some sense be identified with the change of signs.

Let us now identify the work $w = s(n-m)$ and the offset $\delta = s(n-n')$. By Corollary \ref{MappingOffDiagonal}  it follows that $\langle m|\mathcal{F}_{\pm}(|n\rangle\langle n'|)|m'\rangle  = \langle 0|\mathcal{F}_{\pm}(|w/s\rangle\langle w/s-\delta/s|)|-\delta/s\rangle$ for the non-zero elements. By (\ref{adslkjalkd}) these non-zero coefficients satisfy 
\begin{equation*}
\begin{split}
& \langle 0|\mathcal{F}^{\ominus}_{\pm}(|w/s\rangle\langle w/s-\delta/s|)|-\delta/s\rangle \\
& = \langle 0|\mathcal{F}_{\pm}(|-w/s\rangle\langle -w/s+\delta/s|)|+\delta/s\rangle.
\end{split}
\end{equation*}
 Hence, the work parameter $w$ and the offset $\delta$ both change sign due to the $\ominus$ operation. For the diagonal, $\delta =0$, we can thus conclude that $P_{\pm}[\mathcal{F}_{\pm}^{\ominus}](w) = \langle 0|\mathcal{F}_{\pm}^{\ominus}(|w/s\rangle\langle w/s|)|0\rangle = \langle 0|\mathcal{F}_{\pm}(|-w/s\rangle\langle -w/s|)|0\rangle = P_{\pm}[\mathcal{F}_{\pm}](-w)$. Hence for the diagonal elements, the mapping $\ominus$ implements the transformation  $P_{\pm}(w)\mapsto P_{\pm}(-w)$.

\subsection{Examples}

The main purpose of these examples is to show that there exist setups of Hamiltonians, states, and unitary operators that satisfy Assumptions \ref{Def1} and Assumptions \ref{Def2}. We also take the opportunity to construct discretizations of paths of Hamiltonian within these models.  The reason for why these demonstrations  have been postponed until this section is that the energy-translation invariant systems have  properties that make them convenient for constructing explicit examples. These examples are only sketched, and the details of the straightforward but in some cases somewhat long-winding confirmations are left to the reader.

\subsubsection{\label{minimalwithout} A minimal example without time-reversal}

Here we demonstrate a `minimal' setup that satisfies the conditions in Assumptions \ref{Def1}. 
Let $H^i_{S'}$ and $H^{f}_{S'}$ be Hamiltonians on $\mathcal{H}_{S'}$ for which the eigenvalues are multiples of $s$, i.e.,  $s z^i_{n}$ and $sz^{f}_n$ for $z^i_{n},z^{f}_n\in \mathbb{Z}$, and let $H_E:=s\sum_{j}j|j\rangle\langle j|$ be the energy-ladder. Let $|c_i\rangle, |c_f\rangle\in\mathcal{H}_C$ be two orthonormal states, and let 
\begin{equation*}
\begin{split}
H_{S'C} := & H^{i}_{S'}\otimes|c_i\rangle\langle c_i| + H^{f}_{S'}\otimes|c_f\rangle\langle c_f|,\\ 
H := &  H_{S'C}\otimes \hat{1}_E + \hat{1}_{SC}\otimes H_E,
\end{split}
\end{equation*}
and $H^{\perp} := 0$. Let $U_{S'}$ be an arbitrary unitary operator on $\mathcal{H}_{S'}$ and define
$U := U_{S'}\otimes|c_f\rangle\langle c_i| + \overline{U}_{S'}\otimes|c_i\rangle\langle c_f|$. (The unitary operator $\overline{U}_{S'}$ plays no direct role in the protocol, but is there to make $U$ unitary.)
As one can see, $U[\hat{1}_{S'}\otimes |c_i\rangle\langle c_i|] = U_{S'}\otimes |c_f\rangle\langle c_i| = [\hat{1}_{S'}\otimes |c_f\rangle\langle c_f|]U$. 
This $U$ would in general not be energy conserving. However, the unitary operator $V(U)$ is by construction energy conserving on $S'CE$, i.e., $[H,V(U)] = 0$. 

The operators $\hat{1}_{S'}\otimes |c_i\rangle\langle c_i|$ and  $\hat{1}_{S'}\otimes |c_f\rangle\langle c_f|$ are block-diagonal with respect to the energy eigenspaces of $H_{S'C}$. Due to this one can confirm that  
\begin{equation*}
\begin{split}
V(\hat{1}_{S'}\otimes |c_i\rangle\langle c_i|) =  &\hat{1}_{S'}\otimes |c_i\rangle\langle c_i|\otimes\hat{1}_E,\\
V(\hat{1}_{S'}\otimes |c_f\rangle\langle c_f|) =  &\hat{1}_{S'}\otimes |c_f\rangle\langle c_f|\otimes\hat{1}_E.
\end{split}
\end{equation*}
By combining this observation with the general property $V(A)V(B) = V(AB)$ and the perfect control of $U$,  it follows that 
$V(U)[\hat{1}_{S'}\otimes |c_i\rangle\langle c_i|\otimes\hat{1}_E]
=  [\hat{1}_{S'}\otimes |c_i\rangle\langle c_i|\otimes\hat{1}_E]V(U)$.
Hence, the fact that $U$ satisfies perfect control implies that $V :=V(U)$ also satisfies the condition for perfect control, and we can conclude that this setup satisfies all conditions of Assumptions \ref{Def1}.

\subsubsection{\label{discretizedwithout} Discretized paths of Hamiltonians without time-reversal}

The intermediate Crooks relation only requires us to consider the end-points of the dynamics. It may nevertheless be useful to see how one can construct a discretized model of a parametric family of Hamiltonians that satisfies Assumptions \ref{Def1}.

Given a family of Hamiltonians $H(x)$ for $x\in [0,1]$ with $H(0) = H^{i}_{S'}$ and $H(1) = H^{f}_{S'}$ we discretize the path into $L+1$ steps, such that we get a sequence of Hamiltonians $H_l:=H(l/L)$ for $l=0,\ldots, L$.
Given the energy spacing $s$ in the energy ladder, we find approximate Hamiltonians $\widetilde{H}_{l}$ that have the eigenvalues $s z^{(l)}_{n}$ for $z^{(l)}_{n}\in \mathbb{Z}$ with corresponding orthonormal eigenvectors $\{|\chi^{(l)}_{n}\rangle\}_{n}$. (For  more details on the transition from  $H_l$ to  $\widetilde{H}_{l}$, see Section VII.C.1 in the Supplementary Material of \cite{Aberg13}.) We let $\{|c_{l}\rangle\}_{l=0}^{L}$ be a set of orthonormal elements spanning the Hilbert space $\mathcal{H}_{C}$ of the control system $C$, and define
\begin{equation}
\begin{split}
H_{S'C} := & \sum_{l=0}^{L}\widetilde{H}_l\otimes|c_l\rangle\langle c_l|,\\
H := & H_{S'C}\otimes \hat{1}_E + \hat{1}_{SC}\otimes H_E,
\end{split}
\end{equation}
for the energy ladder $H_E$. To compare with Assumptions \ref{Def1} we have
\begin{equation}
\begin{split}
& |c_{i}\rangle := |c_{0}\rangle,\quad |c_{f}\rangle := |c_{L}\rangle,\\
& H^{i}_{S'} := \widetilde{H}_{0},\quad H^{f}_{S'} := \widetilde{H}_{L},\\
& H^{\perp} := \sum_{l=1}^{L-1}\widetilde{H}_l \otimes |c_l\rangle\langle c_l|.
\end{split}
\end{equation}
In the following we shall define a unitary operator $U$ on $S'C$ that generates one single step along the discretization. The propagation along the path is obtained by iterating $U$ such that the entire evolution along the $L$-step discretization is generated by $U^{L}$.
Let $U_1,\ldots, U_L$ be arbitrary unitary operators on $\mathcal{H}_{S'}$ and define
$U := \sum_{l=0}^{L-1}U_{l}\otimes |c_{l+1}\rangle\langle c_l| + U_{L}\otimes |c_{0}\rangle\langle c_{L}|$. 
One can confirm that $U$ is unitary.
 The unitary operator $V(U)$ is energy conserving, and 
analogous to Appendix \ref{minimalwithout} one can confirm that 
$V(\hat{1}_{S'}\otimes |c_{l}\rangle\langle c_l|) = \hat{1}_{S'}\otimes |c_{l}\rangle\langle c_l|\otimes \hat{1}_E$, as well as $V(U)^{L}[\hat{1}_{S'}\otimes|c_i\rangle\langle c_i|\otimes\hat{1}_E]
=  [\hat{1}_{S'}\otimes|c_f\rangle\langle c_f|\otimes\hat{1}_E]V(U)^{L}$.
Hence, $V := V(U)^L$ satisfies the conditions in Assumptions \ref{Def1}.

\subsubsection{\label{minimalwith} A minimal example with time-reversal}

Here we consider a setup that satisfies the conditions in Assumptions \ref{Def2}.
With $H^{i}_{S},H^{f}_{S'}$ as in Appendix \ref{minimalwithout}, and $H_E$ the energy-ladder, 
 let $\{|c_{i+}\rangle, |c_{i-}\rangle, |c_{f+}\rangle, |c_{f-}\rangle\}$ be orthonormal elements spanning the Hilbert space $\mathcal{H}_{C}$ of the control system $C$. Let
\begin{equation}
\begin{split}
H_{S'C} := &  H^{i}_{S'} \otimes P^{i}_{C} + H^{f}_{S'} \otimes P^{f}_{C},\\
 P_{C}^{i} := &  |c_{i+}\rangle\langle c_{i+}| + |c_{i-}\rangle\langle c_{i-}|,\\
 P_{C}^{f} := &  |c_{f+}\rangle\langle c_{f+}| + |c_{f-}\rangle\langle c_{f-}|,\\
H := & H_{S'C}\otimes \hat{1}_E + \hat{1}_{SC}\otimes H_E,
\end{split}
\end{equation}
and $H^{\perp}: = 0$.
We next turn to the time-reversals.
On $\mathcal{H}_C$ we define 
\begin{equation}
\begin{split}
Y := & |c_{i+}\rangle\langle c_{i-}| + |c_{i-}\rangle\langle c_{i+}|\\
 &  + |c_{f+}\rangle\langle c_{f-}| + |c_{f-}\rangle\langle c_{f+}|.
\end{split}
\end{equation}
As one can see, $Y$ is a unitary operator on $\mathcal{H}_C$. Define the basis 
\begin{equation}
\begin{split}
\mathcal{B}:= & \{|\chi^{i}_{n}\rangle|c_{i+}\rangle\}_{n}\cup\{|\chi^{i}_{n}\rangle|c_{i-}\rangle\}_{n}\\
& \cup\{|\chi^{f}_{n}\rangle|c_{f+}\rangle\}_{n}\cup\{|\chi^{f}_{n}\rangle|c_{f-}\rangle\}_{n},
\end{split}
\end{equation}
where $\{|\chi^{i}_{n}\rangle\}_{n}$ is an orthonormal eigenbasis of $H^{i}_{S'}$ and $\{|\chi^{f}_{n}\rangle\}_{n}$ is an orthonormal eigenbasis of $H^{f}_{S'}$.
Define 
$\mathcal{T}_{S'C}(Q) := [\hat{1}_{S'}\otimes Y]Q^{\tau}[\hat{1}_{S'}\otimes Y^{\dagger}]$,
where $\tau$ denotes the transpose with respect to the basis $\mathcal{B}$. Note that $\hat{1}_{S'}\otimes Y$ is complex symmetric with respect to the basis $\mathcal{B}$ (and the space $\mathcal{H}_{S'C}$ on which it operates is finite-dimensional). 
Hence, according to Proposition \ref{TcharactFinite} it follows that $\mathcal{T}_{S'C}$ is a time-reversal (and is moreover the transpose with respect to some basis).
One can verify that $\mathcal{T}_{S'C}(H_{S'C}) =  H_{S'C}$,
and furthermore
\begin{equation}
\begin{split}
\mathcal{T}_{S'C}(\hat{1}_{S'}\otimes|c_{i+}\rangle\langle c_{i+}|)  = & \hat{1}_{S'}\otimes|c_{i-}\rangle\langle c_{i-}|,\\
\mathcal{T}_{S'C}(\hat{1}_{S'}\otimes|c_{f+}\rangle\langle c_{f+}|)  = &\hat{1}_{S'}\otimes|c_{f-}\rangle\langle c_{f-}|.
\end{split}
\end{equation}
Moreover, we define the time-reversal $\mathcal{T}_E$ on the energy reservoir as the transpose with respect to the basis $\{|j\rangle\}_{j\in\mathbb{Z}}$, and thus it is the case that $\mathcal{T}_E(H_E) = H_E$. We define the global time-reversal as $\mathcal{T} :=\mathcal{T}_{S'C}\otimes \mathcal{T}_E$. Let $U^{+}$ and $\overline{U}^{+}$ be arbitrary unitary operators on $\mathcal{H}_{S'}$, and define
\begin{equation}
\label{reoioireoire}
\begin{split}
U := & U^{+}\otimes |c_{f+}\rangle\langle c_{i+}| + \overline{U}^{+}\otimes|c_{i+}\rangle\langle c_{f+}|\\
& + U^{-}\otimes |c_{i-}\rangle\langle c_{f-}| + \overline{U}^{-}\otimes|c_{f-}\rangle\langle c_{i-}|,
\end{split}
\end{equation}
where $U^{-} :=\sum_{nn'}|\chi^{i}_n\rangle\langle \chi^{f}_{n'}|U^{+}|\chi^{i}_n\rangle\langle \chi^{f}_{n'}|$ and $\overline{U}^{-} :=\sum_{nn'}|\chi^{f}_n\rangle\langle \chi^{i}_{n'}|\overline{U}^{+}|\chi^{f}_n\rangle\langle \chi^{i}_{n'}|$. 
One can confirm that $U^{-}$, $\overline{U}^{-}$, and $U$ are unitary, and thus $V(U)$ is an energy conserving unitary operator.  Moreover, one can confirm that $\mathcal{T}_{S'C}(U) = U$. 
Since  $\mathcal{T}_{S'C}(H_{S'C}) =  H_{S'C}$ (and since $\mathcal{H}_{S'C}$ is finite-dimensional) we know by Lemma \ref{kldsdslk}  that $\mathcal{T}(V(U)) = V(\mathcal{T}_{S'C}(U)) = V(U)$. 
Hence, the dynamics is time-reversal symmetric. With a reasoning analogous to Appendix \ref{minimalwithout}, one can also show that $V(U)[\hat{1}_{S'}\otimes |c_{i+}\rangle\langle c_{i+}|\otimes\hat{1}_{E}] = [\hat{1}_{S'}\otimes |c_{f+}\rangle\langle c_{f+}|\otimes\hat{1}_{E}]V(U)$. Hence, all the conditions  of Assumptions \ref{Def2} are satisfied.

\subsubsection{\label{discretizedwith} Discretized paths of Hamiltonians with time-reversal}

Here we modify the setup of Appendix \ref{discretizedwithout} such that it incorporates time-reversal, and 
 satisfies the conditions in Assumptions \ref{Def2}.

We let $\{|c_{l}^{\pm}\rangle\}_{l=0}^{L}$ be a set of orthonormal elements spanning the Hilbert space $\mathcal{H}_{C}$ of the control system $C$. For the family of Hermitian operators  $(\widetilde{H}_{l})_{l=0}^{L}$, let
\begin{equation}
\begin{split}
H_{S'C} := & \sum_{l=0}^{L} \widetilde{H}_l \otimes P^{l}_{C},\\
 P_{C}^{l} := &  |c^{+}_{l}\rangle\langle c^{+}_{l}| + |c^{-}_{l}\rangle\langle c^{-}_{l}|,\\
H := & H_{S'C}\otimes \hat{1}_E + \hat{1}_{SC}\otimes H_E.
\end{split}
\end{equation}
To compare with Assumptions \ref{Def2} we have
\begin{equation*}
\begin{split}
& |c_{i\pm}\rangle := |c^{\pm}_{0}\rangle,\quad |c_{f\pm}\rangle := |c^{\pm}_{L}\rangle,\\
& H^{i}_{S'} := \widetilde{H}_{0},\quad H^{f}_{S'} := \widetilde{H}_{L},\\
& P^{i}_C := P_{C}^{0},\quad P^{f}_C := P_{C}^{L},\\ 
& H^{\perp} := \sum_{l=1}^{L-1}\widetilde{H}_l \otimes P^{l}_{C}.
\end{split}
\end{equation*}
We next turn to the time-reversals.
On $\mathcal{H}_C$ we define 
$Y := \sum_{l=0}^{L}(|c^{+}_{l}\rangle\langle c^{-}_{l}| + |c^{-}_{l}\rangle\langle c^{+}_{l}|)$.
One can confirm that $Y$ is a unitary operator. 
Note that $\mathcal{B}:=\{|\chi^{(l)}_{n}\rangle|c^{+}_l\rangle\}_{l,n}\cup \{|\chi^{(l)}_{n}\rangle|c^{-}_l\rangle\}_{l,n}$
 is an orthonormal basis of $\mathcal{H}_{S'C}$. We define the time-reversal 
$\mathcal{T}_{S'C}(Q) := [\hat{1}_{S'}\otimes Y]Q^{\tau}[\hat{1}_{S'}\otimes Y^{\dagger}]$,
where $\tau$ denotes the transpose with respect to the basis $\mathcal{B}$.
One can show that $\mathcal{T}_{S'C}(H_{S'C}) =  H_{S'C}$ and 
\begin{equation*}
\begin{split}
\mathcal{T}_{S'C}(\hat{1}_{S'}\otimes|c^{+}_s\rangle\langle c^{+}_s|)  =  \hat{1}_{S'}\otimes|c^{-}_s\rangle\langle c^{-}_s|,\quad s = 0,\ldots, L,
\end{split}
\end{equation*}
which thus in particular include the end-point control states $|c_{i\pm}\rangle$ and $|c_{f\pm}\rangle$.
The time-reversal $\mathcal{T}_E$ is defined as the transpose with respect to the basis $\{|j\rangle\}_{j\in\mathbb{Z}}$ of the energy-ladder.
Let $U_{0}^{+},\ldots, U_{L}^{+}$ be arbitrary unitary operators on $\mathcal{H}_{S'}$, and define
\begin{equation*}
\begin{split}
U^{-}_{L} := & \sum_{nn'}|\chi^{(L)}_{n'}\rangle\langle \chi^{(0)}_{n}|U^{+}_{L} |\chi^{(L)}_{n'}\rangle\langle \chi^{(0)}_{n}|, \\
U^{-}_{s} := & \sum_{nn'}|\chi^{(s)}_{n'}\rangle\langle \chi^{(s+1)}_{n}|U^{+}_{s}|\chi^{(s)}_{n'}\rangle\langle \chi^{(s+1)}_{n}|,\,\, s = 0,\ldots,L-1,
\end{split}
\end{equation*}
as well as the unitary operator 
\begin{equation*}
\begin{split}
U := & U^{+}_{L}\otimes |c^{+}_{0}\rangle\langle c^{+}_{L}| + \sum_{s=0}^{L-1}  U^{+}_{s}\otimes |c^{+}_{s+1}\rangle\langle c^{+}_{s}| \\
& +  U^{-}_{L}\otimes|c^{-}_{L}\rangle\langle c^{-}_{0}| +\sum_{s=0}^{L-1} U^{-}_{s}\otimes |c^{-}_{s}\rangle\langle c^{-}_{s+1}|. 
\end{split}
\end{equation*}
One can confirm that $U$ is invariant with respect to $\mathcal{T}_{S'C}$.
By an analogous reasoning as in Appendix \ref{minimalwith} one can show that the energy conserving unitary operator $V(U)$ satisfies $\mathcal{T}(V(U)) = V(U)$. By the general properties of time-reversals it thus follows that the operator $V:=V(U)^{L}$ is an energy conserving, time-reversal symmetric, and unitary operator. The proof that $V$ also satisfies perfect control $V[\hat{1}_{S'}\otimes|c_{i+}\rangle\langle c_{i+}|\otimes\hat{1}_E] = [\hat{1}_{S'}\otimes|c_{f+}\rangle\langle c_{f+}|\otimes\hat{1}_E]V$ can be done analogously as to Appendix \ref{discretizedwithout}, and thus we can conclude that  $V:=V(U)^{L}$ satisfies all the conditions of Assumptions \ref{Def2}.

\subsubsection{\label{nonorthogonalminimalwith} A minimal example with time-reversal and non-orthogonal control states}
In Appendix \ref{SecResetting} it was pointed out that a state and its time-reversal do not necessarily have to be orthogonal to each other. It was also claimed that it is possible to find a setup that satisfies Assumptions \ref{Def2} and in addition is such that the pair of control states $|c_{i+}\rangle$ and $|c_{i-}\rangle$ are not orthogonal to each other (and analogously for $|c_{f+}\rangle$ and $|c_{f-}\rangle$). Here we demonstrate this claim, and we also discuss how additional assumptions may enforce orthogonality.

Let $H^{i}_{S'}$ and $H^{f}_{S'}$ be as in Appendix \ref{minimalwithout} with eigenvalues $sz^i_n$ and eigenvectors $|\chi^i_n\rangle$, as well as eigenvalues $sz^f_n$ and eigenvectors $|\chi^f_n\rangle$, respectively. We let $H_E$ be the energy ladder, with eigenvalues $sj$ and eigenstates $|j\rangle$. We let the control space $\mathcal{H}_C$ be four-dimensional, with an orthonormal basis $\{|i\rangle|0\rangle, |i\rangle|1\rangle, |f\rangle|0\rangle, |f\rangle|1\rangle\}$.

Define $\mathcal{T}_{S'C}$ as the transpose with respect to the basis $\{|\chi^i_n\rangle|i\rangle|0\rangle, |\chi^i_n\rangle|i\rangle|1\rangle\}_n\cup \{|\chi^f_{n'}\rangle|f\rangle|0\rangle, |\chi^f_{n'}\rangle|f\rangle|1\rangle\}_{n'}$, and let $\mathcal{T}_E$ be the transpose with respect to the orthonormal basis $\{|j\rangle\}_{j\in\mathbb{Z}}$.
For $\alpha,\beta,\gamma,\delta\in \mathbb{C}$, such that  $|\alpha|^2 + |\beta|^2 = 1$ and $|\gamma|^2+ |\delta|^2 = 1$, let 
\begin{equation*}
\begin{split}
|c_{i+}\rangle := |i\rangle(\alpha|0\rangle +\beta|1\rangle),\quad |c_{i-}\rangle := |i\rangle(\alpha^*|0\rangle +\beta^*|1\rangle),\\
|c_{f+}\rangle := |f\rangle(\gamma|0\rangle +\delta|1\rangle),\quad |c_{f-}\rangle := |f\rangle(\gamma^{*}|0\rangle +\delta^{*}|1\rangle).
\end{split}
\end{equation*}
The vectors $|c_{i+}\rangle$ and $|c_{i-}\rangle$ are typically not orthogonal for generic choices of $\alpha,\beta$, and analogously for the pair $|c_{f+}\rangle$, $|c_{f-}\rangle$. It is convenient to also define 
\begin{equation*}
\begin{split}
|\overline{c}_{i+}\rangle := |i\rangle(\beta^{*}|0\rangle -\alpha^{*}|1\rangle),\quad |\overline{c}_{i-}\rangle := |i\rangle(\beta|0\rangle -\alpha|1\rangle),\\
|\overline{c}_{f+}\rangle := |f\rangle(\delta^{*}|0\rangle -\gamma^{*}|1\rangle),\quad |\overline{c}_{f-}\rangle := |f\rangle(\delta|0\rangle -\gamma|1\rangle).
\end{split}
\end{equation*}
One can note that $\{|c_{i+}\rangle ,|\overline{c}_{i+}\rangle\}$ and $\{|c_{i-}\rangle ,|\overline{c}_{i-}\rangle\}$ both form orthonormal bases of the subspace $\Sp\{|i\rangle|0\rangle,|i\rangle|1\rangle\}$. Similarly, each of $\{|c_{f+}\rangle ,|\overline{c}_{f+}\rangle\}$ and $\{|c_{f-}\rangle ,|\overline{c}_{f-}\rangle\}$ forms an orthonormal basis of $\Sp\{|f\rangle|0\rangle, |f\rangle|1\rangle\}$.
With 
\begin{equation*}
\begin{split}
H_{S'C} := & H^i_{S'}\otimes P^i_C + H^{f}_{S'}\otimes P^f_C,\\
P^i_C := & |i\rangle\langle i|\otimes (|0\rangle\langle 0|+|1\rangle\langle 1|),\\
P^f_C := & |f\rangle\langle f|\otimes (|0\rangle\langle 0|+|1\rangle\langle 1|),
\end{split}
\end{equation*}
one can confirm that $\mathcal{T}_{S'C}(H_{S'C}) = H_{S'C}$, and also that 
\begin{equation*}
\begin{split}
\mathcal{T}_{S'C}(\hat{1}_{S'}\otimes |c_{i+}\rangle\langle c_{i+}|) = & \hat{1}_{S'}\otimes |c_{i-}\rangle\langle c_{i-}|,\\ \mathcal{T}_{S'C}(\hat{1}_{S'}\otimes |c_{f+}\rangle\langle c_{f+}|) = & \hat{1}_{S'}\otimes |c_{f-}\rangle\langle c_{f-}|,
\end{split}
\end{equation*}
and thus these control states are time-reversals of each other. 
Let $U^{+}$ and $\overline{U}^{+}$ be unitary operators on $\mathcal{H}_{S'}$ and define
\begin{equation*}
\begin{split}
U^{-} := \sum_{n,n'}|\chi^{i}_n\rangle\langle\chi^f_{n'}|U^{+}|\chi^i_{n}\rangle\langle \chi^f_{n'}|,\\
\overline{U}^{-} := \sum_{n,n'}|\chi^{i}_n\rangle\langle\chi^f_{n'}|\overline{U}^{+}|\chi^i_{n}\rangle\langle \chi^f_{n'}|.
\end{split}
\end{equation*}
Since $U^{+}$ and $\overline{U}^{+}$ are unitary, it follows that $U^{-}$ and $\overline{U}^{-}$ are also unitary. 
By the unitarity of  $U^{+}$, $\overline{U}^{+}$, $U^{-}$, $\overline{U}^{-}$ together with the fact that $\{|c_{i+}\rangle, |\overline{c}_{i+}, |c_{f-}\rangle,|\overline{c}_{f-}\rangle\}$ as well as $\{|c_{f+}\rangle,|\overline{c}_{f+}\rangle, |c_{i-}\rangle,|\overline{c}_{i-}\rangle\}$ are orthonormal bases of $\mathcal{H}_{C}$, it follows that the following operator is unitary
\begin{equation*}
\begin{split}
U := & U^{+}\otimes |c_{f+}\rangle\langle c_{i+}| + \overline{U}^{+}\otimes |\overline{c}_{f+}\rangle\langle\overline{c}_{i+}|\\
& +U^{-}\otimes |c_{i-}\rangle\langle c_{f-}| + \overline{U}^{-}\otimes |\overline{c}_{i-}\rangle\langle\overline{c}_{f-}|.
\end{split}
\end{equation*}
One can also confirm that $\mathcal{T}_{S'C}(U) = U$ and  $U[\hat{1}_{S'}\otimes |c_{i+}\rangle\langle c_{i+}|] = [\hat{1}_{S'}\otimes |c_{f+}\rangle\langle c_{f+}|]U$. 
Although $U$ thus is time-reversal symmetric and satisfies perfect control with respect to the designated control states, it is typically not energy conserving. However, since  $\mathcal{T}_{S'C}(H_{S'C}) =  H_{S'C}$ (and since $\mathcal{H}_{S'C}$ is finite-dimensional) we know by Lemma \ref{kldsdslk}  that $\mathcal{T}(V(U)) = V(\mathcal{T}_{S'C}(U)) = V(U)$. 
Moreover, by construction, $[V(U),H_{S'C}\otimes\hat{1}_E+ \hat{1}_{S'C}\otimes H_E] = 0$, and thus $V(U)$ is both time-reversal symmetric and energy conserving.

Since $\hat{1}_{S'}\otimes|c_{i+}\rangle\langle c_{i+}|$ and $\hat{1}_{S'}\otimes|c_{f+}\rangle\langle c_{f+}|$
are block diagonal with respect to the eigenspaces of $H_{S'C}$, it follows that $V(\hat{1}_{S'}\otimes|c_{i+}\rangle\langle c_{i+}|) 
=  \hat{1}_{S'}\otimes|c_{i+}\rangle\langle c_{i+}|\otimes \hat{1}_E$ and $V(\hat{1}_{S'}\otimes|c_{f+}\rangle\langle c_{f+}|) 
=  \hat{1}_{S'}\otimes|c_{f+}\rangle\langle c_{f+}|\otimes \hat{1}_E$. 
With a reasoning analogous to Appendix \ref{minimalwithout}, one can also show that $V(U)[\hat{1}_{S'}\otimes |c_{i+}\rangle\langle c_{i+}|\otimes\hat{1}_{E}] = [\hat{1}_{S'}\otimes |c_{f+}\rangle\langle c_{f+}|\otimes\hat{1}_{E}]V(U)$. Hence, all the conditions  of Assumptions \ref{Def2} are satisfied. Moreover, this is achieved with control states where $|c_{i+}\rangle$ is not necessarily orthogonal to  $|c_{i-}\rangle$, and where $|c_{f+}\rangle$ is not necessarily orthogonal to  $|c_{f-}\rangle$.
This includes the special case that $|c_{i+}\rangle$ and $|c_{i-}\rangle$ are parallel, which happens if the phase factors of $\alpha$ and $\beta$ are identical. 
One may wonder how this parallelity fits with the idea that the control states represent the forward and reverse propagation of control parameters. First of all, if $|c_{i-}\rangle$ is parallel to $|c_{i+}\rangle$, and 
 $|c_{f-}\rangle$ parallel to $|c_{f+}\rangle$, then perfect control implies that repeated applications of $V(U)$ swap the control back and forth between $|c_{i+}\rangle\langle c_{i+}|$  and $|c_{f+}\rangle\langle c_{f+}|$. Moreover, one should observe that the above example, and indeed Assumptions  \ref{Def2}, only concerns the mapping between the initial and final control states, and does not include any requirements concerning what happens at any potential intermediate states of the process. Let us now additionally assume (much as in Appendix \ref{discretizedwith}) a sequence of control states $\{|c^{\pm}_l\rangle\}_l$.  We  do moreover demand perfect control for all states in the forward path, as well as perfect control for all states in the reverse path, i.e., 
\begin{equation*}
\begin{split}
V(U)[\hat{1}_{S'E}\otimes |c^{+}_l\rangle\langle c^{+}_l|] = [\hat{1}_{S'E}\otimes |c^{+}_{l+1}\rangle\langle c^{+}_{l+1}|]V(U),\\
V(U)[\hat{1}_{S'E}\otimes |c^{-}_l\rangle\langle c^{-}_l|] = [\hat{1}_{SE}\otimes |c^{-}_{l-1}\rangle\langle c^{-}_{l-1}|]V(U).
\end{split}
\end{equation*}
By using the unitarity of the global evolution $V(U)$ it follows that 
\begin{equation*}
\begin{split}
& [\hat{1}_{S'E}\otimes |c^{-}_l\rangle\langle c^{-}_l|][\hat{1}_{S'E}\otimes |c^{+}_l\rangle\langle c^{+}_l|] \\
&= V^{\dagger}(U)[\hat{1}_{SE}\otimes |c^{-}_{l-1}\rangle\langle c^{-}_{l-1}|][\hat{1}_{S'E}\otimes |c^{+}_{l+1}\rangle\langle c^{+}_{l+1}|]V(U).
\end{split}
\end{equation*}
Consequently, if we would demand that $\Sp\{|c^{+}_{l+1}\rangle,|c^{-}_{l+1}\rangle\}$ should be orthogonal to $\Sp\{|c^{+}_{l-1}\rangle,|c^{-}_{l-1}\rangle\}$, then it follows that $|c^{+}_{l}\rangle$ must be  orthogonal to $|c^{-}_{l}\rangle$. Hence, perfect control of the reverse and forward paths conspires with the assumed orthogonality of the control spaces to enforce orthogonality between the forward and reverse control states.

\section{\label{SecConditional} Conditional fluctuation relations}

Here we consider a generalized type of fluctuation relation that naturally includes non-equilibrium states. 
This extension may at first sight seem rather radical. However, our quantum fluctuation relation in Proposition \ref{PropQuantumCrooks} strictly speaking already requires initial non-equilibrium states, due to the control system, as discussed in Appendix \ref{TheInitalStates}.

As we have seen in Appendix \ref{SecConditionsOnE}, the assumption of perfect control is a rather strong condition, and may require an energy reservoir spectrum that is unbounded from below (as well as above).  The conditional fluctuation relations allow us to abolish the perfect control (see Appendix \ref{AbolishPerfectControl}). Not only can we avoid unbounded spectra, but we can also base the conditional fluctuation relations on finite-dimensional Hilbert spaces (see Appendix \ref{ExampleTwoQubits} for an explicit example).

\subsection{\label{SecTheGibbsMap}The Gibbs map and the partition map}

For a given operator $A$ we define the Gibbs map $\mathcal{G}_{A}$ and the partition map $\mathcal{Z}_{A}$ by 
\begin{equation*}
\mathcal{G}_{A}(Q) :=  \frac{1}{\mathcal{Z}_{A}(Q)}\mathcal{J}_{A}(Q),\quad \mathcal{Z}_{A}(Q) :=  \Tr \mathcal{J}_{A}(Q).
\end{equation*}
 By construction, $\mathcal{G}_{A}(Q)$ is a density operator whenever $Q$ is a positive operator (modulo the existence of $\mathcal{Z}_{A}(Q)$). In the special case that $A = \beta H$ for $\beta \geq 0$, then $\mathcal{Z}_{\beta H}(\hat{1}) = Z_{\beta}(H)$ and $\mathcal{G}_{\beta H}(\hat{1}) = G_{\beta}(H)$.  

An immediate question is what class of density operators that can be reached by the Gibbs map. 
If $H$ is a bounded Hermitian operator, then $e^{\pm \beta H}$ is also bounded, and thus $\Vert \mathcal{J}_{-\beta H}(\rho) \Vert < +\infty$, where $\Vert Q\Vert := \sup_{\Vert \psi\Vert}\Vert Q|\psi\rangle\Vert$ denotes the standard operator norm. For an arbitrary density operator $\rho$ (which by virtue of being trace class also is bounded, see e.g.~\cite{Davies76}) let $Q := \mathcal{J}_{-\beta H}(\rho)/\Vert \mathcal{J}_{-\beta H}(\rho) \Vert$.  By construction $0\leq Q \leq 1$, and one can confirm that that $\mathcal{G}_{\beta H}(Q) = \rho$. Hence, for bounded Hermitian operators $H$, one can reach all density operators via the Gibbs map (and thus in particular if the Hilbert space is finite-dimensional). 
 The issue becomes more complicated if $H$ is unbounded, since we have to take into account the domain of definition of $H$ and of $e^{\pm \beta H}$. More generally it may be the case that the Gibbs map does not generate the entire set of density operators. Although we do use the Gibbs map for unbounded $H$, we will nevertheless not consider this question further in this investigation.

One can further note that $\mathcal{G}_{\beta H}$ is a many-to-one map from the set of POVM elements, although in a relatively mild sense. If $Q$ is a POVM element, then $rQ$ for $0< r <1$ is also a valid POVM element. The Gibbs map $\mathcal{G}_{\beta H}$ maps both $Q$ and $rQ$ to the same density operator.

\subsection{Without time-reversal}

As mentioned earlier, the fact that we drop the assumption of perfect control implies a simpler structure. The first simplification is that we here only need to consider two subsystems: the energy reservoir $E$ and the rest $\widetilde{S}$. For the general theory there is no need for any further partitioning into  subsystems, but in order to relate to the results in previous sections we would let  $\widetilde{S} = S'C = SBC$.

\begin{Assumptions}
\label{DefNonEq}
Let $\mathcal{H}_{\widetilde{S}}$ and $\mathcal{H}_E$ be complex Hilbert spaces
\begin{itemize}
\item Let  $H_{\widetilde{S}}$ and $H_E$ be Hermitian operators on $\mathcal{H}_{\widetilde{S}}$ and $\mathcal{H}_{E}$, respectively, and let 
\begin{equation*}
H := H_{\widetilde{S}}\otimes\hat{1}_E  + \hat{1}_{\widetilde{S}}\otimes H_{E}.
\end{equation*}
\item Let $V$ be a unitary operator on $\mathcal{H}_{\widetilde{S}}\otimes\mathcal{H}_{E}$ such that $[H,V] = 0$.
\item Let $Q_{\widetilde{S}}^{i}$ and  $Q_{\widetilde{S}}^{f}$ be operators on $\mathcal{H}_{\widetilde{S}}$ such that $0\leq Q_{\widetilde{S}}^{i} \leq \hat{1}_{\widetilde{S}}$ and $0\leq Q_{\widetilde{S}}^{f} \leq \hat{1}_{\widetilde{S}}$.
\end{itemize}
\end{Assumptions}
The operators  $Q_{\widetilde{S}}^{i}$ and $Q_{\widetilde{S}}^{f}$ play dual roles in this analysis. 
First, they correspond to control measurements. For example, we can form the two POVMs $\{Q_{\widetilde{S}}^{i},\hat{1}_{\widetilde{S}}-Q_{\widetilde{S}}^{i}\}$ and  $\{Q_{\widetilde{S}}^{f},\hat{1}_{\widetilde{S}}-Q_{\widetilde{S}}^{f}\}$. (Nothing in this formalism forces us to necessarily use binary POVMs. See the discussion at the end of Appendix \ref{SecCondWithTimeRev}.) In these POVMs,  $Q_{\widetilde{S}}^{i}$ and $Q_{\widetilde{S}}^{f}$ are the `successful' outcomes, and the CPMs $\tilde{\mathcal{F}}$ and $\tilde{\mathcal{R}}$, defined in (\ref{deftildeFR}) below, generate the corresponding (non-normalized) post-measurement states of the reservoir conditioned on these successful outcomes.

The second role of $Q_{\widetilde{S}}^{i}$ and $Q_{\widetilde{S}}^{f}$ is that the they parametrize initial states via the Gibbs map $\mathcal{G}_{\beta H_{\widetilde{S}}}$. These roles are swapped within the pair, such that for the reverse process, $Q_{\widetilde{S}}^{f}$ gives the initial state, and $Q_{\widetilde{S}}^{i}$ the measurement.

We define the completely positive maps 
\begin{equation}
\label{deftildeFR}
\begin{split}
\tilde{\mathcal{F}}(\sigma) = & \Tr_{\widetilde{S}}([Q_{\widetilde{S}}^{f}\otimes\hat{1}_E]V[\mathcal{G}_{\beta H_{\widetilde{S}}}(Q_{\widetilde{S}}^{i})\otimes \sigma]V^{\dagger}),\\
\tilde{\mathcal{R}}(\sigma) =  & \Tr_{\widetilde{S}}([Q_{\widetilde{S}}^{i}\otimes\hat{1}_E]V^{\dagger}[\mathcal{G}_{\beta H_{\widetilde{S}}}(Q_{\widetilde{S}}^{f})\otimes \sigma]V).
\end{split}
\end{equation}
The following is the counterpart of Lemma \ref{Transition}, and is left without proof.
\begin{Lemma}
\label{GenTransition}
With $V$, $H_{\widetilde{S}}$, and $H_E$ as in Assumptions \ref{DefNonEq}, it is the case that for every  $\alpha \in \mathbb{C}$
\begin{equation}
\label{jhvhvbnhj}
\begin{split}
& V[e^{\alpha H_{\widetilde{S}}}\otimes \hat{1}_E] = [e^{\alpha H_{\widetilde{S}}}\otimes e^{\alpha H_E}]V[\hat{1}_{\widetilde{S}}\otimes e^{-\alpha H_{E}}],\\
& [e^{\alpha H_{\widetilde{S}}}\otimes \hat{1}_E]V = [\hat{1}_{\widetilde{S}}\otimes e^{-\alpha H_E}]V[e^{\alpha H_{\widetilde{S}}}\otimes e^{\alpha H_E}].
\end{split}
\end{equation}
\end{Lemma}

\begin{Proposition}
\label{NonEqConditional}
With Assumptions \ref{DefNonEq}, the CPMs $\tilde{\mathcal{F}}$ and $\tilde{\mathcal{R}}$ as defined in (\ref{deftildeFR}) satisfy
\begin{equation}
\label{adskjvhb}
\begin{split}
 \mathcal{Z}_{\beta H_{\widetilde{S}}}(Q_{\widetilde{S}}^{i})\tilde{\mathcal{F}} = \mathcal{Z}_{\beta H_{\widetilde{S}}}(Q_{\widetilde{S}}^{f})\mathcal{J}_{\beta H_E}\tilde{\mathcal{R}}^{*}\mathcal{J}_{\beta H_E}^{-1}.
\end{split}
\end{equation}
\end{Proposition}
\begin{proof}
By comparing the definition of the CPM $\tilde{\mathcal{R}}$  with Eqs.~(\ref{conjremark1}) and (\ref{conjremark2}) in Appendix \ref{ConjugateCPMs}, we can conclude that
\begin{equation}
\begin{split}
& \tilde{\mathcal{R}}^{*}(Y)\\
& =    \Tr_{\widetilde{S}}([\mathcal{G}_{\beta H_{\widetilde{S}}}(Q_{\widetilde{S}}^{f})\otimes \hat{1}_E]V[Q_{\widetilde{S}}^{i}\otimes Y]V^{\dagger})\\
& =  \frac{1}{\mathcal{Z}_{\beta H_{\widetilde{S}}}(Q_{\widetilde{S}}^{f})}\Tr_{\widetilde{S}}\Big([Q_{\widetilde{S}}^{f}\otimes \hat{1}_E]\\
&\quad [e^{-\beta H_{\widetilde{S}}/2}\otimes\hat{1}_E]V [Q_{\widetilde{S}}^{i}\otimes Y] V^{\dagger} [e^{-\beta H_{\widetilde{S}}/2}\otimes\hat{1}_E]\Big)\\
& [\textrm{Lemma \ref{GenTransition}}]\\
& = \frac{1}{\mathcal{Z}_{\beta H_{\widetilde{S}}}(Q_{\widetilde{S}}^{f})}e^{\beta H_E/2}\Tr_{\widetilde{S}}\Big([Q_{\widetilde{S}}^{f}\otimes \hat{1}_E]V\\
& \quad \quad [e^{-\beta H_{\widetilde{S}}/2} Q_{\widetilde{S}}^{i}e^{-\beta H_{\widetilde{S}}/2}\otimes \mathcal{J}_{\beta H_E}(Y) ] V^{\dagger}\Big)e^{\beta H_{E}/2}\\
& = \frac{\mathcal{Z}_{\beta H_{\widetilde{S}}}(Q_{\widetilde{S}}^{i})}{\mathcal{Z}_{\beta H_{\widetilde{S}}}(Q_{\widetilde{S}}^{f})} \mathcal{J}_{\beta H_E}^{-1}\circ \tilde{\mathcal{F}}\circ \mathcal{J}_{\beta H_E}(Y).
\end{split}
\end{equation}
Hence $\mathcal{Z}_{\beta H_{\widetilde{S}}}(Q_{\widetilde{S}}^{f})\tilde{\mathcal{R}}^{*} =  \mathcal{Z}_{\beta H_{\widetilde{S}}}(Q_{\widetilde{S}}^{i}) \mathcal{J}_{\beta H_E}^{-1}\circ \tilde{\mathcal{F}}\circ \mathcal{J}_{\beta H_E}$.
By multiplying from the left with $ \mathcal{J}_{\beta H_E}$ and from the right with $ \mathcal{J}_{\beta H_E}^{-1}$ we obtain Eq.~(\ref{adskjvhb}).
\end{proof}

\subsection{\label{SecCondWithTimeRev}With time-reversal}

\begin{Assumptions}
\label{AssumpCondTimeRev}
Let $\mathcal{H}_{\widetilde{S}}$ and $\mathcal{H}_E$ be complex Hilbert spaces.
 Let $\mathcal{T}_{\widetilde{S}}$ and $\mathcal{T}_{E}$ be time-reversals on $\widetilde{S}$ and $E$, respectively, and let $\mathcal{T} := \mathcal{T}_{\widetilde{S}}\otimes \mathcal{T}_{E}$.
\begin{itemize}
\item Let  $H_{\widetilde{S}}$ and $H_E$ be Hermitian operators on $\mathcal{H}_{\widetilde{S}}$ and $\mathcal{H}_{E}$, respectively, and let 
\begin{equation}
H := H_{\widetilde{S}}\otimes\hat{1}_E  + \hat{1}_{\widetilde{S}}\otimes H_{E}.
\end{equation}
\item Let $V$ be a unitary operator on $\mathcal{H}_{\widetilde{S}}\otimes\mathcal{H}_{E}$ such that $[H,V] = 0$.
\item Let $Q_{\widetilde{S}}^{i+}$ and  $Q_{\widetilde{S}}^{f+}$ be  operators on $\mathcal{H}_{\widetilde{S}}$ such that $0\leq Q_{\widetilde{S}}^{i+} \leq \hat{1}_{\widetilde{S}}$ and $0\leq Q_{\widetilde{S}}^{f+} \leq \hat{1}_{\widetilde{S}}$.
\item Let  $\mathcal{T}_{E}(H_E) = H_E$, $\mathcal{T}_{\widetilde{S}}(H_{\widetilde{S}}) = H_{\widetilde{S}}$, and $\mathcal{T}(V) = V$. Define $Q_{\widetilde{S}}^{i-} := \mathcal{T}_{\widetilde{S}}(Q_{\widetilde{S}}^{i+})$ and $Q_{\widetilde{S}}^{f-} := \mathcal{T}_{\widetilde{S}}(Q_{\widetilde{S}}^{f+})$.
\end{itemize}
\end{Assumptions}
These assumptions are constructed such that the triple $V,Q_{\widetilde{S}}^{i}:=Q_{\widetilde{S}}^{i+},Q_{\widetilde{S}}^{f}:=Q_{\widetilde{S}}^{f+}$ satisfies Assumptions \ref{DefNonEq}. Simultaneously, the triple $V,Q_{\widetilde{S}}^{i}:= Q_{\widetilde{S}}^{f-},Q_{\widetilde{S}}^{f} := Q_{\widetilde{S}}^{i-}$ also satisfies Assumptions \ref{DefNonEq}. 

By Lemma \ref{PreservesIdentity} and Lemma \ref{ReorderHermImplyPos} it follows that  $0\leq Q_{\widetilde{S}}^{i+} \leq \hat{1}_{\widetilde{S}}$ implies $0\leq Q_{\widetilde{S}}^{i-} \leq \hat{1}_{\widetilde{S}}$ and analogously for $Q_{\widetilde{S}}^{f-}$. 

\begin{Lemma}
\label{LemmaSimpleRelations}
Let $\mathcal{T}$ be a time-reversal and $A$ an operator such that $\mathcal{T}(A) = A$ and $A^{\dagger} = A$. Then
\begin{equation}
\label{SimpleRelations}
\begin{split}
& \mathcal{J}_{A}\mathcal{T} =  \mathcal{T}\mathcal{J}_{A},\\
& \mathcal{Z}_{A}\big(\mathcal{T}(Q)\big) =  \mathcal{Z}_{A}(Q),\\
& \mathcal{G}_{A}\big(\mathcal{T}(Q)\big) =  \mathcal{T}\big(\mathcal{G}_{A}(Q)\big),\\
& \Tr\big(Q\mathcal{J}_{A}(R)\big) =  \Tr\big(\mathcal{J}_{A}(Q)R\big).
\end{split}
\end{equation}
\end{Lemma}

Define the CPMs 
\begin{equation}
\label{deftildeFRpm}
\begin{split}
\tilde{\mathcal{F}}_{+}(\sigma) = & \Tr_{\widetilde{S}}([Q_{\widetilde{S}}^{f+}\otimes \hat{1}_E]V[\mathcal{G}_{\beta H_{\widetilde{S}}}(Q_{\widetilde{S}}^{i+})\otimes \sigma]V^{\dagger}),\\
\tilde{\mathcal{R}}_{+}(\sigma) = & \Tr_{\widetilde{S}}([Q_{\widetilde{S}}^{i+}\otimes \hat{1}_E]V^{\dagger}[\mathcal{G}_{\beta  H_{\widetilde{S}}}(Q_{\widetilde{S}}^{f+})\otimes \sigma]V),\\
\tilde{\mathcal{F}}_{-}(\sigma) = & \Tr_{\widetilde{S}}([Q_{\widetilde{S}}^{i-}\otimes \hat{1}_E]V[\mathcal{G}_{\beta H_{\widetilde{S}}}(Q_{\widetilde{S}}^{f-})\otimes \sigma]V^{\dagger}),\\
\tilde{\mathcal{R}}_{-}(\sigma) = & \Tr_{\widetilde{S}}([Q_{\widetilde{S}}^{f-}\otimes \hat{1}_E]V^{\dagger}[\mathcal{G}_{\beta H_{\widetilde{S}}}(Q_{\widetilde{S}}^{i-})\otimes \sigma]V).
\end{split} 
\end{equation}
The CPMs $\tilde{\mathcal{F}}_{\pm}$ and $\tilde{\mathcal{R}}_{\pm}$ describe the \emph{unnormalized} mapping from the input to the output, conditioned on the successful control measurement. The corresponding success probabilities are given by the traces. (In Appendix \ref{ClassCondJarzynskiRel} we shall consider these success probabilities in the special case of energy translation invariance.) Analogous to Lemma \ref{ChannelReversal} one can prove the following.
\begin{Lemma}
\label{ReversalRelation}
With Assumptions \ref{AssumpCondTimeRev}, the CPMs  $\tilde{\mathcal{R}}_{+}$ and $\tilde{\mathcal{F}}_{-}$ as defined in Eq.~(\ref{deftildeFRpm}) are related as
\begin{equation}
\mathcal{T}_E\tilde{\mathcal{R}}_{+} = \tilde{\mathcal{F}}_{-}\mathcal{T}_E
\end{equation}
and thus
\begin{equation}
\tilde{\mathcal{R}}^{*}_{+} = \tilde{\mathcal{F}}^{\ominus}_{-}.
\end{equation}
\end{Lemma}
The proof is obtained by combining the definition of $\tilde{\mathcal{R}}_{+}$ in Eq.~(\ref{deftildeFRpm})  with the general fact that $ \mathcal{T}_1\big(\Tr_2(\rho)\big) = \Tr_{2}\big([\mathcal{T}_1\otimes\mathcal{T}_2](\rho)\big)$, the time-reversal symmetry, 
and Lemma \ref{LemmaSimpleRelations}.

\begin{Proposition}[Conditional quantum fluctuation relation]
\label{PropConditionalFluctuationThm}
With Assumptions \ref{AssumpCondTimeRev}, the CPMs $\tilde{\mathcal{F}}_{+}$ and $\tilde{\mathcal{F}}_{-}$, as defined in Eq.~(\ref{deftildeFRpm}), are related as
\begin{equation}
\label{ForwardRel}
\mathcal{Z}_{\beta H_{\widetilde{S}}}(Q_{\widetilde{S}}^{i})\tilde{\mathcal{F}}_{+} = \mathcal{Z}_{\beta H_{\widetilde{S}}}(Q_{\widetilde{S}}^{f})\mathcal{J}_{\beta H_E}\tilde{\mathcal{F}}_{-}^{\ominus}\mathcal{J}_{\beta H_E}^{-1}.
\end{equation}
Here we use the notation $\mathcal{Z}_{\beta H_{\widetilde{S}}}(Q_{\widetilde{S}}^{i}):=  \mathcal{Z}_{\beta H_{\widetilde{S}}}(Q_{\widetilde{S}}^{i+})=  \mathcal{Z}_{\beta H_{\widetilde{S}}}(Q_{\widetilde{S}}^{i-})$ and $\mathcal{Z}_{\beta H_{\widetilde{S}}}(Q_{\widetilde{S}}^{f}):=  \mathcal{Z}_{\beta H_{\widetilde{S}}}(Q_{\widetilde{S}}^{f+})=  \mathcal{Z}_{\beta H_{\widetilde{S}}}(Q_{\widetilde{S}}^{f-})$.
\end{Proposition}
\begin{proof}
The triple $V$, $Q_{\widetilde{S}}^{i+}$, and $Q_{\widetilde{S}}^{f+}$ from Assumptions \ref{AssumpCondTimeRev} satisfies Assumptions  \ref{DefNonEq} with $Q_{\widetilde{S}}^{i}:= Q_{\widetilde{S}}^{i+}$, $Q_{\widetilde{S}}^{f} := Q_{\widetilde{S}}^{f+}$. Hence, Proposition \ref{NonEqConditional} is applicable and yields 
$\mathcal{Z}_{\beta H_{\widetilde{S}}}(Q_{\widetilde{S}}^{i})\tilde{\mathcal{F}}_{+} = \mathcal{Z}_{\beta H_{\widetilde{S}}}(Q_{\widetilde{S}}^{f})\mathcal{J}_{\beta H_E}\tilde{\mathcal{R}}_{+}^{*}\mathcal{J}_{\beta H_E}^{-1}$. 
The application of Lemma \ref{ReversalRelation} to the above equation results in equation (\ref{ForwardRel}).

Due to Lemma \ref{LemmaSimpleRelations} we know that $\mathcal{Z}_{\beta H_{\widetilde{S}}}(Q_{\widetilde{S}}^{i-})=  \mathcal{Z}_{\beta H_{\widetilde{S}}}\big(\mathcal{T}_{\widetilde{S}}(Q_{\widetilde{S}}^{i+})\big) = \mathcal{Z}_{\beta H_{\widetilde{S}}}(Q_{\widetilde{S}}^{i+})$.
\end{proof}

Note that Proposition \ref{PropConditionalFluctuationThm} only makes a statement concerning pairs of measurement operators. It does not make any assumptions on the POVMs that these measurement operators may be members of. For example, instead of basing the induced CPMs  in (\ref{deftildeFRpm}) on the pair $(Q_{\widetilde{S}}^{i+},Q_{\widetilde{S}}^{f+})$ we could equally well obtain fluctuation relations for each of the pairs 
$(Q_{\widetilde{S}}^{i+},\hat{1}-Q_{\widetilde{S}}^{f+})$, $(\hat{1}-Q_{\widetilde{S}}^{i+},Q_{\widetilde{S}}^{f+})$, and $(\hat{1}-Q_{\widetilde{S}}^{i+},\hat{1}-Q_{\widetilde{S}}^{f+})$.
There is also no need to assume that the POVMs are binary.
For two POVMs $\{Q_{\widetilde{S},k}^{i+}\}_{k}$ and $\{Q_{\widetilde{S},l}^{f+}\}_{l}$ one can, for each possible combination of POVM elements $(Q_{\widetilde{S},k}^{i+},Q_{\widetilde{S},l}^{f+})$, construct the corresponding  CPMs $\tilde{\mathcal{F}}_{\pm}^{(k,l)}$ as in (\ref{deftildeFRpm}), where each of these pairs satisfies the conditional fluctuation relation (\ref{ForwardRel}).

\subsection{Generally no decoupling of diagonals}

In Appendix \ref{DecouplingDiagonals} we showed that the channels $\mathcal{F}_{\pm}$ induced on the reservoir are such that the  dynamics decouples along the modes of coherence. One may thus wonder whether something similar is true for the conditional CPMs $\tilde{\mathcal{F}}_{\pm}$. In Appendix \ref{DecouplingDiagonals} the starting point was Lemma \ref{ComWithCom}, which shows that $\mathcal{F}_{\pm}$ and $\mathcal{R}_{\pm}$ commute with the commutator with respect to $H_E$. The following Lemma shows that this is generally not true for the conditional CPMs $\tilde{\mathcal{F}}_{\pm}$ and $\tilde{\mathcal{R}}_{\pm}$, and thus we cannot expect to have a separation of the dynamics of the different diagonals of the density matrix. (For an explicit example of such `mixing' of diagonal and off-diagonal elements, see  Appendix \ref{ExampleTwoQubits}.)
\begin{Lemma}
\label{nkjvvkv}
With Assumptions \ref{AssumpCondTimeRev}, the CPM $\tilde{\mathcal{F}}_{+}$  defined in Eq.~(\ref{deftildeFRpm}) satisfies the following relation
\begin{equation}
\begin{split}
&  [H_E,\tilde{\mathcal{F}}_{+}(\sigma)]  =  \tilde{\mathcal{F}}_{+}([H_E,\sigma])\\
& + \Tr_{\widetilde{S}}\Big((Q_{\widetilde{S}}^{f+}\otimes \hat{1}_E)V \big(\mathcal{G}_{\beta H_{\widetilde{S}}}([H_{\widetilde{S}},Q_{\widetilde{S}}^{i+}])\otimes \sigma\big)V^{\dagger}\Big)\\
&+ \Tr_{\widetilde{S}}\Big(\big([H_{\widetilde{S}},Q_{\widetilde{S}}^{f+}]\otimes \hat{1}_E\big)V\big(\mathcal{G}_{\beta H_{\widetilde{S}}}(Q_{\widetilde{S}}^{i+})\otimes \sigma\big)V^{\dagger}\Big).
\end{split}
\end{equation}
Analogous statements hold for $\tilde{\mathcal{F}}_{-}$ and $\tilde{\mathcal{R}}_{\pm}$.
\end{Lemma}
The proof is obtained via $H = H_{\widetilde{S}}\otimes\hat{1}_E + \hat{1}_{\widetilde{S}}\otimes H_E$ and the energy conservation $[H,V] = 0$.

\subsection{\label{NeverthelessDiagonal} Nevertheless  diagonal and off-diagonal conditional fluctuation relations}

Although there is no decoupling, there do strictly speaking still exist counterparts to  (\ref{diagonalCrooksRel}) and (\ref{grkgsslkjg}). To see this, suppose that $H_E$ is non-degenerate, with a complete orthonormal eigenbasis $\{|n\rangle\}_n$. Define the general transition matrix $\tilde{q}_{\pm}(mm'|nn') := \langle m|\tilde{\mathcal{F}}_{\pm}(|n\rangle\langle n'|)|m'\rangle$ for arbitrary $m,n,m',n'$. Assuming that $\mathcal{T}_E$ is the transpose with respect to $\{|n\rangle\}_n$, we can rewrite the conditional fluctuation relation (\ref{ForwardRel}) in its `matrix form'
\begin{equation}
\label{ForwardRelMatrixForm}
\begin{split}
& \mathcal{Z}_{\beta H_{\widetilde{S}}}(Q_{\widetilde{S}}^{i})\tilde{q}_{+}(mm'|nn')\\
& = e^{\beta (E_n+E_{n'}-E_{m}-E_{m'})/2} \mathcal{Z}_{\beta H_{\widetilde{S}}}(Q_{\widetilde{S}}^{f})\tilde{q}_{-}(nn'|mm').
\end{split}
\end{equation}
Nothing  prevents us from defining $\tilde{p}_{\pm}(m|n) :=  \tilde{q}_{\pm}(mm|nn)= \langle m|\tilde{\mathcal{F}}_{\pm}(|n\rangle\langle n|)|m\rangle$, and 
$\tilde{q}_{\pm}^{\delta}(m|n) := \tilde{q}_{\pm}(mm'|nn')$ for $E_{n} -E_{n'} = E_{m}-E_{m'}=\delta$, and as special cases of (\ref{ForwardRelMatrixForm}) write  
 \begin{equation}
\label{diagonalCrooksReltilde}
 \mathcal{Z}_{\beta H_{\widetilde{S}}}(Q_{\widetilde{S}}^{i})\tilde{p}_{+}(m|n) = e^{\beta (E_n-E_m)}\mathcal{Z}_{\beta H_{\widetilde{S}}}(Q_{\widetilde{S}}^{f}) \tilde{p}_{-}(n|m),
\end{equation}
\begin{equation}
\label{offdiagonalCrooksReltilde}
 \mathcal{Z}_{\beta H_{\widetilde{S}}}\!(Q_{\widetilde{S}}^{i})\tilde{q}_{+}^{\delta}(m|n)  =  e^{\beta (E_n-E_{m})}\mathcal{Z}_{\beta H_{\widetilde{S}}}\!(Q_{\widetilde{S}}^{f})\tilde{q}_{-}^{\delta}(n|m).
\end{equation}
We can view (\ref{diagonalCrooksReltilde}) and (\ref{offdiagonalCrooksReltilde}) as the conditional counterparts to (\ref{diagonalCrooksRel}) and (\ref{grkgsslkjg}). However, while  (\ref{diagonalCrooksRel}) and (\ref{grkgsslkjg}) describe the dynamics within the decoupled diagonals, we cannot interpret  (\ref{diagonalCrooksReltilde}) and (\ref{offdiagonalCrooksReltilde}) in the same manner, due to the lack of decoupling.
More generally, the objects  $\tilde{p}_{\pm}(m|n)$ and $\tilde{q}_{+}^{\delta}(m|n)$ must be interpreted with more caution than their counterparts $p_{\pm}(m|n)$ and $q_{+}^{\delta}(m|n)$. 
For example, the $\tilde{p}_{\pm}(m|n)$ is the probability to detect the energy reservoir in eigenstate $m$, given that the reservoir was prepared in state $|n\rangle\langle n|$, and that the control measurement is successful.
We can interpret  $p_{\pm}(m|n)$ in a similar manner (without the control measurement). However,  due to the decoupling we can also interpret $p_{\pm}(m|n)$ as describing the evolution of the diagonal elements of the density matrix, \emph{irrespective} of the initial state.   More precisely, we know that the probability  $\langle m|\mathcal{F}_{\pm}(\sigma)|m\rangle$ to detect a specific final energy eigenstate $m$ only depends on the diagonal elements $\langle n|\sigma|n\rangle$. However,  $\langle m|\tilde{\mathcal{F}}_{\pm}(\sigma)|m\rangle$  not only depends on the diagonal elements of the input, but on the entire state $\sigma$. 
In other words, we can no longer claim that an initial energy measurement would not perturb a final energy measurement.

\subsection{\label{SecDiagonalMeasuremnts} The special case  $[H_{\widetilde{S}},Q_{\widetilde{S}}^{i\pm}] =0$ and $[H_{\widetilde{S}},Q_{\widetilde{S}}^{f\pm}] = 0$}

In the case that the measurement operators $Q_{\widetilde{S}}^{i\pm}$ and $Q_{\widetilde{S}}^{f\pm}$  commute with $H_{\widetilde{S}}$ we can regain several of the properties of the unconditional fluctuation relations. By a direct application of Lemma \ref{nkjvvkv} we get $[H_E,\tilde{\mathcal{F}}_{\pm}(\sigma)] =  \tilde{\mathcal{F}}_{\pm}([H_E,\sigma])$.
Analogously to how we obtained Corollary \ref{MappingOffDiagonal} from Lemma \ref{ComWithCom}, and with analogous assumptions, we also regain the decoupling of the different modes of coherence, i.e.,  
$\langle m'|\tilde{\mathcal{F}}_{\pm}(|n'\rangle\langle n|)|m\rangle = 0$, if $E_{m'}-E_{n'}  \neq E_{m}-E_{n}$. In this case (\ref{diagonalCrooksReltilde}) and (\ref{offdiagonalCrooksReltilde})  thus gain a status analogous to   (\ref{diagonalCrooksRel}) and (\ref{grkgsslkjg}), in the sense that they describe the dynamics within each decoupled diagonal.

\subsubsection{\label{SecClassicalConditional}A classical conditional Crooks relation}

By additionally assuming the energy translation invariant model as in Appendix \ref{SecEnergyIndependence} one regains the energy translation invariance of the induced CPMs  $\Delta^{j}\tilde{\mathcal{F}}_{\pm}(\sigma){\Delta^{\dagger}}^{k} =\tilde{\mathcal{F}}_{\pm}\big( \Delta^{j}\sigma{\Delta^{\dagger}}^{k}\big)$.
Analogous to Appendix \ref{SecRegainingClassicalCrooks} we can also define
\begin{equation}
\label{tildebvdfsmbn}
\tilde{P}_{\pm}(w) := \sum_{j,j': j-j' = w/s}\tilde{p}_{\pm}(j'|j)\langle j|\sigma|j\rangle,
\end{equation}  
where $\tilde{P}_{\pm}(w)$ can be interpreted as the probability that the energy reservoir looses the energy $w$ and that the control measurement is successful. The decoupling again guarantees the stability of the detection probabilities under repeated energy measurements.  We obtain the classical conditional Crooks relation
\begin{equation}
\label{ClassicalConditional1}
\mathcal{Z}_{\beta H_{\widetilde{S}}}(Q_{\widetilde{S}}^{i})\tilde{P}_{+}(w)  =   e^{\beta w}\mathcal{Z}_{\beta H_{\widetilde{S}}}(Q_{\widetilde{S}}^{f})\tilde{P}_{-}(-w)
\end{equation}
in a manner very similar to what we did in Appendix \ref{SecRegainingClassicalCrooks}.

\subsubsection{\label{ClassCondJarzynskiRel}A classical conditional Jarzynski relation}

The control measurements  generally do not succeed with unit probability. However, for the energy translation symmetric case, with diagonal measurement operators, the success probabilities $\Tr\mathcal{F}_{\pm}(\sigma)$ and $\Tr\mathcal{R}_{\pm}(\sigma)$ become independent of the state $\sigma$ of the energy reservoir. To see this, we first define the following transition probabilities that do not involve the energy reservoir.
\begin{equation}
\begin{split}
f_{+} :=  & \Tr\big(Q_{\widetilde{S}}^{f+}U\mathcal{G}_{\beta H_{\widetilde{S}}}(Q_{\widetilde{S}}^{i+})U^{\dagger}\big),\\
r_{+} : = & \Tr\big(Q_{\widetilde{S}}^{i+}U^{\dagger}\mathcal{G}_{\beta H_{\widetilde{S}}}(Q_{\widetilde{S}}^{f+})U\big), \\
f_{-} := & \Tr\big(Q_{\widetilde{S}}^{i-}U\mathcal{G}_{\beta H_{\widetilde{S}}}(Q_{\widetilde{S}}^{f-})U^{\dagger}\big),\\
 r_{-}: = &  \Tr\big(Q_{\widetilde{S}}^{f-}U^{\dagger}\mathcal{G}_{\beta H_{\widetilde{S}}}(Q_{\widetilde{S}}^{i-})U\big).
\end{split}
\end{equation}
In words $f_{+}$ is the probability that we would obtain the `successful' outcome when we measure the POVM $\{Q_{\widetilde{S}}^{f+},\hat{1}_{\widetilde{S}}-Q_{\widetilde{S}}^{f+}\}$ if the initial state $\mathcal{G}_{\beta H_{\widetilde{S}}}(Q_{\widetilde{S}}^{i+})$ is evolved under $U$.

One should keep in mind that since $Q_{\widetilde{S}}^{f\pm}$ and $Q_{\widetilde{S}}^{i\pm}$ commute with $H_{\widetilde{S}}$, it follows that the above expressions do not involve coherences with respect to the energy eigenbases. For example, 
if $H_{\widetilde{S}}$ is non-degenerate with eigenstates $|\psi_n\rangle$ then $f_{+} 
=   \sum_{nn'} \langle\psi_n|Q_{\widetilde{S}}^{f+}|\psi_{n}\rangle|\langle \psi_{n}| U |\psi_{n'}\rangle|^2 \langle \psi_{n'}| \mathcal{G}_{\beta H_{\widetilde{S}}}(Q_{\widetilde{S}}^{i+})|\psi_{n'}\rangle
$.

One can confirm that the CPMs  $\tilde{\mathcal{F}}_{\pm}$ and $\tilde{\mathcal{R}}_{\pm}$ in (\ref{deftildeFRpm}) with $V := V(U)$ as in (\ref{DefVU}), and $[Q_{\widetilde{S}}^{i\pm},H_{\widetilde{S}}] = 0$,  $[Q_{\widetilde{S}}^{f\pm},H_{\widetilde{S}}] = 0$,  satisfy the following relations
\begin{equation}
\label{nkjllv}
\Tr\tilde{\mathcal{F}}_{\pm}(\sigma) = f_{\pm}\Tr(\sigma),\quad \Tr\tilde{\mathcal{R}}_{\pm}(\sigma) = r_{\pm}\Tr(\sigma),
\end{equation}
\begin{equation}
\label{ndfknfd}
\tilde{\mathcal{F}}_{\pm}(\hat{1}_E) =   f_{\pm}\hat{1}_E,\quad\tilde{\mathcal{R}}_{\pm}(\hat{1}_E) =   r_{\pm}\hat{1}_E. 
\end{equation}
Hence, the success probabilities of the control measurements are independent of the state of the energy reservoir.
It is instructive to write the expression for $\mathcal{F}_{+}$ in (\ref{nkjllv}) in full
\begin{equation*}
\begin{split}
&  \Tr\big([Q_{\widetilde{S}}^{f+}\otimes \hat{1}_E]V(U)[\mathcal{G}_{\beta H_{\widetilde{S}}}(Q_{\widetilde{S}}^{i+})\otimes \sigma]V^{\dagger}(U)\big)\\
& =  \Tr\big(Q_{\widetilde{S}}^{f+}U\mathcal{G}_{\beta H_{\widetilde{S}}}(Q_{\widetilde{S}}^{i+})U^{\dagger}\big),
\end{split}
\end{equation*}
for $\Tr\sigma =1$.

Hence, in terms of the success probability, the experiment involving the energy reservoir behaves as if it was a simpler experiment not including the reservoir, where the unitary $V(U)$ is replaced by $U$ (but one should keep in mind that this relies on the assumptions that $Q_{\widetilde{S}}^{f+}$ and $Q_{\widetilde{S}}^{i+}$ commute with $H_{\widetilde{S}}$, and thus block-diagonalize with respect to the energy eigenspaces of $H_{\widetilde{S}}$).

By the classical conditional Crooks relation (\ref{ClassicalConditional1}) it follows  that $\mathcal{Z}_{\beta H_{\widetilde{S}}}(Q_{\widetilde{S}}^{i})\sum_{w} e^{-\beta w}\tilde{P}_{+}(w) = \mathcal{Z}_{\beta H_{\widetilde{S}}}(Q_{\widetilde{S}}^{f})\sum_{w}\tilde{P}_{-}(-w)$. 
As opposed to the unconditional case, $\sum_{w}\tilde{P}_{-}(-w)$ is generally not equal to $1$, but equates to the success probability of the reverse process, i.e., $\sum_{w}\tilde{P}_{-}(-w) =  \sum_{w}\tilde{p}_{-}(w/s|0)  = \Tr\tilde{\mathcal{F}}_{-}(|0\rangle\langle 0|) = f_{-}$. Thus
\begin{equation}
\label{nlkdafrehe}
\sum_{w} e^{-\beta w}\tilde{P}_{+}(w)
 = f_{-}\frac{\mathcal{Z}_{\beta H_{\widetilde{S}}}(Q_{\widetilde{S}}^{f})}{\mathcal{Z}_{\beta H_{\widetilde{S}}}(Q_{\widetilde{S}}^{i})}.
\end{equation}
Given that the control measurement succeeds, we can define the conditional probability for the work cost $w$ as $P_{+}(w|Q_{\widetilde{S}}^{f+}) := \tilde{P}_{+}(w)/\sum_{w'}\tilde{P}_{+}(w') = \tilde{P}_{+}(w)/f_{+}$. By using this we can rewrite (\ref{nlkdafrehe}) in the more symmetric form
\begin{equation}
\label{gijtwoi}
\big\langle e^{-\beta W}\big| Q_{\widetilde{S}}^{f+}\big\rangle
 = \frac{f_{-}}{f_{+}}\frac{\mathcal{Z}_{\beta H_{\widetilde{S}}}(Q_{\widetilde{S}}^{f})}{\mathcal{Z}_{\beta H_{\widetilde{S}}}(Q_{\widetilde{S}}^{i})},
\end{equation}
where we use the notation $\langle e^{-\beta W}|Q_{\widetilde{S}}^{f+}\rangle := \sum_{w} e^{-\beta w}P_{+}(w|Q_{\widetilde{S}}^{f+})$.

The conditional Jarzynski relation in (\ref{gijtwoi}) does remind of the Jarzynski equality under feedback control \cite{Sagawa10,Toyabe10,Morikuni11}, and it may be worthwhile to investigate this potential link further. However, we will not do so in this investigation.

\subsection{Examples}

\subsubsection{\label{AbolishPerfectControl} Abolishing perfect control}

As discussed in Appendix \ref{SecConditionsOnE} the perfect control (see equations (\ref{PerfectControl}) or (\ref{PerfectControlTimeRev}) and (\ref{PerfectControlTimeRev2})) can lead to an energy reservoir that has an unbounded spectrum from both above and below. From a physical point of view a spectrum that is unbounded from below is somewhat uncomfortable. With the conditional fluctuation relations we no longer need to assume perfectly functioning control systems. We could for example let the transition from initial to final control state fail if the energy reservoir runs out of energy, so to speak. Here we demonstrate this in the case of a harmonic oscillator as the energy reservoir. 

We make use of a model that was introduced in quite a detail in section IV in the Supplemental Material of \cite{Aberg13}. Due to this, only the briefest description will be provided here. The main point is that one can construct a family of unitary operators $V_{+}(U)$ that act identically to $V(U)$ as long as the energy in the energy reservoir is high enough. (We do not employ the `injection' of energy that was used in \cite{Aberg13}, but instead allow the procedure to fail.)

Apart from the Hamiltonian for the energy reservoir, $H_E^{+}:= s\sum_{j=0}^{+\infty}j|j\rangle\langle j|$ and the class of unitary operators $V_{+}(U)$, the rest of the model is as in Appendix \ref{minimalwith}.
We let
$H_{S'C} := H^{i}_{S'}\otimes P^{i}_C + H^{f}_{S'}\otimes P^{f}_C$, 
and let  $sz^{i}_n$, $|\chi^i_n\rangle$ be the eigenvalues and eigenvectors of $H^{i}_{S'}$, and $sz^{f}_n$, $|\chi^f_n\rangle$ the  eigenvalues and eigenvectors of $H^{f}_{S'}$.
To simplify the notation we let $\{|\psi_n\rangle\}_{n=1}^{4N}$ denote the orthonormal set 
$\{|\chi^i_n\rangle|c_{i+}\rangle, |\chi^i_n\rangle|c_{i-}\rangle, |\chi^f_n\rangle|c_{f+}\rangle, |\chi^f_n\rangle|c_{f-}\rangle\}_{n=1}^{N}$. Similarly, let $z_n$ denote the combined set of numbers $\{z^{i}_n\}_n\cup\{z^{f}_n\}_n$.  Furthermore define $z_{\textrm{max}} = \max_n z_n$ and $z_{\textrm{min}} = \min_n z_n$.
The projectors   $\{P^{(l)}_{+}\}_{l\geq z_{\textrm{min}}}$ onto the eigenspaces of $H_{S'C}\otimes \hat{1}_E + \hat{1}_{S'C}\otimes H_E$ are
\begin{equation*}
P^{(l)}_{+} =  \sum_{n:l\geq z_n}|\psi_n\rangle\langle \psi_n|\otimes |l-z_n\rangle\langle l-z_n|,\quad \forall l\geq z_{\textrm{min}}.
\end{equation*}
For all $l\geq z_{\textrm{max}}$ this simplifies to 
$P^{(l)}_{+} =  \sum_{n=1\ldots N}|\psi_n\rangle\langle \psi_n|\otimes |l-z_n\rangle\langle l-z_n|$. 
Let us define the following map
\begin{equation*}
\begin{split}
V_{+}(U) := & \sum_{l\geq z_{\textrm{max}}}V_{l}(U)+ \sum_{l=z_{\textrm{min}}}^{z_{\textrm{max}}-1}P^{(l)}_{+},\\
V_{l}(U) := & \sum_{n,n'}|\psi_n\rangle\langle\psi_n|U|\psi_{n'}\rangle\langle\psi_{n'}|\otimes |l\!-\!z_n\rangle\langle l\!-\!z_{n'}|.
\end{split}
\end{equation*}
For each unitary operator $U$ on $\mathcal{H}_{S'C}$ it follows that $V_{+}(U)$ is unitary on $\mathcal{H}_{S'CE}$. 
We let  $\mathcal{T}_E$ be the transpose with respect to $\{|j\rangle\}_{j\geq 0}$, and let $\mathcal{T}_{\widetilde{S}}$ be  such that $\mathcal{T}_{S'C}(H_{S'C}) = H_{S'C}$. (Note Lemma \ref{TOnHermitean}.) From this it follows that $\mathcal{T} :=\mathcal{T}_{S'C}\otimes\mathcal{T}_{E}$ is such that  $\mathcal{T}\big(V_{+}(U)\big) = V_{+}\big(\mathcal{T}_{S'C}(U)\big)$. 

One can also confirm that $V_{+}(U)|\psi\rangle|j\rangle = V(U)|\psi\rangle|j\rangle$ for all $|\psi\rangle\in\mathcal{H}_{S'C}$ and all $j\geq z_{\textrm{max}}-z_{\textrm{min}}$, where $V(U)$ is defined in (\ref{DefVU}).

More generally, in terms of the eigenprojectors $P^{(l)}_{+}$ of the total Hamiltonian $H$, it is the case that $V_{+}(U)P^{(l)}_{+} = V(U)P^{(l)}_{+}$ for $l\geq z_{\textrm{max}}$, while $V_{+}(U)P^{(l)}_{+} = P^{(l)}_{+}$ for $z_{\textrm{max}}-1\geq l\geq  z_{\textrm{min}}$.
One can say that $V_{+}$ `censors' our choice of $U$ in the sense that if $U$ entails an energy change that we cannot afford, then $V_{+}(U)$ avoids to perform the too expensive parts of the operation.

Define the projector $P^{\geq z_{\textrm{max}}-z_{\textrm{min}}}:= \sum_{j\geq z_{\textrm{max}}-z_{\textrm{min}}}|j\rangle\langle j|$. The condition $P^{\geq z_{\textrm{max}}-z_{\textrm{min}}}\sigma P^{\geq z_{\textrm{max}}-z_{\textrm{min}}} = \sigma$ guarantees that the actions of $V_+(U)$ and $V(U)$ are identical. 

As an example, let us modify the setup in Appendix \ref{minimalwith} such that we replace $V(U)$ with $V_{+}(U)$. For the sake of illustration we consider an extreme case where $H^{i}_{S'} := s|\chi^{i}_1\rangle\langle\chi^{i}_{1}| + 2s|\chi^{i}_2\rangle\langle\chi^{i}_{2}|$ and $H^{f}_{S'} := 3s|\chi^{f}_1\rangle\langle\chi^{f}_{1}| + 4s|\chi^{f}_2\rangle\langle\chi^{f}_{2}|$. For these choices, all transitions from the initial to the final Hamiltonian require energy. Moreover, $z_{\textrm{max}} = 4$ and $z_{\textrm{min}} = 1$. One furthermore finds that 
\begin{equation*}
\begin{split}
P^{(1)}_{+} = & |\chi^{i}_1\rangle\langle \chi^{i}_1|\otimes P^{i}_C\otimes |0\rangle\langle 0|,\\  
P^{(2)}_{+} = & |\chi^{i}_1\rangle\langle \chi^{i}_1|\otimes P^{i}_C\otimes |1\rangle\langle 1| + |\chi^{i}_2\rangle\langle \chi^{i}_2|\otimes P^{i}_C\otimes |0\rangle\langle 0|,\\
P^{(3)}_{+} = & |\chi^{i}_1\rangle\langle \chi^{i}_1|\otimes P_C^{i}\otimes |2\rangle\langle 2| + |\chi^{i}_2\rangle\langle \chi^{i}_2|\otimes P_C^{i}\otimes |1\rangle\langle 1|\\
& +  |\chi^{f}_1\rangle\langle \chi^{f}_1|\otimes P_C^{f}\otimes |0\rangle\langle 0|,\\
P^{(l)}_{+} = & |\chi^{i}_1\rangle\langle \chi^{i}_1|\otimes P_C^{i}\otimes |l-1\rangle\langle l-1|\\
&  + |\chi^{i}_2\rangle\langle \chi^{i}_2|\otimes P_C^{i}\otimes |l-2\rangle\langle l-2|\\
& +  |\chi^{f}_1\rangle\langle \chi^{f}_1|\otimes P_C^{f}\otimes |l-3\rangle\langle l-3|\\
& +  |\chi^{f}_2\rangle\langle \chi^{f}_2|\otimes P_C^{f}\otimes |l-4\rangle\langle l-4|,\quad l\geq 4.
\end{split}
\end{equation*}
Hence, if the control is in state $|c_{i+}\rangle$ and the reservoir is in the vacuum state $|0\rangle$, we find that $V_{+}(U)[\rho\otimes|c_{i+}\rangle\langle c_{i+}|\otimes|0\rangle\langle 0|]V_{+}(U)^{\dagger} = \rho\otimes|c_{i+}\rangle\langle c_{i+}|\otimes|0\rangle\langle 0|$ no matter what $U$ we feed into it, and  thus the control measurement will always signal `fail'. On the other hand, if there are four or more quanta of energy in the reservoir, then $V_{+}(U)[\rho\otimes|c_{i+}\rangle\langle c_{i+}|\otimes \sigma]V_{+}(U)^{\dagger} = V(U)[\rho\otimes|c_{i+}\rangle\langle c_{i+}|\otimes \sigma]V(U)^{\dagger}$. In particular, if  we would could choose $U$ as  in (\ref{reoioireoire}), then the control measurement would always succeed.

\subsubsection{\label{ExampleTwoQubits} $\widetilde{S}$ and $E$ as single qubits}

The conditional fluctuation relations allow us to treat energy reservoirs with finite-dimensional Hilbert spaces. Here we consider the extreme case where both $\widetilde{S}$ and the energy reservoir $E$ are single spin-half particles. 

We assume that the spins are associated with magnetic moments, and that they are affected by a constant external magnetic field, such that they are in resonance, i.e., the splitting of the eigenenergies are identical. More precisely, 
\begin{equation*}
H_{\widetilde{S}} = H_E =  -\frac{1}{2}s|0\rangle\langle 0| + \frac{1}{2}s|1\rangle\langle 1|,
\end{equation*}
for some $s>0$. The identical energy gap implies that there exist non-trivial energy conserving unitary operations with respect to $H:= H_{\widetilde{S}}\otimes\hat{1}_E + \hat{1}_{\widetilde{S}}\otimes H_E$. More specifically, $H$ has the eigenenergies $-s,0,s$ and corresponding energy eigenspaces $\Sp\{|0,0\rangle\}$, $\Sp\{|0,1\rangle,|1,0\rangle\}$, and $\Sp\{|1,1\rangle\}$. An energy conserving unitary operator thus has to be block diagonal with respect to these energy eigenspaces, 
\begin{equation*}
\begin{split}
V 
= &  e^{i\chi_{-}}|0\rangle\langle 0|\otimes |0\rangle\langle 0| + e^{i\chi_{+}}|1\rangle\langle 1|\otimes |1\rangle\langle 1|\\
& + U_{1,1}|0\rangle\langle 0|\otimes|1\rangle\langle 1| + U_{1,2}|0\rangle\langle 1|\otimes|1\rangle\langle 0|\\
& + U_{2,1}|1\rangle\langle 0|\otimes|0\rangle\langle 1|  + U_{2,2}|1\rangle\langle 1|\otimes |0\rangle\langle 0|,
\end{split}
\end{equation*}
where $\chi_{+},\chi_{-}\in\mathbb{R}$, and where $U = [U_{j,k}]_{j,k = 1,2}$ is a unitary $2\times 2$ matrix.

Let us choose $\mathcal{T}_{\widetilde{S}}$ and $\mathcal{T}_{E}$ as the transpose with respect to the eigenbasis $\{|0\rangle, |1\rangle\}$ of each space, respectively. Thus $\mathcal{T}_{\widetilde{S}}(H_{\widetilde{S}}) = H_{\widetilde{S}}$ and $\mathcal{T}_{E}(H_{E}) = H_{E}$. Let $\mathcal{T} := \mathcal{T}_{\widetilde{S}}\otimes\mathcal{T}_E$. One can confirm that  $\mathcal{T}(V) = V$ if and only if 
\begin{equation*}
U = e^{i\chi}\left[\begin{matrix}
-e^{-i\delta}\cos\theta & \sin\theta\\
\sin\theta & e^{i\delta}\cos\theta
\end{matrix}\right],
\end{equation*}
where $\chi,\delta,\theta\in\mathbb{R}$. (There are no restrictions on $\chi_{+},\chi_{-}$.)

The expansion of  the CPMs $\tilde{\mathcal{F}}_{\pm}$ for arbitrary $Q_{\widetilde{S}}^{i\pm}$ and $Q_{\widetilde{S}}^{f\pm}$ in terms of the $\{|0\rangle,|1\rangle\}$ basis  results in  remarkably bulky and unilluminating expressions. Therefore we shall here only consider the simpler special case where 
$Q_{\widetilde{S}}^{i+} :=  \hat{1}_{\widetilde{S}}$, $Q_{\widetilde{S}}^{f+} :=  |\psi\rangle\langle\psi|$, $|\psi\rangle := \frac{1}{\sqrt{2}}|0\rangle + i\frac{1}{\sqrt{2}}|1\rangle$.
Consequently $Q_{\widetilde{S}}^{i-} =  \mathcal{T}_{\widetilde{S}}(Q_{\widetilde{S}}^{i+}) =  \hat{1}_{\widetilde{S}}$, 
$Q_{\widetilde{S}}^{f-} =   \mathcal{T}_{\widetilde{S}}(Q_{\widetilde{S}}^{f+}) = |\psi^{*}\rangle\langle\psi^{*}|$, $|\psi^{*}\rangle =   \frac{1}{\sqrt{2}}|0\rangle - i\frac{1}{\sqrt{2}}|1\rangle$. 
The partition maps take the values
$\mathcal{Z}_{\beta H_{\widetilde{S}}}(Q_{\widetilde{S}}^{i}) =   e^{\beta E/2} + e^{-\beta E/2}$ and 
$\mathcal{Z}_{\beta H_{\widetilde{S}}}(Q_{\widetilde{S}}^{f}) =  (e^{\beta E/2} + e^{-\beta E/2})/2$.

The Gibbs map applied to $Q_{\widetilde{S}}^{i+} = \hat{1}_{\widetilde{S}}$ gives the initial state of the forward process
\begin{equation*}
\begin{split}
 \mathcal{G}_{\beta H_{\widetilde{S}}}(Q_{\widetilde{S}}^{i+}) 
 = &  \frac{e^{\beta s/2}|0\rangle\langle 0| + e^{-\beta s/2}|1\rangle\langle 1|}{e^{\beta s/2} + e^{-\beta s/2}},
\end{split}
\end{equation*}
which is the Gibbs state of $H_{\widetilde{S}}$.
The initial state of the reversed process is
\begin{equation*}
\begin{split}
\mathcal{G}_{\beta H_{\widetilde{S}}}(Q_{\widetilde{S}}^{f-}) 
= &\frac{1}{e^{\beta s/2} + e^{-\beta s/2}}\Big( e^{\beta s/2}|0\rangle\langle 0| \\
& + e^{-\beta s/2}|1\rangle\langle 1| +i|0\rangle\langle 1| - i|1\rangle\langle 0|\Big).
\end{split}
\end{equation*}

For the unitary operator $V$ we assume that $\chi = \chi_{\pm} = 0$, $\delta = 0$ (while we let $\theta$ be arbitrary). This results in 
\begin{equation*}
\begin{split}
\tilde{\mathcal{F}}_{+}(\sigma) = &  \frac{1}{e^{\beta s/2} + e^{-\beta s/2}}(V_{0+}\sigma V_{0+}^{\dagger} + V_{1+}\sigma V_{1+}^{\dagger}),\\
V_{0+}  := & \frac{e^{\beta s/4}}{\sqrt{2}} (|0\rangle\langle 0|  - \cos\theta |1\rangle\langle 1|  - i\sin\theta |0\rangle\langle 1|),\\
V_{1+}  := & \frac{ e^{-\beta s/4}}{\sqrt{2}}(|1\rangle\langle 1|  +\cos\theta |0\rangle\langle 0|+i\sin\theta |1\rangle\langle 0|),\\
\end{split}
\end{equation*}
and
\begin{equation*}
\begin{split}
\tilde{\mathcal{F}}_{-}(\sigma) = & \frac{1}{e^{\beta s/2} + e^{-\beta s/2}}(V_{0-}\sigma V_{0-}^{\dagger} + V_{1-}\sigma V_{1-}^{\dagger}),\\
V_{0-} := &  e^{\beta s/4}(|0\rangle\langle 0|  - \cos\theta|1\rangle\langle 1|)\\
& - i\sin\theta  e^{-\beta s/4}|1\rangle\langle 0|,\\
V_{1-} := &  e^{-\beta s/4}(|1\rangle\langle 1|  +\cos\theta|0\rangle\langle 0|)\\
&  +i\sin\theta e^{\beta s/4} |0\rangle\langle 1|.
\end{split}
\end{equation*}
By a slightly tedious but straightforward calculation one can confirm that  $\tilde{\mathcal{F}}_{\pm}$  satisfy the conditional fluctuation relation (\ref{ForwardRel}), as we already know that they should, due to Proposition \ref{PropConditionalFluctuationThm}. 

One can also confirm that $\tilde{\mathcal{F}}_{\pm}$ provide examples for the fact that the conditional maps in general do not decouple the evolution of diagonal and off-diagonal elements. For example
\begin{equation*}
\begin{split}
 \tilde{\mathcal{F}}_{+}(|0\rangle\langle 0|)
= &  \frac{1}{2}\frac{1}{e^{\beta s/2} + e^{-\beta s/2}}\bigg[ \sin^2\theta  e^{-\beta s/2} |1\rangle\langle 1|\\
& + (\cos^2\theta e^{-\beta s/2} +  e^{\beta s/2} ) |0\rangle\langle 0|\\
&  +i \sin\theta \cos\theta e^{-\beta s/2}( |1\rangle\langle 0|-|0\rangle\langle 1|)\bigg],
\end{split}
\end{equation*}
\begin{equation*}
\begin{split}
\tilde{\mathcal{F}}_{-}(|0\rangle\langle 1|)
= &  \frac{\cos\theta }{e^{\beta s/2} + e^{-\beta s/2}}\bigg[(e^{-\beta s/2}-e^{\beta s/2})|0\rangle\langle 1|\\
& + i\sin\theta( |1\rangle\langle 1|-|0\rangle\langle 0|)\bigg].
\end{split}
\end{equation*}
Hence, diagonal and off-diagonal elements get mixed.
In particular, a diagonal state such as $|0\rangle\langle 0|$ can be turned into a state $ \tilde{\mathcal{F}}_{+}(|0\rangle\langle 0|)$ with off-diagonal elements, which is a consequence of the non-diagonal measurement operator. 
Note that not only is $E$ initially in a diagonal state, but the global state $G(H_E)\otimes |0\rangle\langle 0|$ is also diagonal with respect to the global energy eigenspaces.

\section{\label{SecAltformTransProb} Alternative formulation}

Up to now we have focused on the dynamics of the energy reservoir, and formulated all our results in terms of channels or CPMs induced on this system. Here we shall take a step back and briefly re-examine the structure of these fluctuation theorems from a global point of view.

\subsection{\label{GlobalSymmetry}Global invariance}

Let us for a moment forget the division into systems, heat baths, and energy reservoirs, and consider one single system with a global Hamiltonian $H$, and a unitary evolution $V$ that is energy conserving $[H,V]=0$, and where this system satisfies a time-reversal symmetry $\mathcal{T}(H) = H$, $\mathcal{T}(V) = V$. For any pair of global measurement operators $Q^{i+}$ and $Q^{f+}$, it is the case that
\begin{equation}
\label{GlobalInvariance}
\begin{split}
& \Tr\big(Q^{f+}V\mathcal{J}_{\beta H}(Q^{i+})V^{\dagger}\big) \\
= &  \Tr\big[\mathcal{T}\big(Q^{f+}V e^{-\beta H/2} Q^{i+}e^{-\beta H/2} V^{\dagger}\big)\big]\\
= &  \Tr\big[Q^{i-} e^{-\beta H/2} V Q^{f-}V^{\dagger} e^{-\beta H/2} \big]\\
= & \Tr\big(Q^{i-}V\mathcal{J}_{\beta H}(Q^{f-})V^{\dagger}\big),
\end{split}
\end{equation}
where as usual $Q^{i-} :=\mathcal{T}(Q^{i+})$ and $Q^{f-}: =\mathcal{T}(Q^{f+})$.
In other words, (\ref{GlobalInvariance}) expresses an invariance of the quantity $\Tr\big(Q^{f}V\mathcal{J}_{\beta H}(Q^{i})V^{\dagger}\big)$ with respect to the transformation $(Q^{i},Q^{f})\mapsto (Q^{i'},Q^{f'}) := \big(\mathcal{T}(Q^{f}),\mathcal{T}(Q^{i})\big)$. All our fluctuation relations can in some sense be regarded as special cases of this global invariance, which here emerges from the combination of time-reversal symmetry and energy conservation. 

Time-reversal symmetry alone is not enough to derive this invariance; energy conservation also is needed. However, if it would be the case that $V = e^{-itH/\hbar}$, then it follows that $[H,V] = 0$, and the assumption $\mathcal{T}(H) = H$ would automatically yield $\mathcal{T}(V) = V$. Hence, in this case time-reversal symmetry would be enough. 

The relation (\ref{GlobalInvariance}) can be rewritten  as the global fluctuation relation (\ref{MainGlobal}) in the main text.

\subsection{\label{SecInaccessible} Factorization of non-interacting degrees of freedom}

The global symmetry in (\ref{GlobalInvariance}) does not explain how 
we can express fluctuation theorems in terms of channels or CPMs on the relevant systems. Throughout this  investigation we have repeatedly used the fact that the exponential function factorizes over non-interacting degrees of freedom. To be more precise, suppose that the global system be decomposed into two subsystems $1$ and $2$ (i.e., $\mathcal{H} = \mathcal{H}_1\otimes\mathcal{H}_2$) and assume a non-interacting Hamiltonian $H = H_{1}\otimes\hat{1}_{2} + \hat{1}_{1}\otimes H_{2}$, with the consequence that 
$\mathcal{J}_{\beta H} = \mathcal{J}_{\beta H_1}\otimes \mathcal{J}_{\beta H_2}$.
For product measurement operators $Q = Q_{1}\otimes Q_{2}$  this results in the factorization of the partition map
$\mathcal{Z}_{\beta H}(Q_1\otimes Q_{2}) = \mathcal{Z}_{\beta H_1}(Q_1) \mathcal{Z}_{\beta H_2}(Q_{2})$, 
and thus also for the Gibbs map $\mathcal{G}_{\beta H}(Q_1\otimes Q_{2}) = \mathcal{G}_{\beta H_1}(Q_1)\otimes \mathcal{G}_{\beta H_2}(Q_{2})$.
In other words, the two states $\mathcal{G}_{\beta H_1}(Q_1)$ and $\mathcal{G}_{\beta H_2}(Q_{2})$ can be prepared separately on respective system. 

One way in which the factorization could fail would be if the global Hamiltonian is interacting, and 
this is the topic of Appendix \ref{SecApproximateFluct}. In Appendix \ref{SecGeneralizedGibbs} we briefly discuss the more exotic alternative that thermal states would not be characterized by Gibbs states.

\subsection{\label{GibbsInaccess}Inaccessible degrees of freedom}

In the context of statistical mechanics we would typically deal with large numbers of degrees of freedom that we have no access to, e.g., a heat bath. This means that we neither have direct access to prepare arbitrary states on these degrees of freedom, nor to make arbitrary measurements on them. (One can of course imagine some form of partial accessibility, but to keep things simple we here assume  `all or nothing'.)
Thus imagine that the total system is divided into two subsystem $1$ and $2$, where $2$ is inaccessible to us, and $1$ is completely accessible. We also assume that the global Hamiltonian is non-interacting $H = H_1\otimes \hat{1}_2+ \hat{1}_1\otimes H_2$.

Since system $2$ is inaccessible to us, all available measurement operators are of the form $Q = Q_{1}\otimes \hat{1}_2$, i.e., we can only perform the trivial measurement on system $2$. Correspondingly, a trivial preparation would be an equilibrium state $G_{\beta}(H_2)$ on the unaccessible degrees of freedom, with the philosophy is that nature provides equilibrium states `for free'.  Hence, all possible initial states that we would be able to prepare would be of the form $\rho = \rho_1\otimes G_{\beta}(H_2)$.

To further highlight how the factorization property enters in the treatment of the unaccessible degrees of freedom, let us take a closer look at the global fluctuation relation in (\ref{MainGlobal}). Suppose that the forward process would be characterized by a measurement operator of the allowed form $Q^{f+} = Q^{f+}_1\otimes \hat{1}_2$. Then we know that the initial state of the reverse process would be given by $\mathcal{G}_{\beta H}(Q^{f-}_1\otimes \hat{1}_2)$, where we have assumed $\mathcal{T} = \mathcal{T}_1\otimes\mathcal{T}_2$. Due to the factorization property we know that  $\mathcal{G}_{\beta H}(Q^{f-}_1\otimes \hat{1}_2) = \mathcal{G}_{\beta H_1}(Q^{f-}_1)\otimes \mathcal{G}_{\beta H_2}(\hat{1}_2) = \mathcal{G}_{\beta H_1}(Q^{f-}_1)\otimes G_{\beta}(H_2)$. In other words, trivial measurements on the unaccessible degrees of freedom get mapped to trivial preparations on these systems, and vice versa. However, this would no longer be true if $\mathcal{G}_{\beta H}(Q_1\otimes \hat{1}_2)$ would not factorize. 

This line of reasoning may also fail if  $\mathcal{G}_{\beta H_2}(\hat{1}_2)$ would \emph{not} correspond to the equilibrium state on system $2$. This would be the case if we would choose  $\beta$ in our fluctuation relations to be different from the actual  $\beta'$ of the heat bath. From a purely mathematical point of view, the fluctuation relations are of course valid for all values of $\beta$ irrespective of whether they correspond to the actual temperature or not. However, in this case $G_{\beta'}(H_2)$ rather than $G_{\beta}(H_2)$ would be the true equilibrium state. Hence, it would require an active intervention on system $2$ to prepare the state  $\mathcal{G}_{\beta H_1}(Q^{f-}_1)\otimes G_{\beta}(H_2)$.

\subsection{\label{SecGeneralizedGibbs} The issue with non-exponential generalizations of the Gibbs maps}

As an illustration of the particular role of the exponential function and the Gibbs distribution, let us imagine that we attempt to use some other form of function to describe 
a generalized type of equilibrium states.
(This should not be confused with other types of generalizations, such as Jaynes  \cite{Jaynes57a,Jaynes57b} and recent approaches to multiple conserved quantities
\cite{Lostaglio15d,Halpern15c,Guryanova15,PerernauLlobet15}.)
 More precisely, let $g$ be any reasonable (possibly complex valued) function, and define the generalized map
$\mathcal{J}^{g}_{\beta H}(Q) := g(\beta H)Qg(\beta H)^{\dagger}$,
as well as the generalized Gibbs and partition maps
\begin{equation*}
\mathcal{G}^{g}_{\beta H}(Q) := \frac{\mathcal{J}^{g}_{\beta H}(Q)}{\mathcal{Z}^{g}_{\beta H}(Q)},\quad  \mathcal{Z}^{g}_{\beta H}(Q) := \Tr  \mathcal{J}^{g}_{\beta H}(Q). 
\end{equation*}
Here $\mathcal{G}_{g}(\hat{1}) = g(\beta H)g(\beta H)^{\dagger}/\Tr\big(g(\beta H)g(\beta H)^{\dagger}\big)$ would presumably take the role of a generalized form of equilibrium state. 

One could imagine constructing fluctuation relations for this generalized setup, and it is indeed straightforward to repeat the derivation of (\ref{GlobalInvariance}) to obtain
\begin{equation*}
\Tr\big(Q^{f+}V\mathcal{J}^{g}_{\beta H}(Q^{i+})V^{\dagger}\big) = \Tr\big(Q^{i-}V\mathcal{J}^g_{\beta H}(Q^{f-})V^{\dagger}\big).
\end{equation*}
At first sight this may seem like an endless source of non-standard fluctuation theorems for hypothetical non-Gibbsian distributions. However, since the factorization property fails for general non-exponential choices of $g$,  we can for example not reproduce the reasoning in Appendix \ref{GibbsInaccess}.

\section{\label{SecPrecorrelations} Pre-correlations}

Here we provide some further details on the example of section \ref{MainCorrelatedSE} in the main text, where we describe a setup similar to the one for the quantum Crooks relation (\ref{QuantumCrooks}), but where we wish to allow for the possibility that $S$ and $E$ are pre-correlated. The quantum Crooks relation (\ref{QuantumCrooks}) is formulated in terms of channels on the energy reservoir $E$. However, this implicitly assumes that $E$ initially is uncorrelated with the degrees of freedom it is about to interact with. Hence, in the present case we cannot express fluctuation relations in terms of channels or CPMs on $E$ alone, but one  alternative is to formulate a fluctuation relation in terms of channels on the joint system $SE$. 

We let the global Hamiltonian be
\begin{equation*}
\begin{split}
H:= & H^{i}_{SE}\otimes\hat{1}_B\otimes P^i_C + H^{f}_{SE}\otimes\hat{1}_B\otimes P^f_C \\
&   +  \hat{1}_{SE}\otimes H_B\otimes \hat{1}_{C},
\end{split}
\end{equation*}
where $H^{i}_{SE}$ and $H^{f}_{SE}$ are the initial and final Hamiltonians of the combined system and energy reservoir (where we can let these be non-interacting if we so wish). Moreover, like in Assumptions \ref{Def2}, $P^i_C$ and $P^f_C$ are the projectors onto the two orthogonal spaces $\Sp\{|c_{i+}\rangle,|c_{i-}\rangle\}$ and $\Sp\{|c_{f+}\rangle,|c_{f-}\rangle\}$, respectively. 

The total time-reversal is of the form $\mathcal{T}:= \mathcal{T}_{SE}\otimes\mathcal{T}_B\otimes \mathcal{T}_C$. (This is different from the decomposition $\mathcal{T}_{SBC}\otimes\mathcal{T}_E$ that we use in Appendix \ref{SecAQuantumFlctnThrm}.) We assume
 $\mathcal{T}_B(H_B) = H_B$.
(We do not need to assume $\mathcal{T}_{SE}(H^{i}_{SE}) = H^{i}_{SE}$ or $\mathcal{T}_{SE}(H^{f}_{SE}) = H^{f}_{SE}$.)
We furthermore let $\mathcal{T}_{C}$ be such that 
\begin{equation*}
\begin{split}
\mathcal{T}_{C}(|c_{i+}\rangle\langle c_{i+}|) = & |c_{i-}\rangle\langle c_{i-}|,\\
 \mathcal{T}_{C}(|c_{f+}\rangle\langle c_{f+}|) = & |c_{f-}\rangle\langle c_{f-}|.
\end{split}
\end{equation*}
We assume a global unitary evolution operator $V$ that is energy conserving $[V,H] = 0$, time-reversal symmetric $\mathcal{T}(V)=V$, and satisfies perfect control 
\begin{equation*}
V[\hat{1}_{SBE}\otimes |c_{i+}\rangle\langle c_{i+}|] = [\hat{1}_{SBE}\otimes |c_{f+}\rangle\langle c_{f+}|]V.
\end{equation*} 
We define the forward and reverse channels 
\begin{equation*}
\begin{split}
\overline{\mathcal{F}}_{+}(\chi) := & \Tr_{CB}(V [|c_{i+}\rangle\langle c_{i+}|\otimes G(H_B)\otimes\chi] V^{\dagger}),\\
\overline{\mathcal{F}}_{-}(\chi) := & \Tr_{CB}(V [|c_{f-}\rangle\langle c_{f-}|\otimes G(H_B)\otimes\chi] V^{\dagger}).
\end{split}
\end{equation*}
By using the relation $\mathcal{T}_{SE}\Tr_{CB} Q = \Tr_{CB}\mathcal{T}(Q)$, perfect control, and energy conservation, it follows, much as in previous derivations, that 
\begin{equation*}
\begin{split}
& \mathcal{T}_{SE}\overline{\mathcal{F}}_{-}^{*}(\chi) \\
& =  \Tr_{CB}\Big(\mathcal{T}\big([|c_{f-}\rangle\langle c_{f-}|\otimes G(H_B)\otimes\hat{1}_{SE}]\\
&\quad\quad\quad\quad\times V^{\dagger} [\hat{1}_C\otimes \hat{1}_B\otimes\chi] V\big)\Big)\\
& = \frac{1}{Z(H_B)} \Tr_{CB}\big([|c_{f+}\rangle\langle c_{f+}|\otimes e^{-\beta H_B}\otimes\hat{1}_{SE}]\\
&\quad\quad\quad\quad\quad\quad\times  V[\hat{1}_C\otimes\hat{1}_B\otimes\mathcal{T}_{SE}(\chi)]V^{\dagger}\big)\\
& = \frac{1}{Z(H_B)} \Tr_{CB}\big([|c_{f+}\rangle\langle c_{f+}|\otimes\hat{1}_B\otimes   e^{\beta H^{f}_{SE}/2}]  e^{-\beta H/2}\\
& \quad\quad\quad\quad\quad\quad\times V[\hat{1}_C\otimes\hat{1}_B\otimes\mathcal{T}_{SE}(\chi)]V^{\dagger}\\
& \quad\quad\quad\quad\quad\quad\times  e^{-\beta H/2}[|c_{f+}\rangle\langle c_{f+}|\otimes\hat{1}_B\otimes   e^{\beta H^{f}_{SE}/2}]\big)\\
 & = \frac{1}{Z(H_B)} e^{\beta H^{f}_{SE}/2}\Tr_{CB}\big(V[|c_{i+}\rangle\langle c_{i+}|\otimes\hat{1}_B\otimes \hat{1}_{SE}] \\
& \quad\quad\quad\quad\quad\quad \times e^{-\beta H/2}[\hat{1}_C\otimes\hat{1}_B\otimes\mathcal{T}_{SE}(\chi)]e^{-\beta H/2}\\
& \quad\quad\quad\quad \quad\quad \times [ |c_{i+}\rangle\langle c_{i+}|\otimes\hat{1}_B\otimes \hat{1}_{SE}]V^{\dagger} 
\big)e^{\beta H^{f}_{SE}/2}\\
 & =  \mathcal{J}_{\beta H^{f}_{SE}}^{-1}\bigg(\overline{\mathcal{F}}_{+}\Big(\mathcal{J}_{\beta H^{i}_{SE}}\big(\mathcal{T}_{SE}(\chi)\big)\Big)\bigg),
\end{split}
\end{equation*}
which can be rewritten as $\overline{\mathcal{F}}_{+} = \mathcal{J}_{\beta H_{SE}^{f}}\overline{\mathcal{F}}_{-}^{\ominus}\mathcal{J}^{-1}_{\beta H_{SE}^{i}}$.

\section{\label{SecApproximateFluct}Approximate fluctuation relations}

As discussed in Appendix \ref{MakeSense} we have up to now separated the role of the global Hamiltonian $H$ as characterizing energy, from its role as generator of the time evolution. The latter role we have so far assigned to the unitary operator $V$, with the restriction that it should be energy conserving $[H,V] = 0$. 
Here we consider the modifications needed to re-join these two roles in the sense that we let $V = e^{-it H/\hbar}$ (which  automatically satisfies the condition for energy conservation $[H,V] = 0$).  

Some issues appear when we try to fit  $V = e^{-it H/\hbar}$ with our previous assumptions on the structure of $H$.  For example, in Assumptions \ref{Def2} we explicitly assumed that the global Hamiltonian is of the form 
$H =  H_{S'C}\otimes\hat{1}_E + \hat{1}_{S'C}\otimes H_E$,
where $H_{S'C} =  H^{i}_{S'}\otimes P^{i}_C+ H^{f}_{S'}\otimes P^f_C + H^{\perp}$.
We furthermore assumed that the initial state of the control system in the forward process has support in the subspace onto which $P^{i}_C$ projects.
If the global evolution is given by $V = e^{-itH/\hbar}$, this implies that the control system will never leave the initial subspace, and thus fails to satisfy the assumption of perfect control. In other words, we cannot obtain the relevant dynamics with these combinations of assumptions.

This particular issue does not apply to the setting of the conditional fluctuation relations in Assumptions \ref{AssumpCondTimeRev}, since we there neither assume a special structure of $H_{S'C}$, nor of the initial state. 
However,  in the conditional setup we still assume that the global Hamiltonian is of the form
$H = H_{\widetilde{S}}\otimes\hat{1}_E  + \hat{1}_{\widetilde{S}}\otimes H_{E}$.
In other words, we assume that there is no interaction between $\widetilde{S} = S'C$ and the energy reservoir $E$. Needless to say, if the global evolution would be given by $V= e^{-it H/\hbar}$, then the state of the energy reservoir would be left unaffected by whatever happens in $S'C$, and the whole idea of the energy reservoir thus becomes meaningless. Hence, we do again find that an evolution of the form $V = e^{-itH/\hbar}$ clashes with our general assumptions in the sense that it generates a trivial dynamics.

\subsection{\label{SecApproxFluctuationRelations}A general notion of approximate fluctuation relations}

 To highlight the general structure we consider a separation into two anonymous subsystems $1$ and $2$. The reason for this is that we will consider different ways to partition the subsystems  $S$, $B$, $C$, and $E$.

\subsubsection{\label{SecApproxGeneral}The approximation}

In Appendix \ref{SecInaccessible} we pointed out that for non-interacting Hamiltonians $H = H_1\otimes\hat{1}_2 + \hat{1}_1\otimes H_2$, the function $\mathcal{J}_{\beta H}$ satisfies the factorization property $\mathcal{J}_{\beta H}(Q_1\otimes Q_2) = \mathcal{J}_{\beta H_1}(Q_1)\otimes \mathcal{J}_{\beta H_2}(Q_2)$. Since we in this section abandon these convenient non-interacting Hamiltonians, the question is what is supposed to replace them. 
Intuitively, the idea is that for positive operators with \emph{suitable} support, the action of the global Hamiltonian $H$ can be approximated by $H^{i}_{1}\otimes\hat{1}_{2} + \hat{1}_{1}\otimes H^{i}_2$, for some local Hamiltonians $H^{i}_{1}$ and  $H^i_2$. Similarly, for another suitable class of operators, the action of the global Hamiltonian can be approximated by  $H^{f}_{1}\otimes\hat{1}_{2} + \hat{1}_{1}\otimes H^{f}_2$, for Hamiltonians $H^{f}_{1}$ and  $H^{f}_2$.  (For the sake of generality and flexibility we here allow different initial and final Hamiltonians on both subsystems.)  

The more exact formulation is based on the $\mathcal{J}_{\beta H}$ map.
We assume a product time reversal $\mathcal{T} =\mathcal{T}_1\otimes\mathcal{T}_2$ and local approximate Hamiltonians $H_{1}^{i}$, $H_{1}^{f}$,  $H_{2}^{i}$, $H_{2}^{f}$  such that 
\begin{equation}
\label{CondHandT}
\begin{split}
& \mathcal{T}(H) =  H,\\
&  \mathcal{T}_{1}(H_{1}^{i}) =  H_{1}^{i},\quad \mathcal{T}_{1}(H_{1}^{f}) =  H_{1}^{f},\\
& \mathcal{T}_{2}(H^i_2) =  H^i_2,\quad  \mathcal{T}_2(H^f_2) =  H^f_2,
\end{split}
\end{equation}
Suppose that for the measurement operators $Q_{1}^{i+}, Q_{2}^{i+}, Q_{1}^{f+}, Q_{2}^{f+}$ the approximate factorization holds 
\begin{equation}
\label{Approximation}
\begin{split}
\mathcal{J}_{\beta H}(Q_{1}^{i+}\otimes Q_{2}^{i+})\approx  & \mathcal{J}_{\beta H^i_1}(Q_{1}^{i+})\otimes  \mathcal{J}_{\beta H^i_2}(Q_{2}^{i+}),\\
\mathcal{J}_{\beta H}(Q_{1}^{f+}\otimes Q_{2}^{f+})\approx  & \mathcal{J}_{\beta H^f_1}(Q_{1}^{f+})\otimes  \mathcal{J}_{\beta H^f_2}(Q_{2}^{f+}).
\end{split}
\end{equation}
The idea is that under these conditions we should obtain the following approximate global fluctuation relation
\begin{equation}
\label{nndsdsvk}
\begin{split}
&  \mathcal{Z}_{\beta H^i_1}(Q_1^{i})\mathcal{Z}_{\beta H^i_2}(Q_2^{i})P^{V}_{\beta H^i}[ Q_{1}^{i+}\otimes Q_{2}^{i+} \rightarrow Q_{1}^{f+}\otimes Q_{2}^{f+}]\\
\approx &   \mathcal{Z}_{\beta H^f_1}(Q_1^{f})\mathcal{Z}_{\beta H^f_2}(Q_2^{f})P^{V}_{\beta H^f}[ Q_{1}^{f-}\otimes Q_{2}^{f-} \rightarrow Q_{1}^{i-}\otimes Q_{2}^{i-}],
\end{split}
\end{equation}
where $H^{i} := H^i_1\otimes\hat{1}_2 + \hat{1}_1\otimes H^i_2$, $H^{f} := H^f_1\otimes\hat{1}_2 + \hat{1}_1\otimes H^f_2$, and where $V = e^{-it H/\hbar}$ for some $t\in \mathbb{R}$, and $Q_{1}^{i-}  := \mathcal{T}_{1}(Q_{1}^{i+})$, $Q_{2}^{i-}  := \mathcal{T}_{2}(Q_{2}^{i+})$, $Q_{1}^{f-}  := \mathcal{T}_{1}(Q_{1}^{f+})$, $Q_{2}^{f-}  := \mathcal{T}_{2}(Q_{2}^{f+})$.

It is straightforward to make an informal `derivation' of  (\ref{nndsdsvk}), which combines the approximations in (\ref{Approximation}) with the assumed properties (\ref{CondHandT}) of the time reversal together with $[V,H] =0$ and $\mathcal{T}(V) =  V$ 
\begin{equation*}
\begin{split}
 &\Tr\big([Q_{1}^{f+}\otimes Q_{2}^{f+} ]V[\mathcal{J}_{\beta H_1^{i}}(Q_{1}^{i+})\otimes \mathcal{J}_{\beta H_2^{i}}(Q_{2}^{i+})]V^{\dagger}\big) \\
\approx &\Tr\big([Q_{1}^{f+}\otimes Q_{2}^{f+} ]V\mathcal{J}_{\beta H}(Q_{1}^{i+}\otimes Q_{2}^{i+})V^{\dagger}\big) \\
= &\Tr\big(
[Q_{1}^{i-}\otimes Q_{2}^{i-}]\mathcal{J}_{\beta H}(V[Q_{1}^{f-}\otimes Q_{2}^{f-}]V^{\dagger})
\big) \\
= &\Tr\big(
[Q_{1}^{i-}\otimes Q_{2}^{i-}]V\mathcal{J}_{\beta H}(Q_{1}^{f-}\otimes Q_{2}^{f-})V^{\dagger}
\big) \\
\approx &\Tr(
[Q_{1}^{i-}\otimes Q_{2}^{i-}]V[\mathcal{J}_{\beta H_1^f}(Q_{1}^{f-})\otimes \mathcal{J}_{\beta H_2^f}(Q_{2}^{f-})]V^{\dagger}
).
\end{split}
\end{equation*}
In the final step to obtain (\ref{nndsdsvk}) we use assumptions (\ref{CondHandT}) and Lemma \ref{LemmaSimpleRelations} to get 
$\mathcal{Z}_{\beta H_1^i}(Q^{i-}_1) = \mathcal{Z}_{\beta H_1^i}\big(\mathcal{T}_1(Q^{i+}_1)\big) =: \mathcal{Z}_{\beta H_1^i}(Q^{i}_1)$, and analogously for the other partition maps.

\subsubsection{\label{formalizationsAppr} Quantitative formulation of the approximation}

Here we consider one  way to make the approximation as expressed by (\ref{Approximation})  and (\ref{nndsdsvk}) quantitative. We here assume that the underlying Hilbert spaces are finite-dimensional.  Define 
\begin{equation}
\label{Defdidf}
\begin{split}
 & d_{H_{1},H_{2}}^{H}(Q_{1},Q_{2})  \\
& := \big\Vert \mathcal{J}_{\beta H_1}(Q_{1})\otimes \mathcal{J}_{\beta H_2}(Q_{2})  - \mathcal{J}_{\beta H}(Q_{1}\otimes Q_{2}) \big\Vert_1, 
\end{split}
\end{equation}
where we use the trace norm $\Vert Q\Vert_1 := \Tr\sqrt{Q^{\dagger}Q}$.
In the following we will also use the standard operator norm $\Vert Q\Vert :=\sup_{\Vert \psi\Vert = 1}\Vert Q|\psi\rangle\Vert$.

By Lemma \ref{Tpreservesbound}, we can conclude that $\Vert Q_{1}^{i-}\Vert =  \Vert Q_{1}^{i+}\Vert$ and analogous for $Q_{2}^{i\pm}$, $Q_{1}^{f\pm}$, $Q_{2}^{f\pm}$. 
Moreover, one can confirm that  $d^{H}_{H_1^{i},H_2^{i}}(Q_{1}^{i-},Q_{2}^{i-}) = d^{H}_{H_1^{i},H_2^{i}}(Q_{1}^{i+},Q_{2}^{i+})$ and $d^{H}_{H_1^{f},H_2^{f}}(Q_{1}^{f-}, Q_{2}^{f-}) = d^{H}_{H_1^{f},H_2^{f}}(Q_{1}^{f+}, Q_{2}^{f+})$.

As a bit of a technical side-remark one may note that none of the proofs explicitly use the assumption that $V = e^{-it H/\hbar}$. We do still only rely on  $[V,H] = 0$ and $\mathcal{T}(V) = V$. However, due to the more forgiving structure we can now assume that  $V = e^{-it H/\hbar}$ and  yet have a non-trivial evolution on the energy reservoir.

\begin{Proposition}
\label{ApproxFlctnRel}
Let $H$, $H^{i}_1$, $H^{i}_2$, $H^{f}_1$, $H^{f}_2$ be Hermitian operators, let $\mathcal{T}, \mathcal{T}_{1}, \mathcal{T}_2$ be time reversals that satisfy the conditions in (\ref{CondHandT}), and let $V$ be a unitary operator such that $[H,V] = 0$ and $\mathcal{T}(V) = V$ (which in particular allows us to choose $V = e^{-it H/\hbar}$ for some $t\in\mathbb{R}$). Then 
\begin{equation*}
\begin{split}
&\bigg| \mathcal{Z}_{\beta H^i_1}(Q_1^{i})\mathcal{Z}_{\beta H^i_2}(Q_2^{i})P^{V}_{\beta H^i}[ Q_{1}^{i+}\otimes Q_{2}^{i+} \rightarrow Q_{1}^{f+}\otimes Q_{2}^{f+}]\\
 &   -\mathcal{Z}_{\beta H^i_1}(Q_1^{f})\mathcal{Z}_{\beta H^i_2}(Q_2^{f})P^{V}_{\beta H^f}[ Q_{1}^{f-}\otimes Q_{2}^{f-} \rightarrow Q_{1}^{i-}\otimes Q_{2}^{i-}]\bigg|\\
& \leq    \Vert Q_{1}^{f+}\Vert \Vert Q_{2}^{f+} \Vert  d^{H}_{H_1^{i},H_2^{i}}(Q_{1}^{i+},Q_{2}^{i+})  \\
& +   \Vert Q_{1}^{i+}\Vert \Vert Q_{2}^{i+}\Vert d^{H}_{H_1^{f},H_2^{f}}(Q_{1}^{f+}, Q_{2}^{f+}),
\end{split}
\end{equation*}
where $H^{i} := H^i_1\otimes\hat{1}_2 + \hat{1}_1\otimes H^i_2$, $H^{f} := H^f_1\otimes\hat{1}_2 + \hat{1}_1\otimes H^f_2$, and 
where  $Q_{1}^{i-}  := \mathcal{T}_{1}(Q_{1}^{i+})$, $Q_{2}^{i-}  := \mathcal{T}_{2}(Q_{2}^{i+})$, $Q_{1}^{f-}  := \mathcal{T}_{1}(Q_{1}^{f+})$, $Q_{2}^{f-}  := \mathcal{T}_{2}(Q_{2}^{f+})$.
\end{Proposition}
\begin{proof}
As the first step one can confirm the identity
\begin{equation}
\label{vnaknd}
\begin{split}
&\Tr\big([Q_{1}^{f+}\otimes Q_{2}^{f+} ]V\mathcal{J}_{\beta H}(Q_{1}^{i+}\otimes Q_{2}^{i+})V^{\dagger}\big) \\
 = &\Tr\big(
[Q_{1}^{i-}\otimes Q_{2}^{i-}]V\mathcal{J}_{\beta H}(Q_{1}^{f-}\otimes Q_{2}^{f-})V^{\dagger}
\big).
\end{split}
\end{equation}
For a more compact notation let $\sigma_i := \mathcal{J}_{\beta H^i_1}(Q_{1}^{i+})\otimes \mathcal{J}_{\beta H^i_2}(Q_{2}^{i+})$ and
$\sigma_f :=  \mathcal{J}_{\beta H_{1}^f}(Q_{1}^{f-})\otimes \mathcal{J}_{\beta H_{2}^f}(Q_{2}^{f-})$.
Then 
\begin{equation*}
\begin{split}
&\big|\Tr([Q_{1}^{f+}\!\otimes\! Q_{2}^{f+} ]V\sigma_iV^{\dagger})  - \Tr(
[Q_{1}^{i-}\!\otimes\! Q_{2}^{i-}]V\sigma_f V^{\dagger}
)\big| \\
&\quad\quad[\textrm{By (\ref{vnaknd}) and the triangle inequality}]\\
&\leq   \Big|\Tr\Big(V^{\dagger}[Q_{1}^{f+}\otimes Q_{2}^{f+} ]V\big(\sigma_i  - \mathcal{J}_{\beta H}(Q_{1}^{i+}\otimes Q_{2}^{i+}) \big)\Big)\Big| \\
&+\Big|\Tr\Big(
V^{\dagger}[Q_{1}^{i-}\otimes Q_{2}^{i-}]V\big( \mathcal{J}_{\beta H}(Q_{1}^{f-}\otimes Q_{2}^{f-})  -\sigma_f\big)
\Big)\Big|\\
&\quad\quad[\textrm{By the general relation $|\Tr(AQ)|\leq \Vert A\Vert \Vert Q\Vert_1$}]\\
&\leq   \Vert Q_{1}^{f+}\Vert \Vert Q_{2}^{f+} \Vert  d^{H}_{H_1^{i},H_2^{i}}(Q_{1}^{i+},Q_{2}^{i+})  \\
& +   \Vert Q_{1}^{i+}\Vert \Vert Q_{2}^{i+}\Vert  d^{H}_{H_1^{f},H_2^{f}}(Q_{1}^{f+}, Q_{2}^{f+}).
\end{split}
\end{equation*}
\end{proof}

\subsection{\label{SecApproxConditional} Approximate conditional fluctuation relations}

Here we consider approximate versions of the conditional fluctuation relations in Appendix \ref{SecConditional}. The general case is treated in the following subsection, and in the next we consider the more concrete special case of a particle as control system.

\subsubsection{The general case}

 Here we identify system $1$ of the previous section with $\widetilde{S}$, and system $2$ with $E$.
We thus assume
\begin{equation}
\label{transpham}
\begin{split}
& \mathcal{T}(H) =  H,\\
&  \mathcal{T}_{\widetilde{S}}(H_{\widetilde{S}}^{i}) =  H_{\widetilde{S}}^{i},\quad \mathcal{T}_{\widetilde{S}}(H_{\widetilde{S}}^{f}) =  H_{\widetilde{S}}^{f},\\
& \mathcal{T}_E(H^i_E) =  H^i_E,\quad  \mathcal{T}_E(H^f_E) =  H^f_E.
\end{split}
\end{equation}
Similar to the conditional fluctuation relations in Appendix \ref{SecConditional} we define CPMs on the energy reservoir conditioned on the successful control measurements
\begin{equation}
\begin{split}
\tilde{\mathcal{F}}_{+}(\sigma) := & \Tr_{\widetilde{S}}([Q_{\widetilde{S}}^{f+}\otimes \hat{1}_E]V[\mathcal{G}_{\beta H^i_{\widetilde{S}}}(Q_{\widetilde{S}}^{i+})\otimes \sigma]V^{\dagger}),\\
\tilde{\mathcal{F}}_{-}(\sigma) := & \Tr_{\widetilde{S}}([Q_{\widetilde{S}}^{i-}\otimes \hat{1}_E]V[\mathcal{G}_{\beta H^f_{\widetilde{S}}}(Q_{\widetilde{S}}^{f-})\otimes \sigma]V^{\dagger}).
\end{split} 
\end{equation}
One can verify that
\begin{equation*}
\begin{split}
& P^{\tilde{\mathcal{F}}_{+}}_{\beta H^i_E}[Q^{i+}_E\rightarrow Q^{f+}_E]=  P^{V}_{\beta H^i}[Q_{\widetilde{S}}^{i+}\otimes Q^{i+}_E\rightarrow Q_{\widetilde{S}}^{f+}\otimes Q^{f+}_E],\\
& P^{\tilde{\mathcal{F}}_{-}}_{\beta H^f_E}[Q^{f-}_E\rightarrow Q^{i-}_E]=  P^{V}_{\beta H^f}[Q_{\widetilde{S}}^{f-}\otimes Q^{f-}_E\rightarrow Q_{\widetilde{S}}^{i-}\otimes Q^{i-}_E], 
\end{split}
\end{equation*}
which together with Proposition \ref{ApproxFlctnRel} yield
\begin{equation}
\label{dsjflkfkldlfdk}
\begin{split}
&\bigg| \mathcal{Z}_{\beta H^i_{\widetilde{S}}}(Q_{\widetilde{S}}^{i})\mathcal{Z}_{\beta H^i_E}(Q_E^{i})P^{\tilde{\mathcal{F}}_{+}}_{\beta H^i_E}[Q^{i+}_E\rightarrow Q^{f+}_E]\\
&   -\mathcal{Z}_{\beta H^f_{\widetilde{S}}}(Q_{\widetilde{S}}^{f})\mathcal{Z}_{\beta H^i_E}(Q_E^{f}) P^{\tilde{\mathcal{F}}_{-}}_{\beta H^f_E}[Q^{f-}_E\rightarrow Q^{i-}_E]\bigg|\\
& \leq    \Vert Q_{\widetilde{S}}^{f+}\Vert \Vert Q_{E}^{f+} \Vert  d^{H}_{H_{\widetilde{S}}^{i},H_E^{i}}(Q_{\widetilde{S}}^{i+},Q_{E}^{i+})  \\
& +   \Vert Q_{\widetilde{S}}^{i+}\Vert \Vert Q_{E}^{i+}\Vert d^{H}_{H_{\widetilde{S}}^{f},H_E^{f}}(Q_{\widetilde{S}}^{f+}, Q_{E}^{f+}).
\end{split}
\end{equation}
With (\ref{dsjflkfkldlfdk}) as starting point one can also obtain an approximate fluctuation relation in terms of channels, i.e., more in the spirit of  (\ref{ConditionalFluctuationThm}). Let us define the following measure of difference between maps 
\begin{equation}
\textrm{Diff}(\phi_1,\phi_2):=\sup_{Q^i\geq 0,Q^f\geq 0}\frac{\big|\Tr\big(Q^f \phi_2(Q^i)\big)-\Tr\big(Q^f\phi_1(Q^{i})\big)\big|}{\Vert Q^i\Vert \Vert Q^f\Vert}.
\end{equation}
This particular choice of measure has no deeper reason than that it makes the derivations simple. Also define
\begin{equation}
\label{nlkvnadf2}
\begin{split}
D_i(Q_{\widetilde{S}}^{i+}) :=  & \sup_{Q\geq 0}\frac{1}{\Vert Q\Vert}d^{H}_{H_{\widetilde{S}}^{i},H_E^{i}}(Q_{\widetilde{S}}^{i+},Q),\\
D_f(Q_{\widetilde{S}}^{f+}) :=  & \sup_{Q\geq 0}\frac{1}{\Vert Q\Vert}d^{H}_{H_{\widetilde{S}}^{f},H_E^{f}}(Q_{\widetilde{S}}^{f+},Q).
\end{split}
\end{equation}
By the approximate conditional fluctuation relation (\ref{dsjflkfkldlfdk}) it follows that
\begin{equation}
\label{nvdsnlkdslkn}
\begin{split}
& \textrm{Diff}\big(
\mathcal{Z}_{\beta H^i_{\widetilde{S}}}(Q_{\widetilde{S}}^{i})\tilde{\mathcal{F}}_{+}\circ\mathcal{J}_{\beta H^i_E}\, ,\, \mathcal{Z}_{\beta H^f_{\widetilde{S}}}(Q_{\widetilde{S}}^{f})\mathcal{J}_{\beta H^f_E}\circ \tilde{\mathcal{F}}^{\ominus}_{-}\big)\\
&\quad\quad \leq  \Vert Q_{\widetilde{S}}^{f+}\Vert D_i(Q_{\widetilde{S}}^{i+})  +  \Vert Q_{\widetilde{S}}^{i+}\Vert D_f(Q_{\widetilde{S}}^{f+}).
\end{split}
\end{equation}
Hence, this is a quantitative version of the approximate conditional relation (\ref{MainApproxFluctThm}) in the main text.

In (\ref{dsjflkfkldlfdk}) we only consider the error for a specific pair $Q^{i+}_E,Q^{f+}_E$, while in  (\ref{nvdsnlkdslkn}) we ask for the worst case error over the entire set of positive semi definite operators. It may very well be the case that pointwise error can be small for some specific choice of operators, while uniform error would be large. 
Hence, the formulation via the transition probabilities can be more `forgiving' than the formulation via channels. This should be compared with the non-approximate case, where the choice largely is a matter of convenience.

\subsubsection{\label{controlparticle1}The special case of a control-particle}

Imagine a joint Hamiltonian of the form
\begin{equation}
\label{GlobalHamiltonian}
H = \frac{1}{2M_C}\hat{P}_{C}^2\otimes\hat{1}_{S'E} + H_{S'E}(\hat{X}_{C}).
\end{equation}
Here $M_C$ is the mass and $\hat{X}_{C},\hat{P}_{C}$ are the canonical position and momentum operators of the control particle.
Assume furthermore that  $H_{S'E}(x)$ is of the form 
\begin{equation}
\label{FinalInitial}
H_{S'E}(x) = \left\{\begin{matrix} H^{i}_{S'}\otimes\hat{1}_E + \hat{1}_{S'} \otimes H^i_E,\quad x \leq x_{i},\\
 H^{f}_{S'}\otimes\hat{1}_E + \hat{1}_{S'} \otimes H^f_E,\quad x \geq x_{f},
\end{matrix}\right.
\end{equation}
while in the interval $[x_i,x_f]$ there is some non-trivial dependence on $x$ (where we assume $x_i < x_f$).
The Hamiltonians $H^{i}_{\widetilde{S}}$ and $H^{f}_{\widetilde{S}}$ would in this case be
\begin{equation}
\label{kvald}
\begin{split}
H^{i}_{\widetilde{S}} := & \frac{1}{2M_C}\hat{P}_{C}^2\otimes\hat{1}_{S'} + \hat{1}_C\otimes H^{i}_{S'},\\
H^{f}_{\widetilde{S}} := & \frac{1}{2M_C}\hat{P}_{C}^2\otimes\hat{1}_{S'} + \hat{1}_C\otimes H^{f}_{S'}.
\end{split}
\end{equation}
 Hence, outside the `interaction region' $[x_i,x_f]$ all the three systems  systems $S'$, $E$, and $C$ are non-interacting.

For the sake of illustration we consider the special case of measurement operators of the form
\begin{equation}
\label{gaajgjga}
Q_{\widetilde{S}}^{i\pm} :=  Q_{C}^{i\pm}\otimes Q_{S'}^{i\pm},\quad Q_{\widetilde{S}}^{f\pm} :=  Q_{C}^{f\pm}\otimes Q_{S'}^{f\pm},
\end{equation}
where $Q_{S'}^{i\pm}$,  $Q_{S'}^{f\pm}$ are measurement operators on $S'$, and where $Q_{C}^{i\pm}$, $Q_{C}^{f\pm}$ are measurement operators on $C$ which are `concentrated' in the regions $(-\infty,x_i]$ and $[x_f,+\infty)$, respectively.
It seems intuitively reasonable  that the further down from $x_i$ that  $Q_{C}^{i+}$ is supported, the better is the approximation
\begin{equation}
\label{aklalkv}
\begin{split}
& \mathcal{J}_{\beta H}( Q_{\widetilde{S}}^{i+}\otimes Q)\\
& \approx   \mathcal{J}_{\beta H^{i}_{\widetilde{S}}}(Q_{\widetilde{S}}^{i+})\otimes \mathcal{J}_{\beta H^i_E}(Q)\\
& =  \mathcal{J}_{\beta \hat{K}}(Q_{C}^{i+})\otimes  \mathcal{J}_{\beta H^{i}_{S'}}(Q_{S'}^{i+})\otimes \mathcal{J}_{\beta H^i_E}(Q),
\end{split}
\end{equation}
where we have introduced the notation
$\hat{K} := \hat{P}_{C}^2/(2M_C)$,
 and used the fact that (\ref{kvald}) defines non-interacting Hamiltonians between $S'$ and $C$. 

The approximate fluctuation relation (\ref{nvdsnlkdslkn})  takes the form
\begin{equation*}
\begin{split}
& \textrm{Diff}\Big(
 \mathcal{Z}_{\beta \hat{K}}(Q_{C}^{i+})\mathcal{Z}_{H^i_{S'}}(Q_{S'}^{i+})
\tilde{\mathcal{F}}_{+}\circ \mathcal{J}_{\beta H^i_E},\\ 
&\quad\quad\quad \mathcal{Z}_{\beta \hat{K}}(Q_{C}^{f+})\mathcal{Z}_{H^f_{S'}}(Q_{S'}^{f+})
\mathcal{J}_{\beta H^f_E}\circ\tilde{\mathcal{F}}^{\ominus}_{-}
\Big)\\
& \leq  \Vert Q_{C}^{f+}\Vert \Vert Q_{S'}^{f+}\Vert D_i(Q_{C}^{i+}\otimes Q_{S'}^{i+})  \\
& \quad\quad\quad +  \Vert Q_{C}^{i+}\Vert \Vert Q_{S'}^{i+}\Vert D_f(Q_{C}^{f+}\otimes Q_{S'}^{f+}).
\end{split}
\end{equation*}
Note that if $Q_{C}^{f+}$ is a spatial translation of $Q_{C}^{i+}$, i.e., such that  $Q_{C}^{f+} = e^{ir\hat{P}_C}Q_{C}^{i+}e^{-ir\hat{P}_C}$ for some $r\in\mathbb{R}$, then $\mathcal{Z}_{\beta \hat{K}}(Q_{C}^{i+}) = \mathcal{Z}_{\beta \hat{K}}(Q_{C}^{f+})$. 

One may wonder how to choose the time-reversals $\mathcal{T}_{\widetilde{S}}$ and $\mathcal{T}_E$. One possibility is if there exist orthonormal complete bases of $\mathcal{H}_{S'}$ and $\mathcal{H}_E$, such that the family of Hamiltonians $H_{S'E}(x)$, $H^{i}_{S'}$, $H^{f}_{S'}$, $H^i_E$, and $H^f_E$ are real valued matrices in these bases.
In this case we can choose $\mathcal{T}_E$ and $\mathcal{T}_{S'}$ as transpositions with respect to these bases, and $\mathcal{T}_{C}$ as the transposition with respect to the position representation, and let $\mathcal{T}_{\widetilde{S}} = \mathcal{T}_{C}\otimes\mathcal{T}_{S'}$. This guarantees that 
(\ref{transpham}) holds. 
 
As a further comment one may observe that what we assign as being the energy reservoir is largely a matter of choice. In this particular example it is clear that also the control particle can donate energy.

\subsection{\label{SecApproximateTransProb} Joint control and energy reservoir}

Here we turn to a setting where the control  and the energy reservoir are one and the same system. Hence, instead of dividing the global system between $E$ and $S'C$, we here divide it into $CE$ and $S'$. 

\subsubsection{The general case}
Here we let system $1$ (in Appendix \ref{SecApproxGeneral}) be $S'$ and system $2$ be $CE$.
We assume   
\begin{equation}
\label{CondHandTtrans}
\begin{split}
& \mathcal{T}(H) =  H,\\
&  \mathcal{T}_{S'}(H_{S'}^{i}) =  H_{S'}^{i},\quad \mathcal{T}_{S'}(H_{S'}^{f}) =  H_{S'}^{f},\\
& \mathcal{T}_{CE}(H^i_{CE}) =  H^i_{CE},\quad  \mathcal{T}_{CE}(H^f_{CE}) =  H^f_{CE},
\end{split}
\end{equation}
as well as $[V,H] = 0$ and $\mathcal{T}(V) = V$. By Proposition \ref{ApproxFlctnRel}
\begin{equation}
\label{ApproxFluctQuantitative}
\begin{split}
&\bigg| \mathcal{Z}_{\beta H^i_{S'}}(Q_{S'}^{i})\mathcal{Z}_{\beta H^i_{CE}}(Q_{CE}^{i}) \mathcal{P}^{+} \\
&  -\mathcal{Z}_{\beta H^f_{S'}}(Q_{S'}^{f})\mathcal{Z}_{\beta H^f_{CE}}(Q_{CE}^{f}) \mathcal{P} ^{-} \bigg|\\
& \leq    \Vert Q_{S'}^{f+}\Vert \Vert Q_{CE}^{f+} \Vert  d^{H}_{H_{S'}^{i},H_{CE}^{i}}(Q_{S'}^{i+},Q_{CE}^{i+})  \\
& +   \Vert Q_{S'}^{i+}\Vert \Vert Q_{CE}^{i+}\Vert d^{H}_{H_{S'}^{f},H_{CE}^{f}}(Q_{S'}^{f+}, Q_{CE}^{f+}),
\end{split}
\end{equation}
where we introduce the following short hand notation for the transition probabilities
\begin{equation}
\label{jdslkfajd}
\begin{split}
\mathcal{P}^{+} :=& P^{V}_{\beta H^i}[ Q_{S'}^{i+}\otimes Q_{CE}^{i+} \rightarrow Q_{S'}^{f+}\otimes Q_{CE}^{f+}],\\
\mathcal{P}^{-} := & P^{V}_{\beta H^f}[ Q_{S'}^{f-}\otimes Q_{CE}^{f-} \rightarrow Q_{S'}^{i-}\otimes Q_{CE}^{i-}].
\end{split}
\end{equation}
Hence, (\ref{ApproxFluctQuantitative}) is the quantitative version of (\ref{MainApproxFluctQuantitative}) in the main text.

\subsubsection{\label{SecControlparticleAgain} A single particle as both control and energy reservoir}

In Appendix \ref{controlparticle1} we considered the case of a particle  whose motion implements the time-dependent Hamiltonian in Crooks relation. There we regarded the degrees of freedom of the particle as separate from the designated energy reservoir.  This made it possible  to express CPMs on the energy reservoir conditioned on the control measurements on the control particle. An intuitively reasonable alternative would be that the motion of the control particle also fuels the process, i.e., it is the initial kinetic energy of the particle that that drives the whole non-equilibrium process. 
The global Hamiltonian can in this case be chosen to be
\begin{equation}
\label{GlobalHamiltonianCE}
H = \frac{1}{2M_{CE}}\hat{P}_{CE}^2\otimes\hat{1}_{S'} + H_{S'}(\hat{X}_{CE}),
\end{equation}
where $M_{CE}$ is the mass, and $\hat{X}_{CE},\hat{P}_{CE}$ are the canonical position and momentum operators of the particle. 

We assume
\begin{equation}
\label{FinalInitialCE}
H_{S'}(x) = \left\{ \begin{matrix}H^{i}_{S'},\quad x \leq x_{i},\\
 H^{f}_{S'},\quad x \geq x_{f},
\end{matrix}\right.
\end{equation}
i.e., the Hamiltonian for $S'$ is constant outside the interaction region. Since the particle moreover is free outside the interaction region, we can take  the initial and final approximate Hamiltonians for $CE$ to be
$H_{CE}^i = H_{CE}^f = \frac{1}{2M_{CE}}\hat{P}_{CE}^2 =:\hat{K}$.
 For operators $Q_{CE}^{i+}$ and $Q_{CE}^{f+}$ that are well localized outside the interaction region it seems reasonable that $\mathcal{J}_{\beta H}(Q_{S'}^{i+}\otimes Q_{CE}^{i+})\approx  \mathcal{J}_{\beta H^{i}_{S'}}(Q_{S'}^{i+})\otimes \mathcal{J}_{\beta \hat{K}}(Q_{CE}^{i+})$ and   $\mathcal{J}_{\beta H}(Q_{S'}^{f+}\otimes Q_{CE}^{f+})\approx \mathcal{J}_{\beta H^{f}_{S'}}(Q_{S'}^{f+})\otimes \mathcal{J}_{\beta \hat{K}}(Q_{CE}^{f+})$.
Under these conditions we thus get the approximate fluctuation relation 
$\mathcal{Z}_{\beta H^i_{S'}}(Q_{S'}^{i+}) \mathcal{Z}_{\beta \hat{K}}(Q_{CE}^{i+})\mathcal{P}^{+} \approx  \mathcal{Z}_{\beta H^f_{S'}}(Q_{S'}^{f-}) \mathcal{Z}_{\beta \hat{K}}(Q_{CE}^{f-})\mathcal{P}^{-}$, with the quantitative version in (\ref{ApproxFluctQuantitative}).

\subsubsection{\label{SecNumericalEvaluation}Numerical evaluation}

\begin{figure}
 \includegraphics[width= 8cm]{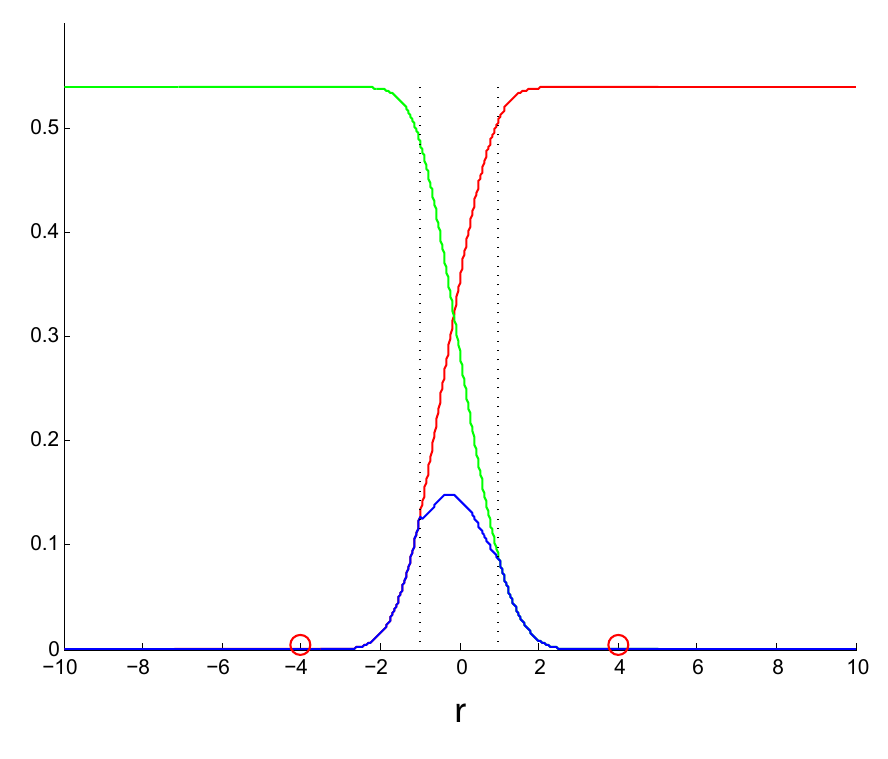} 
   \caption{\label{FigTestConditions}  {\bf Approximate factorization.}   
To assess the quality of the approximate factorization we evaluate 
$d^{H}_{\hat{K},E_0\sigma_z/2}(\hat{1}_{S'}\otimes |\alpha\rangle\langle \alpha|)$ (red curve), $d^{H}_{\hat{K},E_0\sigma_x/4}(\hat{1}_{S'}\otimes |\alpha\rangle\langle \alpha|)$ (green curve), 
$d^{H}_{\hat{K},E_0\overline{n}(ry_0)\cdot\overline{\sigma}/2}(\hat{1}_{S'}\otimes |\alpha\rangle\langle \alpha|)$ (blue curve)  for $\alpha := r+2i$ with $r \in [-10,10]$. The dotted black lines correspond to the borders of the interaction region. The red circles are the positions of the coherent states that gives the measurement operators $Q_{CE}^{i+} = |-4 +2i\rangle\langle -4+2i|$ and  $Q_{CE}^{f+} = |4 +2i\rangle\langle 4+2i|$.
 }
\end{figure}

To make the approximate fluctuation relations a bit more concrete we here make a numerical evaluation of a special case of the combined control and energy reservoir particle in the previous section.  We consider a single particle of mass $M$ that is restricted to move along the $y$-axis, and this spatial degree of freedom is taken as the combined control and energy reservoir $CE$. The particle also carries a magnetic moment corresponding to a spin-half degree of freedom, which we interpret as system $S'$. (A single spin is of course somewhat ridiculous regarded as a combined system and heat bath, but this example only serves to illustrate the formalism,  for which the sizes of the participating systems do not matter.)  We assume that the spin interacts with an external magnetic field that is time-independent, but is a function of $y$.
The total Hamiltonian can be expressed in terms of the differential operator
\begin{equation}
H = -\frac{\hbar^2}{2M}\frac{d^2}{dy^2} + \frac{1}{2}E_0\overline{\sigma}_{S'}\cdot \overline{n}(y),
\end{equation}
where $\overline{n}(y)$ determines the strength and direction of the external magnetic field
as a function of $y$, and where $\overline{\sigma}_{S'} = (\sigma^{(x)}_{S'},\sigma^{(y)}_{S'},\sigma^{(z)}_{S'})$ are the Pauli spin operators, with $\sigma^{(x)} = |0\rangle\langle 1|+|1\rangle\langle 0|$, $\sigma^{(y)} = i|0\rangle\langle 1|-i|1\rangle\langle 0|$, $\sigma^{(z)} = |1\rangle\langle 1|-|0\rangle\langle 0|$, with $\{|0\rangle,|1\rangle\}$ being  the eigenbasis of $\sigma^{(z)}$.
For $\Vert \overline{n}\Vert = 1$ it follows that $E_0$ is the excitation energy of the spin. 

To get a particularly simple model we here assume that $\overline{n}(y) = (0,0,1)$ for  $y < -y_0$,  
$\overline{n}(y) = [3/4-y/(4y_0)]\big(\sin(\frac{\pi (y+y_0)}{4 y_0}),0,\cos(\frac{\pi (y+y_0)}{4 y_0})\big)$ for $-y_0 \leq y \leq y_0$, and $\overline{n}(y) = (1/2,0,0)$ for $y_0 < y$.
Hence, for all positions below $-y_0$ the magnetic field is directed along the $z$-axis.  Within the interaction region $[-y_0,y_0]$ the field rotates in the $xz$-plane until it aligns with the $x$-axis at $y_0$, simultaneously as it decreases in strength to the half.
We choose $H^{i}_{CE} = H^{f}_{CE} = \hat{P}^2/(2M) =: \hat{K}$, $H^{i}_{S'} := E_0\sigma_z/2$, and $H^{f}_{S'} := E_0\sigma_x/4$.

Define $\mathcal{T}_{CE}$ to be the transpose with respect to the coordinate representation, and thus $\mathcal{T}_{CE}(\hat{K}) = \hat{K}$. For the chosen $\overline{n}(y)$, the Hamiltonian $\overline{n}(y)\cdot\overline{\sigma}$ is  represented as a real valued matrix in the eigenbasis of $\sigma_z$. Hence,  $\mathcal{T}_{S'}$ can be chosen as the transpose with respect to the eigenbasis of $\sigma_z$, and thus  $\mathcal{T}_{S'}\big(\overline{n}(y)\cdot\overline{\sigma}\big) = \overline{n}(y)\cdot\overline{\sigma}$, $\mathcal{T}_{S'}(H^{i}_{S'}) = H^{i}_{S'}$, and $\mathcal{T}_{S'}(H^{f}_{S'}) = H^{f}_{S'}$. We also get $\mathcal{T}(H) = H$. 

As measurement operators on $CE$ we choose projectors onto coherent states $Q_{CE}^{i+} := |\alpha_i\rangle\langle\alpha_i|$, $Q_{CE}^{f+} := |\alpha_f\rangle\langle\alpha_f|$, 
where the corresponding wave-functions are 
\begin{equation*}
 \psi_{\alpha}(y) = \frac{1}{(2\pi)^{1/4}}\frac{1}{\sqrt{\sigma}}e^{-\Imag(\alpha)^2}\exp[-\frac{1}{4}(\frac{y}{\sigma}-2\alpha)^2].
\end{equation*}
Here, $\sigma$ is the standard deviation, and $2\sigma\Real(\alpha)$ the expectation value, of the corresponding Gaussian distribution $|\psi_{\alpha}(y)|^2$. (With the coherent state defined as the displaced ground state of a harmonic oscillator, $\sigma$ is determined by the parameters of the chosen oscillator.) Similarly, $\hbar\Imag(\alpha)/\sigma$ is the average momentum of the coherent state. One can confirm that
$\mathcal{T}_{CE}(|\alpha\rangle\langle\alpha|)  = |\alpha^{*}\rangle\langle\alpha^{*}|$, and thus the time-reversal changes the sign of the momentum, but leaves the position intact.
For the spin degree of freedom we let the measurement operators be 
$Q_{S'}^{i+} = Q_{S'}^{f+} = \hat{1}_{S'}$.
We choose (somewhat arbitrarily) the parameters such that $\hbar^2/(ME_0y_0^2) = 0.1$, $\beta E_0 = 1$,  and such that the standard deviation is $\sigma = y_0/2$. This means that the typical thermal energy $kT$ is equal to the excitation energy of the spin (for $y\leq -y_0$), and the  width of the wave-packet is of the same order as the size of the interaction region.
Figure \ref{FigTestConditions} displays the numerical evaluation of the factorization error $d_{H_{1},H_{2}}^{H}(\hat{1}_{S'}, |\alpha\rangle\langle \alpha|)$ as defined in (\ref{Defdidf}). Here we choose $H_1:= \hat{K}$, and as $H_2$ we choose $E_0\sigma_z/2$ (red curve),  and $E_0\sigma_x/4$ (green curve), for $\alpha := r+2i$ with $r \in [-10,10]$.  Since the standard deviation is $\sigma = y_0/2$ it follows that the spatial wave-packet $\psi_{\alpha}$ is centered at $ry_0$. For the sake of comparison we also include the `local' approximation $H_2:=E_0\overline{n}(ry_0)\cdot\overline{\sigma}/2$ (blue curve), where for a state centered at the location $ry_0$ we approximate the total Hamiltonian $H$ with the non-interacting Hamiltonian $\hat{K} + E_0\overline{n}(ry_0)\cdot\overline{\sigma}/2$. By construction, the blue curve coincides with the green for $r\leq -1$, and with the red for $1\leq r$. Maybe unsurprisingly, the local approximation is better than the others inside the interaction region.

We now turn to the test of the approximate fluctuation relation, and for this purpose we choose an evolution time $t$ such that $tE_0/\hbar = 21.5$. (This particular choice happens to make the transition probability for the forward process fairly large.) In Fig.~\ref{FigDynamics} the final states of the particle in the forward and reverse processes are depicted, where we use the measurement operators $Q_{CE}^{i+} := |-4 +2i\rangle\langle -4+2i|$,  $Q_{CE}^{f+} := |4 +2i\rangle\langle 4+2i|$, and $Q_{S'}^{i+} = Q_{S'}^{f+} := \hat{1}_{S'}$.  The transition probabilities [in equation (\ref{jdslkfajd})] of the forward and reverse processes are $\mathcal{P}^{+} \approx 0.36$ and $\mathcal{P}^{-} \approx 0.39$, respectively. The error in the approximate fluctuation relation, as defined by the left hand side of equation (\ref{ApproxFluctQuantitative}) becomes $|Z_i\mathcal{P}^{+}- Z_f\mathcal{P}^{-}|\approx 1.6*10^{-8}$ where
$Z_i := Z_{\beta}(E_0\sigma_z/2) \mathcal{Z}_{\beta \hat{K}}(Q_{CE}^{i+})$ and $Z_f := Z_{\beta}(E_0\sigma_x/4) \mathcal{Z}_{\beta\hat{K}}(Q_{CE}^{f-})$. An estimate of the relative error is
$|Z_i\mathcal{P}^{+}- Z_f\mathcal{P}^{-}|/(|Z_i\mathcal{P}^{+}| + |Z_f\mathcal{P}^{-}|)\approx 2.2*10^{-8}$. One can also calculate the upper bound in the right hand side of  (\ref{ApproxFluctQuantitative}), which  becomes $d^{H}_{H_{S'}^{i},H_{CE}^{i}}(Q_{S'}^{i+},Q_{CE}^{i+})  + d^{H}_{H_{S'}^{f},H_{CE}^{f}}(Q_{S'}^{f+}, Q_{CE}^{f+})\approx 3.1*10^{-6}$ (where we use the fact that $\Vert Q_{S'}^{i+}\Vert =1$, $\Vert Q_{S'}^{f+} \Vert=1$, $\Vert Q_{CE}^{i+}\Vert =1$, $\Vert Q_{CE}^{f+}\Vert=1$).

\begin{figure}
 \includegraphics[width= 8cm]{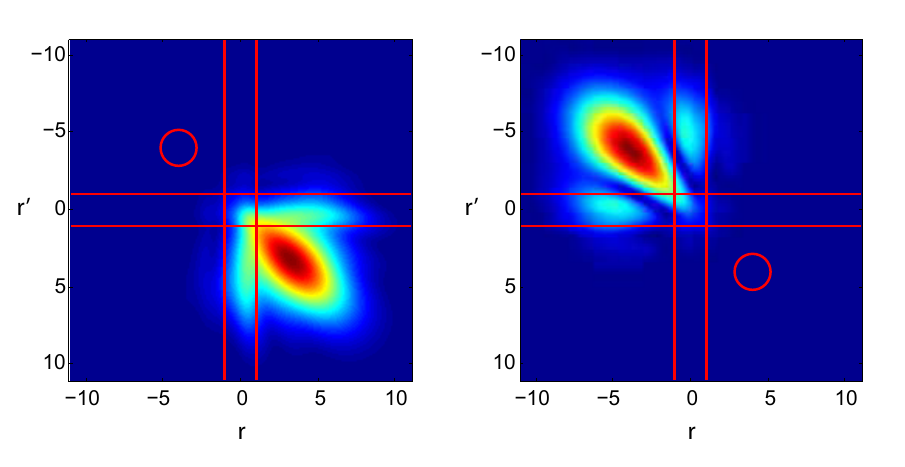} 
   \caption{\label{FigDynamics}  {\bf The forward and reverse processes.}   
The density operators of the combined control and energy reservoir particle in the position representation (the absolute values of the matrix elements) at the end of the forward (left) and reverse (right) process. The pairs of horizontal and vertical lines show the borders of  the interaction region. The red circles indicate the positions of the initial states $\mathcal{G}_{\beta \hat{K}}(Q_{CE}^{i+})$ and $\mathcal{G}_{\beta \hat{K}}(Q_{CE}^{f+})$ for the initial and final measurement operators $Q_{CE}^{i+} = |-4 +2i\rangle\langle -4+2i|$,  $Q_{CE}^{f+} = |4 +2i\rangle\langle 4+2i|$. \newline
The error in the approximate fluctuation relation is small (approximately $1.6*10^{-8}$), which corresponds to  the fact that the measurement operators that are well separated from the interaction region. One should compare this with the final wave packets, where one  can clearly see that these have significant weights within the interaction region. This illustrates the fact that it is the properties of the measurement operators, rather than the final states, that matter for the quality of the approximation. \newline 
One can also note that the final states are not mirror images of each other. Hence, the symmetry discussed in Appendix \ref{GlobalSymmetry} does not imply that the wave packets of the forward and reverse processes have to be symmetric images of each other.
}
\end{figure}

\subsubsection{\label{SecApproxFreeEnergyDiff} Approximate free energy differences}
One can use the setting of the joint control and energy reservoir to approximately evaluate the free energy difference between the final and initial Hamiltonians $H^{f}_{S'}$ and $H^{i}_{S'}$, respectively, or equivalently, the quotient of the partition functions $Z(H^{f}_{S'})/Z(H^{i}_{S'})$. 
Since $Q^{i\pm}_{S'} = Q^{f\pm}_{S'} = \hat{1}$, we obtain  $\mathcal{Z}_{\beta H^i_{S'}}(Q^{i}_{S'}) =  Z_{\beta}(H^{i}_{S'})$ and $\mathcal{Z}_{\beta H^f_{S'}}(Q^{f}_{S'}) =  Z_{\beta}(H^{f}_{S'})$. 
If the heat bath $B$ moreover is non-interacting with $S$ (or if $S = S'$), then $Z(H^{f}_{S'})/Z(H^{i}_{S'}) = Z(H^{f}_{S})/Z(H^{i}_{S})$, and (\ref{MainApproxFluctQuantitative}) can be rewritten as
\begin{equation}
\label{nsvdnsbslfk}
\begin{split}
 \frac{Z_{\beta}(H^f_{S})}{Z_{\beta}(H^i_{S})}\approx  & \frac{\mathcal{Z}_{\beta \hat{K}}(Q_{CE}^{i})}{\mathcal{Z}_{\beta \hat{K}}(Q_{CE}^{f})}\frac{P^{V}_{\beta H^i}[ \hat{1}\otimes Q_{CE}^{i+}\rightarrow \hat{1}\otimes Q_{CE}^{f+}]}{P^{V}_{\beta H^f}[ \hat{1}\otimes Q_{CE}^{f-}\rightarrow \hat{1}\otimes Q_{CE}^{i-}]}. 
\end{split}
\end{equation}
Hence, this enables us to approximately determine the free energy difference between the final and initial Hamiltonian.

For the very same setting as in the previous section, the true quotient is  $\frac{Z(H^f)}{Z(H^i)} = (e^{-1/4} + e^{1/4})/(e^{-1/2} +e^{1/2})\approx 0.91$. For the evaluation of the right hand side of  (\ref{nsvdnsbslfk}) we can use that $Q^{i\pm}_{CE}$ and $Q^{f\pm}_{CE}$  are space translations of each other, and thus $\mathcal{Z}_{\beta \hat{K}}(Q^{i\pm}_{CE}) = \mathcal{Z}_{\beta \hat{K}}(Q^{f\pm}_{CE})$, which thus cancel in (\ref{nsvdnsbslfk}).
Numerical evaluation yields a difference between the left and right hand side of  (\ref{nsvdnsbslfk}) that is approximately $4*10^{-8}$.

\section{\label{SecTurningTables} Assuming a global fluctuation relation}

\subsection{\label{SecComp}Comparisons}

Here we investigate the class of channels $\mathcal{F}$ that satisfy the relation
\begin{equation}
\label{ReductionHeatBathAgain}
\mathcal{F}\mathcal{J}_{\beta H} = \mathcal{J}_{\beta H}\mathcal{F}^{\ominus},
\end{equation}
i.e., the class of channels that satisfy the fluctuation relation  (\ref{ReductionHeatBath}) obtained in section \ref{SecMainTimeRevesalSymmetricThermalOperations}. More generally, we investigate the relations between this class, the thermal operations \cite{Janzing00,Janzing06,Horodecki11,Brandao13b,Gour,Brandao13,Renes14,Faist15,Lostaglio14b,Perry16,Scharlau16,Lostaglio16}, time-reversal symmetric thermal operations, as well as the Gibbs preserving maps \cite{Faist15}  (see Fig.~\ref{FigRelations}). For the sake of clarity we here state the definitions that we employ.
\begin{itemize}
\item A channel $\mathcal{F}$ on a Hilbert space $\mathcal{H}$ is a thermal operation with respect to $\beta\geq 0$, and   a Hermitian operator $H$ on $\mathcal{H}$, if here exists an ancillary Hilbert space $\mathcal{H}_B$,  a Hermitian operator $H_B$ on $\mathcal{H}_B$, and a unitary operator $U$ on $\mathcal{H}\otimes\mathcal{H}_B$, such that $[U,H\otimes \hat{1}_B+\hat{1}\otimes H_B] = 0$, and $\mathcal{F}(\rho) = \Tr_B(U[\rho\otimes G_{\beta}(H_B)]U^{\dagger})$ for all $\rho$ on $\mathcal{H}$. 

\item A channel $\mathcal{F}$ on $\mathcal{H}$ is a time-reversal symmetric thermal operation with respect to $\beta\geq 0$,  a Hermitian operator $H$ on  $\mathcal{H}$, and a time-reversal $\mathcal{T}$ on $\mathcal{H}$, if $\mathcal{T}(H) = H$, and if there exists an ancillary Hilbert space $\mathcal{H}_B$,  a Hermitian operator $H_B$ on $\mathcal{H}_B$, and a time-reversal $\mathcal{T}_B$ on $\mathcal{H}_B$, and a unitary operator $U$ on $\mathcal{H}\otimes\mathcal{H}_B$, such that $[U,H\otimes \hat{1}_B+\hat{1}\otimes H_B] = 0$, $\mathcal{T}_B(H_B) = H_B$, and $\mathcal{T}_{\mathrm{tot}}(U) = U$, where $\mathcal{T}_{\mathrm{tot}}:=\mathcal{T}\otimes\mathcal{T}_B$, and $\mathcal{F}(\rho) = \Tr_B(U[\rho\otimes G_{\beta}(H_B)]U^{\dagger})$ for all $\rho$ on 
$\mathcal{H}$. 

\item A channel $\mathcal{F}$ on $\mathcal{H}$ is Gibbs preserving with respect to $\beta\geq 0$,  and a Hermitian operator $H$ on  $\mathcal{H}$, if $\mathcal{F}\big(\mathcal{G}_{\beta}(H)\big) = \mathcal{G}_{\beta}(H)$.
\end{itemize}

By these definitions it is clear that we get a family of thermal operations, as well as a family of Gibbs preserving maps, for each choice of  $\beta$ and $H$ on a given Hilbert space. Similarly, we get a class of time-reversal symmetric thermal operations for each choice of $\beta$, $H$, and $\mathcal{T}$ such that $\mathcal{T}(H)=H$. Analogously we get a class of channels that satisfy (\ref{ReductionHeatBathAgain}), for each $\beta$, $H$, and $\mathcal{T}$ (where $\ominus$ is defined via $\mathcal{T}$). However, in this case we have the choice whether to additionally demand that $\mathcal{T}(H) = H$ or not.  In the setup of section \ref{SecMainTimeRevesalSymmetricThermalOperations}, the condition  $\mathcal{T}(H) = H$ is satisfied by assumption (since the starting point is time-reversal symmetric thermal operations). Moreover, for the applications of (\ref{ReductionHeatBathAgain}) in Appendices \ref{SecGlobalChnlIdeal} and \ref{GlobalFlctnThmApproximate} we assume that $\mathcal{T}(H) = H$ holds. (It is in any case a convenient assumption, since it implies that $\mathcal{T}\mathcal{J}_{\beta H} = \mathcal{J}_{\beta H}\mathcal{T}$.) Most of the results in this section are based on  (\ref{ReductionHeatBathAgain}) for $\mathcal{T}(H) = H$ (i.e. class F in Fig.~\ref{FigRelations}), with the exception of Lemma \ref{LemmaFixpointCHannel}, which holds true for every $\mathcal{T}$, regardless of its  relation to $H$.

\begin{figure}
 \includegraphics[width= 7cm]{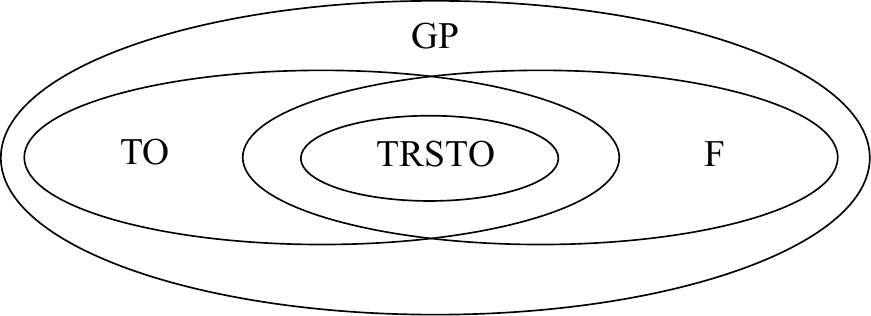} 
   \caption{\label{FigRelations} 
 {\bf Classes of channels.}  Schematic description of the relation between  the time-reversal symmetric thermal operations (TRSTO), the thermal operations (TO), the  Gibbs preserving maps (GP), and the set of channels (F) that satisfy the fluctuation relation  (\ref{ReductionHeatBathAgain}), for a given $\beta$, $H$, and $\mathcal{T}$ such that $\mathcal{T}(H) = H$.
 We know from section  \ref{SecMainTimeRevesalSymmetricThermalOperations} that $\mathrm{F}\supseteq \mathrm{TRSTO}$, and it is clear from the definitions that $\mathrm{TO}\supseteq \mathrm{TRSTO}$. Moreover $\mathrm{GP}\supseteq \mathrm{TO}$ \cite{Faist15}. In Lemma \ref{LemmaFixpointCHannel} it is shown that $\mathrm{GP}\supseteq \mathrm{F}$.
For specific choices of $\beta$, $H$, and $\mathcal{T}$, we moreover find explicit examples of channels in the sets $\mathrm{TO}\setminus \mathrm{F}$,  $\mathrm{F}\setminus \mathrm{TO}$, and $\mathrm{GP}\setminus \big(\mathrm{TO}\cup \mathrm{F}\big)$, i.e., there are cases where these sets are non-empty.  The drawing suggests that $(\mathrm{TO}\cap\mathrm{F})\setminus \mathrm{TRSTO}$ would be nonempty. However, it is not clear whether this is the case or not.
 In Appendix \ref{SecConvTimerevsymm} it is shown that TRSTO is convex. The set TO is also convex \cite{Lostaglio14b}, and one can realize that GP and F are convex.
}
\end{figure}

\subsubsection{\label{SecTimeRevSymmThrmOp}Time reversal symmetric thermal operations form a proper subset of thermal operations}

From the definitions in the previous section it is clear that all time-reversal symmetric thermal operations are thermal operations. 
Here we show that there exist  thermal operations that are not time-reversal symmetric thermal operations. We do this by proving a more general statement, namely that there exist thermal operations that do not satisfy (\ref{ReductionHeatBathAgain}) for any choice of time-reversal  $\mathcal{T}$ such that $\mathcal{T}(H) = H$. (This is more general since we know from section \ref{SecMainTimeRevesalSymmetricThermalOperations} that all time-reversal symmetric thermal operations satisfy (\ref{ReductionHeatBathAgain}).)

 We consider two non-interacting two-level systems in resonance, $H = H_B = E_0|0\rangle\langle 0| + E_{1}|1\rangle\langle 1|$, where we assume $E_1 > E_0$, and thus each of these systems is non-degenerate.  The joint Hamiltonian $H_{\textrm{tot}} = H\otimes \hat{1}_B + \hat{1}\otimes H_B$ has the three eigenspaces $\Sp\{|00\rangle\}$, $\Sp\{|01\rangle,|10\rangle\}$, and $\Sp\{|11\rangle\}$, and the most general energy conserving unitary operators on these two systems can be written
\begin{equation*}
\begin{split}
V = & e^{i\theta_0}|00\rangle\langle 00| + e^{i\theta_1}|11\rangle\langle 11|+ \left[\begin{matrix}|01\rangle & |10\rangle\end{matrix}\right] U \left[\begin{matrix} \langle 01|\\ \langle 10| \end{matrix}\right],
\end{split}
\end{equation*}
where $\theta_0$ and $\theta_1$ are arbitrary real numbers, and $U := \left[\begin{smallmatrix}U_{00} & U_{01}\\ U_{10} & U_{11}\end{smallmatrix}\right]$ is an arbitrary unitary $2\times 2$ matrix. We write the Gibbs state on the heat bath as
$G_{\beta}(H_B) = \lambda|0\rangle\langle 0| + (1-\lambda)|1\rangle\langle 1|$,
 where we for the sake of notational convenience have introduced $\lambda: = e^{-\beta E_0}/(e^{-\beta E_0}+ e^{-\beta E_1})$.
We can thus construct the following thermal operations on $S$ 
\begin{equation*}
\begin{split}
&\mathcal{F}(\rho) =  \Tr_{B}\big(V[\rho\otimes G_{\beta}(H_B)]V^{\dagger}\big)\\
& =  (1-\lambda) |U_{10}|^2 \langle 0|\rho|0\rangle |1\rangle\langle 1| \\
& + \lambda  |U_{01}|^2\langle 1|\rho|1\rangle|0\rangle\langle 0|\\
& + (1-\lambda) (e^{i\theta_1}|1\rangle\langle 1| +  U_{00}|0\rangle\langle 0|)\rho (e^{-i\theta_1}|1\rangle\langle 1| + U_{00}^{*}|0\rangle \langle 0|)\\
& + \lambda  (e^{i\theta_0}|0\rangle\langle 0|  + U_{11}|1\rangle  \langle 1|)\rho (e^{-i\theta_0}|0\rangle\langle 0|  + U_{11}^{*} |1\rangle \langle 1|).
\end{split}
\end{equation*}
Let us now consider any time-reversal $\mathcal{T}$ such that $\mathcal{T}(H) = H$. Since $H$ is assumed to be non-degenerate, it follows by Lemma \ref{TOnHermitean} that we must have $\mathcal{T}(|0\rangle\langle 0|) = |0\rangle\langle 0|$ and $\mathcal{T}(|1\rangle\langle 1|) = |1\rangle\langle 1|$.  
We can use this to show that 
\begin{equation*}
\begin{split}
& \mathcal{F}\mathcal{J}_{\beta H}(\rho)  - \mathcal{J}_{\beta H}\mathcal{F}^{\ominus}(\rho) \\
& = \lambda (1-\lambda) Z(H_B) (|U_{10}|^2-|U_{01}|^2)\\
& \quad \quad \times\big(\langle 0|\rho|0\rangle |1\rangle\langle 1| -\langle 1|\rho|1\rangle|0\rangle\langle 0|\big).
\end{split}
\end{equation*}
Hence, if $|U_{10}| \neq |U_{10}|$, then the thermal operation $\mathcal{F}$ does not  satisfy (\ref{ReductionHeatBathAgain}) for any choice of $\mathcal{T}$ such that $\mathcal{T}(H) = H$ (thus providing an element in the set $\mathrm{TO}\setminus\mathrm{F}$ described in Fig.~\ref{FigRelations}). 
From the discussion in section \ref{SecMainTimeRevesalSymmetricThermalOperations} it follows that every time-reversal symmetric thermal operation with respect to $\beta$, $H$, $\mathcal{T}$  with  $\mathcal{T}(H) = H$ has to satisfy (\ref{ReductionHeatBathAgain}). We can thus conclude that there exist thermal operations that are not time-reversal symmetric thermal operations.

\subsubsection{\label{SecExampleBeyondThermal}
There exist channels that satisfy (\ref{ReductionHeatBathAgain}), but are not thermal operations}

As pointed out in section \ref{SecMainTimeRevesalSymmetricThermalOperations}, all time-reversal symmetric thermal operations satisfy (\ref{ReductionHeatBathAgain}), where we by definition of the time-reversal symmetric thermal operations must have $\mathcal{T}(H) = H$. An immediate question is if  all channels that satisfy  (\ref{ReductionHeatBathAgain}) for $\mathcal{T}(H) = H$ also are time-reversal symmetric thermal operations. 
Here we show that there exist channels that satisfy (\ref{ReductionHeatBathAgain}) for $\beta$, $H$, and $\mathcal{T}$ with $\mathcal{T}(H) = H$, but that are not  thermal operations (and thus not time-reversal symmetric thermal operations).
The counterexample is based on the fact that thermal operations cannot map states that are diagonal in the energy eigenbasis of $H$, to states that have off-diagonal elements \cite{Lostaglio14b}.

This example moreover shows that  (\ref{ReductionHeatBathAgain}) admits channels that are not enhanced (or generalized) thermal operations \cite{Cwiklinski15}. The latter are Gibbs preserving channels $\mathcal{F}$ that also 
satisfy  $\mathcal{F}([H,\rho]) = [H,\mathcal{F}(\rho)]$ for all operators $\rho$, i.e., they are time-translation symmetric \cite{Lostaglio14b}. By the same reasoning as in Appendix \ref{DecouplingDiagonals}, the channels $\mathcal{F}$ decouple along the modes of coherence, and thus cannot create off-diagonal elements from diagonal inputs, for non-degenerate $H$.

Select some $\beta\geq 0$, Hermitian operator $H$, and time-reversal $\mathcal{T}$ on a Hilbert space $\mathcal{H}$, such that $\mathcal{T}(H) = H$.
Let $\{A_k\}_k$ be a POVM, and define the channel
\begin{equation}
\label{QCQ}
\mathcal{F}(\rho) = \sum_k \mathcal{G}_{\beta H}\big(\mathcal{T}(A_k)\big)\Tr(A_k\rho).
\end{equation}
This type of channels can be implemented by a measurement of the POVM $\{A_k\}_k$ on the input state $\rho$, followed by a preparation of state $\mathcal{G}_{\beta H}\big(\mathcal{T}(A_k)\big)$ conditioned on outcome $k$. (These are  `entanglement breaking' channels \cite{Holevo,Horodecki03,Ruskai03}).
By using the condition $\mathcal{T}(H) = H$, one can confirm that all channels (\ref{QCQ}) satisfy (\ref{ReductionHeatBathAgain}).
Let us consider the qubit case with $H = E_0|0\rangle\langle 0|+ E_1|1\rangle\langle1|$, with $E_1 > E_0$, and the channel (\ref{QCQ}) for the two-element POVM $\{A_1,A_2\}$ with $A_2 = \hat{1}-A_1$, where
\begin{equation*}
\begin{split}
A_1 := &  \frac{1}{2}(1+r)|0\rangle\langle 0| +\frac{1}{2}(1-r)|1\rangle\langle 1|\\ 
        & + \frac{1}{2}\eta\sqrt{1-r^2}(e^{i\theta}|0\rangle\langle 1| + e^{-i\theta}|1\rangle\langle 0|),
\end{split}
\end{equation*}
for $-1\leq r\leq 1$ and $-1\leq\eta\leq 1$.
One can confirm that the resulting channel $\mathcal{F}$ satisfies
\begin{equation*}
\begin{split}
& \langle 0|\mathcal{F}(|0\rangle\langle 0|)|1\rangle \\
 & = \frac{2\eta r\sqrt{1-r^2} e^{-\beta E_1}e^{-\beta (E_0+E_1)/2}e^{-i\theta}}{
[(e^{-\beta E_0} + e^{-\beta E_1})^2      - r^2(e^{-\beta E_0}-e^{-\beta E_1})^2  ]
},
\end{split}
\end{equation*}
and thus maps a diagonal state to a non-diagonal state for suitable choices of $r$ and $\eta$. Hence, $\mathcal{F}$ generally cannot be a thermal operation, or enhanced thermal operation. We have thus constructed examples of channels that satisfy (\ref{ReductionHeatBathAgain}) for $\mathcal{T}(H) = H$, but which are not thermal operations (and thus are elements of the set $\mathrm{F}\setminus\mathrm{TO}$ described in Fig.~\ref{FigRelations}). Consequently, these channels are also not time-reversal symmetric thermal operations.

\subsubsection{\label{SecImplyGibbs}  (\ref{ReductionHeatBathAgain}) versus Gibbs preservation}

The following lemma shows that every channel that satisfies (\ref{ReductionHeatBathAgain}) is a Gibbs preserving map. For general discussions on Gibbs preserving maps, see \cite{Faist15}.
\begin{Lemma}
\label{LemmaFixpointCHannel}
Let $\mathcal{F}$ be a channel that satisfies  (\ref{ReductionHeatBathAgain}) with respect to $\beta\geq 0$ and Hermitian operator $H$ (and an arbitrary time reversal $\mathcal{T}$). Then $\mathcal{F}(e^{-\beta H}) = e^{-\beta H}$.
Hence, if $Z_{\beta}(H)$ is finite, then  $\mathcal{F}\big(G_{\beta}(H)\big) = G_{\beta}(H)$.
\end{Lemma}

\begin{proof}
Since $\mathcal{F}$ is a channel, and thus by definition is trace preserving, it follows that $\mathcal{F}^{*}(\hat{1}) =\hat{1}$, and thus $\mathcal{F}^{\ominus}(\hat{1}) = \hat{1}$. If we apply both sides of  (\ref{ReductionHeatBathAgain}) to $\hat{1}$, then we thus obtain $\mathcal{F}(e^{-\beta H}) = e^{-\beta H}$.
\end{proof}

Since all channels that satisfy  (\ref{ReductionHeatBathAgain}) are Gibbs preserving, one may wonder whether all Gibbs preserving channels satisfy  (\ref{ReductionHeatBathAgain}). However, this is not the case, at least if one restricts to the class for which $\mathcal{T}(H) = H$. In order to demonstrate this we make use of the class of single qubit channels used in \cite{Faist15} to show that not all Gibbs preserving maps are thermal operations.
We consider a single qubit with Hamiltonian $H = E_0|0\rangle\langle 0| + E_1|1\rangle\langle 1|$, with $E_1 > E_0$, resulting in the Gibbs state $G(H) = \lambda|0\rangle\langle 0| + (1-\lambda)|1\rangle\langle 1|$, with $\lambda := e^{-\beta E_0}/(e^{-\beta E_0} +e^{-\beta E_1})$, and $\lambda > 1-\lambda$. For any qubit density operator $\eta$, it was shown in \cite{Faist15} that the map
$\mathcal{F}(\rho) := \langle 0|\rho|0\rangle \frac{1}{\lambda}\big(G(H)  -(1-\lambda)\eta\big) + \langle 1|\rho|1\rangle \eta$
is a Gibbs preserving channel.  Since we assume $\mathcal{T}(H) = H$, for a non-degenerate $H$, it follows by Lemma \ref{TOnHermitean} that $\mathcal{T}(|0\rangle\langle 0|) = |0\rangle \langle 0|$ and  $\mathcal{T}(|1\rangle\langle 1|) = |1\rangle \langle 1|$. One can use this to show that
$ \langle 0|[\mathcal{F}\mathcal{J}(\rho)  - \mathcal{J}\mathcal{F}^{\ominus}(\rho)]|1\rangle
=  (1-\lambda)Z(H)(\langle 1|\rho|1\rangle - \langle 0|\rho|0\rangle) \langle 0|\eta|1\rangle$,
and thus $\mathcal{F}$ does not satisfy (\ref{ReductionHeatBathAgain}) whenever $\langle 0|\eta|1\rangle \neq 0$  and $\lambda\neq 1$. Moreover, under these conditions $\mathcal{F}$ also fails to be a thermal operation \cite{Faist15}. 
Hence, we have an example of a Gibbs preserving map that neither is a thermal operation, nor belongs to the class of  operations that satisfy  (\ref{ReductionHeatBathAgain}) for $\mathcal{T}(H) = H$. (Thus, this example is an element of the set $\mathrm{GP}\setminus(\mathrm{TO}\cup\mathrm{F})$ described in Fig.~\ref{FigRelations}.)

\subsubsection{\label{SecConvTimerevsymm}Time-reversal symmetric thermal operations form a convex set}

Here we show that the time-reversal symmetric thermal operations, with respect to given $\beta$, $H$, and $\mathcal{T}$, form a convex set.  Note that a proof of the convexity of the set of general thermal operations can be found in Appendix C of \cite{Lostaglio14b}. One can also realize that the set of Gibbs preserving maps with respect to a given $H$ and $\beta$ is a convex set, as well as the set of channels that satisfy $\mathcal{F}\mathcal{J} = \mathcal{J}\mathcal{F}^{\ominus}$ with respect to given $\beta$, $H$, and $\mathcal{T}$ (irrespective of whether $\mathcal{T}(H) = H$ or not).
\begin{Lemma}
\label{TimeRevSymmThermalConv}
Suppose that $\mathcal{F}_0$ and $\mathcal{F}_1$ are time-reversal symmetric thermal operations with respect to the Hamiltonian $H$, time-reversal $\mathcal{T}$, and $\beta \geq 0$ (which by definition requires that $\mathcal{T}(H) = H$). Then $\lambda \mathcal{F}_0 + (1-\lambda)\mathcal{F}_1$,  for $0\leq \lambda\leq 1$,  is also a time-reversal symmetric thermal operation with respect to $H$, $\mathcal{T}$,  and $\beta$.
\end{Lemma}
The proof below is based on a two-dimensional ancillary system with a Hamiltonian that is adapted to the weights of the convex combination. An alternative, along the lines of the above mentioned proof in Appendix C of \cite{Lostaglio14b}, would be to instead consider a completely degenerate Hamiltonian on a sufficiently large ancillary Hilbert space, and let the unitary evolution yield the (arbitrarily well approximated) weights. Such an alternative may potentially be useful in the context of the issues discussed in Appendix \ref{SecCaution}.

\begin{proof}
Since  $\mathcal{F}_0$ and $\mathcal{F}_1$ are time-reversal symmetric thermal operations, there
 exist Hilbert spaces $\mathcal{H}_{B0}$ and $\mathcal{H}_{B1}$ and Hamiltonians $H_{B0}$ and $H_{B1}$, and time-reversals $\mathcal{T}_{B0}$ and $\mathcal{T}_{B1}$ such that $\mathcal{T}_{B0}(H_{B0}) = H_{B0}$ and $\mathcal{T}_{B1}(H_{B1}) = H_{B1}$. Moreover, there exist unitaries $V_0$ and $V_1$ with $[V_0,H\otimes \hat{1}_{B0} + \hat{1}\otimes H_{B0}] = 0$ and  $[V_1,H\otimes \hat{1}_{B1} + \hat{1}\otimes H_{B1}] = 0$. Furthermore, $[\mathcal{T}\otimes\mathcal{T}_{B0}](V_0) = V_0$, $[\mathcal{T}\otimes\mathcal{T}_{B1}](V_1) = V_1$, as well as  $\mathcal{F}_0(\rho) =  \Tr_{B0}(V_0[\rho\otimes G_{\beta}(H_{B0})]V_0^{\dagger})$ and 
 $\mathcal{F}_1(\rho) =  \Tr_{B1}(V_1[\rho\otimes G_{\beta}(H_{B1})]V_1^{\dagger})$.
Let $\mathcal{H}_{\tilde{B}}$ be a two-dimensional space with orthonormal basis $\{|0\rangle,|1\rangle\}$, equipped with the Hamiltonian 
\begin{equation*}
H_{\tilde{B}} := E|0\rangle\langle 0|  +\Big(E+\frac{1}{\beta}\ln\frac{\lambda}{1-\lambda}\Big)|1\rangle\langle 1|,
\end{equation*}
with $E$ being some arbitrary reference energy.
The Hamiltonian $H_{\tilde{B}}$ is constructed in such a way that $G_{\beta}(H_{\tilde{B}}) = \lambda|0\rangle\langle 0|  +(1-\lambda)|1\rangle\langle 1|$.
We furthermore define $\mathcal{T}_{\tilde{H}}$ as the transpose with respect to $\{|0\rangle,|1\rangle\}$. On $\mathcal{H}_{B'}:= \mathcal{H}_{B0}\otimes\mathcal{H}_{B1}\otimes\mathcal{H}_{\tilde{B}}$ we define
\begin{equation*}
\begin{split}
H_{B'} := & H_{B0}\otimes \hat{1}_{B1}\otimes\hat{1}_{\tilde{B}} +  \hat{1}_{B0}\otimes H_{B1}\otimes \hat{1}_{\tilde{B}} \\
&  + \hat{1}_{B0}\otimes \hat{1}_{B1}\otimes H_{\tilde{B}}
\end{split}
\end{equation*}
and $\mathcal{T}_{B'}:=\mathcal{T}_{B0}\otimes\mathcal{T}_{B1}\otimes\mathcal{T}_{\tilde{B}}$.
We let the global Hamiltonian be $H_{\mathrm{tot}}  :=   H\otimes \hat{1}_{B0}\otimes \hat{1}_{B1}\otimes\hat{1}_{\tilde{B}} + \hat{1}\otimes H_{B'}$, and the global unitary be $V := V_0\otimes \hat{1}_{B1}\otimes |0\rangle_{\tilde{B}}\langle 0| + V_1\otimes \hat{1}_{B0}\otimes |1\rangle_{\tilde{B}}\langle 1|$.  
One can confirm that $\mathcal{T}_{B'}(H_{B'}) = H_{B'}$,
 $[\mathcal{T}\otimes\mathcal{T}_{B'}](H_{\mathrm{tot}}) = H_{\mathrm{tot}}$, $[\mathcal{T}\otimes\mathcal{T}_{B'}](V) = V$, and $[H_{\mathrm{tot}},V] = 0$, as well as $\Tr_{B'}(V[\rho\otimes G(H_{B'})]V^{\dagger}) = \lambda \mathcal{F}_0(\rho) + (1-\lambda)\mathcal{F}_1(\rho)$. Hence, the convex combination is a time-reversal symmetric thermal operation. 
\end{proof}

\subsection{\label{SecGlobalChnlIdeal}Conditional fluctuation relations again}

Here we consider the counterpart of the conditional fluctuation relations in Appendix \ref{SecCondWithTimeRev}, but where we,  instead of a time-reversal symmetric energy conserving global unitary evolution with a heat bath, assume that the evolution on $SCE$ is determined by  channels that satisfy (\ref{ReductionHeatBath}). Analogous to Appendix \ref{SecApproximateFluct}, here we use the `anonymous' division of $SCE$ into two subsystems $1$ and $2$.
\begin{Assumptions}
\label{IdealChannelDef}
Let $\mathcal{H}_{1}$ and $\mathcal{H}_2$ be complex Hilbert spaces. Let $\mathcal{T}_{1}$ and $\mathcal{T}_2$ be time-reversals on $1$ and $2$, respectively, and let $\mathcal{T}_{SCE} := \mathcal{T}_{1}\otimes\mathcal{T}_2$.
\begin{itemize}
\item Let $H_{1}$ and $H_2$ be Hermitian operators on $\mathcal{H}_{1}$ and $\mathcal{H}_2$, respectively, and let
\begin{equation}
H_{SCE} := H_{1}\otimes \hat{1}_2 + \hat{1}_1\otimes H_{2}.
\end{equation}
\item Let $\mathcal{F}$ be a channel  such that 
\begin{equation}
\label{flbmmbdf}
\mathcal{F}\mathcal{J}_{\beta H_{SCE}}=  \mathcal{J}_{\beta H_{SCE}}\mathcal{F}^{\ominus}, 
\end{equation}
where $\ominus$ is defined with respect to $\mathcal{T}_{SCE}$.
\item Let $Q^{i+}_{2}$ and $Q^{f+}_{2}$ be operators on $\mathcal{H}_{2}$ such that $0\leq Q^{i+}_{2}\leq \hat{1}_{2}$ and $0\leq Q^{f+}_{2}\leq \hat{1}_{2}$.
\item Let $\mathcal{T}_1(H_1) = H_1$, $\mathcal{T}_{2}(H_{2}) = H_{2}$. Define $Q^{i-}_{2} := \mathcal{T}_{2}(Q^{i+}_{2})$ and $Q^{f-}_{2} := \mathcal{T}_{2}(Q^{f+}_{2})$.
\end{itemize}
\end{Assumptions}
One may note that the above assumptions imply that $\mathcal{T}_{SCE}(H_{SCE}) = H_{SCE}$.

With a channel $\mathcal{F}$ on the joint system $12$, we define the following CPMs on system $1$
\begin{equation}
\label{mdsflfmsb}
\begin{split}
\tilde{\mathcal{F}}_{+}(\sigma) = & \Tr_{2}\Big([\hat{1}_{1}\otimes Q^{f+}_{2}]\mathcal{F}\big(\sigma\otimes\mathcal{G}_{\beta H_{2}}(Q^{i+}_2) \big)\Big),\\
\tilde{\mathcal{F}}_{-}(\sigma) = & \Tr_{2}\Big([\hat{1}_{1}\otimes Q^{i-}_{2}]\mathcal{F}\big(\sigma\otimes \mathcal{G}_{\beta H_{2}}(Q^{f-}_2) \big)\Big).
\end{split}
\end{equation}

The following is the counterpart of Proposition \ref{PropConditionalFluctuationThm}.
\begin{Proposition}
\label{PropIdealGlbChnl}
With Assumptions \ref{IdealChannelDef}, the CPMs $\tilde{\mathcal{F}}_{+}$ and $\tilde{\mathcal{F}}_{-}$ as defined in Eq.~(\ref{mdsflfmsb}) are related as 
\begin{equation}
\label{fbfbbfs}
\mathcal{Z}_{\beta H_{2}}(Q^{i}_{2})  \tilde{\mathcal{F}}_{+} \mathcal{J}_{\beta H_{1}} = \mathcal{Z}_{\beta H_{2}}(Q^{f}_{2})\mathcal{J}_{\beta H_{1}}\tilde{\mathcal{F}}^{\ominus}_{-}.
\end{equation}
\end{Proposition}
Here we again have made use of Lemma \ref{LemmaSimpleRelations} to define $\mathcal{Z}_{\beta H_{2}}(Q_{2}^{i}):=  \mathcal{Z}_{\beta H_{2}}(Q_{2}^{i+})=  \mathcal{Z}_{\beta H_{2}}(Q_{2}^{i-})$ and $\mathcal{Z}_{\beta H_{2}}(Q_{2}^{f}):=  \mathcal{Z}_{\beta H_{2}}(Q_{2}^{f+})=  \mathcal{Z}_{\beta H_{2}}(Q_{2}^{f-})$. The proof of Proposition \ref{PropIdealGlbChnl} is obtained if one first observes the relation
\begin{equation*}
\begin{split}
\mathcal{Z}_{\beta H_{2}}(Q^{f}_2)\tilde{\mathcal{F}}_{-}^{\ominus}(Q) 
= &  \Tr_{2}\big( [\hat{1}_1\otimes \mathcal{J}_{\beta H_{2}}(Q^{f+}_2)]\mathcal{F}^{\ominus}(Q\otimes Q^{i+}_{2})\big)\\
\end{split}
\end{equation*}
and uses this together with (\ref{flbmmbdf}), the definition (\ref{mdsflfmsb}), and $\mathcal{J}_{\beta H_{SCE}} = \mathcal{J}_{\beta H_1}\otimes\mathcal{J}_{\beta H_2}$.

\subsection{\label{GlobalFlctnThmApproximate}An approximate version}
Analogous to what we did in Appendix \ref{SecApproximateFluct} we can use the assumption of the global fluctuation relation (\ref{ReductionHeatBath}) in order to derive approximate fluctuation relations.

\begin{Proposition}
\label{GlobalApprox}
Let $H^{i}_1$, $H^{f}_1$ be Hermitian operators on the complex Hilbert space $\mathcal{H}_1$ with a time reversal $\mathcal{T}_1$, and let $H^{i}_2$, $H^{f}_2$ be Hermitian operators on the complex Hilbert space $\mathcal{H}_2$ with time-reversal $\mathcal{T}_2$. $H$ is a Hermitian operator on $\mathcal{H}_1\otimes\mathcal{H}_2$. Let $\mathcal{T} := \mathcal{T}_1\otimes\mathcal{T}_2$, $\mathcal{T}_1$, $\mathcal{T}_2$ satisfy the conditions (\ref{CondHandT}). Let $\mathcal{F}$ be a channel on $\mathcal{H}_1\otimes\mathcal{H}_2$ such that $\mathcal{F}\mathcal{J}_{\beta H} = \mathcal{J}_{\beta H}\mathcal{F}^{\ominus}$. Then 
\begin{equation}
\label{GlobalApproxFlctnRel}
\begin{split}
\Big\vert & \mathcal{Z}_{\beta H^i_1}(Q^{i}_1)\mathcal{Z}_{\beta H^i_2}(Q^{i}_2)P^{\mathcal{F}}_{\beta H^i}[Q^{i+}_1 \otimes Q^{i+}_2\rightarrow Q^{f+}_1 \otimes Q^{f+}_2]\\
& -  \mathcal{Z}_{\beta H^f_1}(Q^{f}_1)\mathcal{Z}_{\beta H^f_2}(Q^{f}_2)P^{\mathcal{F}}_{\beta H^f}[Q^{f-}_1 \otimes Q^{f-}_2\rightarrow Q^{i-}_1 \otimes Q^{i-}_2]\Big\vert \\ 
& \leq  \Vert Q^{f+}_1\Vert \Vert Q^{f+}_2\Vert d^{H}_{H^i_1,H^i_2}(Q^{i+}_1, Q^{i+}_2)\\
& \quad +  \Vert Q^{i+}_1\Vert  \Vert Q^{i+}_2\Vert d^{H}_{H^f_1,H^f_2}(Q^{f+}_1, Q^{f+}_2),
\end{split}
\end{equation}
where $H^{i}:= H^{i}_1\otimes \hat{1}_2 + \hat{1}_1\otimes H^{i}_2$, and $H^{f}:= H^{f}_1\otimes \hat{1}_2 + \hat{1}_1\otimes H^{f}_2$, and where $Q^{i-}_1 :=\mathcal{T}_1(Q^{i+}_1)$, $Q^{i-}_2 :=\mathcal{T}_2(Q^{i+}_2)$, $Q^{f-}_1 :=\mathcal{T}_1(Q^{f+}_1)$, $Q^{f-}_2 :=\mathcal{T}_2(Q^{f+}_2)$, and where $d^{H}_{H_1,H_2}$ is as defined in (\ref{Defdidf}).
\end{Proposition}
The proof proceeds analogously to that of Proposition \ref{ApproxFlctnRel} in Appendix \ref{SecApproximateFluct}, but with the unitary channel $V\cdot V^{\dagger}$ substituted with the channel $\mathcal{F}$. The analogue of (\ref{vnaknd}) is proved via (\ref{flbmmbdf}). For the further derivations one can use the general relation $\Vert \mathcal{F}(X)\Vert_1 \leq \Vert X\Vert_1$ for channels $\mathcal{F}$, which can be found in \cite{Raginsky02}.

\section{\label{FlctnThrmsForMstrEqs}Implicit heat baths: Fluctuation relations for Markovian master equations}

In Appendix \ref{SecTurningTables} we introduced the generalization where we model the effect of the heat bath by assuming that the induced quantum channel satisfies a global fluctuation relation. In this section we re-express this generalization in terms of master equations.

Suppose that we have a (sufficiently smooth) family of completely positive maps $\{\mathcal{F}_t\}_{t\geq 0}$ that satisfies the master equation
$\frac{d}{dt}\mathcal{F}_t = \mathcal{L}\mathcal{F}_t$, $\mathcal{F}_{0} = \mathcal{I}$,
for some generator $\mathcal{L}$. 
If each $\mathcal{F}_t = e^{t\mathcal{L}}$  satisfies $\mathcal{F}_t\mathcal{J}_{\beta H}=  \mathcal{J}_{\beta H}\mathcal{F}_t^{\ominus}$,
 then, by differentiation at $t =0$, it follows that 
\begin{equation}
\label{GeneratorRelation}
\mathcal{L}\mathcal{J}_{\beta H}=  \mathcal{J}_{\beta H}\mathcal{L}^{\ominus}.
\end{equation} 
 Vice versa, if $\mathcal{L}$ satisfies (\ref{GeneratorRelation}), then the family of channels $\mathcal{F}_{t} = e^{t\mathcal{L}}$ satisfies $\mathcal{F}_t\mathcal{J}_{\beta H}=  \mathcal{J}_{\beta H}\mathcal{F}_t^{\ominus}$, which can be shown by repeatedly applying (\ref{GeneratorRelation}) to the components in the Taylor expansion of the exponential.  
Hence, if a generator satisfies (\ref{GeneratorRelation}), then we can apply the fluctuation relations developed in Appendix  \ref{SecTurningTables}  to the resulting channels $\mathcal{F}_{t} = e^{t\mathcal{L}}$.

\subsection{\label{ExamplesSatisfyingCond}Examples}
Here we consider a few examples of generators that satisfy the condition  (\ref{GeneratorRelation}).

\subsubsection{\label{SecModelForThrmls}A model of thermalization}

Assume that $\mathcal{T}$ is the transpose with respect to an energy eigenbasis $\{|k\rangle\}_k$ of $H$, with corresponding eigenvalues $E_k$.
Assume that the Lindbladians in (\ref{nvsfklbnmain}) are given by 
\begin{equation}
\label{ThermalizingLindbladian}
L_{k',k} = \sqrt{r(k'|k)}|k'\rangle \langle k|,
\end{equation}
where $r(k'|k)\geq 0$, thus resulting in the generator 
\begin{equation}
\label{dvmvmv}
\begin{split}
\mathcal{L}(\rho) =&  -\frac{i}{\hbar}[H,\rho] \\
& +  \sum_{k',k}r(k'|k) |k'\rangle \langle k|\rho |k\rangle\langle k'|  \\
& -\frac{1}{2}\sum_{k',k}r(k'|k)|k\rangle\langle k|\rho  -\frac{1}{2}\sum_{k',k}r(k'|k)\rho|k\rangle\langle k|.
\end{split}
\end{equation}
This is a special case of  the Davies generators \cite{Davies74,Roga10,Temme13}, where one may note that the evolution of the diagonal elements $p_{l} := \langle l|\rho|l\rangle$ are decoupled from the off-diagonal elements, and satisfy the classical master equation $\frac{d}{dt}p_{l} =  \sum_{k}r(l|k)p_k  -\sum_{k}r(k|l)p_l$.

 One can furthermore show that 
\begin{equation*}
\begin{split}
& \mathcal{L}\mathcal{J}_{\beta H}(Q)-  \mathcal{J}_{\beta H}
\mathcal{L}^{\ominus}(Q) \\
& =        \sum_{k',k} \big[ r(k'|k)e^{-\beta E_{k}}-  r(k|k')  e^{-\beta E_{k'}}  \big] |k'\rangle \langle k|Q |k\rangle \langle k'|.
\end{split}
\end{equation*}
Hence, if the rates $r(k'|k)$ of the classical master equation satisfy detailed balance,
\begin{equation}
\label{ClassicalDetailedBalance}
r(k'|k)e^{-\beta E_{k}} =  r(k|k')e^{-\beta E_{k'}},
\end{equation}
then $\mathcal{L}$ satisfies (\ref{GeneratorRelation}).

A special case of (\ref{ThermalizingLindbladian}), with $r(k'|k) = rG_{\beta}(H)_{k'}$ for all $k',k$ yields 
\begin{equation}
\label{nsdlvvd}
\mathcal{L}(\rho) =  -\frac{i}{\hbar}[H,\rho]  +r G_{\beta}(H)\Tr(\rho)- r \rho. 
\end{equation}

\subsubsection{\label{SecModelDecoherence}A model of decoherence}

Suppose that we have a Hamilton operator $H$ such that $\mathcal{T}(H) = H$, and a collection of observables $D_k$ that commute with $H$, and are time-reversal symmetric, i.e., $D_k^{\dagger} = D_k$, $[H,D_k] = 0$, and $\mathcal{T}(D_k) = D_k$.  If we take $D_k$ as the Lindblad operators in (\ref{nvsfklbnmain}), then for a non-degenerate  $H$, the resulting master equation is such that the off-diagonal elements of the density operator in the energy eigenbasis decay exponentially, while the diagonal elements remain invariant. 
Moreover it is the case that $\mathcal{L}^{\ominus} = \mathcal{L}$ and $[\mathcal{J}_{\beta H},\mathcal{L}] = 0$, which implies that (\ref{GeneratorRelation}) is satisfied.
A special case is a single spin-half particle with Hamiltonian $H := \frac{1}{2}E\sigma_z$, and a single Lindblad operator $D = \sqrt{r}\sigma_z$ for some $r\geq 0$.  In this case the generator takes the form $\mathcal{L}(\rho) = -i[H,\rho]/\hbar + r\sigma_z\rho \sigma_z -r\rho$.

\subsubsection{Thermalizing harmonic oscillator}

Consider a bosonic mode with annihilation and creation operators $a$ and $a^{\dagger}$, and Hamiltonian $H :=\hbar\omega(a^{\dagger}a + \frac{1}{2}\hat{1})$ for $\omega>0$. One can model the thermalization of this mode by the following generator (see e.g.~\cite{Ford96})  with $\Gamma>0$
\begin{equation}
\label{mgklbklfgm}
\begin{split}
\mathcal{L}(\rho) = & -\frac{i}{\hbar}[\hbar\omega a^{\dagger}a,\rho]\\
&  + \Gamma(1+n_B)(a\rho a^{\dagger} -\frac{1}{2}a^{\dagger}a\rho -\frac{1}{2}\rho a^{\dagger}a) \\
& +\Gamma n_B (a^{\dagger}\rho a -\frac{1}{2}aa^{\dagger}\rho -\frac{1}{2}\rho aa^{\dagger}),
\end{split}
\end{equation}
where $n_B = \frac{1}{e^{\beta\hbar\omega}-1}$ is the average number of quanta in the thermal state of the mode. Hence, in this case we have the two Lindblad operators $L_1 := \sqrt{\Gamma(1+n_B)}a$ and $L_2 := \sqrt{\Gamma n_B}a^{\dagger}$.

If we assume that $\mathcal{T}$ is the transpose with respect to the number basis of the mode, then $\mathcal{T}(a) = a^{\dagger}$, and one can confirm that (\ref{mgklbklfgm}) satisfies (\ref{GeneratorRelation}), where the relations 
$ae^{-\beta \omega a^{\dagger}a/2} =  e^{-\beta \omega/2} e^{-\beta\omega a^{\dagger}a/2}a$, 
$e^{-\beta \omega a^{\dagger}a/2} a^{\dagger} =  e^{-\beta \omega/2}a^{\dagger} e^{-\beta\omega a^{\dagger}a/2}$, and $e^{-\beta\hbar\omega} = n_B/(1 + n_B)$ are useful.

\subsection{\label{SecAssemblingGenerators} Composing generators that satisfy (\ref{GeneratorRelationMain})}

In Appendix \ref{SecImplyGibbs} we found that if a channel $\mathcal{F}$ satisfies (\ref{ReductionHeatBath}), then $\mathcal{F}$ is Gibbs preserving (see Lemma \ref{LemmaFixpointCHannel}). The following lemma provides the corresponding statement for generators.
\begin{Lemma}
\label{FixPointOfL}
Let $H$  be a Hermitian operator on the complex Hilbert space $\mathcal{H}$. If $\mathcal{L}$ satisfies (\ref{GeneratorRelation}), and moreover is such that $\mathcal{L}^{*}(\hat{1}) = 0$, then $\mathcal{L}(e^{-\beta H}) = 0$. Hence, if $Z_{\beta}(H)$ is finite, then  $\mathcal{L}\big(G_{\beta}(H)\big) = 0$.
\end{Lemma}
One may note that the condition $\mathcal{L}^{*}(\hat{1}) = 0$ is another way of saying that $\Tr\mathcal{L}(\rho) = 0$, which means that the generator yields a trace preserving evolution, which by construction is satisfied by generators on the Lindblad  form (\ref{nvsfklbnmain}).

Although a direct consequence of linearity, the following lemma is convenient since it can be applied even if the two generators $\mathcal{L}_a$ and $\mathcal{L}_b$ do not commute, and it thus may be difficult to evaluate $e^{t(\mathcal{L}_a +\mathcal{L}_b)}$, even if we can evaluate $e^{t\mathcal{L}_a}$ and $e^{t\mathcal{L}_b}$ separately. (An example of such non-commuting generators can be found in Appendix \ref{DissipativeSpins}.)
\begin{Lemma}
\label{LemmaAdditionGenerators}
Let $H$  be a Hermitian operator on the complex Hilbert space $\mathcal{H}$. Let $\mathcal{L}_a$ and $\mathcal{L}_b$ be linear maps such that $\mathcal{L}_a\mathcal{J}_{\beta H} =  \mathcal{J}_{\beta H}\mathcal{L}_a^{\ominus}$ and 
$\mathcal{L}_b\mathcal{J}_{\beta H} =  \mathcal{J}_{\beta H}\mathcal{L}_b^{\ominus}$.
Then $\mathcal{L}: = \mathcal{L}_a +\mathcal{L}_b$ satisfies $\mathcal{L}\mathcal{J}_{\beta H} =  \mathcal{J}_{\beta H}\mathcal{L}^{\ominus}$.
\end{Lemma}

\begin{Lemma}
\label{LemmaClosedEvolution}
Let $\tilde{H}$ be a Hermitian operator, and let $\mathcal{T}$ be a time-reversal such that $\mathcal{T}(\tilde{H}) = \tilde{H}$. Let $H$ be a Hermitian operator such that $[H,\tilde{H}] = 0$, then  $\mathcal{L}_{\tilde{H}}\mathcal{J}_{\beta H} = \mathcal{J}_{\beta H}\mathcal{L}^{\ominus}_{\tilde{H}}$, where $\mathcal{L}_{\tilde{H}}(\rho) := -i[\tilde{H},\rho]/\hbar$, for all operators $\rho$.
\end{Lemma}

The following proposition can be used to `glue'  generators on different subsystems via an interaction Hamiltonian in such a way that the global generator satisfies (\ref{GeneratorRelation}). 
\begin{Proposition}
\label{GlueGenerators}
Let $H_1$ and $H_2$ be Hermitian operators on the complex Hilbert spaces $\mathcal{H}_1$ and $\mathcal{H}_2$, respectively. Define $H := H_1\otimes\hat{1}_2 + \hat{1}_1\otimes H_2$. 
Let $\mathcal{T} := \mathcal{T}_1\otimes\mathcal{T}_2$ where $\mathcal{T}_1$ and $\mathcal{T}_2$ are time-reversals on $\mathcal{H}_1$ and $\mathcal{H}_2$, respectively. Let $\beta \geq 0$. Let $\mathcal{L}_1$ and $\mathcal{L}_2$ be generators on $\mathcal{H}_1$ and $\mathcal{H}_2$, respectively, such that 
\begin{equation}
\label{nfdjkfjkgd}
\begin{split}
\mathcal{L}_1\mathcal{J}_{\beta H_1} = & \mathcal{J}_{\beta H_1}\mathcal{L}^{\ominus}_1,\quad 
\mathcal{L}_2\mathcal{J}_{\beta H_2} =  \mathcal{J}_{\beta H_2}\mathcal{L}^{\ominus}_2.
\end{split}
\end{equation}
Let $H_{\mathrm{int}}$ be a Hermitian operator on $\mathcal{H}_1\otimes\mathcal{H}_2$ such that $[H, H_{\mathrm{int}}] = 0$, and $\mathcal{T}(H_{\mathrm{int}}) = H_{\mathrm{int}}$, and define $\mathcal{L}_{\mathrm{int}}(\rho) := -i[H_{\mathrm{int}},\rho]/\hbar$, for all operators $\rho$.
Then $\mathcal{L}:= \mathcal{L}_1\otimes \mathcal{I}_2 + \mathcal{I}_1\otimes\mathcal{L}_2 + \mathcal{L}_{\mathrm{int}}$ is such that $\mathcal{L}\mathcal{J}_{\beta H} = \mathcal{J}_{\beta H}\mathcal{L}^{\ominus}$.
\end{Proposition}

\begin{proof}
We first note that (\ref{nfdjkfjkgd}) implies $[\mathcal{L}_1\otimes \mathcal{I}_2]\mathcal{J}_{\beta H} = \mathcal{J}_{\beta H}[\mathcal{L}_1\otimes \mathcal{I}_2]^{\ominus}$, with the analogous statement for $\mathcal{I}_1\otimes \mathcal{L}_2$.
 By Lemma \ref{LemmaAdditionGenerators} we can conclude that $\tilde{\mathcal{L}} := \mathcal{L}_1\otimes \mathcal{I}_2 + \mathcal{I}_1\otimes\mathcal{L}_2$ satisfies $\tilde{\mathcal{L}}\mathcal{J}_{\beta H} = \mathcal{J}_{\beta H}\tilde{\mathcal{L}}^{\ominus}$. Since $[H, H_{\textrm{int}}] = 0$ and $\mathcal{T}(H_{\textrm{int}}) = H_{\textrm{int}}$ it follows by Lemma \ref{LemmaClosedEvolution} that
$\mathcal{L}_{\textrm{int}}\mathcal{J}_{\beta H} = \mathcal{J}_{\beta H}\mathcal{L}^{\ominus}_{\textrm{int}}$.
We can thus again use Lemma \ref{LemmaAdditionGenerators} to conclude that $\mathcal{L} = \tilde{\mathcal{L}} + \mathcal{L}_{\textrm{int}}$ satisfies $\mathcal{L}\mathcal{J}_{\beta H} = \mathcal{J}_{\beta H}\mathcal{L}^{\ominus}$. 
\end{proof}

\begin{Corollary}
\label{CorGlue}
With the assumptions as in Proposition \ref{GlueGenerators}, and if moreover 
$\mathcal{T}_1(H_1) = H_1$ and $\mathcal{T}_2(H_2) = H_2$, then all the conditions  in Assumptions \ref{IdealChannelDef} are satisfied for each channel  $\mathcal{F}_t := e^{t\mathcal{L}}$, $t \geq 0$, and  thus Proposition \ref{PropIdealGlbChnl} is applicable.
\end{Corollary}

One may observe that by Lemma \ref{FixPointOfL} it follows that the generator $\mathcal{L}$ in Proposition \ref{GlueGenerators} has $G_{\beta}(H_1)\otimes G_{\beta}(H_2)$ as a fixpoint, even though there is an interaction term. In the general case $G_{\beta}(H_1\otimes\hat{1}_2 + \hat{1}_1\otimes H_2 + H_{\mathrm{int}})$ would not be a fixpoint; an example being the generator in Appendix \ref{DissipativeSpins}.

\subsection{\label{SecDecouplingDiagAgain}Decoupling of diagonals again}

In  Appendices \ref{DecouplingDiagonals} and \ref{SecDiagonalMeasuremnts} we discussed the decoupling of diagonals, or modes of coherence. In the following section we show that the condition (\ref{GeneratorRelationMain}) on the generators is not strong enough to guarantee such decompositions. In section \ref{secRegainDecoupl} we discuss a sufficient condition for regaining the decoupling.

\subsubsection{An example}

Consider a Hamiltonian $H$ on  a finite-dimensional Hilbert space with eigenvalues $E_n$ and eigenbasis $|n\rangle$. For the sake of convenience we introduce the notation $\mathcal{C}_H(\rho) := [H,\rho]$.
In the following we say that a map $\mathcal{F}$ satisfies the mode decomposition with respect to $H$, if $\langle n|\mathcal{F}(|m\rangle\langle m'|)|n'\rangle = 0$ for every $n,n',m,m'$ such that $E_{m}-E_{m'}\neq E_{n}-E_{n'}$. (If $H$ in addition is non-degenerate, then the decomposition implies that diagonal elements can only be mapped to diagonal elements.)

Let $\mathcal{L}$ be a generator, and let $\mathcal{F}_t :=e^{t\mathcal{L}}$. If $\mathcal{F}_t$ satisfies the mode-decomposition with respect to $H$ for all $t\geq 0$, then it follows that $\mathcal{C}_H\mathcal{L}  = \mathcal{L}\mathcal{C}_H$. (To see this, one can first note that the mode decomposition implies that $\mathcal{C}_H\mathcal{F}_t = \mathcal{F}_t\mathcal{C}_H$ for each $t\geq 0$. By differentiation at $t=0$, one obtains $\mathcal{C}_H\mathcal{L}  = \mathcal{L}\mathcal{C}_H$.) By negation if follows that if $\mathcal{C}_H\mathcal{L}  \neq \mathcal{L}\mathcal{C}_H$, then there must exist some $t\geq 0$ for which $\mathcal{F}_t$ fails to satisfy the mode-decomposition.

For every channel $\Phi$ it is the case that $\mathcal{L}:= \Phi-\mathcal{I}$ is a Lindblad generator. (To see this, take the operators in a Kraus representation of $\Phi$ as the Lindblad operators of $\mathcal{L}$.) 
Recall the class of channels (\ref{QCQ}) defined in Appendix \ref{SecExampleBeyondThermal}. Since every such channel $\Phi$ satisfies (\ref{ReductionHeatBath}) we can conclude that $\mathcal{L}:=\Phi-\mathcal{I}$ satisfies (\ref{GeneratorRelationMain}).
For $A_k := |\psi_k\rangle\langle\psi_k|$, where $\{|\psi_k\rangle\}_k$ is some orthonormal basis,  one can confirm that $[\mathcal{C}_H\Phi-\Phi\mathcal{C}_H](\hat{1}) = -\mathcal{J}_{\beta H}\mathcal{T}(R)$, where
$R := [H,\sum_{k}|\psi_k\rangle\langle \psi_k|/\langle\psi_k|e^{-\beta H}|\psi_k\rangle]$.
Since both $\mathcal{J}_{\beta H}$ and $\mathcal{T}$ are invertible, it follows that a non-zero $R$ implies $\mathcal{C}_H\Phi\neq \Phi\mathcal{C}_H$.
 One can convince oneself that if $H$ is non-degenerate, then it is indeed possible to find an orthonormal basis $\{|\psi_k\rangle\}_k$ such that $R\neq 0$. Consequently both the resulting channel $\Phi$ and the generator $\mathcal{L}:=\Phi-\mathcal{I}$ fail to commute with $\mathcal{C}_H$. By the reasoning above, we know that there must exist some time $t\geq 0$ such that $\mathcal{F}_t:=e^{t\mathcal{L}}$ fails to satisfy the mode decomposition, although $\mathcal{L}$ satisfies (\ref{GeneratorRelationMain}).

The above construction provides an example of a channel that fails the decoupling. However, in section \ref{SecMainJCwithDissipation}, we primarily focus on the decoupling of the conditional CPMs $\tilde{F}_{\pm}$ defined in (\ref{euezuie}). To construct an example also for this case, take the above $\mathcal{L}$ as the generator for $E$, and make the trivial extension $\mathcal{L}_{SCE}:= \mathcal{L}\otimes \mathcal{I}_{SC}$, and let $H_{SCE}:= H\otimes\hat{1}_{SC} + \hat{1}\otimes H_{SC}$, for some arbitrary Hamiltonian $H_{SC}$ with $\mathcal{T}_{SC}(H_{SC}) = H_{SC}$. We can conclude that $\mathcal{L}_{SCE}$ satisfies (\ref{GeneratorRelationMain}) with respect to $H_{SCE}$, while the conditional CPMs $\tilde{F}_{\pm}$ generally fail to decouple.

\subsubsection{\label{secRegainDecoupl}Regaining the decoupling}
If $\mathcal{L}$ is such that $\mathcal{C}_H\mathcal{L} = \mathcal{L}\mathcal{C}_H$, where $\mathcal{C}_H(\rho) := [H,\rho]$, then the evolution $\mathcal{F}_t:=e^{t\mathcal{L}}$ is time-translation symmetric \cite{Lostaglio17}, i.e.,  $e^{is\mathcal{C}_H}\mathcal{F}_t =\mathcal{F}_t e^{is\mathcal{C}_H}$ for all $t\geq 0$ and all $s$, or equivalently $e^{-isH}\mathcal{F}_{t}(\rho)e^{isH} = \mathcal{F}_t(e^{-isH}\rho e^{is H})$ for all $\rho$, $t\geq 0$ and $s$.

Let us now assume that $H:= H_1\otimes\hat{1}_2 +\hat{1}_1\otimes H_2$, and define the conditional CPM $\tilde{\mathcal{F}}_{+}(\rho):= \Tr_2\big([\hat{1}\otimes Q^{f+}_2]\mathcal{F}_t(\rho\otimes\mathcal{G}_{\beta H_2}(Q^{i+}_{2}))\big)$, for some measurement operators $Q^{f+}_2$, $Q^{i+}_2$. With the assumptions $\mathcal{C}_H\mathcal{L} = \mathcal{L}\mathcal{C}_H$, $[Q^{f+}_2,H_2]=0$, and $[Q^{i+}_2,H_2] = 0$, it follows that $\mathcal{C}_{H_1}\tilde{\mathcal{F}}_{+} = \tilde{\mathcal{F}}_{+}\mathcal{C}_{H_1}$. (One can prove this by using $e^{-isH}\mathcal{F}_{t}(\rho)e^{isH} = \mathcal{F}_t(e^{-isH}\rho e^{is H})$ in the definition of $\tilde{\mathcal{F}}_{+}$, and differentiate at $s=0$.) By the same reasoning as in Appendix \ref{DecouplingDiagonals} one can conclude that the CPM $\tilde{\mathcal{F}}_{+}$ decouples along the modes of coherence.

It is straightforward to establish the counterpart to the gluing of generators discussed in Appendix \ref{SecAssemblingGenerators}. Suppose that $\mathcal{C}_{H_1}\mathcal{L}_1 = \mathcal{L}_1\mathcal{C}_{H_1}$, $\mathcal{C}_{H_2}\mathcal{L}_2 = \mathcal{L}_2\mathcal{C}_{H_2}$, and $[H_{\mathrm{int}},H_1\otimes \hat{1}_2 + \hat{1}_1\otimes H_2] = 0$. One can confirm that $\mathcal{C}_{H_1\otimes \hat{1}_2 + \hat{1}_1\otimes H_2}\mathcal{L} = \mathcal{L}\mathcal{C}_{H_1\otimes \hat{1}_2 + \hat{1}_1\otimes H_2}$, where $\mathcal{L}:= \mathcal{L}_1\otimes\mathcal{I}_2 + \mathcal{I}_1\otimes\mathcal{L}_2 + \mathcal{L}_{\mathrm{int}}$, for $\mathcal{L}_{\mathrm{int}}(\rho):= -i[H_{\mathrm{int}},\rho]/\hbar$.

\subsection{\label{SecGeneratorThrmOp}Generators of thermal operations and time-reversal symmetric thermal operations}

As mentioned in Section \ref{SecMainGenThrmOp}, a drawback with relying on the condition (\ref{GeneratorRelationMain}) is that it is unclear what it implies concerning the evolution of resources.  Here we develop some tools that enable us to show that a generator yields thermal operations. We also find a class of generators that results in time-reversal symmetric thermal operations.

\subsubsection{Generators of thermal operations}

Recall that we defined thermal operations in Appendix \ref{SecTimeRevSymmThrmOp}.
\begin{Lemma}
\label{ThermOpComposition}
Let $\mathcal{F}_1$ and $\mathcal{F}_2$  be two thermal operations with respect to the same Hamiltonian $H$ and  the same $\beta$. Then  the composition $\mathcal{F}_2\circ\mathcal{F}_1$ is also a thermal operation corresponding to $H$ and $\beta$.
\end{Lemma}
\begin{proof}
Since $\mathcal{F}_j$ is a thermal operation on the underlying Hilbert space $\mathcal{H}$, then there exists an ancillary Hilbert space $\mathcal{H}_{Bj}$, a corresponding Hamiltonian $H_{Bj}$, and  a unitary $U_{j}$ on $\mathcal{H}\otimes\mathcal{H}_{Bj}$ such that $[U_{j}, H\otimes \hat{1}_{Bj} + \hat{1}\otimes H_{Bj}] = 0$ and $\mathcal{F}_j(\rho) = \Tr_{Bj}(U_j[\rho\otimes G_{\beta}(H_{Bj})]U^{\dagger}_j)$. On the Hilbert space $\mathcal{H}\otimes\mathcal{H}_{B1}\otimes\mathcal{H}_{B2}$ we construct the unitary $U := [U_2\otimes \hat{1}_{B1}][U_{1}\otimes\hat{1}_{B2}]$, and the Hamiltonian $H_B :=  H_{B1}\otimes\hat{1}_{B2} + \hat{1}_{B1}\otimes H_{B2}$. One can confirm that $[U,H\otimes\hat{1}_{B1}\otimes\hat{1}_{B2}+ \hat{1}\otimes H_B] = 0$, and $\mathcal{F}_2\big(\mathcal{F}_1(\rho)\big) =\Tr_{B1,B2}(U[\rho\otimes G_{\beta}(H_B)]U^{\dagger})$. Thus $\mathcal{F}_2\circ\mathcal{F}_1$ is a thermal operation with respect to $H$ and $\beta$.
\end{proof}

For a given $\beta \geq 0$ and Hermitian operator $H$, we say that a linear map $\mathcal{L}$ is a generator of thermal operations with respect to $\beta$ and $H$, if $e^{t\mathcal{L}}$ is a thermal operation with respect to $\beta$ and $H$, for each $t\geq 0$.

A more or less direct consequence of Lemma  \ref{ThermOpComposition} is the following.
\begin{Lemma}
\label{LemmaCommutingGenerators}
If $\mathcal{L}_1$ and $\mathcal{L}_2$ are generators of thermal operations with respect to $H$ and $\beta\geq 0$, and if $[\mathcal{L}_1,\mathcal{L}_2] = 0$, then $\mathcal{L}_1+\mathcal{L}_2$ is also a generator of thermal operations with respect to $H$ and $\beta$.
\end{Lemma}

\subsubsection{\label{SecGenForThrmOp}An example}

Here we consider the class of generators (\ref{dvmvmv}) in Appendix \ref{SecModelForThrmls}, but restricted to the special case of a two-level system with $H := \frac{1}{2}E\sigma_z$. With the condition (\ref{ClassicalDetailedBalance}) there are three remaining free parameters $r_0,r_1,r\geq 0$, where
\begin{equation}
\begin{split}
& r(0|0) := r_0,\quad r(1|1) := r_1,\\
& r(1|0) := re^{-\beta E/2},\quad r(0|1) := re^{\beta E/2}.
\end{split}
\end{equation}
One can realize that the resulting generator $\mathcal{L}$ in (\ref{dvmvmv}) can be decomposed as $\mathcal{L} = \mathcal{L}_{H} + \mathcal{L}_{r} +\mathcal{L}_{r_0,r_1}$, where $\mathcal{L}_{H}(\rho) := -i[H,\rho]/\hbar$, 
\begin{equation*}
\begin{split}
 \mathcal{L}_{r}(\rho) :=&  re^{-\beta E/2} |1\rangle \langle 0|\rho |0\rangle\langle 1|+  re^{\beta E/2} |0\rangle \langle 1|\rho |1\rangle\langle 0|  \\
& -\frac{1}{2}re^{-\beta E/2}|0\rangle\langle 0|\rho   -\frac{1}{2}re^{-\beta E/2}\rho|0\rangle\langle 0|\\
& -\frac{1}{2} re^{\beta E/2}|1\rangle\langle 1|\rho   -\frac{1}{2} re^{\beta E/2}\rho|1\rangle\langle 1|,\\
\mathcal{L}_{r_0,r_1}(\rho) := & -\frac{r_0+r_1}{2}\big(|0\rangle\langle 0|\rho|1\rangle\langle 1|  +|1\rangle\langle 1|\rho |0\rangle\langle 0| \big).
\end{split}
\end{equation*}
The generator $\mathcal{L}_{r}$ yields the family of channels
\begin{equation*}
\begin{split}
 e^{t\mathcal{L}_{r}}(\rho) 
& =  G_{\beta}(H)\Tr(\rho)\\
& - \frac{1}{Z}e^{-rt Z} \sigma_z(e^{-\beta E/2}\langle 0|\rho|0\rangle - e^{\beta E/2} \langle 1|\rho|1\rangle)\\
 & + e^{ -\frac{t}{2} rZ}|1\rangle\langle 1|\rho|0\rangle\langle 0| +  e^{ -\frac{t}{2}rZ}|0\rangle\langle 0|\rho|1\rangle\langle 1|,
\end{split}
\end{equation*}
where $Z:= e^{\beta E/2} + e^{-\beta E/2}$. 
This family of channels can be obtained via an ancillary two-level system with Hamiltonian $H_B = \frac{1}{2}E\sigma_z$. One can confirm that the unitary operators
\begin{equation*} 
\begin{split}
V_t  & := |0\rangle\langle 0|\otimes |0\rangle\langle 0|   + |1\rangle\langle 1|\otimes |1\rangle\langle 1|\\ 
&+e^{-\frac{1}{2}rtZ}\big(|0\rangle\langle 0|\otimes |1\rangle\langle 1| + |1\rangle\langle 1|\otimes |0\rangle\langle 0|\big)\\
& + \sqrt{1-e^{-rtZ}}\big(|0\rangle\langle 1|\otimes |1\rangle\langle 0|  - |1\rangle\langle 0|\otimes |0\rangle\langle 1|\big)
\end{split}
\end{equation*}
satisfy $[V_t,H\otimes\hat{1}_B+ \hat{1}\otimes H_B] = 0$ and that  $ e^{t\mathcal{L}_{r}}(\rho) = \Tr_B(V_t[\rho \otimes G_{\beta}(H_B)]V_t^{\dagger})$. We can thus conclude that $\mathcal{L}_{r}$ is a generator of thermal operations.

Similarly it is the case that 
\begin{equation}
\begin{split}
& e^{t\mathcal{L}_{r_0,r_1}} =  (1-e^{ -\frac{t}{2}(r_0 + r_1)})\mathcal{D}    +e^{ -\frac{t}{2}(r_0 + r_1)}\mathcal{I}, \\
& \mathcal{D}(\rho) :=|0\rangle\langle 0|\rho|0\rangle \langle 0|  +|1\rangle\langle 1|\rho|1\rangle\langle 1|.
\end{split}
\end{equation}
In other words,  $e^{t\mathcal{L}_{r_0,r_1}} $ is a convex combination of the identity mapping and the pinching $\mathcal{D}$.
For an ancillary two-level system with a degenerate Hamiltonian $H_B  \propto \hat{1}_B$, the unitary operator $U = |0\rangle\langle 0|\otimes\sigma_x + |1\rangle\langle 1|\otimes \sigma_y$, is globally energy conserving, and $\mathcal{D}(\rho) = \Tr_{B}(U[\rho\otimes \frac{1}{2}\hat{1}_B]U^{\dagger})$. Hence, $\mathcal{D}$ is a thermal operation. Since $\mathcal{D}$ and $\mathcal{I}$ are thermal operations, it follows that their convex combination $e^{t\mathcal{L}_{r_0,r_1}}$
also is a thermal operation for each $t\geq 0$ (see  Appendix C in \cite{Lostaglio14b}). Hence, $\mathcal{L}_{r_0,r_1}$ is a generator of thermal operations.

One can furthermore confirm that $\mathcal{L}_H$, $\mathcal{L}_{r}$, and $\mathcal{L}_{r_0,r_1}$ commute with each other. Hence, we can apply Lemma \ref{LemmaCommutingGenerators}, which implies that $\mathcal{L}$ also is a generator of thermal operations.

It is plausible that the general class of generators in (\ref{dvmvmv}) that satisfy  (\ref{ClassicalDetailedBalance}) are generators of thermal operations. However, we will not consider this question here.

\subsubsection{Non-commuting generators of thermal operations}

In the following we wish to generalize the addition of generators to the case when these do not necessarily commute. The idea is that we shall use the Trotter decomposition, i.e., the relation $e^{t\mathcal{L}_1+t\mathcal{L}_2} = \lim_{n\rightarrow \infty}(e^{t\mathcal{L}_1/n}e^{t\mathcal{L}_2/n})^{n}$. By repeated applications of  Lemma \ref{ThermOpComposition} it follows that $(e^{t\mathcal{L}_1/n}e^{t\mathcal{L}_2/n})^{n}$ is a thermal operation for each $n$. However, it is not clear (at least not to the author) whether the set of thermal operations is closed, i.e., that the limit map  $e^{t\mathcal{L}_1+t\mathcal{L}_2}$ actually is a thermal operation, as we have defined them.
One may wonder whether it would not be possible to  define a `limit bath' for $n\rightarrow\infty$. However, one can realize that the use of Lemma  \ref{ThermOpComposition} in the proof of Proposition \ref {AddGeneratorsForThermalOp} entails an indefinitely increasing number of thermal ancillary systems as $n$ increases. Hence, the `limit object' that would realize the limit channel for this particular construction would correspond to an infinite tensor product. Here, we do not consider the rather technical issue of whether one can make sense of that limit object or not, but we leave this as an open question. 
 To be on the safe side, here we allow for the possibility that the set of thermal operations is not closed, and settle for maps that  generate channels in the closure. The closure is defined with respect to the choice of  norm in the following bound, which is obtained from Theorem 3 in \cite{Suzuki76}. 
\begin{equation}
\label{SuzukiBound}
\begin{split}
& \Vert e^{\sum_{j=1}^{p} A_j} -(e^{A_1/n}\cdots e^{A_p/n} )^{n}\Vert \\
& \quad\quad \leq \frac{2}{n}(\sum_{j=1}^{p}\Vert A_j\Vert)^2 e^{\frac{n+2}{n}\sum_{j=1}^{p}\Vert A_j\Vert},
\end{split}
\end{equation}
where $A_1,\ldots,A_p$ are bounded operators with respect to the norm $\Vert\cdot\Vert$ of some Banach algebra.

\begin{Proposition}
\label{AddGeneratorsForThermalOp}
Let $\mathcal{L}_1, \mathcal{L}_2,\ldots,\mathcal{L}_p$ be generators of thermal operations with respect to a Hermitian operator $H$ and $\beta \geq 0$.
If these generators are bounded with respect to a  norm $\Vert\cdot\Vert$, then $\mathcal{L}: = \sum_{j=1}\mathcal{L}_j$ is a generator of maps in the closure (with respect to $\Vert \cdot\Vert$) of the set of thermal operations with respect to  $H$ and $\beta$.
\end{Proposition}
\begin{proof}
We let $A_j :=  t\mathcal{L}_j$ in (\ref{SuzukiBound}). 
Since we assume that  $\Vert\mathcal{L}_1\Vert, \Vert\mathcal{L}_2\Vert,\ldots,\Vert\mathcal{L}_p\Vert$ are finite, if follows that the right hand side of  (\ref{SuzukiBound}) goes to zero as $n\rightarrow\infty$ for each fixed $t$. Hence, for each fixed $t$ it is the case that  
\begin{equation}
\label{bfndndfb}
\lim_{n\rightarrow\infty}\Vert e^{t\mathcal{L}} -(e^{t\mathcal{L}_1/n}\cdots e^{t\mathcal{L}_p/n} )^{n}\Vert =0. 
\end{equation}
Since $\mathcal{L}_1,\ldots,\mathcal{L}_p$ are assumed to be generators for thermal operations with respect to $H$ and $\beta$, it follows that $e^{t\mathcal{L}_1/n},\cdots,  e^{t\mathcal{L}_p/n}$ are all thermal operations. By Lemma \ref{ThermOpComposition} we can conclude 
 that $(e^{t\mathcal{L}_1/n}\cdots e^{t\mathcal{L}_p/n})^{n}$ is also a thermal operation with respect to $H$ and $\beta$. By (\ref{bfndndfb}) we know that for every neighborhood of $e^{t\mathcal{L}}$  (for a fixed $t$) with respect to $\Vert \cdot\Vert$, there exists an $n$ such that the thermal operation $(e^{t\mathcal{L}_1/n}\cdots e^{t\mathcal{L}_p/n})^{n}$ is in that neighborhood. 
We can conclude that for every $t\geq 0$, the channel  $e^{t\mathcal{L}}$ is in the closure of the set of thermal operations. Hence, $\mathcal{L}$ is a generator of maps in the closure of the thermal operations with respect to $H$ and $\beta$.
\end{proof}

Here we provide a general method to glue generators of thermal operations, much analogous to what Proposition \ref{GlueGenerators} does for generators that satisfy (\ref{GeneratorRelationMain}), but with the caveat that we can only prove that the resulting generator produces channels in the closure of the set of thermal operations. 
In the following we define  $\mathcal{L}_{\mathrm{int}}(\rho) := -i[H_{\mathrm{int}},\rho]$, for all operators $\rho$.
\begin{Proposition}
\label{GlueGeneratorsForThermalOp}
Let $H_1$ and $H_2$ be Hermitian operators on the complex Hilbert spaces $\mathcal{H}_1$ and $\mathcal{H}_2$, respectively. Let $\beta \geq 0$. Let $\mathcal{L}_1$ be a generator of thermal operations with respect to $H_1$ and $\beta$,  and let $\mathcal{L}_2$ be a generator of thermal operations with respect to  $H_2$ and $\beta$. Moreover, let $H_{\mathrm{int}}$ be a Hermitian operator such that $[H_{\mathrm{int}}, H_1\otimes\hat{1}_2 + \hat{1}_1\otimes H_2] =0$. Then $\mathcal{L}_1\otimes \mathcal{I}_2 + \mathcal{I}_1\otimes\mathcal{L}_2 + \mathcal{L}_{\mathrm{int}}$ is a generator of maps in the closure of the thermal operations with respect to $H_1\otimes\hat{1}_2 + \hat{1}_1\otimes H_2$ and $\beta$.
\end{Proposition}

\begin{proof}
One can first observe that if $\mathcal{L}_1$ is a generator with respect to $H_1$, then $\mathcal{L}_1\otimes\mathcal{I}_2$ is a generator with respect to $H_1\otimes \hat{1}_2$. Since $\mathcal{L}_1\otimes \mathcal{I}_2$ and $\mathcal{I}_1\otimes\mathcal{L}_2$ commute, it follows by Lemma \ref{LemmaCommutingGenerators} that $\mathcal{L}_1\otimes \mathcal{I}_2 + \mathcal{I}_1\otimes\mathcal{L}_2$ is a generator of thermal operations with respect to $H_1\otimes\hat{1}_2 + \hat{1}_1\otimes H_2$.\\
Next we observe that if $[H_{\mathrm{int}}, H_1\otimes\hat{1}_2 + \hat{1}_1\otimes H_2] =0$, then it follows that $\mathcal{L}_{\mathrm{int}}$ is a generator of thermal operations with respect to $H_1\otimes\hat{1}_2 + \hat{1}_1\otimes H_2$. By Proposition \ref{AddGeneratorsForThermalOp}  we can conclude that by adding $\mathcal{L}_1\otimes \mathcal{I}_2 + \mathcal{I}_1\otimes\mathcal{L}_2$ and $\mathcal{L}_{\mathrm{int}}$ we obtain a generator of maps in the closure of the thermal operations.
\end{proof}

\subsubsection{\label{SecGeneratorTimeRevThrmOp} Generators of time-reversal symmetric thermal operations}

Recall the definition of time-reversal symmetric thermal operations in Appendix \ref{SecTimeRevSymmThrmOp}.

For a given $\beta \geq 0$, Hermitian operator $H$, and time-reversal $\mathcal{T}$, such that $\mathcal{T}(H) = H$, we say that a linear map $\mathcal{L}$ is a generator of time-reversal symmetric thermal operations, if $e^{t\mathcal{L}}$ is a time-reversal symmetric thermal operation with respect to $\beta$, $H$, and $\mathcal{T}$, for all $t\geq 0$.

 Unfortunately, it seems difficult to directly generalize Lemma \ref{ThermOpComposition}, and thus also Proposition \ref{AddGeneratorsForThermalOp}.
The issue is that even if the unitary operators $U_1$ and $U_2$ in the proof of Lemma \ref{ThermOpComposition} would be time-reversal symmetric with respect to the time-reversals $\mathcal{T}\otimes\mathcal{T}_{B1}$ and $\mathcal{T}\otimes\mathcal{T}_{B2}$, respectively, it is not clear that  $U := [U_2\otimes \hat{1}_{B1}][U_{1}\otimes\hat{1}_{B2}]$ would be time-reversal symmetric under $\mathcal{T}\otimes\mathcal{T}_{B1}\otimes\mathcal{T}_{B2}$. It is conceivable that the set of time-reversal symmetric thermal operations is not closed under composition of channels. Although an interesting issue, we leave it as an open question, and here  explicitly construct a family of generators that yield time-reversal symmetric thermal operations.

For the purpose of establishing our example, let $H$ be a non-degenerate Hamiltonian on a finite-dimensional Hilbert space, and let  $\mathcal{T}$ be a time-reversal such that $\mathcal{T}(H) = H$. We focus on the generator (\ref{nsdlvvd}), i.e., $\mathcal{L}(\rho) =  -\frac{i}{\hbar}[H,\rho]  +r G_{\beta}(H)\Tr(\rho)- r \rho$.
The solution to the master equation corresponding to this generator is given by the family of channels
$\mathcal{F}_t(\rho) = e^{-rt}e^{-itH/\hbar}\rho e^{itH/\hbar} + (1-e^{-rt})G_{\beta}(H)\Tr(\rho)$.

In the following we show that the channel $\mathcal{F}_t$ for each $t$ is a time-reversal symmetric thermal operation. (It is conceivable that the same is true for the channels obtained from the more general generators in (\ref{dvmvmv}), but we will not consider this question here.) Our first observation is that the channel $\mathcal{F}_{t}$ for each $t$ is a convex combination of the Hamiltonian evolution $\rho \mapsto e^{-itH/\hbar}\rho e^{itH/\hbar}$  and the replacement map $\mathcal{R}(\rho):= G_{\beta}(H)\Tr(\rho)$. The Hamiltonian evolution is a time-reversal symmetric thermal operation. In the following we shall show that $\mathcal{R}$ also is a time-reversal symmetric thermal operation.

We let $\mathcal{H}_{B}$ be isomorphic to $\mathcal{H}$. With respect to some selected isomorphism, let the Hamiltonian  on $\mathcal{H}_B$ be a `copy' of the Hamiltonian $H$. Similarly we let the time reversal on $\mathcal{H}_B$ be the isomorphic copy of the one on $\mathcal{H}$. Let $|k\rangle$ simultaneously denote both the energy eigenbasis of $H$ and  $H_{B}$, in such a way that $|k\rangle$ corresponds to the same energy eigenvalue in both cases. 
By assumption, $H$ is non-degenerate. Lemma \ref{TOnHermitean} thus yields $\mathcal{T}(|k\rangle\langle k|) = |k\rangle\langle k|$. Corollary \ref{CorrBasisInvariant} implies that there exist real numbers $\theta_k$ such that $\mathcal{T}(|k\rangle\langle l|) = e^{i(\theta_l-\theta_k)}|l\rangle\langle k|$.
Define the unitary operator $W := \sum_{k,l}|k\rangle\langle l|\otimes |l\rangle\langle k|$ (a swap-operator). It follows that $[\mathcal{T}\otimes\mathcal{T}](W) = W$, and thus $W$ is time-reversal symmetric. One can also see that $[\mathcal{T}\otimes\mathcal{T}](H\otimes\hat{1}+\hat{1}\otimes H) = H\otimes\hat{1}+\hat{1}\otimes H$, and $[H\otimes\hat{1}+\hat{1}\otimes H,W] = 0$. Moreover, $\Tr_B(W[\rho\otimes G_{\beta}(H)]W^{\dagger}) = G_{\beta}(H)\Tr(\rho) = \mathcal{R}(\rho)$. Hence, $\mathcal{R}$  is a time-reversal symmetric thermal operation. Since $\mathcal{F}_t$ is a convex combination of time-reversal symmetric thermal operations, it follows by Lemma \ref{TimeRevSymmThermalConv} that $\mathcal{F}_t$ is also a time-reversal symmetric thermal operation. We can conclude that $\mathcal{L}$ is a generator of time-reversal symmetric thermal operations.

\subsubsection{\label{SecCaution}A word of caution}
The definition of generators of thermal operations (or time-reversal symmetric thermal operations, or their closures) as we have stated it, only requires that the induced channel $\mathcal{F}_t$ is a (time-reversal symmetric) thermal operation for each single time $t$.  In other words, the definition is `pointwise' in the sense that it in principle allows for a different physical setup for each time $t$. We made use of this for the construction in the previous section, where the convex combination of the Hamiltonian evolution and the replacement map is obtained via Lemma \ref{TimeRevSymmThermalConv}, where the weight $\lambda := e^{-rt}$ is obtained via the ancillary Hamiltonian, which has to be adapted for each $t$. (Although as mentioned in Appendix \ref{SecConvTimerevsymm}, one can obtain a proof via an alternative technique.)

This adaptive construction should be put in contrast with another scenario (cf.~the discussion in Appendix \ref{MakeSense}) where we have a fixed interaction Hamiltonian $H_{\mathrm{int}}$ such that $[H_{\mathrm{int}},H_{\textrm{tot}}]=0$, where $H_{\textrm{tot}}:= H\otimes \hat{1}_B + \hat{1}_{SCE}\otimes H_B$, and where the evolution is determined  by $H_{\mathrm{evol}}:= H_{\mathrm{int}} + H_{\textrm{tot}}$. 
For all $t\geq 0$, the channels $\mathcal{F}_{t}(\sigma) = \Tr_{B}(e^{-itH_{\mathrm{evol}}/\hbar}[\sigma\otimes G_{\beta}(H_B)]e^{itH_{\mathrm{evol}}/\hbar})$ are by construction thermal operations with respect to $\beta$ and $H$. In this case we thus have a single physical setup, in the sense of a fixed global Hamiltonian $H_{\mathrm{evol}}$ that generates the evolution for all times.
 However, one would not generally expect that the resulting family of maps $\{\mathcal{F}_t\}_{t\geq 0}$ would satisfy a time-independent Lindblad master equation. (In derivations of Markovian master equations one tends to consider limits of infinite heat baths, see e.g.~\cite{BreuerPetruccione,Davies76}.)

The question is if a generator of thermal operations (in the point-wise sense) always can be implemented by a time-independent physical setup of the type described above (allowing for an infinite heat bath).
We will not explore this question here. However, a potential starting point for such investigations could be the notion of Davies generators \cite{Davies74,Roga10,Temme13}.

\subsection{\label{DissipativeSpins}Further details on the two thermalizing spins in section \ref{SecMainTwoCoupledSpins}}

We consider  two two-level systems, or the internal degrees of freedom of two spin-half particles, with Hamiltonians $H_1 := \frac{1}{2}E\sigma_{z1}$ and $H_{2} := \frac{1}{2}E\sigma_{z2}$. We also assume an interaction Hamiltonian of the form $H_{\mathrm{int}} := \lambda |01\rangle\langle 10| + \lambda|10\rangle\langle 01|$. We let $\mathcal{T}_1$ and $\mathcal{T}_2$ be the transposes in the eigenbasis of $\sigma_{z1}$ and $\sigma_{z2}$, respectively.
On the separate spins we assume thermalizing generators $\mathcal{L}_1$ and $\mathcal{L}_2$ as defined in (\ref{MainTwoSpin}). 
We know from Appendix \ref{SecModelForThrmls} that $\mathcal{L}_1$ and $\mathcal{L}_2$ separately satisfy (\ref{GeneratorRelation}). 
By adding an interaction Hamiltonian, via the generator $\mathcal{L}_{\mathrm{int}}(\rho) := -i\lambda[H_{\mathrm{int}},\rho]/\hbar$,  we obtain the global generator $\mathcal{L}:= \mathcal{L}_1\otimes\mathcal{I}_2 +\mathcal{I}_1\otimes\mathcal{L}_2 + \mathcal{L}_{\mathrm{int}}$, where $\mathcal{I}$ denotes the identity map. By construction, $[\hat{1}_{1}\otimes H_2 + H_1\otimes\hat{1}_2 + \hat{1}_1\otimes, H_{\mathrm{int}}] = 0$, and the conditions of Proposition \ref{GlueGenerators} are satisfied, and thus Proposition \ref{PropIdealGlbChnl} is applicable to all the induced channels $\mathcal{F}_t = e^{t\mathcal{L}}$. On one of the reduced systems, e.g., system $1$, we can generate the conditional dynamics $\tilde{\mathcal{F}}_{\pm}$ as defined in (\ref{mdsflfmsb}) for each $\mathcal{F}_t$, where $\tilde{\mathcal{F}}_{\pm}$ satisfy (\ref{fbfbbfs}). In other words, with one of the particles acting as the energy reservoir for the other, we regain the conditional fluctuation relation.

The generator $\mathcal{L}$ does not only satisfy  (\ref{GeneratorRelation}), but also yields channels in the closure of the thermal operations with respect to $H_1\otimes \hat{1}_2 + \hat{1}_1\otimes H_2$ and $\beta$. To see this, we know from Appendix \ref{SecGeneratorTimeRevThrmOp} that $\mathcal{L}_1$ and $\mathcal{L}_2$ are generators for time-reversal symmetric thermal operations, which is a particular case of thermal operations. Since  $[\hat{1}_{1}\otimes H_2 + H_1\otimes\hat{1}_2, H_{\mathrm{int}}] = 0$, we can use Proposition \ref{GlueGeneratorsForThermalOp} to conclude that $\mathcal{L}$ generates channels in the closure of the thermal operations with respect to $\hat{1}_{1}\otimes H_2 + H_1\otimes\hat{1}_2$ and $\beta$.

One can moreover confirm that $[H_1,\mathcal{L}_1(\rho)] = \mathcal{L}_1([H_1,\rho])$, $[H_2,\mathcal{L}_2(\rho)] = \mathcal{L}_2([H_2,\rho])$, and  we know that $[\hat{1}_{1}\otimes H_2 + H_1\otimes\hat{1}_2, H_{\mathrm{int}}] = 0$. If we assume $[H_2,Q_{2}^{i\pm}] = 0$  and $[H_2,Q_{2}^{f\pm}] = 0$, then the reasoning outlined in Appendix \ref{SecDecouplingDiagAgain} yields that the CPMs $\tilde{\mathcal{F}}_{\pm}$ decouple with respect to the modes of coherence. Since system $1$ in the present case is a single two-level system, it follows that the mode structure becomes particularly simple. There is the diagonal mode, corresponding to $\{|0\rangle\langle 0|,|1\rangle\langle 1|\}$, and the two off-diagonal modes corresponding to $\{|0\rangle\langle 1|\}$ and $\{|1\rangle\langle 0|\}$, respectively.  Due to the decoupling it follows that, e.g., $\langle 0| \tilde{\mathcal{F}}_{\pm} (|0\rangle\langle 1|)|0\rangle =0$.  It must be emphasized though, that the decoupling can only be expected to hold if $Q_{2}^{i\pm}$ and $Q_{2}^{f\pm}$ are diagonal in the energy eigenbasis of  $H_2$.

On a different note one can confirm that $\mathcal{L}_{\mathrm{int}} (\mathcal{L}_1\otimes\mathcal{I}_2 + \mathcal{I}_1\otimes\mathcal{L}_2) \neq (\mathcal{L}_1\otimes\mathcal{I}_2 + \mathcal{I}_1\otimes\mathcal{L}_2)\mathcal{L}_{\mathrm{int}}$. Hence, even though $H_{\mathrm{int}}$ commutes with $H_1\otimes \hat{1}_2 + \hat{1}_1\otimes H_2$, this does not imply that $\mathcal{L}_{\mathrm{int}}$ commutes with  $\mathcal{L}_1\otimes\mathcal{I}_2 + \mathcal{I}_1\otimes\mathcal{L}_2$. As a further remark along these lines, one may observe that Lemma \ref{FixPointOfL} yields that $G_{\beta}(H_1)\otimes G_{\beta}(H_2)$ is a fixpoint of $\mathcal{L}$. By explicit evaluation one  can show  that $G_{\beta}(H_1\otimes\hat{1}_2 + \hat{1}_1\otimes H_2 + H_{\mathrm{int}})$ is not a fixpoint. Hence, in spite of the coupling term $H_{\textrm{int}}$, the fixpoint still remains the product of the local Gibbs states.

\subsection{\label{TwoSpinGlobalThrm}Example of the approximate fluctuation relation: Two weakly interacting spins with global thermalization}
The example of the exact fluctuation relation in Appendix \ref{DissipativeSpins} requires that the spins are resonant and that the interaction Hamiltonian commutes with the sum of the local Hamiltonians.  Here we relax those assumptions, and instead apply the global approximate fluctuation relation of Proposition \ref{GlobalApprox}.

\begin{figure}
 \includegraphics[width= 9cm]{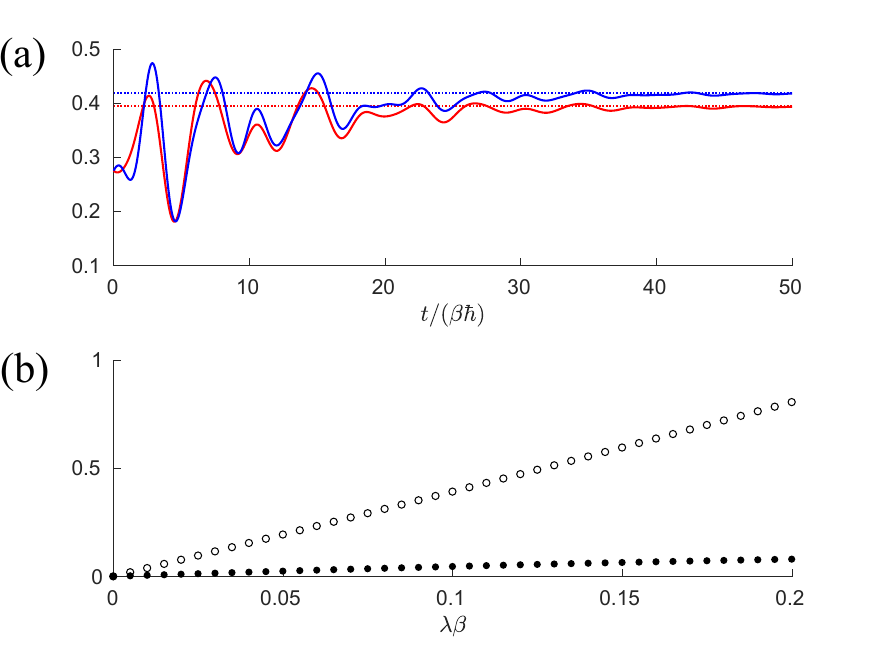} 
   \caption{\label{FigApproxMstrEq} 
 {\bf Approximate fluctuation relation with error bound.}   
(a) $f(t)$ is plotted as the red solid line, and $r(t)$ the blue solid line, for the interval $t/(\beta\hbar)\in[0,50]$.
In the ideal case $f$ and $r$ would be identical.
In the interval the maximal error is $\max_{t/(\beta\hbar)\in[0,50]}|f(t)-r(t)|\approx 0.0795$. This should be compared with the upper bound  in Proposition \ref{GlobalApprox}, which is approximately $0.806$. 
The dotted lines correspond to the limits $\lim_{t\rightarrow +\infty}f(t)$ and  $\lim_{t\rightarrow +\infty}r(t)$, where the system approaches the fixpoint of the master equation.\\
(b) For otherwise fixed settings, the maximum error $\max_{t/(\beta\hbar)\in[0,50]}|f(t)-r(t)|$ is plotted (filled circles) against the value of $\lambda\beta$. The empty circles correspond to the error bound in Proposition \ref{GlobalApprox}
}
\end{figure}

In Appendix \ref{SecControlparticleAgain} we discussed the approximate fluctuation relations in terms of interaction regions, where the quality of the approximation depends on how far out in the non-interacting region that the operators $Q^{i\pm}$ and $Q^{f\pm}$ are localized. Here we take to opportunity to consider an alternative setting, more in the spirit of perturbation theory, where the quality of the approximation rather depends on the interaction strength.

Similar to Appendix \ref{DissipativeSpins} we consider  two  spin-half particles, but where we allow for different excitation energies, i.e., we have the Hamiltonians $H_1 := \frac{1}{2}E_1\sigma_{z1}$ and $H_{2} := \frac{1}{2}E_2\sigma_{z2}$. We also have an interaction Hamiltonian $H_{\textrm{int}}$, resulting in the global Hamiltonian
\begin{equation*}
\begin{split}
H_{\mathrm{tot}} :=   H_0 + \lambda H_{\textrm{int}},\quad H_0 :=  H_1\otimes\hat{1}_2 + \hat{1}_1\otimes H_2. 
\end{split}
\end{equation*}
Instead of the two local generators that we used in Appendix \ref{DissipativeSpins}, here we apply a global relaxation
\begin{equation}
\label{nbfkorgnv}
\mathcal{L}(\rho) = -\frac{i}{\hbar}[H_{\mathrm{tot}},\rho] + rG_{\beta}(H_{\mathrm{tot}})\Tr(\rho) -r\rho.
\end{equation}
 We let $\mathcal{T}_1$ and $\mathcal{T}_2$ be the transposes in the eigenbasis of $\sigma_{z1}$ and $\sigma_{z2}$, respectively.  We assume that the interaction Hamiltonian is such that $[\mathcal{T}_1\otimes \mathcal{T}_2](H_{\textrm{int}}) = H_{\textrm{int}}$. According to the results in Appendix \ref{SecGeneratorTimeRevThrmOp} it follows that $\mathcal{L}$ satisfies (\ref{GeneratorRelation}).

If we let $H^i_1 := H^{f}_1 := H_1$ and $H^i_2 := H^f_2 := H_2$, then all the conditions of Proposition \ref{GlobalApprox} are satisfied for the channels $\mathcal{F}_t=e^{t\mathcal{L}}$, thus yielding the approximate fluctuation relation  (\ref{GlobalApproxFlctnRel}).
Since the local approximate Hamiltonians $H^i_1$, $H^i_2$, $H^f_1$, $H^f_2$ are obtained by disregarding the interaction term, it does intuitively seem reasonable that the error should be small when $\lambda$ is small.

As a concrete example we let $H_{\mathrm{int}} :=  \sigma_{x1}\otimes\sigma_{x2}$, where this interaction Hamiltonian is chosen such that it does not commute with $H_0$, but is such that $[\mathcal{T}_1\otimes\mathcal{T}_2](\sigma_{x1}\otimes\sigma_{x2}) = \sigma_{x1}\otimes\sigma_{x2}$.
For each channel $\mathcal{F}_t = e^{t\mathcal{L}}$, Proposition \ref{GlobalApprox} bounds the difference between 
\begin{equation}
\label{mblmbmb}
\begin{split}
f(t) := & \mathcal{Z}_{\beta H_1}(Q^{i}_1)\mathcal{Z}_{\beta H_2}(Q^{i}_2) \mathcal{P}^{+},\\
r(t):= & \mathcal{Z}_{\beta H_1}(Q_{1}^{f})\mathcal{Z}_{\beta H_{2}}(Q_{2}^{f}) \mathcal{P} ^{-},
\end{split}
\end{equation}
where
\begin{equation*}
\begin{split}
\mathcal{P}^{+} :=& P^{\mathcal{F}_t}_{\beta H_0}[ Q_{1}^{i+}\otimes Q_{2}^{i+} \rightarrow Q_{1}^{f+}\otimes Q_{2}^{f+}],\\
\mathcal{P}^{-} := & P^{\mathcal{F}_t}_{\beta H_0}[ Q_{1}^{f-}\otimes Q_{2}^{f-} \rightarrow Q_{1}^{i-}\otimes Q_{2}^{i-}].
\end{split}
\end{equation*}
For the calculations yielding Fig.~\ref{FigApproxMstrEq}, we assume
\begin{equation*}
E_1\beta = 1,\quad E_2\beta = 1.5,\quad \lambda\beta = 0.2,\quad r\beta\hbar = 0.1,
\end{equation*}
where these are dimensionless groups of parameters.
In Fig.~\ref{FigApproxMstrEq}(a) the pair $f(t), r(t)$ is plotted for times $t/(\beta\hbar)\in[0,50]$, where the operators $Q^{i+}_1, Q^{i+}_2, Q^{f+}_1, Q^{f+}_2$ have been chosen as projectors onto pure states that are selected independently according to the Haar distribution. In the limit of large evolution times, the state of the system evolves to $G_{\beta}(H_{\mathrm{tot}})$, and thus $f(t)$ and $r(t)$ each approaches a limit value, which correspond to the dotted lines. In Fig.~\ref{FigApproxMstrEq}(b) the maximum error over the given time interval is plotted as a function of the interaction strength $\lambda$. This is compared with the bound from  Proposition \ref{GlobalApprox}. Analogous to what we found in Appendix \ref{SecNumericalEvaluation} concerning the error bound in Proposition \ref{ApproxFlctnRel} (and more specifically concerning the bound in (\ref{ApproxFluctQuantitative})), the error bound in  Proposition \ref{GlobalApprox} (empty circles in  Fig.~\ref{FigApproxMstrEq}(b)) seems to be somewhat pessimistic compared to the actual error (filled circles). As expected, the error decreases for decreasing $\lambda$.

\subsection{\label{SecMoreWidelyAppl}More widely applicable approximate fluctuation relations?}

The approximate relation in Proposition \ref{GlobalApprox} allows us to apply the fluctuation relations in a wider setting than the exact relations. Moreover, Proposition \ref{GlobalApprox} provides an analytical bound (although potentially crude) that  is time-independent, which thus allows us to estimate the quality of the approximate fluctuation relation irrespective of how long we allow the system to evolve. However, this comes at the price that the conditions of Proposition \ref{GlobalApprox} have to be satisfied. In particular, for a given generator $\mathcal{L}$ we have to find a Hermitian operator $H$ such that  $\mathcal{L}\mathcal{J}_{\beta H} = \mathcal{J}_{\beta H}\mathcal{L}^{\ominus}$, and this can be challenging. It would thus be desirable to find generalizations that would be more easily applied. As an indication that such generalizations may be possible, here we numerically test the approximate fluctuation relation in a scenario where Proposition \ref{GlobalApprox} is not obviously applicable.

Here we combine the two local generators of Appendix \ref{DissipativeSpins} with the Hamiltonian part that was used in Appendix \ref{TwoSpinGlobalThrm}. In other words, we assume a generator of the form
\begin{equation}
\label{sdvmmvs}
\begin{split}
\mathcal{L}(\rho) := & -\frac{iE_1}{2\hbar}[\sigma_{z1}\otimes\hat{1}_2,\rho] +r_1G_1\otimes\Tr_1(\rho) -r_1\rho\\
 & -\frac{iE_2}{2\hbar}[\hat{1}_1\otimes \sigma_{z2},\rho] +r_2\Tr_2(\rho) \otimes G_2 -r_2\rho\\
& -\frac{i}{\hbar}\lambda[\sigma_{x1}\otimes\sigma_{x2}].
\end{split}
\end{equation}
For the parameters we choose
\begin{equation*}
\begin{split}
& E_1\beta = 1,\,\, E_2\beta = 1.5,\,\, \lambda\beta = 0.2,\,\, r_1\beta\hbar = 0.1,\,\, r_2\beta\hbar = 0.2.
\end{split}
\end{equation*}
Like in Appendix \ref{TwoSpinGlobalThrm} we choose  $Q^{i+}_1, Q^{i+}_2, Q^{f+}_1, Q^{f+}_2$ as projectors onto random pure states.  Analogous to  Fig.~\ref{FigApproxMstrEq} we do in Fig.~\ref{FigNoBound}(a) compare the functions $f(t)$ and $r(t)$ defined in (\ref{mblmbmb}), but where we now evaluate these functions for the generator (\ref{sdvmmvs}). A comparison of Fig.~\ref{FigNoBound} and Fig.~\ref{FigApproxMstrEq} suggests that the behaviors on a qualitative level are rather similar, which gives some indication that a generalization of Proposition \ref{GlobalApprox} could be possible.
For example, one could imagine introducing some systematic approximation at the level of the global fluctuation relation, i.e.,  replacing the condition  $\mathcal{L}\mathcal{J}_{\beta H} = \mathcal{J}_{\beta H}\mathcal{L}^{\ominus}$ with some approximate version. However, we will not consider this question further in this investigation.

\begin{figure}
 \includegraphics[width= 8cm]{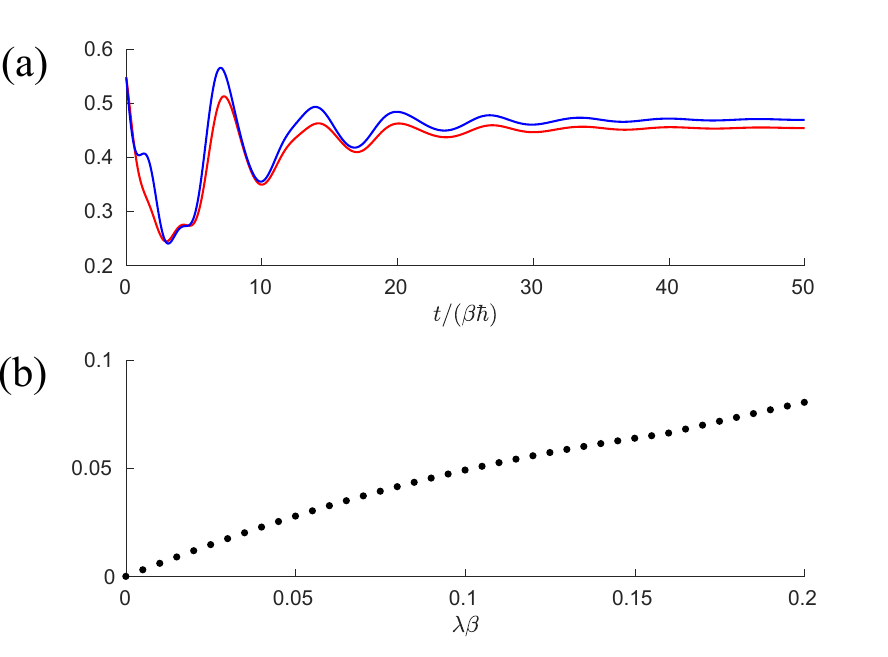} 
   \caption{\label{FigNoBound}  
{\bf Approximate fluctuation relation without bound.}   
(a) Analogous to Fig.~\ref{FigApproxMstrEq}, the functions $f(t)$ (red solid line) and $r(t)$ (blue solid line) are as defined in (\ref{mblmbmb}), but for the new generator (\ref{sdvmmvs}). These are plotted for the interval $t/(\beta\hbar)\in[0,50]$.\\
(b) The maximum error $\max_{t/(\beta\hbar)\in[0,50]}|f(t)-r(t)|$ is plotted (filled circles) against the value of $\lambda\beta$.
}
\end{figure}

\subsection{\label{SecQEFT} Further details on the example in section \ref{SecMainQEFT}}

We let $\{|n\rangle\}_{n}$ be an orthonormal eigenbasis of the Hamiltonian $H_{SC}$, with $H_{SC}|n\rangle = h_n|n\rangle$. The corresponding generator is $\mathcal{L}_{H_{SC}}(\rho) := -\frac{i}{\hbar}[H_{SC},\rho]$. The set of eigenstates of $H_{SC}$ are partitioned into the three subsets $D_0$, $D_1$, and $D_2$, where $D_0$ represents the ground state basin, $D_1$ the desired meta-stable configurations, and $D_2$ all the other states. We  construct the corresponding projectors $P_j :=\sum_{n\in D_j}|n\rangle\langle n|$. For each such set of eigenstates we also define a corresponding partial partition function $Z_j := \sum_{n\in D_j}e^{-\beta h_n} = \mathcal{Z}_{\beta H_{SC}}(P_j)$.

For each of $j = 0,1,2$ we assume a generator of the form
\begin{equation}
\label{localdissipator}
\begin{split}
\mathcal{L}_{j}(\rho) := &   \sum_{k,k'\in D_j}r_j(k'|k) |k'\rangle \langle k|\rho |k\rangle\langle k'|  \\
& -\frac{1}{2}\sum_{k,k'\in D_j}r_j(k'|k)|k\rangle\langle k|\rho \\
&  -\frac{1}{2}\sum_{k,k'\in D_j}r_j(k'|k)\rho|k\rangle\langle k|,
\end{split}
\end{equation}
where we assume that each set $\{r_j(k'|k)\}_{k,k'\in D_j}$ satisfies the detailed balance (\ref{ClassicalDetailedBalance}) described in Appendix \ref{SecModelForThrmls}.
By construction $\mathcal{L}_j$ only causes transitions within $D_j$,  leading to a local thermalization within this basin. We moreover assume a global thermalization $\mathcal{L}_{\mathrm{global}}$ on the same form as (\ref{localdissipator}), but where the sums span over all eigenstates. We model the slow equilibration between the basins by choosing the transition rates $r(k'|k)$ of the global dissipator much smaller than those for the local dissipators.

We assume a two-level energy reservoir with Hamiltonian $H_E := E_0|0\rangle\langle 0| + E_{1}|1\rangle\langle 1|$, and corresponding generator $\mathcal{L}_{H_E}(\rho) := -\frac{i}{\hbar}[H_E,\rho]$. 
The interaction between the energy reservoir and the system is generated via a resonant coupling that causes a transition from the ground state $|0\rangle$ to a selected state $|n^{*}\rangle$ in the desired basin, $n^{*}\in D_1$, such that $E_1-E_0 = h_{n^*}-h_0$. More precisely, the interaction Hamiltonian is of the form
\begin{equation}
\label{InteractHamil}
H_{\mathrm{int}} := \lambda|n^{*}\rangle\langle 0|\otimes|0\rangle\langle 1|+ \lambda |0\rangle\langle n^{*}|\otimes|1\rangle\langle 0|.
\end{equation}
This interaction Hamiltonian is, so to speak, the handle by which we push the system towards the desired basin.
The total generator is
\begin{equation}
\label{GeneratorQEFT}
\mathcal{L} := \mathcal{L}_{H_{SC}}  +\mathcal{L}_{H_E}+ \mathcal{L}_{H_\textrm{int}} + \mathcal{L}_{0} + \mathcal{L}_{1} +\mathcal{L}_{2} + \mathcal{L}_{\mathrm{global}}.
\end{equation}
We choose the time-reversals on both the system and energy reservoir as the transpose with respect to the corresponding energy eigenbasis. One can confirm that the generators $\mathcal{L}_j$ and $\mathcal{L}_{\mathrm{global}}$ all satisfy (\ref{GeneratorRelation}) with respect to $H_{SC}$, by virtue of being special cases of the generator in Appendix \ref{SecModelForThrmls}. By Lemma \ref{LemmaAdditionGenerators} it follows that their sum also satisfies (\ref{GeneratorRelation}) with respect to $H_{SC}$. Lemma \ref{LemmaClosedEvolution} yields that  $\mathcal{L}_{H_E}$ satisfies (\ref{GeneratorRelation}) with respect to $H_{E}$. Since $[H_{SC}\otimes\hat{1}_E+ \hat{1}_{SC}\otimes H_E,H_{\mathrm{int}}] =0$ it follows by Proposition \ref{GlueGenerators},  that $\mathcal{L}$  satisfies (\ref{GeneratorRelation}) with respect to $H_{SC}\otimes\hat{1}_E+ \hat{1}_{SC}\otimes H_E$.
Hence, each $\mathcal{F}_t = e^{t\mathcal{L}}$ satisfies the global fluctuation relation (\ref{nbmvbmn}).

We use the fluctuation relation (\ref{nbmvbmn}) in order to determine the partition function quotient $Z_1/Z_0$ between the desired meta-stable basin $D_1$ and the ground state basin $D_0$. For this purpose we assign $Q^{i+}_{SC}  := P_1$ and  $Q^{f+}_{SC} := P_2$ in   (\ref{nbmvbmn}). By $\mathcal{Z}_{\beta H_{SCE}}(Q^{i\pm}) = Z_0\mathcal{Z}_{\beta H_E}(Q^{i\pm})$ and  $\mathcal{Z}_{\beta H_{SCE}}(Q^{f\pm}) = Z_1\mathcal{Z}_{\beta H_E}(Q^{f\pm})$ we obtain (\ref{ndfklbklnbfdmain}) and (\ref{sdvlkvsdmain}) in the main text.

In the following we describe the specific choices for the numerical evaluation presented in Fig.~\ref{FigQEFT2}.
In order to mimic a somewhat `messy' system with no particular structure, up to the general picture painted in the main text, we select energy levels and transition rates randomly. We let $D_0$ contain $5$ states, $D_1$ consist of $5$ states, and $D_2$ consist of $20$ states. The eigenvalues in $D_0$ are constructed by drawing all $h_n\beta$ independently from the uniform distribution in the interval $[0,2]$, after which the whole spectrum is shifted such that the lowest eigenvalue is at zero, yielding the global ground state. The elements of $D_1$ are such that all $h_n\beta$ are drawn independently and uniformly in the interval $[3,5]$, and those of $D_2$ from the interval $[0,4]$. We choose $n^*$ as the lowest energy level in $D_1$, i.e., the `local' ground state. These constructions implement the idea that the desired basin $D_1$ is higher up in energy compared to the ground state basin $D_0$, and that there are several alternative states in $D_2$ that are energetically favorable. For the transition rates $r_j(k'|k)$ and $r(k'|k)$ we first construct a symmetric matrix $A$ with real non-negative elements, and let $r(k'|k) := A_{k',k} e^{\beta (h_{k}-h_{k'})/2}/(\beta\hbar)$. (The division by $\beta\hbar$ is there in order to make $A$ unit-free.) One can confirm that this guarantees that $r(k'|k)$ satisfies the detailed balance (\ref{ClassicalDetailedBalance}) with respect to $e^{-\beta h_k}$. 
For our implementation we let each independent element of $A$ be assigned as the absolute value of a random number drawn from the Gaussian distribution with zero mean and unit variance. The same procedure is repeated for the local transition rates, but where we additionally take into account that all transitions to the other basins should be zero. In order to model the slower global transitions we furthermore multiply the global $r(k'|k)$ with $0.005$, thus making these global transitions about $200$ times slower than the local transitions within each basin. Finally, we choose the interaction strength $\lambda$ such that $\lambda\beta = 1$.
 For the particular instance in Fig.~\ref{FigQEFT2} we have  $Z_0 \approx 2.27$, $Z_1 \approx 0.0728$, and thus $Z_1/Z_0 \approx 0.0320$.

\subsection{\label{DetailsJC} Further details on the JC-model in section \ref{SecMainJCwithDissipation}}

\subsubsection{$\mathcal{L}$ satisfies (\ref{GeneratorRelationMain})}
Here we show that  the generator $\mathcal{L}$ defined in (\ref{ndfbkjnjkdfb}) satisfies 
(\ref{GeneratorRelationMain}) (which is the same as (\ref{GeneratorRelation})). 
We know from Appendix \ref{ExamplesSatisfyingCond} that the various components in (\ref{ndfbkjnjkdfb}) satisfy (\ref{GeneratorRelation}) with respect to $H_{SC}$. Hence, by Lemma \ref{LemmaAdditionGenerators} we can conclude that $\mathcal{L}_{SC}$ satisfies (\ref{GeneratorRelation}) with respect to $H_{SC}$. The generator $\mathcal{L}_E$ trivially satisfies (\ref{GeneratorRelation}) with respect to $H_{E}$. Since $[H_{SCE},H_{\mathrm{int}}]=0$, 
Proposition \ref{GlueGenerators} yields $\mathcal{L}\mathcal{J}_{\beta H_{SCE}} = \mathcal{J}_{\beta H_{SCE}}\mathcal{L}^{\ominus}$. Since the conditions of Corollary \ref{CorGlue} are satisfied, it follows that all  channels $\mathcal{F}_{t} := e^{t\mathcal{L}}$ for $t\geq 0$ satisfy Eq.~(\ref{flbmmbdf}) in Assumptions \ref{IdealChannelDef}, and thus Proposition \ref{PropIdealGlbChnl} is applicable. We can thus conclude that the conditional fluctuation relation (\ref{fbfbbfsMain}) is satisfied on the energy reservoir for the conditional maps $\tilde{\mathcal{F}}_{\pm}$, based on the underlying evolution $\mathcal{F}_{t}$.

\subsubsection{Decomposition of the dynamics into modes of coherence}
As mentioned in the main text, the generator $\mathcal{L}$ yields  dynamics that decomposes with respect to the modes of coherence.
To see this, one can first confirm that 
$[H_{SC},\mathcal{L}_{SC}(\rho)] = \mathcal{L}_{SC}([H_{SC},\rho])$ and $[H_{E},\mathcal{L}_{E}(\rho)] = \mathcal{L}_{E}([H_E,\rho])$. Since $[H_1\otimes \hat{1}_2 + \hat{1}_1\otimes H_2,H_{\mathrm{int}}] = 0$, and $[H_{SC}, Q_{SC}^{i\pm}] = 0$ and $[H_{SC}, Q_{SC}^{f\pm}] = 0$,  it follows by  the argument outlined in Appendix \ref{SecDecouplingDiagAgain} that the CPMs $\tilde{\mathcal{F}}_{\pm}$ decouple along the modes of coherence.

\subsubsection{The case of no decoupling}

\begin{figure}
 \includegraphics[width= 8cm]{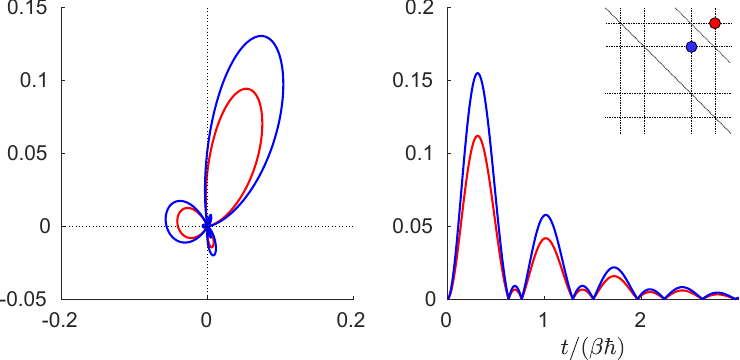} 
   \caption{\label{FigJCcoffNoDiag} 
 {\bf Off-diagonal fluctuation relation in the case of no decoupling.}
The fluctuation theorem (\ref{davnndvaf}) relates the evolution of the coherences carried by two off-diagonal elements along two different displaced diagonals. Hence, as opposed to the case in Fig.~\ref{FigJCcoff}, the two off-diagonal elements belong to two different modes of coherence. For the setting in Fig.~\ref{FigJCcoff} the quantities $d_{\pm}$ defined in (\ref{davnndvaf}) would be identically zero, while for the new measurement operators (\ref{newmeasurmntop}) they become non-trivial, and satisfy the fluctuation relation in  (\ref{davnndvaf}). The plot on the right displays $|d_{+}|$ (red curve) and $|d_{-}|$ (blue curve), while the left depicts the trajectories  $d_{+}$ (red curve) and $d_{-}$ (blue curve) in the complex plane. Like in  Fig.~\ref{FigJCcoff} the proportionality of the absolute values, and identical phase factors, predicted by (\ref{davnndvaf}) are visible. 
}
\end{figure}

As mentioned in the main text, the off-diagonal fluctuation relations are valid even if there is no decoupling. The decoupling merely makes certain fluctuation relations trivial. To illustrate this, let us  consider the pair of off-diagonal elements $|19 \rangle\langle 23|$ and $|20\rangle\langle 22|$, which one can realize belong to two different modes of coherence. The relation (\ref{fbfbbfsMain}) in this case yields
\begin{equation}
\label{davnndvaf}
\begin{split}
& \mathcal{Z}_{\beta H_{SC}}(Q^{i}_{SC}) d_{+}  = \mathcal{Z}_{\beta H_{SC}}(Q^{f}_{SC}) d_{-},\\
& d_{+} := \langle 20| \tilde{\mathcal{F}}_{+} \big(|19 \rangle\langle 23|\big) |22\rangle,\,\, d_{-}:=\langle 19|\tilde{\mathcal{F}}_{-} (|20\rangle\langle 22|)|23 \rangle.
\end{split}
\end{equation}
For diagonal measurement operators, both $d_{+}$ and $d_{-}$ vanish, and thus trivially satisfy (\ref{davnndvaf}). However, let us consider a new couple of measurement operators that are not diagonal in the energy eigenbasis, namely 
\begin{equation}
\label{newmeasurmntop}
\begin{split}
Q^{i+}_{2} := & |\psi_i\rangle\langle\psi_i|,\quad |\psi_i\rangle := (|0\rangle +|1\rangle)/\sqrt{2},\\
Q^{f+}_{2} := & |\psi_f\rangle\langle\psi_f|,\quad  |\psi_f\rangle := (|0\rangle +2|1\rangle)/\sqrt{5}.
\end{split}
\end{equation}
Figure \ref{FigJCcoffNoDiag} illustrates the fluctuation relation (\ref{davnndvaf}) for this choice of  $Q^{i+}_{2}$ and $Q^{f+}_{2}$.

\subsubsection{Is $\mathcal{L}$ a generator of maps in the closure of thermal operations?}

One can argue that the total generator $\mathcal{L}$ should reasonably yield maps in  the closure of the set of thermal operations with respect to $H_{SC}\otimes\hat{1}_E + \hat{1}_E\otimes H_E$. First of all, one can note that $\mathcal{L}_{SC}$ is a special case of the class of generators considered in Appendix \ref{SecGenForThrmOp} (with $r_0$, $r_1$ such that $r_0 +r_1 = 4\kappa$), and thus is a generator of thermal operations with respect to $H_{SC}$. Moreover, $\mathcal{L}_E$ is only the Hamiltonian evolution of $E$ and is thus trivially a thermal operation with respect to $H_E$. By construction, $H_{\mathrm{int}}$ commutes with  $H_{SC}\otimes\hat{1}_E + \hat{1}_E\otimes H_E$, and thus one may be tempted to apply Proposition \ref{AddGeneratorsForThermalOp}. However, this proposition is strictly speaking not applicable since the Hamiltonians $H_E$ and $H_{\mathrm{int}}$, and thus the generators $\mathcal{L}_{E}$ and $\mathcal{L}_{\mathrm{int}}$, are unbounded. However, it appears reasonable that if the initial state has bounded energy, then $\mathcal{L}$ could be truncated to a finite number of energy eigenstates, to an arbitrarily good accuracy. For that truncated  system, Proposition \ref{AddGeneratorsForThermalOp} would reasonably be applicable. This argument suggests that $\mathcal{L}$ is some sense could be a generator of channels in the closure of thermal operations. However, strictly speaking this remains to be proved.
Alternatively, one could consider some generalization of Trotter's decomposition for unbounded operators (see e.g. \cite{Trotter59,Tosio78,Ichinose04}). However, we leave such generalizations as an open question.

\section{\label{SecAddtnlRmrks}Some additional remarks}

\subsection{\label{AppDetailedBalance}Detailed balance}

One way of obtaining fluctuation relations in the classical case is via stochastic dynamics that satisfies detailed balance \cite{ReviewJarzynski,CrooksTheorem,Crooks98}. Apart from Section \ref{SecMainMarkovian} and Appendix \ref{SecModelForThrmls} we have not referred much to detailed balance in our discussions, or used it in the derivations. The reason is that energy conservation and time reversal symmetry in some sense supersedes detailed balance, which we here demonstrate briefly.

Let $H_1$ and $H_2$ be non-degenerate Hermitian operators on a finite-dimensional Hilbert space. We let the global Hamiltonian be non-interacting $H = H_1\otimes\hat{1}_2 + \hat{1}_1\otimes H_2$, and assume an energy conserving unitary evolution $[H,V] = 0$, and a product time-reversal $\mathcal{T} = \mathcal{T}_1\otimes\mathcal{T}_2$ with $\mathcal{T}_1(H_1) = H_1$ and $\mathcal{T}_2(H_2) = H_2$, and $\mathcal{T}(V) = V$. 
Assuming that system $2$ is in equilibrium we can define the transition probability of changing the state of system $1$ from eigenstate $n$ to $n'$ (for a non-degenerate $H_1$) as
\begin{equation}
p(n'|n) :=  \Tr([|n'\rangle\langle n'|\otimes \hat{1}_2]V[|n\rangle\langle n|\otimes G(H_2)]V^{\dagger}).
\end{equation}
By using the energy conservation (and the observation $G(H)[|n\rangle\langle n|\otimes \hat{1}_2] = G_n(H_1)|n\rangle\langle n|\otimes G(H_2)$) it follows that  
\begin{equation*}
\begin{split}
p(n'|n)G_{n}(H_1) =  & G_{n'}(H_1)\Tr([|n\rangle\langle n|\otimes\hat{1}_2]\\
& V^{\dagger}[|n'\rangle\langle n'|\otimes G(H_2)]V).
\end{split}
\end{equation*}
This is almost what we want, apart from the fact that the evolution is reversed. Here we can make use of the time reversal symmetry (and Lemma \ref{TOnHermitean}) to obtain $p(n'|n)G_{n}(H_1) = p(n|n')G_{n'}(H_1)$, i.e., the transition probability $p(n'|n)$ satisfies detailed balance.

\subsection{\label{HeatBathGibbsState}Heat baths in the Gibbs state}

We do in this investigation often assume that the heat bath initially is in the Gibbs state corresponding to a given temperature. Although not an unusual assumption, e.g., in  derivations of master equations \cite{Caldeira83} and \cite{BreuerPetruccione} (see in particular sections 3.6.2.1 and 4.2.2), it is nevertheless worth considering the justification, especially since one may argue that it is not the heat bath \emph{per se} that is Gibbs distributed, but rather systems that are weakly coupled to it.  One possible argument would be that the environment can be separated into a `near environment' that is relevant on the time scale of the experiment, and a `far environment' (or  `super bath') that puts the near environment into  the Gibbs state (see e.g. the discussions in \cite{Hanggi15}).  Another approach would be to assume that an ideal heat bath in some sense behaves  \emph{as if} it is Gibbs distributed. These issues approach the question of thermalization in closed systems \cite{Polkovnikov11,Gogolin15,DAlessio15}, and along these lines one may speculate that typicality \cite{Popescu06,Goldstein06} could be employed to obtain a more refined analysis of fluctuation relations.

\end{appendix}

\end{document}